\documentclass[11pt,reqno]{amsart}
\usepackage{amsthm,amssymb,amsbsy,epsfig,graphicx,color}
\usepackage{hyperref}

\def\lb{\label}
\def\ba{\begin{eqnarray}}
\def\ea{\end{eqnarray}}

\newcommand{\nc}{\newcommand}
\nc{\be}{\begin{equation}} \nc{\ee}{\end{equation}}
\nc{\bea}{\begin{eqnarray}} \nc{\eea}{\end{eqnarray}}
\nc{\disp}{\displaystyle} \nc{\ade}{\mbox{$A$-$D$-$E$}}
\nc{\calN}{{\cal N}} \nc{\calC}{{\cal C}} \nc{\calM}{{\cal M}}
\nc{\calS}{{\cal S}} \nc{\phit}{\hat{\varphi}}
\nc{\chit}{\hat{\chi}} \nc{\hcalN}{\hat{\calN}}
\nc{\hcalS}{\hat{\calS}} \nc{\hS}{\hat{S}}
\nc{\sigmad}{\sigma^\dagger} \nc{\psid}{\psi^\dagger}

\def\R{\tilde{R}}

\definecolor{IndianRed}{rgb}{0.8,0.36,0.36}

\newtheorem{lemma}{Lemma}
\newtheorem{theorem}{Theorem}
\newtheorem{defn}{Definition}
\newtheorem{prop}{Proposition}
\newtheorem{cor}{Corollary}
\newtheorem{obs}{Observation}
\newtheorem{conj}{Conjecture}

\theoremstyle{definition}
\newtheorem{example}{Example}
\newtheorem{remark}{Remark}

\font\tenmsb=msbm10\font\sevenmsb=msbm7 \font\fivemsb=msbm5
\newfam\msbfam
\textfont\msbfam=\tenmsb \scriptfont\msbfam=\sevenmsb
\scriptscriptfont\msbfam=\fivemsb

\def\ket#1{|#1\rangle}
\def\e{{\rm e}}

\def\i{{\rm i}}
\def\H{\mathcal{H}}
\def\T{\mathcal{T}}

\def\D{\mathcal{D}}
\def\B{\mathcal{B}}

\numberwithin{equation}{section}

\begin{document}
\begin{titlepage}
\title{Factorised solutions of Temperley-Lieb $q$KZ equations on a segment}

\author{Jan de Gier}
\address{Department of Mathematics and Statistics, The University
of Melbourne, VIC 3010, Australia}
\email{degier@ms.unimelb.edu.au}

\author{Pavel Pyatov}
\address{Bogoliubov Laboratory of Theoretical Physics, Joint Institute
for Nuclear Research, 141980 Dubna, Moscow Region, Russia}
\email{pyatov@theor.jinr.ru}

\date{\today}

\begin{abstract}
We study the $q$-deformed Knizhnik-Zamolodchikov equation in path representations of the Temperley-Lieb algebras. We consider two types of open boundary conditions, and in both cases we derive factorised expressions for the solutions of the qKZ equation in terms of Baxterised Demazurre-Lusztig operators. These expressions are alternative to known integral solutions for tensor product representations. The factorised expressions reveal the algebraic structure within the qKZ equation, and effectively reduce it to a set of truncation conditions on a single scalar function. The factorised expressions allow for an efficient computation of the full solution once this single scalar function is known.
We further study particular polynomial solutions for which certain additional factorised expressions give weighted sums over components of the solution. In the homogeneous limit, we formulate positivity conjectures in the spirit of Di Francesco and Zinn-Justin. We further conjecture relations between weighted sums and individual components of the solutions for larger system sizes. 
\end{abstract}

\maketitle

\end{titlepage}

{\footnotesize\tableofcontents}
\vfill\newpage

\section{Introduction}
The $q$-deformed Knizhnik-Zamolodchikov equations (qKZ) are widely recognized as
important tools in the computation of form factors in integrable quantum field theories
\cite{S2} and correlation functions in conformal field theory and
solvable lattice models \cite{JM}. They can be derived using the
representation theory of affine quantum groups \cite{FR} or,
equivalently and using a dual setup, from the affine and double affine
Hecke algebra \cite{Ch}. The qKZ
equations have been extensively studied in tensor product modules
of affine quantum groups or Hecke algebras, and much is known about their
solutions in the case of cyclic boundary conditions \cite{V,TV}.
We refer to \cite{EFK} and references
therein for extensive literature on the qKZ equations.

Recently interest has arisen in the qKZ equation in the context of
the Razumov-Stroganov conjectures. These relate the integrable
spin-1/2 quantum XXZ spin chain \cite{S00,RazuS00} in condensed
matter physics and the O(1) loop model \cite{BatchGN01,RazuS01} in
statistical mechanics, to alternating sign matrices and plane
partitions \cite{Bress}. Further developments surrounding the
Razumov-Stroganov conjectures include progress on loop models
\cite{Jan,MNGB,DFZJ04,DFZJZ03,DFZJZ06,ZJ07,DFZJ07} and quantum XXZ
spin chains \cite{Magic,RS06,RSPZJ07},  the stochastic raise and
peel models \cite{GNPR, GNPR03, PRGN02, Pavel, APR}, lattice
supersymmetry \cite{BA05,FNS,VW,YangF}, higher spin and higher
rank cases \cite{ZJ06,DFZJ05,ShigechiU}, as well as connections to
the Brauer algebra and (multi)degrees of certain algebraic
varieties \cite{GN,DFZJ04b,KZJ}.

The connection to the qKZ equation was realised by Pasquier
\cite{P05} and Di Francesco and Zinn-Justin \cite{DFZJ05}, by
generalising the Razumov-Stroganov conjectures to include an extra
parameter $q$ or $\tau=-(q+q^{-1})$. In particular, the polynomial
solutions for level one qKZ equations display intriguing
positivity properties and are conjectured to be related to
weighted plane partitions and alternating sign matrices
\cite{P05,DFZJ05,KP,DF06,DFZJ07}.

In the
Razumov-Stroganov context one considers the qKZ equation in a path
representation for $SL(k)$ quotients of the Hecke algebra, using cyclic as well
as open (non-affine) boundary conditions. In the case $k=2$ this
quotient corresponds to the Temperley-Lieb algebra, for which
there is a well known and simple equivalence between the path
representation and its graphical loop, or link pattern,
representation. In this paper we study the qKZ equation for $k=2$
in the path representation and for the two types of open boundary
conditions also considered in \cite{DF06,ZJ07,DFZJ07}.

The solution of the qKZ equation is a function in $N$ variables
$x_i$ $i=1,\ldots, N$ taking values in the path representation.
The components of this vector valued function can be expressed in
a single scalar function which we call the base function. We
derive factorised expressions for the components of the solution
of the qKZ equations for the Temperley-Lieb algebra (referred to
as type A) and the one-boundary Temperley-Lieb algebra (type B)
with arbitrary parameters. The factorised expressions are given in
terms of Baxterised Demazurre-Lusztig operators, acting on the
base function which we assume to be known. The formula for type A
was already derived for Kazhdan-Lusztig elements of Grasmannians
by Kirillov and Lascoux \cite{KiriL}. We further reduce the qKZ
equation to a set of truncation relations that determine the base
function. These relations also appear in a factorised form. We
conjecture that the truncation relations can be recast entirely in
terms of Baxterised elements of the symmetric group.

Restricting to polynomial solutions, the factorised expressions
provide an efficient way for computing explicit solutions. We note
here that polynomial solutions may also be obtained from Macdonald
polynomials \cite{Ch,P05,KP,KT}. Using the factorised expressions,
we compute explicit polynomial solutions of the level one qKZ
equations, recovering and extending the results of Di Francesco
\cite{DF07} in the case of type A, and of Zinn-Justin \cite{ZJ07}
in the case of type B. We would like to emphasize the importance
of such explicit solutions as a basis for experimentation and
discovery of novel results. Based on the explicit solutions, and
in analogy with Di Francesco \cite{DF06,DF07} and Kasatani and
Pasquier \cite{KP}, we formulate new positivity conjectures in the
case of type B for two-variable polynomials in the homogeneous
limit ($x_i \rightarrow 0$) of the solutions of the qKZ equations.
In the inhomogeneous case, the factorised expressions furthermore
suggest to define linear combinations (weighted partial sums) of
the components of the solution in a very natural way. Special
cases of these partial sums are also considered in the homogeneous
limit by Razumov, Stroganov and Zinn-Justin \cite{RS06,RSPZJ07}
and by Di Francesco and Zinn-Justin \cite{DFZJ07}. We conjecture
identities between the partial sums and individual components of
solutions for larger system sizes.

The first three sections of this paper are a review of known
results. We define the Hecke and Temperley-Lieb algebras of type A
and B, the path representations and explain the qKZ equation in
these representations. In Section~4 we state our main theorems
concerning factorised solutions and truncation conditions for the
qKZ equation of type A and B. These results are proved in
Section~6 and Appendix~A. The fifth section contains a list of
conjectures regarding the explicit polynomial solutions of the qKZ
equation. These conjectures relate to the positivity of solutions
in the homogeneous limit, and to natural partial sums over
components of the solution. Our observations are based on explicit
solutions for type A and B which are listed in Appendices B and C.

Throughout the following we will use the notation $[x]_q$ for the usual q-number
\[
[x]_q = \frac{q^x-q^{-x}}{q-q^{-1}}.
\]
The notation $[x]$ will always refer to base $q$.

\subsection*{Acknowledgments}
Our warm thanks go to Arun Ram, Ole Warnaar and Keiichi Shigechi
for useful discussions. JdG thanks the Australian Research Council for
financial assistance. PP was supported by RFBR grant No.
05-01-01086-a and by the DFG-RFBR grant (436 RUS 113/909/0- 1(R)
and 07-02-91561-a)

\section{Relevant algebras and their representations}

\subsection{Iwahori-Hecke algebras}
\subsubsection{Type A}
\begin{defn}
\label{def:HeckeA} The Iwahori-Hecke algebra of type $A_N$,
denoted by $\H_N^{\rm A}(q)$, is the unital algebra defined in
terms of generators $g_i$, $i=1,\ldots,N-1$, and relations
\begin{align}
&(g_i-q)(g_i+q^{-1}) = 0,\qquad g_ig_j = g_jg_i \quad\forall\, i,j:\, |i-j| >1,\nonumber\\
&g_ig_{i+1}g_i = g_{i+1}g_ig_{i+1}. \label{Hdef_typeA}
\end{align}
\end{defn}

Hereafter we always assume
\be
\lb{cond-q}
q\in {\mathbb C}\setminus \{0\},\quad \mbox{ and}\quad
[k]\neq 0\quad \forall\, k=2,3,\dots ,N,
\ee
in which case the algebra $\H_N^{\rm A}(q)$ is semisimple.
It is isomorphic to
the group algebra of the symmetric group ${\mathbb C}[S_N]\simeq \H_N^{\rm A}(1)$.

It is sometimes convenient to use two other presentations of the
algebra $\H_N^{\rm A}(q)$ in terms of the elements $a_i$ and
$s_i$
\[ a_i := q-g_i, \qquad s_i := q^{-1}+g_i\, ,
\qquad i=1,\dots ,N-1.
\]
For each particular value of index $i$
the elements $a_i$ and $s_i$ are mutually orthogonal unnormalised
projectors
\[
a_i s_i = s_i a_i =0,\qquad a_i+s_i=[2]\, ,
\]
generating the subalgebra ${\H}^{\rm A}_2(q)\hookrightarrow
{\H}^{\rm A}_N(q)$. Traditionally they are called
the {\em antisymmetrizer} and the {\em symmetrizer} and
associated, respectively, with the two possible partitions of the
number 2: $\{1^2\}$ and $\{2\}$.

The $\H^{\rm A}_N(q)$ defining relations  (\ref{Hdef_typeA}) in
terms of generators $a_i$, $i=1,\ldots,N-1$, read
\begin{align}
\lb{Hdef_typeA_a}
& a_i^2 = [2]\, a_i,\qquad a_ia_j = a_ja_i \quad\forall\, i,j:\, |i-j| >1,\nonumber \\
& a_ia_{i+1}a_i -a_i = a_{i+1}a_ia_{i+1}-a_{i+1},
\end{align}
and  in terms of $s_i$, $i=1,\ldots,N-1$, they read
\begin{align*}
& s_i^2 = [2]\, s_i,\qquad s_is_j = s_js_i \quad\forall\, i,j:\, |i-j| >1,\nonumber \\
& s_is_{i+1}s_i -s_i = s_{i+1}s_is_{i+1}-s_{i+1}.
\end{align*}

\subsubsection{Type B}

\begin{defn}
\label{def:HeckeB} The Iwahori-Hecke algebra of type $B_N$,
denoted by $\H^{\rm B}_N(q,\omega)$, is the unital algebra
defined in terms of generators $g_i$, $i=0,\ldots,N-1$,
satisfying, besides \eqref{Hdef_typeA}, relations
\begin{align}
& (g_0 + q^\omega)(g_0 + q^{-\omega}) =0, \qquad
g_0 g_i = g_i g_0 \quad\forall\, i>1,\nonumber  \\
& g_0 g_1 g_0 g_1 = g_1 g_0 g_1 g_0 . \lb{Hdef_typeB}
\end{align}
\end{defn}
If in addition to \eqref{cond-q} we assume
$
[\omega \pm k]\neq 0\;\; \forall\,
k\in 0,1,\dots ,N-1,
$
then the algebra $\H^{\rm B}_N(q,w)$ becomes semisimple. Hereafter we do not need the semisimplicity and we
only assume that $[\omega +1]\neq 0$.
%

It is sometimes convenient to use the presentations of $\H^{\rm
B}_N(q,\omega)$ in terms of either the antisymmetrizers $a_i$,
or the symmetrizers $s_i$, supplemented, respectively, by the {\em
boundary generators}
\be \lb{def_e0_s0} a_0 :=
\frac{-q^{-\omega}-g_0}{q^{\omega +1}-q^{-\omega -1}},
\qquad\mbox{or}\qquad s_0 :=\frac{q^{\omega}+g_0}{q^{\omega
+1}-q^{-\omega -1}}. \ee The generators $a_0$ and $s_0$ are
mutually orthogonal unnormalised projectors,
\[
a_0 s_0 = s_0 a_0 =0,\qquad
a_0+s_0=\frac{[\omega]}{[\omega+1]}.
\]
The defining relations
(\ref{Hdef_typeB}) written in terms $a_0$ and $a_i$ read
\begin{align}
&a_0^2= \frac{[\omega]}{[\omega +1]}a_0, \qquad
a_0 a_i = a_i a_0 \quad\forall\, i>1,\nonumber  \\
& a_0 a_1 a_0 a_1 - a_0 a_1 = a_1 a_0 a_1 a_0 -a_1 a_0,
\lb{Hdef_typeB_a}
\end{align}
and written in terms of $s_0$ and $s_i$ they read
\begin{align*}
&s_0^2=  \frac{[\omega]}{[\omega +1]}s_0, \qquad
s_0 s_i = s_i s_0 \quad\forall\, i>1,\nonumber  \\
& s_0 s_1 s_0 s_1 - s_0 s_1 = s_1 s_0 s_1 s_0 -s_1 s_0.
\end{align*}

\subsection{Baxterised elements and their graphical presentation}

It is well known that the defining relations of the Iwahori-Hecke
algebras \eqref{Hdef_typeA}, \eqref{Hdef_typeB} can be generalised
to include a so-called spectral parameter  (see \cite{Jo,IO} and
references therein). This generalisation is sometimes called
Baxterisation and is relevant both in the representation theory of
these algebras as well as in their applications to the theory of
integrable systems.

\subsubsection{Type A}

For the algebra $\H_N^{\rm A}(q)$, the Baxterised elements
$g_i(u)$, $i=1,\dots ,N-1,$ are defined as
\[
g_i(u) :=\, q^{-2u}\, \frac{g_i - q^{2u-1}}{g_i-q^{-2u-1}},
\]
which we can write alternatively as
\be g_i(u)=\,
\frac{q^{u}-[u]\, g_i}{[u+1]}\, =\, \frac{[1-u] + [u]\,
a_i}{[1+u]}\, =\, 1- \frac{[u]}{[u+1]} s_i. \lb{g_u}
\ee
Here $u\in {\mathbb C}\setminus \{-1\}$ is the spectral parameter. It
can be shown that the following relations hold
\begin{align}
&g_i(u)g_i(-u) = 1,\qquad \forall\, u\in {\mathbb C}\setminus \{-1,1\}, \lb{unitarity}\\
&g_i(u)g_{i+1}(u+v)g_i(v) = g_{i+1}(v) g_i(u+v) g_{i+1}(u),
\label{ybe}\\
&g_i(u) g_j(v) = g_j(v)g_i(u) \quad\forall\, i,j:\, |i-j| >1,
\label{commu}
\end{align}
The relations \eqref{g_u}--\eqref{commu} are equivalent to the
defining set of conditions \eqref{Hdef_typeA}. The relations
\eqref{unitarity} and \eqref{ybe} are called, respectively, the
{\em unitarity condition} and the {\em Yang-Baxter equation}.

Notice that the unitarity condition is not valid at the degenerate
points $u= \pm 1$. It is therefore
 useful in certain cases to use a different normalisation for the Baxterised elements,
\be h_i(u) :=\, \frac{[1-u]}{[u]}g_i(-u)\, =\, \frac{[u+1]}{[u]} -
a_i \, =\, s_i -\frac{[u-1]}{[u]}   \lb{h_u}. \ee
Note that now the elements $h_i(u)$ are ill-defined at $u=0$. In this
normalisation we have
\[ h_i(u)h_i(-u) = 1  - \frac{1}{[u]^2},\qquad h_i(1)=s_i, \qquad h_i(-1)= -a_i,
\]
and $h_i(u)$ still satisfies the Yang-Baxter and the
commutativity equations \eqref{ybe} and \eqref{commu}. We also
note that $h_i$ satisfies the simple but very useful identity
\be h_i(u) = h_i(v) + \frac{[v-u]}{[u][v]}. \label{huv} \ee

\subsubsection{Type B} In this case we additionally define
\[
g_0(u) :=\, q^{-2u} \frac{[\tfrac{\omega+\nu}{2}+u](g_0 +
q^{2u-\nu  })}{[\tfrac{\omega+\nu}{2}-u](g_0 + q^{-2u-\nu  })},
\]
which alternatively can be written as
\be g_0(u) =\,\frac{k(u,\nu)  - [2u] [\omega+1] a_0}{k(-u,\nu)}
=\, 1 + \frac{[2u][\omega +1]}{k(-u,\nu)} s_0, \lb{g0_u}
\ee
where $k(u,\nu):= [\tfrac{\omega+\nu}{2}+u][\tfrac{\omega-\nu}{2}+u]$,
and  $\nu$ is an additional arbitrary parameter.

The boundary Baxterised element $g_0(u)$ satisfies relations
\begin{align}
g_0(u) g_0(-u) &= 1,\lb{unitarity2}\\
g_0(v) g_1(u+v) g_0(u) g_1(u-v)  &=  g_1(u-v) g_0(u) g_1(u+v) g_0(v),\lb{rea} \\
g_0(u) g_i(v) &= g_i(v) g_0(u) \quad\forall\, i>1,\nonumber
\end{align}
which are equivalent to the defining relations \eqref{Hdef_typeB}.
Relations \eqref{unitarity2} and \eqref{rea} are called,
respectively, the {\em unitarity condition} and the {\em
reflection equation} \cite{Sk}. An alternative normalisation for the
boundary Baxterised element is
\be h_0(u) :=\, -
\frac{k(u/2,\nu)}{[u][\omega +1]} g_0(-u/2) \, =\,
-\frac{k(-u/2,\nu)}{[u][\omega +1]}-a_0\, =\,
s_0 -\frac{k(u/2,\nu)}{[u][\omega +1]}. \lb{h0_u}
\ee
In this
normalisation we find
\begin{align}
h_0(\omega\pm\nu)\, =\,  -a_0,& \qquad h_0(-\omega\pm\nu)\,=\,s_0, \nonumber
\\
h_0(v-u) h_1(v) h_0(u+v) h_1(u)\,  &=\,   h_1(u) h_0(u+v) h_1(v)
h_0(v-u) . \lb{rea_h}
\end{align}

\subsubsection{Graphical presentation}
For our purposes it is convenient to represent the Baxterised
elements graphically as tiles and the boundary Baxterised element
as a half-tile
\[
\begin{picture}(0,30)(0,21)
\qbezier[25](43,0)(43,25)(43,50)
\qbezier[5](58,40)(58,45)(58,50)
\qbezier[25](73,0)(73,25)(73,50)
\qbezier[6](159,0)(159,6)(159,12)
\qbezier[5](159,40)(159,45)(159,50)
\qbezier[25](174,0)(174,25)(174,50)
\put(34,-10){$\scriptstyle i-\mathit 1$}
\put(70,-10){$\scriptstyle i+\mathit 1$}
\put(157,-11){$\scriptstyle\mathit 0$}
\put(172,-11){$\scriptstyle\mathit 1$}
\end{picture}
g_i(u) =\
\raisebox{-32pt}{\includegraphics[height=50pt]{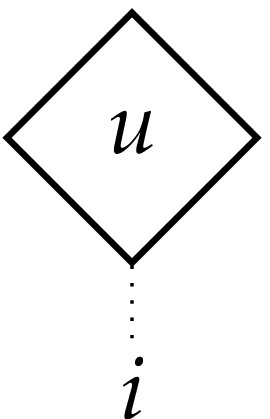}}\
,\quad\qquad g_0(u) =\
\raisebox{-14pt}{\includegraphics[height=34pt]{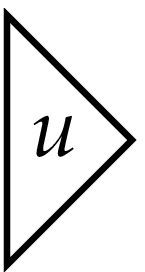}}\ .
\]
The (half-)tiles  are placed on labeled vertical lines and they
can move freely along the lines unless they meet other
(half-)tiles. Multiplication in the algebra corresponds to a
simultaneous placement of several (half-)tiles on the same picture
and a rightwards order of terms in the product corresponds to a
downwards order of (half-)tiles on the picture. In this way, the
Yang-Baxter equation \eqref{ybe} can be depicted as
\be
\raisebox{-40pt}{\includegraphics[height=80pt]{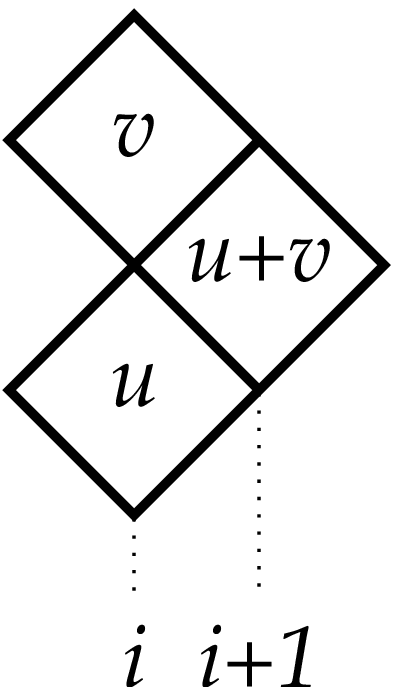}}\;\;\;
=\;\;\; \raisebox{-40pt}{\includegraphics[height=80pt]{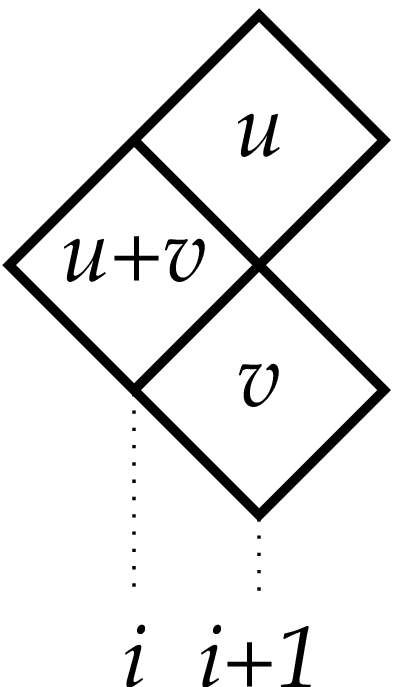}}
\ . \label{ybegraphic} \ee
and the reflection equation can be depicted as
\[
\begin{picture}(0,0)
\qbezier[10](0,-45)(0,-32.5)(0,-20)
\qbezier[2](16,-45)(16,-42.5)(16,-40)
%
\qbezier[2](70,-45)(70,-42.5)(70,-40)
\qbezier[8](86,-45)(86,-35)(86,-25)
%
\put(-3,-55){$\scriptstyle\mathit 0$}
\put(13,-55){$\scriptstyle\mathit 1$}
%
\put(67,-55){$\scriptstyle\mathit 0$}
\put(83,-55){$\scriptstyle\mathit 1$}
\end{picture}
\raisebox{-40pt}{\includegraphics[height=80pt]{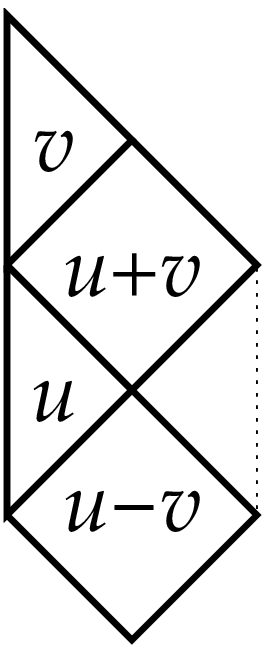}}\quad
=\quad
\raisebox{-40pt}{\includegraphics[height=80pt]{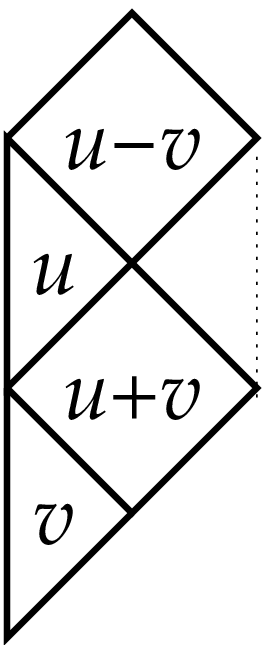}} \
.
\]

For the alternative set of Baxterised elements we will use dashed
pictures,
\be
\label{h_pic}
\phantom{a} \hspace{-55pt}
\begin{picture}(0,30)(0,21)
\qbezier[4](61,0)(61,4)(61,8)
\qbezier[6](159,0)(159,6)(159,12)
{\linethickness{1pt}
\multiput(159,10)(0,8){4}{\qbezier[10](0,0)(0,1.5)(0,3)}
\multiput(159,10)(6,6){3}{\qbezier[10](0,0)(1,1)(2,2)}
\multiput(159,39)(6,-6){3}{\qbezier[10](0,0)(1,-1)(2,-2)} }
\put(58,-11){$\scriptstyle i$}
\put(157,-11){$\scriptstyle\mathit 0$}
\put(162,22){$u$}
\end{picture}
h_i(u) =\
\raisebox{-15pt}{\includegraphics[height=36pt]{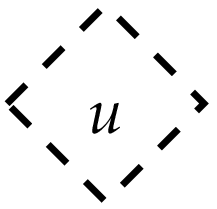}}\, ,
\qquad\quad h_0(u) =\ \vspace{20pt} \ee The picture of the
Yang-Baxter equation for the dashed tiles is the same as
(\ref{ybegraphic}), but the picture of the reflection equation
(\ref{rea_h}) has a different arrangement of the spectral parameters
\be
\lb{rea-h-graphic}
\begin{picture}(100,60)(0,00)
\qbezier[10](0,-20)(0,-10)(0,0)
\qbezier[2](16,-20)(16,-18.5)(16,-17)
%
\qbezier[2](70,-20)(70,-18.5)(70,-17)
\qbezier[10](86,-20)(86,-10)(86,0)
{\linethickness{1pt}
\multiput(0,-1)(0,8){9}{\qbezier[10](0,0)(0,1.5)(0,3)}
\multiput(0,0)(6,6){6}{\qbezier[10](0,0)(1,1)(2,2)}
\multiput(0,0)(6,-6){3}{\qbezier[10](0,0)(1,-1)(2,-2)}
\multiput(0,33.5)(6,6){3}{\qbezier[10](0,0)(1,1)(2,2)}
\multiput(0,33.5)(6,-6){6}{\qbezier[10](0,0)(1,-1)(2,-2)}
\multiput(0,66)(6,-6){6}{\qbezier[10](0,0)(1,-1)(2,-2)}
\multiput(17.8,-14.8)(6,6){3}{\qbezier[10](0,0)(1,1)(2,2)}
\multiput(70,-15)(0,8){9}{\qbezier[10](0,0)(0,1.5)(0,3)}
\multiput(70,18)(6,6){6}{\qbezier[10](0,0)(1,1)(2,2)}
\multiput(70,18)(6,-6){3}{\qbezier[10](0,0)(1,-1)(2,-2)}
\multiput(70,51.5)(6,6){3}{\qbezier[10](0,0)(1,1)(2,2)}
\multiput(70,51.5)(6,-6){6}{\qbezier[10](0,0)(1,-1)(2,-2)}
\multiput(87,66)(6,-6){3}{\qbezier[10](0,0)(1,-1)(2,-2)}
\multiput(70,-15)(6,6){6}{\qbezier[10](0,0)(1,1)(2,2)} }
\put(0,-30){$\scriptstyle\mathit 0$}
\put(15,-30){$\scriptstyle\mathit 1$}
\put(70,-30){$\scriptstyle\mathit 0$}
\put(85,-30){$\scriptstyle\mathit 1$}
%
\put(47,23){$=$} \put(14,-3){$u$} \put(14,30){$v$}
\put(84,14){$v$} \put(84,48){$u$} \put(1.5,14){$\scriptstyle
u\!+\!v$} \put(1.5,47.2){$\scriptstyle v\!-\!u$}
\put(71.5,-1.5){$\scriptstyle v\!-\!u$}
\put(71.5,32.2){$\scriptstyle u\!+\!v$}
\end{picture}
\vspace{30pt} \ee

Let us remark that expressions for the boundary half-tile $g_0(u)$
and the dashed half-tile $h_0(u)$ depend on an arbitrary parameter
$\nu$ which is not shown in the pictures. For the dashed
half-tile we shall exploit this degree of freedom in Section
\ref{sec4.2}, see \eqref{half-tile}.

\subsection{Temperley-Lieb algebras}

The Iwahori-Hecke algebras have a well known series of $SL(2)$
type, or Temperley-Lieb, quotients whose irreducible
representations are classified in the semisimple case by
partitions into one or two parts (i.e., by the Young diagrams
containing at most two rows). The Temperley-Lieb algebra can be
described in terms of equivalence classes of the
generators $a_i$ (different generators belong to different
equivalence classes). Below we use the notation $e_i$ for the
equivalence class of $-a_i$.

\begin{defn}
\label{def:TL} The Temperley-Lieb algebra of type $A_N$, denoted
by $\T^{\rm A}_N(q)$, is the unital algebra defined in terms of
generators $e_i$, $i=1,\ldots,N-1,\,$ satisfying the relations
\begin{align}
&e_i^2 = -[2] e_i, \quad e_ie_j = e_je_i \quad\forall\, i,j:\, |i-j| >1,\nonumber \\
&e_ie_{i\pm1}e_i = e_i. \label{TLdef}
\end{align}
\end{defn}

\begin{defn}
\label{def:1BTL} The Temperley-Lieb algebra of type $B_N$,  $\T^{\rm B}_N(q,\omega)$
(also called the blob algebra \cite{MS}),
is the unital algebra defined in
terms of generators $e_i$, $i=0,\ldots,N-1,\,$  satisfying,
besides \eqref{TLdef}, the relations
\begin{align}
&e_0^2 = -\frac{[\omega]}{[\omega+1]}e_0,
\quad e_0e_i = e_ie_0 \quad\forall\, i >1,\nonumber \\
&e_1e_0e_1 = e_1. \label{Bg0}
\end{align}
\end{defn}

\subsubsection{Graphical presentation} We reserve empty tiles and half-tiles for the generators $e_i$ and $e_0$

\be \label{e_pic}
\begin{picture}(0,10)(16,25)
\qbezier[4](61,0)(61,4)(61,8)
\qbezier[6](158,0)(158,6)(158,12)
%
\put(58,-11){$\scriptstyle i$}
\put(157,-11){$\scriptstyle\mathit 0$}
\end{picture}
e_i =\ \raisebox{-18pt}{\includegraphics[height=36pt]{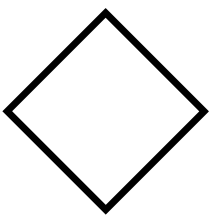}}\, ,
\qquad\qquad e_0 =\
\raisebox{-18pt}{\includegraphics[height=36pt]{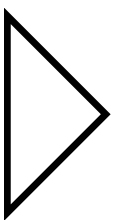}}\, .
\vspace{23pt} \ee The defining relations (\ref{TLdef}) and
(\ref{Bg0}) are depicted, respectively, as
\be \lb{rule-1}
\raisebox{30pt}{$\displaystyle -\frac{1}{[2]}$}\!\!
\begin{picture}(40,72)(0,-15)
{\thicklines \put(10,35){\line(1,1){15}}
\put(10,35){\line(1,-1){15}} \put(40,35){\line(-1,1){15}}
\put(40,35){\line(-1,-1){15}}
\put(10,5){\line(1,1){15}} \put(10,5){\line(1,-1){15}}
\put(40,5){\line(-1,1){15}} \put(40,5){\line(-1,-1){15}} }
\qbezier[2](25,-15)(25,-12.5)(25,-10)
%
\put(24,-23){$\scriptstyle i$}
\end{picture}
\quad \raisebox{30pt}{=} \quad\;\,
\begin{picture}(40,67)(0,-15)
{\thicklines \put(-5,20){\line(1,1){30}}
\put(10,35){\line(1,-1){15}}
\put(40,35){\line(-1,1){15}}
\put(40,35){\line(-1,-1){15}}
\put(10,5){\line(1,1){15}}
\put(-5,20){\line(1,-1){30}}
\put(40,5){\line(-1,1){15}}
\put(40,5){\line(-1,-1){15}} }
\qbezier[10](10,-15)(10,-5)(10,5)
\qbezier[2](25,-15)(25,-12.5)(25,-10)
%
\put(5,-23){$\scriptstyle i-\mathit 1$}
\put(25,-23){$\scriptstyle i$}
\end{picture}
\quad \raisebox{30pt}{=} \quad\;\,
\begin{picture}(40,67)(0,-15)
{\thicklines \put(-5,35){\line(1,1){15}}
\put(-5,35){\line(1,-1){30}}
\put(40,20){\line(-1,1){30}}
\put(-5,5){\line(1,1){30}}
\put(-5,5){\line(1,-1){15}}
\put(40,20){\line(-1,-1){30}} }
\qbezier[10](25,-15)(25,-5)(25,5)
\qbezier[2](10,-15)(10,-12.5)(10,-10)
%
\put(9,-23){$\scriptstyle i$}
\put(19,-23){$\scriptstyle i+\mathit 1$}
\end{picture}
\quad \raisebox{30pt}{=}
\begin{picture}(40,67)(0,-15)
{\thicklines \put(10,20){\line(1,1){15}}
\put(10,20){\line(1,-1){15}}
\put(40,20){\line(-1,1){15}}
\put(40,20){\line(-1,-1){15}} }
\qbezier[10](25,-15)(25,-5)(25,5)
%
\put(24,-23){$\scriptstyle i$}
\end{picture}\;\;\raisebox{30pt}{,}
\vspace{5pt} \ee
and
\be
\lb{rule-2}
\phantom{a}\hspace{20mm}
\begin{picture}(35,65)(0,-15)
{\thicklines \put(5,-15){\line(0,1){30}}
\put(20,0){\line(-1,1){15}}
\put(20,0){\line(-1,-1){15}}
\put(5,15){\line(0,1){30}}
\put(20,30){\line(-1,1){15}}
\put(20,30){\line(-1,-1){15}}
}
\qbezier[3](5,-22)(5,-18.5)(5,-15)
\put(3,-31){$\scriptstyle\mathit 0$}
\end{picture}
\raisebox{25pt}{$\displaystyle =\;-\frac{[\omega ]}{[\omega +1]}$}\,
\begin{picture}(30,65)(0,-15)
{\thicklines \put(5,0){\line(0,1){30}}
\put(20,15){\line(-1,1){15}}
\put(20,15){\line(-1,-1){15}}
}
\qbezier[9](5,-22)(5,-11)(5,0)
\put(3,-31){$\scriptstyle\mathit 0$}
\end{picture}
\raisebox{25pt}{,}\quad\quad
\begin{picture}(35,65)(0,-15)
{\thicklines \put(5,0){\line(0,1){30}} \put(5,0){\line(1,-1){15}}

\put(20,-15){\line(1,1){15}} \put(5,30){\line(1,1){15}}
\put(20,45){\line(1,-1){15}} \put(5,0){\line(1,1){30}}
\put(5,30){\line(1,-1){30}}
}
\qbezier[9](5,-22)(5,-11)(5,0)
\qbezier[3](20,-22)(20,-18.5)(20,-15)
\put(3,-31){$\scriptstyle\mathit 0$}
\put(18,-31){$\scriptstyle\mathit 1$}
\end{picture}
\quad\;\raisebox{25pt}{=}\;\;\,
\begin{picture}(35,65)(0,-15)
{\thicklines \put(5,15){\line(1,-1){15}}
\put(5,15){\line(1,1){15}} \put(20,30){\line(1,-1){15}}
\put(20,0){\line(1,1){15}}
}
\qbezier[9](20,-22)(20,-11)(20,0)
\put(18,-31){$\scriptstyle\mathit 1$}
\end{picture}\;\;\; \raisebox{25pt}{.}
\quad\;\; \vspace{5mm} \ee

\subsubsection{Baxterisation}
Obviously, one can adapt all formulas for the Baxterised elements
from the previous subsection to the case of Temperley-Lieb
algebras by the substitution $a_i\mapsto -e_i$. We shall follow
tradition and will use a special notation --- $R_i(u)$ and $K_0(u)$
--- for the Baxterised elements and the boundary Baxterised element
of the Temperley-Lieb algebras. In the same normalisation used for
$g_i(u)$ (see \eqref{g_u}, \eqref{g0_u}), we have
\begin{align}
R_i(u)& :=\,
\frac{[1-u] - [u]\, e_i}{[1+u]}\, , \nonumber\\
\lb{K_u} K_0(u)& :=\, \frac{k(u,\delta)  + [2u] [\omega+1]
e_0}{k(-u,\delta)},\quad k(u,\delta):=[\tfrac{\omega
+\delta}{2}+u][\tfrac{\omega -\delta}{2}+u].
\end{align}
Here we have intentionally used a different notation to denote an
arbitrary additional parameter, $\delta$ instead of $\nu$ which
was used in the Iwahori-Hecke case, see \eqref{g0_u}. The two
parameters $\delta$ and $\nu$ will play different roles in what
follows below (see the comment after \eqref{rep-s0}).

The elements $R_i(u)$ and $K_0(u)$ satisfy the unitarity
conditions \eqref{unitarity}, \eqref{unitarity2}, the Yang-Baxter
equation \eqref{ybe} and the reflection equation \eqref{rea}. They
are usually called the {\em R-matrix} and the {\em reflection
matrix}. This notation comes from the theory of integrable quantum
spin chains. The path representations of the Temperley-Lieb
algebras, which are introduced in the next subsection and which
are used later on in the qKZ equations, are invariant subspaces of
the state space of certain quantum spin-1/2 XXZ chains, see, e.g.,
\cite{Magic}.

\begin{remark}
\label{rem:notation}
In order to make contact with other notations in the literature,
we note that our notation here correspond to those in \cite{Magic}
if we identify $q=\e^{\i\gamma}$, $q^\omega \rightarrow
\e^{\i\omega}$, $q^\delta \rightarrow-\e^{\i\delta}$, and in \cite{ZJ07} to $q^{\delta}=-\zeta$. Further
useful notations in \cite{Magic} that we shall employ later are:
\[ \tau =-[2],\qquad \tau' = \sqrt{2+[2]} = [2]_{q^{1/2}}\ ,\qquad a =
-\frac{[\omega+1]}{[\frac{\omega-\delta}{2}][\frac{\omega+\delta}{2}]}.
\]
\end{remark}

\subsection{Representations on paths}

We will now decribe an important and well-known representation of
the Temperley-Lieb algebras of type A and type B on Dyck and
Ballot paths respectively.

\subsubsection{Dyck path representation}

\begin{defn}
A Dyck path $\alpha$ of length $N$ is a vector of
$\,(N+1)\,$ local integer heights
\[ \alpha=(\alpha_0,\alpha_1,\ldots,\alpha_N),
\]
such that $\alpha_0=0$, $\alpha_N=0$ for $N$ even and $\alpha_N=1$
for $N$ odd, and the heights are subject to the constraints
$\alpha_i\geq 0$ and $\alpha_{i+1} - \alpha_i =\pm 1$.

By $\D_N$ we denote the set of all Dyck paths of length $N$.
\end{defn}
Each Dyck path $\alpha$ of length $N$ corresponds uniquely to
a word in $w_\alpha\in \T^{\rm A}_N(q)$, represented pictorially
as
\ba
\begin{picture}(0,20)
\linethickness{0.3pt} \put(101,1){\line(0,-1){23}}
\put(235,1){\line(0,-1){23}}
\multiput(97,-14.5)(10.5,0){13}{\linethickness{0.3pt}\put(5,15){\line(1,0){5}}}
\end{picture}
\mbox{N even:}\qquad w_{\alpha}&=&
\raisebox{-33pt}{\includegraphics[height=80pt]{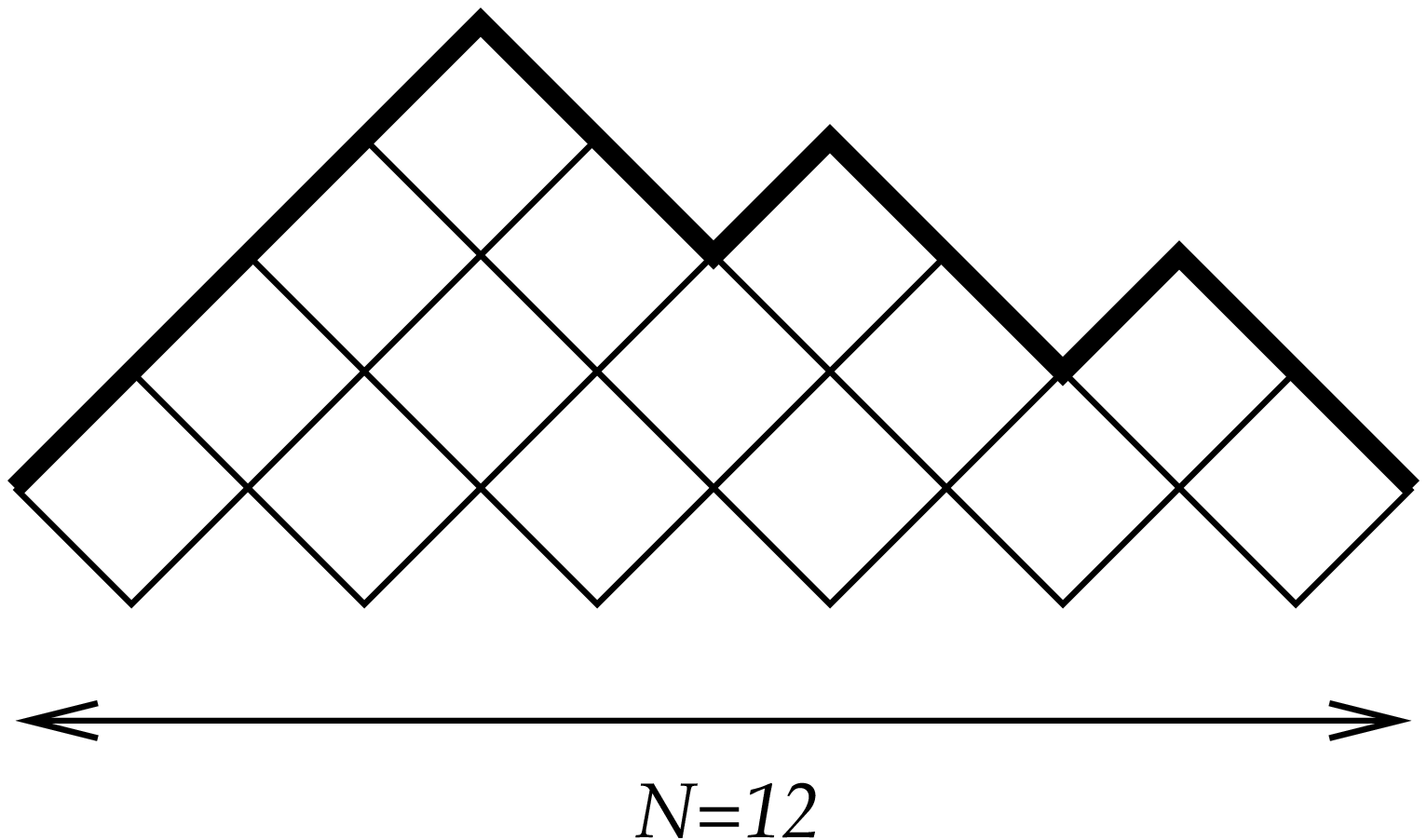}}\; ,
\nonumber
\\[5mm]
\begin{picture}(0,60)(0,-30)
\linethickness{0.3pt} \put(98,0){\line(0,-1){23}}
\put(219,11){\line(0,-1){35}} \put(219,11){\line(1,0){15}}
\put(219,0){\line(1,0){15}} \put(240,4){$\scriptstyle
\alpha_N=\mathit 1$} \put(232,1){\vector(0,1){10}}
\put(232,10){\vector(0,-1){10}} \put(147,-32){$\scriptstyle
N=\mathit 1\mathit 1$} \put(150,-22){\vector(1,0){69}}
\put(160,-22){\vector(-1,0){62}}
\multiput(93,-15)(10.5,0){12}{\linethickness{0.3pt}\put(5,15){\line(1,0){5}}}
\multiput(109,-11)(22,0){5}{\line(1,1){11}}
\multiput(98,0)(22,0){5}{\line(1,-1){11}}
\multiput(98,0)(22,0){6}{\line(1,1){11}}
\multiput(109,11)(22,0){5}{\line(1,-1){11}}
\multiput(109,11)(22,0){5}{\line(1,1){11}}
\multiput(120,22)(22,0){5}{\line(1,-1){11}}
\multiput(120,22)(22,0){3}{\line(1,1){11}}
\multiput(131,33)(22,0){3}{\line(1,-1){11}}
\multiput(131,33)(22,0){1}{\line(1,1){11}}
\multiput(142,44)(22,0){1}{\line(1,-1){11}}
{\linethickness{1.3pt} \qbezier[400](98,0)(120,22)(142,44)
\qbezier[200](142,44)(153,33)(164,22)
\qbezier[100](164,22)(169.5,27.5)(175,33)
\qbezier[200](175,33)(186,22)(197,11)
\qbezier[200](197,11)(202.5,16.5)(208,22)
\qbezier[200](208,22)(213.5,16.5)(219,11) }
\end{picture}
\raisebox{9pt}{$\mbox{N odd:}\qquad
w_{\alpha}$}&\raisebox{9pt}{$=$}&
\hspace{170pt} \raisebox{9pt}{,}
\nonumber
\ea
\noindent where the empty tiles at horizontal
position $i$ are the generators $e_i$, see (\ref{e_pic}).

We now define an action of the algebra $\T^{\rm A}_N(q)$ on a
space which is spanned linearly by states $\ket{\alpha}$ labeled by the Dyck paths, identifying the
states $\ket\alpha$ with the corresponding words $w_\alpha\in
\T^{\rm A}_N(q)$. This action is given by a set of elementary
transformations of pictures shown in (\ref{rule-1}). A typical
example of such an action is given in Figure~\ref{fig:eonSlope}.
\begin{figure}[h]
\includegraphics[height=60pt]{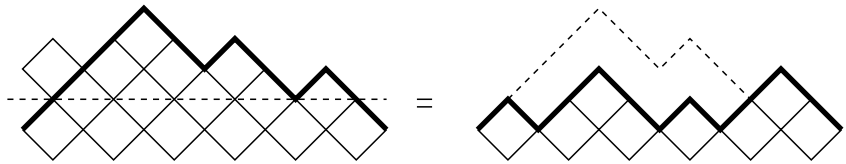}
\caption{The result of $e_i\ket\alpha$ if $\alpha$ has
  a slope at $i$. If $i+r$ is the first position to the right of an
  upward slope at $i$ whose height is equal to that at $i$,
  i.e. $\alpha_{i+r}=\alpha_i>0$, then a layer of
  tiles between $i$ and $i+r$ is peeled off the original path and the
  result is again a Dyck path. A similar peeling mechanism to the left
  works for downward slopes.}
\label{fig:eonSlope}
\end{figure}

In doing so we find the following representation of the algebra
$\T^{\rm A}_N(q)$:
\begin{prop}
\label{prop:ActionOfGeneratorsA} The action of $e_i$ for
$i=1,\ldots,N-1$ on Dyck paths is explicitly given by
\begin{align}
\bullet\ &\text{Local Minimum:}\nonumber\\
&e_i \ket{\ldots,\alpha_i+1,\alpha_i,\alpha_i+1,\ldots} =
\ket{\ldots,\alpha_i+1,\alpha_i+2,\alpha_i+1,\ldots},
\nonumber\\[3mm]
\bullet\ &\text{Local Maximum:}\nonumber\\
&e_i \ket{\ldots,\alpha_i-1,\alpha_i,\alpha_i-1,\ldots} = -[2]\ \ket{\ldots,\alpha_i-1,\alpha_i,\alpha_i-1,\ldots},\nonumber\\[3mm]
\bullet\ &\text{Uphill Slope: $\alpha_{i-1}<\alpha_i<\alpha_{i+1}$.}\nonumber\\
&\text{Let $j>i$ be such that $\alpha_j=\alpha_i$ and $\alpha_l > \alpha_i\;\;\forall\, l:\; i<l<j$, then}\nonumber\\
&e_i \ket{\ldots,\alpha_i-1,\alpha_i,\alpha_i+1,\alpha_{i+2},\ldots,\alpha_j,\ldots} = \nonumber\\
&\hphantom{\text{Uphill Slope:}} \ket{\ldots,\alpha_i-1,\alpha_i,\alpha_i-1,\alpha_{i+2}-2,\ldots,\alpha_{j-1}-2,\alpha_j,\alpha_{j+1},\ldots},\nonumber\\[3mm]
\bullet\ &\text{Downhill Slope: $\alpha_{i-1}>\alpha_i>\alpha_{i+1}$.}\nonumber\\
&\text{Let $k<i$ be such that $\alpha_k=\alpha_i$ and $\alpha_l > \alpha_i\;\;\forall\, l:\; k<l<i$, then}\nonumber\\
&e_i \ket{\ldots,\alpha_k,\ldots,\alpha_{i-2},\alpha_i+1,\alpha_i,\alpha_i-1,\ldots} =\nonumber\\
&\hphantom{\text{Uphill Slope:}}
 \ket{\ldots,\alpha_{k},\alpha_{k+1}-2,\ldots,\alpha_{i-2}-2,\alpha_i-1,\alpha_i,\alpha_i-1,\ldots}
 \nonumber
\end{align}
\end{prop}

\begin{remark}
For generic values of $q$ the Dyck path representation is the
irreducible representation of the Temperley-Lieb algebra $\T^{\rm
A}_N(q)$ corresponding in the conventional classification to the
partition $\{\lfloor\tfrac{ N+1}{2}\rfloor,\lfloor\tfrac{ N}{2}\rfloor\}$.
\end{remark}

\begin{defn}\lb{rem3}
We call the unique Dyck path without local minima in the bulk the
{\em maximal Dyck path} and denote it $\Omega^{\rm A}$.
Explicitly this path reads:
\[ \Omega^{\rm A} =
(0,1,2,\dots ,\lfloor\tfrac{N-1}{2}\rfloor,
\lfloor\tfrac{N+1}{2}\rfloor,\lfloor\tfrac{N-1}{2}\rfloor,\dots
,\epsilon_N), \]
where $\epsilon_N:= N \bmod 2$ is the parity of $N$.
\end{defn}
It is clear from Proposition \ref{prop:ActionOfGeneratorsA} that
the maximal Dyck path plays a role of a highest weight element of
the Dyck path representation.

\subsubsection{Ballot path representation}

\begin{defn}
A Ballot path $\alpha$ of length $N$ is a vector of  $(N+1)$ local
integer heights
\[ \alpha=(\alpha_0,\alpha_1,\ldots,\alpha_N),
\]
such that $\alpha_i \geq 0$, $\alpha_{i+1}-\alpha_i = \pm 1$ and
$\alpha_N=0$.

We denote the set of all Ballot paths of the length $N$ by $\B_N$
.
\end{defn}

Each Ballot path $\alpha$ of length $N$ corresponds uniquely
to a word $w_\alpha\in \T^{\rm B}_N(q,w)$, represented
pictorially as
\ba
\begin{picture}(0,0)
\linethickness{0.3pt} \put(131.5,-7){\line(0,-1){23}}

\put(265.5,-7){\line(0,-1){23}}
\multiput(127,-21)(10.5,0){13}{\linethickness{0.3pt}\put(5,15){\line(1,0){5}}}
\put(131.5,-6){\line(-1,0){15}} \put(131.5,38){\line(-1,0){15}}
\put(118.5,13){\vector(0,1){25}} \put(118.5,19){\vector(0,-1){25}}
\put(95,14){$\scriptstyle \alpha_0=\mathit 4$}
\end{picture}
\mbox{N even:}\qquad w_{\alpha}&=&\hspace{30pt}
\raisebox{-40pt}{\includegraphics[height=80pt]{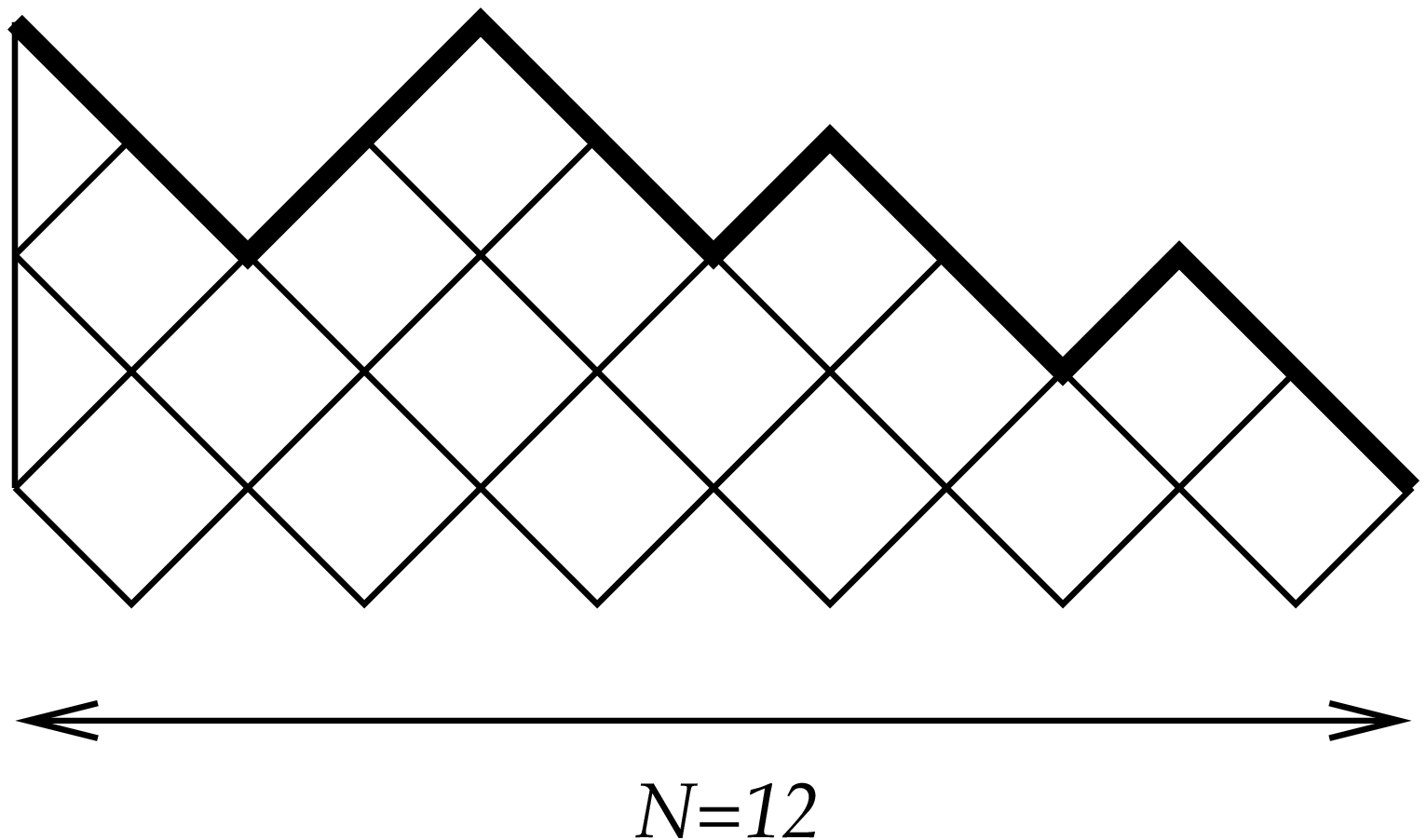}}\;
,\qquad\qquad
\nonumber
\\
\begin{picture}(0,70)(0,-30)
\linethickness{0.3pt} \put(131.5,27){\line(0,-1){57}}
\put(252.5,-7){\line(0,-1){23}}
\multiput(127,-22)(10.5,0){12}{\linethickness{0.3pt}\put(5,15){\line(1,0){5}}}
\put(131.5,-7){\line(-1,0){15}} \put(131.5,27){\line(-1,0){15}}
\put(118.5,2){\vector(0,1){25}} \put(118.5,18){\vector(0,-1){25}}
\put(182.5,-28){\vector(1,0){70}}
\put(201.5,-28){\vector(-1,0){70}} \put(95,7){$\scriptstyle
\alpha_0=\mathit 3$} \put(180,-38){$\scriptstyle\mathit N=\mathit
1\mathit 1$}
\multiput(131.5,-18)(22,0){6}{\line(1,1){11}}
\multiput(142.5,-7)(22,0){5}{\line(1,-1){11}}
\multiput(142.5,-7)(22,0){5}{\line(1,1){11}}
\multiput(131.5,4)(22,0){6}{\line(1,-1){11}}
\multiput(131.5,4)(22,0){5}{\line(1,1){11}}
\multiput(142.5,15)(22,0){5}{\line(1,-1){11}}
\multiput(164.5,15)(22,0){2}{\line(1,1){11}}
\multiput(175.5,26)(22,0){2}{\line(1,-1){11}}
\multiput(131.5,26)(22,0){1}{\line(1,-1){11}}
{\linethickness{1.3pt} \qbezier[200](131.5,26)(142.5,15)(153.5,4)
\qbezier[200](153.5,4)(164.5,15)(176.5,26)
\qbezier[100](176.5,26)(182,20.5)(187.5,15)
\qbezier[100](187.5,15)(193,20.5)(198.5,26)
\qbezier[200](198.5,26)(209.5,15)(220.5,4)
\qbezier[100](220.5,4)(226,9.5)(231.5,15)
\qbezier[200](231.5,15)(242.5,4)(253.5,-7) }
\end{picture}
\raisebox{15pt}{$\mbox{N odd:}\qquad
w_{\alpha}$}&\raisebox{15pt}{$=$}&\hspace{163pt}\raisebox{15pt}{,}
\nonumber
\ea
\noindent where the empty (half-)tiles at horizontal position $i$
($0$) are the generators $e_i$ ($e_0$), see (\ref{e_pic}).

Now we can define an action of the algebra $\T^{\rm
B}_N(q,\omega)$ on the space which is spanned linearly by states $\ket{\alpha}$ labeled by the
Ballot paths, identifying the states $\ket\alpha$ with the
corresponding words $w_\alpha\in \T^{\rm B}_N(q,\omega)$. We
thus find the following representation of the algebra $\T^{\rm
B}_N(q,\omega)$:

\begin{prop}
\label{prop:ActionOfGeneratorsB} The action of $e_i$ for
$i=1,\ldots,N-1$ on Ballot paths is explicitly given by
Proposition~\ref{prop:ActionOfGeneratorsA} in the case of a local
extremum or an uphill slope. In the remaining cases we find
\begin{align}
\bullet\ &\text{Downhill Slope, Type I:}\nonumber\\
&\text{If there exists
$k<i$ such that $\alpha_k=\alpha_i$ and $\alpha_l > \alpha_i\; \forall\, l:\, k<l<i$, then}\nonumber\\
&e_i \ket{\ldots,\alpha_k,\ldots,\alpha_{i-2},\alpha_i+1,\alpha_i,\alpha_i-1,\ldots} =\nonumber\\[1mm]
&\hphantom{\hspace{30mm}}
 \ket{\ldots,\alpha_{k},\alpha_{k+1}-2,\ldots,\alpha_{i-2}-2,\alpha_i-1,\alpha_i,\alpha_i-1,\ldots}
 \nonumber
 \\[3mm]
\bullet\ &\text{Downhill Slope, Type IIa:}\nonumber\\
&\text{If $i$ is odd and $\alpha_k > \alpha_i\; \forall\, k<i$, then}\nonumber\\
&e_i \ket{\alpha_0, \ldots,\alpha_{i-2},\alpha_i+1,\alpha_i,\alpha_i-1,\ldots} =\nonumber\\[1mm]
&\hphantom{\hspace{30mm}}
\ket{\alpha_0-2,\ldots,\alpha_{i-2}-2,\alpha_i-1,\alpha_i,\alpha_i-1,\ldots}.
\label{eq:DownHillIIa}
 \\[3mm]
\bullet\ &\text{Downhill Slope, Type IIb:}\nonumber\\
&\text{If $i$ is even and $\alpha_k > \alpha_i\; \forall\, k<i$, then}\nonumber\\
&e_i \ket{\alpha_0,\ldots,\alpha_{i-2},\alpha_i+1,\alpha_i,\alpha_i-1,\ldots} =\nonumber\\[1mm]
&\hphantom{\hspace{30mm}}
 -\frac{[\omega]}{[\omega+1]}\ \ket{\alpha_0-2,\ldots,\alpha_{i-2}-2,\alpha_i-1,\alpha_i,\alpha_i-1,\ldots}.
 \label{eq:DownHillIIb}
\end{align}
The action of the boundary generator $e_0$ is given by
\begin{align}
\bullet\ &\text{Uphill Slope at $i=0$:}\hspace{97mm}\nonumber\\
&e_0 \ket{\alpha_0,\alpha_0+1,\alpha_{2},\ldots} =
\ket{\alpha_0+2,\alpha_0+1,\alpha_{2},\ldots}
\nonumber
\\[3mm]
\bullet\ &\text{Downhill Slope at $i=0$:}\nonumber\\
&e_0 \ket{\alpha_0,\alpha_0-1,\alpha_{2},\ldots} =
- \frac{[\omega]}{[\omega+1]}\ket{\alpha_0,\alpha_0-1,\alpha_{2},\ldots}
\nonumber
\end{align}
\end{prop}

\begin{remark}
For generic values of $q$ and $\omega$ the Ballot path representation is the
irreducible representation of the Temperley-Lieb algebra $\T^{\rm
B}_N(q,\omega)$ corresponding to
bi-partition $\{\lfloor\tfrac{ N+1}{2}\rfloor\},\, \{\lfloor\tfrac{ N}{2}\rfloor\}$.
\end{remark}

\begin{defn}\lb{rem4}
We call the unique Ballot path without local minima in the bulk
the {\em maximal} Ballot path and denote it by $\Omega^{\rm B}$.
Explicitly this path reads:
\[
\Omega^{\rm B} = (N,N-1,\dots ,2,1,0).
\]
\end{defn}
As follows from the Proposition \ref{prop:ActionOfGeneratorsB} the
maximal Ballot path plays a role of a highest weight element of
the Ballot path representation.

\section{$q$-deformed  Knizhnik-Zamolodchikov equation}

\subsection{Definition}

Let us consider a linear combination $\ket\Psi$ of states $\ket \alpha$ with coefficients
$\psi_\alpha$ taking values in the ring of formal series in $N$ variables $q^{\pm x_i}, \; i=1,2,\dots ,N$:
\[
\ket{\Psi(x_1,\ldots,x_{N})} = \sum_{\alpha}
\psi_\alpha(x_1,\ldots,x_{N}) \ket{\alpha}.
\]
Here $\alpha$ runs over the set of either Dyck (type A), or Ballot (type B)
paths of length $N$.

The {\em qKZ equation}
in the Temperley-Lieb algebra setting
is a system of finite difference equations on the vector $\ket\Psi$.
Actually, we consider the qKZ equation in an alternative
form with permutations in place of finite differences. This is historically the first form it
appeared in literature \cite{S}.
In both types A and B
the qKZ equation reads universally  \cite{ZJ07},
\begin{align}
R_i(x_i-x_{i+1}) \ket\Psi &= \pi_i \ket\Psi,\qquad \forall\, i=1,\ldots,N-1,
\label{qKZTL_TypeB1}\\
K_0(-x_1) \ket\Psi  &= \pi_0 \ket\Psi,
\label{qKZTL_TypeB2}\\
\ket\Psi &= \pi_N\ket\Psi.
\label{qKZTL_TypeB3}
\end{align}
Here $R_i$  are the Baxterised elements of the Temperley-Lieb
algebra, $K_0$ is the boundary Baxterised element in the type B case
and $K_0$ is the identity operator in type A.
The operators $R_i(x_i-x_{i+1})$ and $K_0(-x_1)$ act on states $\ket{\alpha}$, whereas the
operators $\pi_i$ permute or reflect  arguments of
the coefficient functions
\begin{align}
\pi_i \psi_\alpha(\ldots,x_i,x_{i+1},\ldots) &= \psi_\alpha(\ldots,x_{i+1},x_{i},\ldots),
\nonumber\\
\pi_0\psi_\alpha(x_1,\ldots)&= \psi_\alpha(-x_1,\ldots),\label{pi0}\\
\pi_N\psi_\alpha(\ldots,x_N)&= \psi_\alpha(\ldots,-\lambda-x_N).\label{piN}
\end{align}
Here $\lambda\in {\mathbb C}$ is a parameter related to the level of the qKZ equation, see \cite{EFK}.

\begin{remark}
Clearly, the elementary permutations $\pi_i$, $i=1,\dots , N-1,$
are the generators of the symmetric group $S_N$, whereas
$\pi_0$ and $\pi_N$ are left and right boundary reflections. They satisfy the relations
\begin{align*}
&\pi_i^2= 1, \qquad\pi_i\pi_{i+1}\pi_i=\pi_{i+1}\pi_i\pi_{i+1},\qquad\qquad\qquad\qquad\;\;
\pi_i\pi_j =\pi_j\pi_i\quad \forall\, i,j: |i-j|>1,\\
&\pi_0^2=1, \qquad\pi_0\pi_1\pi_0\pi_1 =\pi_1\pi_0\pi_1\pi_0,\qquad\qquad\qquad\qquad\;\;\;
\pi_0\pi_i =\pi_i\pi_0\quad \forall\, i>1,\\
&\pi_N^2=1, \qquad\pi_N\pi_{N-1}\pi_N\pi_{N-1} =\pi_{N-1}\pi_N\pi_{N-1}\pi_N,\qquad
\pi_N\pi_i =\pi_i\pi_N\;\; \forall\, i<N-1.
\end{align*}
Therefore the unitarity conditions \eqref{unitarity}, \eqref{unitarity2},
the Yang-Baxter relation \eqref{ybe} and the reflection equation \eqref{rea} for
the operators $R_i$ and $K_0$
are consistency conditions for the qKZ equation.
\end{remark}

\subsection{Algebraic interpretation}
In this subsection we consider the algebraic content of the qKZ equation.
We follow the lines of the paper \cite{P05}.

\subsubsection{Type A}
Consider equation \eqref{qKZTL_TypeB1}.
Here the R-matrix $R_i(x_i-x_j)$ acts on the states $\ket\alpha$, $\alpha \in \D_{N}$, while the operator $\pi_i$ acts on
the functions $\psi_\alpha(x_1,\ldots,x_{N})$. In other words, the qKZ equation
\eqref{qKZTL_TypeB1} written out
in components becomes
\[
\sum_{\alpha\in \D_{N}}
\psi_\alpha(x_1,\ldots,x_{N}) \, (R_i(x_i-x_{i+1}) \ket\alpha) =
\sum_{\alpha\in \D_{N}} (\pi_i \psi_\alpha)(x_1,\ldots,x_{N})\ket\alpha,
\]
which can be rewritten as
\begin{multline}
\sum_{\alpha\in \D_{N}}
\psi_\alpha(x_1,\ldots,x_{N})\, (-e_i \ket{\alpha}) \\ =  \sum_{\alpha\in \D_{N}}
\left( \frac{[x_{i}-x_{i+1}+1]}{[x_{i}-x_{i+1}]}\pi_i +
\frac{[x_{i}-x_{i+1}-1]}{[x_{i}-x_{i+1}]}\right)
\psi_\alpha(x_1,\ldots,x_{N}) \ket\alpha \\
= \sum_{\alpha\in \D_{N}}
\left(a_i
\psi_\alpha\right)(x_1,\ldots,x_{N}) \ket\alpha.
\label{qKZ2TL_bulk}
\end{multline}
Here we have used notation
\begin{align}
\lb{a-rep}
a_i &:=\, \frac{1}{[x_i-x_{i+1}]}(\pi_i - 1)[x_{i+1}-x_{i}+1]\, \nonumber\\
&\,=\,(\pi_i+1)\frac{[x_i-x_{i+1}-1]}{[x_i-x_{i+1}]}, \qquad i=1,2,\dots ,N-1,
\end{align}
for {\em symmetrising operators} acting on functions in the variables $x_i$ \cite{DKLLST}.
These operators satisfy the Hecke relations \eqref{Hdef_typeA_a}. Moreover, they
generate a {\em faithful}  representation
of the algebra $\H_N^{\rm A}(q)$ in the space of functions in $N$
variables $x_i\; (i=1,\dots ,N)$, thus justifying the use of the identical notation
$a_i$ in \eqref{Hdef_typeA_a} and in \eqref{a-rep}. The generator $g_i=q-a_i$ in this representation is known as the Demazurre-Lusztig operator \cite{Ch}.
The alternative set of generators $s_i$ and the Baxterised elements $h_i(u)$ \eqref{h_u}
in this particular representation read
\begin{align}
s_i &= \frac{[x_{i}-x_{i+1}+1]}{[x_{i}-x_{i+1}]} \,(1-\pi_i),
\nonumber\\
h_i(u) &= \frac{[x_{i}-x_{i+1}+u]}{[u][x_{i}-x_{i+1}]}\, -\, \frac{[x_{i}-x_{i+1}+1]}{[x_{i}-x_{i+1}]}\pi_i .
\lb{h-rep}
\end{align}

Looking back at the equations \eqref{qKZ2TL_bulk} we notice that their solution amounts
to constructing an explicit homomorphism from the Dyck path representation
of the Temperley-Lieb algebra $\T^{\rm A}_N(q)$  into
the functional representation \eqref{a-rep} of the Iwahori-Hecke algebra
$\H^{\rm A}_N(q)$:
\be
\ket{\alpha} \mapsto \psi_\alpha\quad \forall \alpha\in \D_N,
\ee
where $\psi_\alpha$ are the components of the solution $\ket\Psi$ of the qKZ equation
of type A.
\subsubsection{Type B}

In this case we additionally have a non-trivial operator $K_0$
affecting equation \eqref{qKZTL_TypeB2}, which reads in components:
\[
\sum_{\alpha\in \B_{N}}
\psi_\alpha(x_1,\ldots,x_{L}) \, (K_0(-x_1) \ket\alpha) =
\sum_{\alpha\in \B_{N}} (\pi_0 \psi_\alpha)(x_1,\ldots,x_{L})\ket\alpha,
\]
where the summation is taken now over all Ballot paths. Recalling the definition
\[
k(u,\delta)=[u+\frac{\omega+\delta}{2}][u+\frac{\omega-\delta}{2}],
\]
from \eqref{K_u}, this can be rewritten as
\begin{multline}
\sum_{\alpha\in \B_{N}}
\psi_\alpha(x_1,\ldots,x_N)\, (-e_0 \ket{\alpha}) \\ = \sum_{\alpha\in \B_{N}}
\left( \frac{k(x_1,\delta)}{[2x_1][\omega+1]}\pi_0 -
\frac{k(-x_1,\delta)}{[2x_1][\omega+1]}
\right)
\psi_\alpha(x_1,\ldots,x_{N}) \ket\alpha \\
= \sum_{\alpha\in \B_{N}}
\left(a_0
\psi_\alpha\right)(x_1,\ldots,x_{N}) \ket\alpha.
\label{qKZ2TL_bound}
\end{multline}
where we have denoted
\be
\lb{rep-a0}
a_0 := -(\pi_0+1)\frac{k(-x_1,\delta)}{[2x_1][\omega+1]}\, =\, \frac{1}{[2x_1][\omega+1]}(\pi_0-1)k(-x_1,\delta)
\ee
The operator $a_0$ and the operators $a_i$ from \eqref{a-rep}
satisfy the defining relations \eqref{Hdef_typeB_a} for the type B
Iwahori-Hecke algebra. Moreover, the realisation \eqref{rep-a0}, \eqref{a-rep}
gives a {\em faithful} representation of $\H^{B}_N(q,\omega)$.
The generator $s_0$ defined in \eqref{def_e0_s0} and the boundary Baxterised element $h_0(u)$
from \eqref{h0_u} in this realisation read
\be
\lb{rep-s0}
s_0 = \frac{k(x_1,\delta)}{[2x_1][\omega+1]}(1-\pi_0),\qquad
h_0(u) = s_0 - \frac{k(u/2,\nu)}{[u][\omega+1]} .
\ee

Let us stress here the difference between the parameters $\delta$ and $\nu$.
The ``physical" parameter $\delta$ appears in the definition of the boundary Baxterised
element $K_0(u)$ and therefore enters the qKZ equation and the boundary conditions of related
integrable models, see e.g. \cite{Magic,ZJ07}. The parameter $\nu$ is introduced here for the first
time in this section. It plays an auxilliary role and we will fix it later to obtain a convenient
presentation of the solution of the qKZ equation of type B, see \eqref{half-tile} below.

Finally, the relation \eqref{qKZ2TL_bound} states  that
the solution $\ket\Psi$ of the  qKZ equation of type B
encodes an explicit homomorphism from the Ballot path representation
of the Temperley-Lieb algebra $\T^{\rm B}_N(q,\omega)$  into
the functional representation \eqref{a-rep}, \eqref{rep-a0}
 of the Iwahori-Hecke algebra
$\H^{\rm B}_N(q,\omega)$:
\[
\ket{\alpha} \mapsto \psi_\alpha\quad \forall\, \alpha\in \B_N.
\]

\subsection{Preliminary analysis}

\subsubsection{Type A}
\label{subsec3.3.1}
Equation \eqref{qKZ2TL_bulk} breaks up into two cases, depending whether or
not a path $\alpha$ in the sum on the right-hand side  of
\eqref{qKZ2TL_bulk} has a local maximum at $i$, i.e. whether or not it is of the form
$\alpha=(\ldots,\alpha_i-1,\alpha_i,\alpha_i-1,\ldots)$.
We will first look at the case in which it does not.
\bigskip

\noindent\textbf{Case i)}:\ $\alpha$ does not have a local maximum at $i$.

As each term in the left-hand side  of \eqref{qKZ2TL_bulk}
is of the form $e_i\ket\alpha$, and hence corresponds to a local maximum at $i$,

the coefficient of $\ket\alpha$ in the right-hand side of \eqref{qKZ2TL_bulk}
has to equal zero. Hence we obtain
\be
-(a_i \psi_\alpha)(x_1,\ldots,x_{N}) \equiv (h_i(-1)
\psi_\alpha) (x_1,\ldots,x_{N}) =0,
\label{eq:hmonpsi}
\ee
which can be rewritten as
\[
(\pi_i -1) \left\{ [x_{i}-x_{i+1}-1]
\psi_\alpha(x_1,\ldots,x_{N})\right\} =0\quad\text{for}\;
\ket\alpha\not\propto e_i \ket{\alpha'}.
\]
Hence, if $\ket\alpha\not\propto e_i \ket{\alpha'}$ the function
\[
[x_{i}-x_{i+1}-1] \psi_\alpha(x_1,\ldots,x_{N})
\]
is symmetric in $x_i$ and $x_{i+1}$ which implies that
$[x_{i+1}-x_{i}-1]$ divides $\psi_\alpha(x_1,\ldots,x_{N})$ and the ratio is symmetric
with respect to $x_i$ and $x_{i+1}$.

Iterating \eqref{qKZTL_TypeB1} we find
\begin{multline}
\Psi(x_1,\ldots,x_{k-1},x_m,x_{k},\ldots,x_{m-1}, x_{m+1},\ldots x_N) = \\R_{k}(x_{k}-x_m)
\cdots R_{m-1}(x_{m-1}-x_m) \Psi(x_1,\ldots,x_N).
\label{ijperm}
\end{multline}
Consider now the component $\psi_\alpha$ on the LHS of \eqref{ijperm},
where $\alpha$ does not have a local maximum at any $i$ for $k\leq i \leq
m-1$. This component can only arise from the same component on the RHS
of \eqref{ijperm}, on which the R-matrices have acted as multiples of
the identity. Hence, if $\alpha$ is a path which does not have a local maximum for any $k\leq i \leq m-1$, we find
\be
\psi_\alpha(x_1,\ldots,x_{k-1},x_m,x_{k},\ldots,x_{m-1},x_{m+1},\ldots,x_N) =
\prod_{i=k}^{m-1} \frac{[1-x_i+x_m]}{[1+x_i-x_m]}
  \psi_\alpha(x_1,\ldots,x_N).
\label{ijperm2}
\ee
It follows from \eqref{ijperm2} that if $\alpha$ does not have a local maximum between $k$
and $m-1$, then $\psi_\alpha(x_1,\ldots,x_N)$ contains a factor $\prod_{k\leq i < j\leq m}
[1+x_i-x_j]$ and the ratio is symmetric in the variables $x_i$, $k\leq i\leq m$.
An analogous argument can be given
when considering the boundary equations \eqref{qKZTL_TypeB2}, with $K_0=1$, and
\eqref{qKZTL_TypeB3}. In summarising the effects of these considerations, it will be
convenient to introduce the following notation:
\begin{defn}
\label{def:Delta} We denote by $\Delta_\mu^{\pm}$ the following functions:
\[
\Delta_\mu^{\pm} (x_k,\ldots,x_m) := \prod_{k\leq i<j\leq m}
    [\mu+x_i\pm x_j],
\]
where $\mu$ is a parameter.
\end{defn}

\begin{lemma}
\label{lem:slope_factors}

The following hold:
\begin{itemize}
\item[$\bullet$] If $\alpha$ does not have a local maximum between $k$ and $m-1$,
then $\psi_\alpha(x_1,\ldots,x_N)$ contains a factor $\Delta^-_1 (x_k,\ldots,x_m)$
and the ratio is symmetric in the variables
$x_i, \;\; k\leq i\leq m$.\medskip

\item[$\bullet$] If $\alpha$ does not have a local maximum between $1$ and $m-1$,
then $\psi_\alpha(x_1,\ldots,x_N)$ contains a factor
$\Delta_1^-(x_1,\ldots,x_m)\Delta^+_{-1}(x_1,\ldots,x_m)$
and the ratio is an even symmetric function in the variables
$x_i, \;\; 1\leq i\leq m$.

\item[$\bullet$] If $\alpha$ does not have a local maximum between $k$ and $N-1$,
then $\psi_\alpha(x_1,\ldots,x_N)$ contains a factor
$\Delta^-_1(x_k,\ldots,x_N)\Delta^+_{\lambda +1}(x_k,\ldots,x_N)$
and the ratio is an even symmetric function in the variables
$(x_i+\lambda/2), \;\; k\leq i\leq N$.
\end{itemize}
\end{lemma}
\begin{cor}
\label{cor:psi0}
The {\em base} coefficient function $\psi_{\Omega}^{\rm A}$ corresponding to the maximal Dyck path
$\Omega^{\rm A}$ (see Definition~\ref{rem3})
in the solution of the qKZ equation of type A  has the following form:
\begin{multline}
\psi_{\Omega}^{\rm A} (x_1,\ldots,x_N) =
\Delta^-_1 (x_1,\ldots,x_{n})\,
\Delta^+_{-1}(x_1,\ldots,x_{n})\,
\Delta^-_1(x_{n+1},\ldots,x_{N})\\
\Delta^+_{\lambda +1}(x_{n+1},\ldots,x_{N})\,\,
\xi^{\rm A}\Bigl(x_1,\ldots,x_{n}|
x_{n+1}+\tfrac{\lambda}{2},
\ldots ,x_{N}+\tfrac{\lambda}{2}\Bigr),
\label{psi_alpha_D}
\end{multline}
where $n=\lfloor (N+1)/2\rfloor$ and $\xi^{\rm A}(x_1,\dots x_n|x_{n+1},\ldots ,x_N)$ is an even symmetric function separately in the variables
$x_i,\;\; 1\leq i\leq n$ and
$x_j,\;\;n+1\leq j\leq N$.
\end{cor}

\begin{proof}
The path $\Omega^{\rm A}$ does not have a local maximum between $1$ and
$n$ and neither between $n+1$ and $N$, and the result follows immediately from Lemma~\ref{lem:slope_factors}.
\end{proof}

In the sequel we use the following picture to represent
$\psi_{\Omega}^{\rm A}$:
\be
\psi_{\Omega}^{\rm A}\, =\;
\raisebox{-25pt}{
\begin{picture}(230,50)
\put(0,10){\includegraphics[width=100pt]{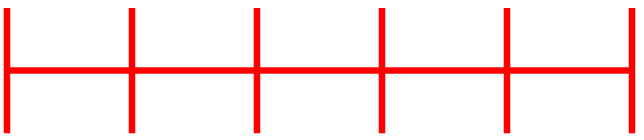}}
\put(120,10){\includegraphics[width=100pt]{Delta.eps}}
\put(-2,35){\small $x_1$}
\put(17,35){\small $x_2$}
\put(53,35){\small $\dots$}
\put(96,35){\small $x_n$}
\put(116,35){\small $x_{n+1}$}
\put(164,35){\small $\dots$}
\put(216,35){\small $x_{N}$}
\end{picture}
}
\label{psi-D}
\ee
Relation \eqref{eq:hmonpsi} for $\psi^{\rm A}_{\Omega}$ can then be displayed
pictorially as
\be
\label{eq:hmonpsi2}
\begin{picture}(300,60)(8,-20)
{\linethickness{1pt}
\multiput(129,-20)(6,6){3}{\qbezier[10](0,0)(1,1)(2,2)}
\multiput(129,9)(6,-6){3}{\qbezier[10](0,0)(1,-1)(2,-2)}
\multiput(129,-20)(-6,6){3}{\qbezier[10](0,0)(-1,1)(-2,2)}
\multiput(129,9)(-6,-6){3}{\qbezier[10](0,0)(-1,-1)(-2,-2)} }
\put(78,25){\small $x_1$}
\put(96,25){\small $\dots$}
\put(116,25){\small $x_i$}
\put(132,25){\small $x_{i+1}$}
\put(156,25){\small $\dots$}
\put(176,25){\small $x_n$}
\put(196,25){\small $x_{n+1}$}
\put(244,25){\small $\dots$}
\put(296,25){\small $x_{N}$}
\put(121,-9.7){$-1$}
\put(-50,8){$(h_i(-1)\psi^{\rm A}_{\Omega})(x_1,\ldots,x_N)=$}
\put(80,0){\includegraphics[width=100pt]{Delta.eps}}
\put(200,0){\includegraphics[width=100pt]{Delta.eps}}
\put(310,8){$=0,$}
\end{picture}
\ee
where $1\leq i<N$, $i\neq n$, and we use the graphical notation \eqref{h_pic} to represent $h_i(-1)$.
\bigskip

\noindent \textbf{Case ii)}: $\alpha$ has a local maximum at $i$.

Now \eqref{qKZ2TL_bulk} gives,
\be
([2]-a_i)\, \psi_\alpha \equiv s_i\, \psi_\alpha \equiv
 h_i(1)\, \psi_\alpha =
 \sum_{\genfrac{}{}{0pt}{}{\beta\neq\alpha}{e_i\beta=\alpha}} \psi_{\beta}\, =\,
 \psi_{\alpha^-}\, +\, \sum_{k=1,2,\dots}\, \psi_{\alpha^{+k}}.
\label{TLei}
\ee
A pictorial definition of the paths $\alpha^-$ and $\alpha^{+k}$,
$k=1,2,\dots $, is given in Figure \ref{alphapm}.
In the example of Figure \ref{alphapm}, the heights
of the paths $\alpha^{+k}$ to the right of the point $i$ are higher than those of $\alpha$.
This happens in case if the path $\alpha$ has a local minimum at the point $i+1$,
i.e., if $\alpha_i=\alpha_{i+1}+1=\alpha_{i+2}$.
Equally well, if the path $\alpha$ has a local minimum at the point $i-1$
($\alpha_i=\alpha_{i-1}+1=\alpha_{i-2}$), the sum in \eqref{TLei} contains
one or several paths $\alpha^{+k}$ whose heights are higher than those of $\alpha$ to the left of the point $i$. The number of paths $\alpha^{+k}$ appearing in the sum \eqref{TLei} depends
on the shape of $\alpha$ and varies from 0 to $\lfloor(N-1)/2\rfloor$.

\begin{figure}[h]
\begin{picture}(200,120)(0,0)
\put(30,100){\line(1,-1){10}}
\put(40,110){\line(1,-1){10}}
\put(30,100){\line(1,1){10}}
\put(40,90){\line(1,1){10}}
\put(37,98){$e_i$}
{\thicklines \put(40,90){\vector(0,-1){10}}}
\put(0,30){\line(1,-1){10}}
\put(10,20){\line(1,1){30}}
\put(40,50){\line(1,-1){10}}
\put(50,40){\line(1,1){30}}
\put(80,70){\line(1,-1){30}}
\put(110,40){\line(1,1){30}}
\put(140,70){\line(1,-1){10}}
\put(150,60){\line(1,1){20}}
\put(170,80){\line(1,-1){40}}
\put(210,40){\line(1,1){20}}
\put(230,60){\line(1,-1){30}}
\put(260,30){\line(1,1){25}}
\qbezier[10](30,40)(35,35)(40,30)
\qbezier[10](40,30)(45,35)(50,40)
\qbezier[40](40,50)(60,70)(80,90)
\qbezier[10](50,60)(55,55)(60,50)
\qbezier[40](80,90)(100,70)(120,50)
\qbezier[30](110,60)(125,75)(140,90)
\qbezier[10](140,90)(145,85)(150,80)
\qbezier[20](150,80)(160,90)(170,100)
\qbezier[50](170,100)(195,75)(220,50)
\linethickness{0.4pt}
\multiput(20,40)(6,0){40}{\line(1,0){4}}
\put(3,40){\vector(0,-1){10}}
\put(0,44){$\scriptstyle\alpha$}
\put(55,77){\vector(0,-1){20}}
\put(49,81){$\scriptstyle \alpha^{+1}$}
\put(96,85){\vector(0,-1){10}}
\put(93,89){$\scriptstyle\alpha^{+2}$}
\put(196,85){\vector(0,-1){10}}
\put(193,89){$\scriptstyle\alpha^{+3}$}
\put(38,20){\vector(0,1){10}}
\put(35,10){$\scriptstyle\alpha^-$}
\end{picture}
\caption{Definition of the Dyck paths $\alpha^-$ and $\alpha^{+k}$, $k\geq 1$.}
\label{alphapm}
\end{figure}
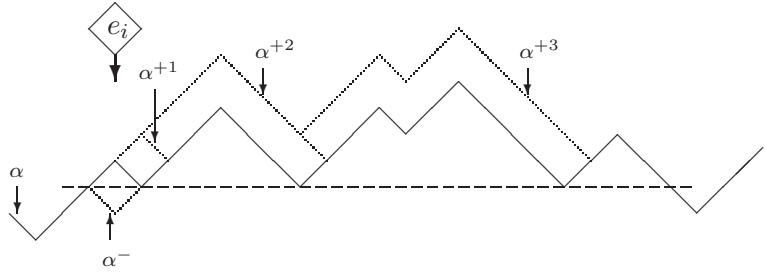

\begin{remark}\lb{important-remark}
An important observation is that the path $\alpha^-$ is absent
in Figure~\ref{alphapm} in case $\alpha_i=1$. In this case
the condition $\alpha^-_i=\alpha_i-2=-1$ implies that the path $\alpha^-$ is no
longer a Dyck path, and that consequently the term $\psi_{\alpha^-}$ in \eqref{TLei} has to vanish.
\end{remark}

A further analysis of equation \eqref{TLei} will be made in
sections \ref{sec4} and \ref{sec6}.

\subsubsection{Type B}
\label{subsec3.3.2}
The analysis of the bulk qKZ equation \eqref{qKZ2TL_bulk}
in case i) is identical to that of type A and we conclude:

\begin{cor}
\label{cor:psi0B}
The {\em base} coefficient function $\psi_{\Omega}^{\rm B}$ corresponding to the maximal Ballot path
$\Omega^{\rm B}$ (see Definition \ref{rem4})
in the solution of the qKZ equation of type B
has the following form:
\be
\lb{psi_alpha_B}
\psi_{\Omega}^{\rm B} (x_1,\ldots,x_N)\, =\, \Delta^-_1(x_1,\ldots,x_{N})\,
\Delta^+_{\lambda +1}(x_{1},\ldots,x_{N})\,\,
\xi^{\rm B}\Bigl(x_1+\tfrac{\lambda}{2},\dots ,x_N+\tfrac{\lambda}{2}\Bigr),
\ee
where $\xi^{\rm B}(x_1,\dots ,x_N)$ is an even symmetric function in all of its variables
$x_i,\;\; 1\leq i\leq N$.
\end{cor}

\begin{proof}
The path $\Omega^{\rm B}$ does not have a local maximum between $1$ and $N$,
so the result follows immediately from Lemma~\ref{lem:slope_factors}.
\end{proof}

In the sequel we use the following picture to represent
$\psi_{\Omega}^{\rm B}$:
\be
\psi_{\Omega}^{\rm B}\, =\;
\raisebox{-25pt}{
\begin{picture}(130,50)
\put(0,10){\includegraphics[width=100pt]{Delta.eps}}
\put(-2,35){\small $x_1$}
\put(17,35){\small $x_2$}
\put(53,35){\small $\dots$}
\put(96,35){\small $x_N$}
\end{picture}
}
\label{psi-B}
\ee

\noindent \textbf{Case ii)}: $\alpha$ has a local maximum at $i$.

For type B, \eqref{qKZ2TL_bulk} gives,
\be
 h_i(1)\, \psi_\alpha =
 \sum_{\genfrac{}{}{0pt}{}{\beta\neq\alpha}{e_i\beta=\alpha}} c_\beta\, \psi_{\beta}\, =\,
\psi_{\alpha^-}\,+\,c_0(i)\, \psi_{\alpha^{+0}}\, +\, \sum_{k=1,2,\dots}  \psi_{\alpha^{+k}}\, ,
\label{TLeiB}
\ee
where a pictorial definition of the paths $\alpha^-$ and $\alpha^{+k}$,
$k\geq 0$,
is given in Figure \ref{alphapm0}. For $\psi_{\alpha^-}$ and $\psi_{\alpha^{+k}}$
the coefficients $c_\beta$ are all equal to $1$,
but $c_{0}$ may be different from $1$. This coefficient is defined by the following rules:
$c_0=0$ if the path $\alpha^{+0}$ is not in the preimage of $\alpha$ under $e_i$,
that is, $\exists\, j<i :\; \alpha_j<\alpha_i-1$.
Otherwise, $c_0(i)=1$ if $i$ is odd and $c_0(i)=-[\omega]/[\omega +1]$ if $i$ is even (this follows from
the rules \eqref{eq:DownHillIIa} and \eqref{eq:DownHillIIb}). Further analysis of case ii) is
postponed to sections \ref{sec4} and \ref{sec6}.

\begin{figure}[h]
\begin{picture}(200,120)(0,0)
\put(170,105){\line(1,-1){10}}
\put(180,115){\line(1,-1){10}}
\put(170,105){\line(1,1){10}}
\put(180,95){\line(1,1){10}}
\put(177,103){$e_i$}
{\thicklines \put(180,95){\vector(0,-1){10}}
\put(0,10){\line(0,1){90}}}
\put(0,70){\line(1,-1){30}}
\put(30,40){\line(1,1){30}}
\put(60,70){\line(1,-1){10}}
\put(70,60){\line(1,1){20}}
\put(90,80){\line(1,-1){40}}
\put(130,40){\line(1,1){20}}
\put(150,60){\line(1,-1){20}}
\put(170,40){\line(1,1){10}}
\put(180,50){\line(1,-1){40}}
\qbezier[10](170,40)(175,35)(180,30)
\qbezier[10](180,30)(185,35)(190,40)
\qbezier[10](160,50)(165,55)(170,60)
\qbezier[30](150,80)(165,65)(180,50)
\qbezier[40](90,100)(110,80)(130,60)
\qbezier[20](70,80)(80,90)(90,100)
\qbezier[30](120,50)(135,65)(150,80)
\qbezier[10](60,90)(65,85)(70,80)
\qbezier[30](0,90)(15,75)(30,60)
\qbezier[40](20,50)(40,70)(60,90)
\linethickness{0.4pt}
\multiput(0,40)(6,0){38}{\line(1,0){4}}
\put(216,26){\vector(0,-1){10}}
\put(213,29){$\scriptstyle\alpha$}
\put(15,91){\vector(0,-1){15}}
\put(9,95){$\scriptstyle \alpha^{+0}$}
\put(76,100){\vector(0,-1){13}}
\put(66,104){$\scriptstyle\alpha^{+3}$}
\put(136,82){\vector(0,-1){15}}
\put(126,86){$\scriptstyle\alpha^{+2}$}
\put(172,75){\vector(0,-1){15}}
\put(158,79){$\scriptstyle\alpha^{+1}$}
\put(178,20){\vector(0,1){10}}
\put(175,10){$\scriptstyle\alpha^-$}
\end{picture}
\caption{Definition of the paths $\alpha^-$ and $\alpha^{+k}$, $k\geq 0$.}
\label{alphapm0}
\end{figure}
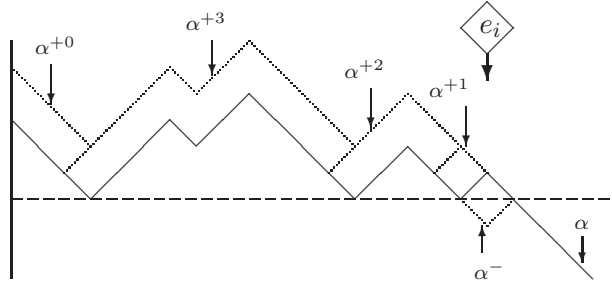

\medskip
Now consider the nontrivial type B boundary qKZ equation \eqref{qKZ2TL_bound}.
As before, the analysis breaks up into two cases, depending whether or
not a path $\alpha$ in the sum on the right-hand side  of
\eqref{qKZ2TL_bound} has a maximum at $0$, i.e. whether or not it is of the form
$\alpha=(\alpha_0,\alpha_0-1,\ldots)$. We will first look at the case
in which it does not.
\bigskip

\noindent \textbf{Case i)}: $\alpha$ does not have a local maximum at $0$.

As each term in the left-hand side  of \eqref{qKZ2TL_bound} is of the form
$e_0\ket\alpha$, and hence corresponds to a maximum at $0$, the coefficient of $\ket\alpha$
in the right-hand side  of \eqref{qKZ2TL_bound} has to equal
zero. We thus obtain
\be
(a_0 \psi_\alpha)(x_1,\ldots,x_{N}) =0,
\label{eq:hmonpsi0}
\ee
which can be rewritten as
\[
(\pi_0-1) \left\{[x_1-\tfrac{\omega + \delta}{2}][x_1-\tfrac{\omega - \delta}{2}]
\psi_\alpha(x_1,\ldots,x_{L})\right\} =0\quad\text{for}\;
\ket\alpha\not\propto e_0 \ket{\alpha'}.
\]
Hence, if $\ket\alpha\not\propto e_0 \ket{\alpha'}$ the function
\[
[x_1-\tfrac{\omega + \delta}{2}][x_1-\tfrac{\omega -
\delta}{2}] \psi_\alpha(x_1,\ldots,x_{N})
\]
is an even function in $x_1$ which implies that
$[x_1+\tfrac{\omega - \delta}{2}][x_1+\tfrac{\omega + \delta}{2}]$ divides
$\psi_\alpha(x_1,\ldots,x_{N})$, and the ratio is even in $x_1$.
\bigskip

\noindent\textbf{Case ii)}: $\alpha$ has a local maximum at $0$.

Now \eqref{qKZ2TL_bound} gives,
\be
 \bigl(\tfrac{[\omega]}{[\omega +1]} -a_0\bigr) \psi_\alpha\, =\, s_0 \psi_\alpha\, =\,
 \psi_{\alpha^{-0}}\, .
\label{TLe0}
\ee
Here by $\alpha^{-0}$ we denote the path coinciding with $\alpha$
everywhere except at the left boundary,
where one has $\alpha^{-0}_0=\alpha_0 -2$, so that $e_0\ket{\alpha^{-0}}=\ket\alpha$.
Relation \eqref{TLe0} in fact implies condition \eqref{eq:hmonpsi0}
as any path with a local maximum at $0$ is  the $\alpha^{-0}$ path for a certain
path $\alpha$.
Therefore, the type B boundary qKZ equation \eqref{qKZ2TL_bound}
is equivalent to the relation \eqref{TLe0}.

\section{Factorised solutions}
\lb{sec4}

We now present factorised formulas, in terms of the Baxterised elements $h_i(u)$, $i=0,1,\ldots,N-1$, for the
coefficients $\psi_{\alpha}$ of the solution $\ket\Psi$ to the qKZ
equation \eqref{qKZTL_TypeB1}-\eqref{qKZTL_TypeB3}.
For type A, such formulas were obtained earlier by Kirillov and Lascoux \cite{KiriL} who
considered factorisation of Kazhdan-Lusztig elements for Grassmanians.

\subsection{Type A}
\label{sec:facsolA}
The factorised expression for $\psi_\alpha$ is most easily
expressed in the following pictorial way. Complement the Dyck path
$\alpha$ with tiles to fill up the
triangle corresponding to the maximal Dyck path $\Omega^{\rm A}$,
as in Figure \ref{fig:qKZsol}.
To each added tile at horizontal position $i$ and height $j$
assign a positive integer number $u_{i,j}$
according to a following rule:
\begin{itemize}\lb{rule-A}
\item
put $u_{i,j}=1$
if in the list of added tiles there are no elements with the coordinates
$(i\pm1,j-1)$;
\item
otherwise, put
$
u_{i,j} = \max\{u_{i+1,j-1},u_{i-1,j-1}\}+1.
$
\end{itemize}
Algorithmically this rule works as follows. First, observe that the added tiles taken
together form a Young diagram $Y_\alpha$ (see Figure \ref{fig:qKZsol}).
In other words, the Young diagram $Y_\alpha$ is the difference
of the maximal Dyck path $\Omega^{\rm A}$ and the Dyck path $\alpha$.
Then, act in the following way:
\begin{itemize}
\item
Assign the integer $1$ to all corner tiles of the Young diagram $Y_\alpha$ and then
remove the corner tiles from the diagram.
\item
Assign the integer $2$ to all corner tiles of the reduced diagram and, again,
remove the filled tiles from the diagram.
\item Continue to repeat the procedure, increasing the integer by 1 at each
consecutive step, until all tiles are removed.
\end{itemize}
Once the assignment of integers is done, define an
ordered product of operators $h_i(u_{ij})$
\be
\lb{H_a}
H_\alpha\, :=\, \prod_{i,j}^{\nearrow u} h_i(u_{ij})\, ,
\ee
where the product is taken over all added tiles and the factors of the product
are ordered in such a way that their arguments $u_{i,j}$
do not decrease when reading from left to right (note that factors with the same argument
commute).

\begin{figure}[h]
\centerline{
\begin{picture}(330,180)
\put(0,0){\includegraphics[width=330pt]{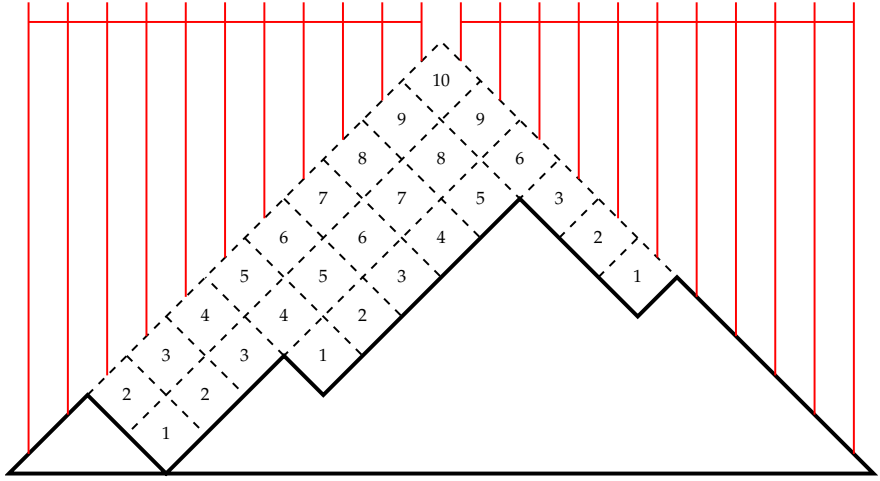}}
\end{picture}
}
\caption{Solution of the qKZ equations of type A.
We use the graphical notation \eqref{h_pic} and \eqref{psi-D} to represent
the Baxterised elements $h_i(k)$, $k=1,2,\dots$, and the coefficient
$\psi_{\Omega}^{\rm A}$. The associated Young diagram $Y_\alpha$ corresponds to the partition
$\{9^2,6,1^3\}$.}
\label{fig:qKZsol}
\end{figure}

\begin{theorem}
\label{th:facsol}
Let $\alpha$ be  a Dyck path of length $N$.
The corresponding coefficient function $\psi_\alpha$ in the
solution of the qKZ equations \eqref{qKZTL_TypeB1}-\eqref{qKZTL_TypeB3}
of type A  is given by
\be
\lb{psi_a-A}
\psi_{\alpha} = H_\alpha \psi_{\Omega}^{\rm A},
\ee
where $\psi_{\Omega}^{\rm A}$ is the base  function corresponding to the maximal Dyck path $\Omega^{\rm A}$
of length $N$ and the factorised operator $H_\alpha$ is defined in (\ref{H_a})
(see also Figure \ref{fig:qKZsol}).
\end{theorem}

The proof of Theorem~\ref{th:facsol} is given in Section~\ref{sec:ProofFacSolA}.
\medskip

\subsubsection{Truncation conditions}
\label{subsec4.1.1}

For $k=1,2,\dots ,\lfloor N/2\rfloor$, let $\beta (k)$ denote the  path of length $N$ which
has only one minimum, occuring at the point $2k-1$, with ${\beta(k)}_{2k-1}=-1$. Note that
$\beta(k)$ is therefore not a Dyck path. The associated Young diagram $Y_{\beta{(k)}}$ is a
$(n-k+1) \times k$ rectangle, $n=\lfloor \tfrac{N+1}{2}\rfloor$. We introduce notation $H^{\rm A}_{k}:=H_{\beta{(k)}}$
for the corresponding factorised operator. An example of $\beta(k)$ and its corresponding
operator $H^{\rm A}_k$ is given
in Figure~\ref{betapath} for $k=3$ and $N=12$.

\begin{figure}[h]
\begin{picture}(0,0)
\qbezier[6](-74,-95)(-74,-92)(-74,-89)
\qbezier[6](-12.5,-108)(-12.5,-105)(-12.5,-102)
\qbezier[6](75,-95)(75,-92)(75,-89)
\put(-77,-103){$\scriptstyle\mathit 0$}
\put(-10,-111){$\scriptstyle\mathit 2k-\mathit 1=\mathit 5$}
\put(56,-103){$\scriptstyle N=\mathit 1\mathit 2$}
\end{picture}
\centerline{\includegraphics[width=150pt]{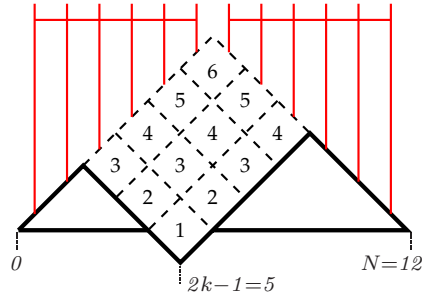}}
\vspace{-2pt}
\caption{The path $\beta(3)$ of the length $N=12$
is drawn in bold. The corresponding factorised operator $H^{\rm A}_3$
is an ordered product of the Baxterised elements $h_i(u_{ij})$ represented by dashed tiles on the
picture. The additional conditions \eqref{alpha_0-add} require that this picture vanishes.}
\label{betapath}
\end{figure}

\begin{prop}
The base coefficient function $\psi_{\Omega}^{\rm A}$
for the solution of the type A  qKZ equation is subject to truncation relations
\be
\lb{alpha_0-add}
H^{\rm A}_{k}\, \psi_{\Omega}^{\rm A} = 0 \quad \forall\, k=1,2,\dots , \lfloor N/2\rfloor.
\ee
Conditions \eqref{psi_alpha_D}, \eqref{H_a}, \eqref{psi_a-A}, \eqref{alpha_0-add}
together are equivalent to the qKZ equations
\eqref{qKZTL_TypeB1}-\eqref{qKZTL_TypeB3} of type A.
\end{prop}

\begin{proof}
Using the techniques described in the proof of Theorem~\ref{th:facsol}, see

Section~\ref{sec:ProofFacSolA}, it follows that the relations \eqref{alpha_0-add} ensure the
vanishing of the contributions $\psi_{\alpha^-}$ in \eqref{TLei} if $\alpha^-$ is
not a Dyck path, see Remark \ref{important-remark}. Explicit
examples are given in subsection \ref{separate}.
\end{proof}

For a particular value of the boundary parameter
$\lambda$ a simple polynomial solution of the conditions
\eqref{alpha_0-add} was found in \cite{DF05}:
\begin{prop}
\lb{prop4}
For $\lambda = -3$, the conditions \eqref{alpha_0-add} admit the simple solution $\xi^{\rm A}
=1$. In this case the coeficients of the qKZ equation, when properly normalised, are
polynomials in variables $z_i = q^{x_i}, \;\; i=1,\dots ,N$.
\end{prop}

Using the factorised formulas we have calculated these solutions for system sizes up to $N=10$.
In the homogeneous limit $x_i \rightarrow 0$ the coefficients
$\psi_\alpha(x_1,\ldots,x_N)$ become polynomials in $\tau=-[2]$. In fact, up to an overall
factor, each $\psi_\alpha$ becomes a polynomial in $\tau^2$ with positive integer
coefficients. These polynomials were considered in \cite{DF07}, where their intriguing
combinatorial content was described. In Appendix~\ref{se:solutions} we present a table of
these polynomials up to $N=10$, and we shall further discuss them in Section 5.
\medskip

\subsection{Type B}\lb{sec4.2}

We will now formulate a factorised solution for type B. This
result was deduced from some exercises we made for small size systems.
As we found these instructive, we have given these in Appendix~\ref{facsolB}.
Analysing the expressions for the coefficient functions in the cases $N=2,3$
we see that in order to write them as a product of the Baxterised elements $h_i(u)$, $i\geq 0$,
we have to fix in a special way the auxiliary parameter
$\nu$ in the definition \eqref{rep-s0} of $h_0(u)$. From now on we therefore specify
\be
\lb{half-tile}
\raisebox{-18pt}{\includegraphics[height=36pt]{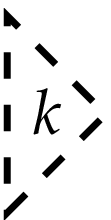}}\, =\, \bar h_0(k)\, :=\,
h_0(k)|_{\nu=\omega+p_k}\, =\,  s_0 -
\frac{\left[\lfloor{k/2}\rfloor\right]
\left[\omega +\lfloor{(k+1)/2}\rfloor\right]}{[k][\omega+1]},
\ee
where $p_k = k \mbox{~mod 2}$. Note that the Baxterised boundary element ${\bar h}_0(u)$ is defined for integer values of its spectral parameter $u\in {\mathbb Z}$ as only such values appear in our considerations.

Now we can repeat the procedure described in the beginning of Section
\ref{sec:facsolA} but with the set of Ballot paths instead of
Dyck paths. For each Ballot path
$\alpha$ we consider its complement to the maximal Ballot path $\Omega^{\rm B}$.
This complement may be thought of as half a transpose symmetric Young diagram, cut along its symmetry
axis. We  fill the complement with (half-)tiles  corresponding to the
(boundary) Baxterised elements $h_i(u_{i,j})$ and ${\bar h}_0(u_{i,j})$
as shown in Figure \ref{fig:qKZsolB}.
\begin{figure}[h]
\centerline{
\begin{picture}(180,190)
\put(0,0){\includegraphics[width=180pt]{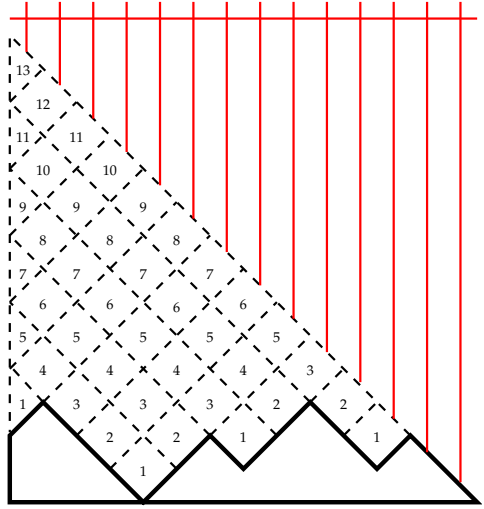}}
\end{picture}
}
\caption{Solution of the qKZ equation of type B.
We use the graphical notation \eqref{h_pic}, \eqref{half-tile} and
\eqref{psi-B}  to represent
the Baxterised elements $h_i(k)$,
the boundary Baxterised elements ${\bar h}_0(k)$
and the  coefficient
$\psi_{\Omega}^{\rm B}$.
}
\label{fig:qKZsolB}
\end{figure}
The rule for assigning integers $u_{i,j}$ to the tiles remains exactly the same as for type A. For the
half-tiles the rule is:
\begin{itemize}
\item put $u_{0,j}=1$ if there is no adjacent tile $h_1$ with the coordinate $(1,j-1)$;
\item otherwise, put  $u_{0,j}=u_{1,j-1}+1$.
\end{itemize}

The operator $H_\alpha$ corresponding to the Ballot path $\alpha$
is given by the same formula \eqref{H_a} as for type A, where now index $i$ may
take also value $0$, thus allowing the boundary operators ${\bar h}_0(u_{0,j})$
enter the product.

\begin{theorem}

\label{th:facsolB}
Let $\alpha$ be a Ballot path of length $N$.
The corresponding coefficient $\psi_\alpha$ of the
solution of type B qKZ equation is given by
\be
\lb{psi_a-B}
\psi_{\alpha} = H_\alpha \psi_{\Omega}^{\rm B},
\ee
where $\psi_{\Omega}^{\rm B}$ is the base function corresponding to
 the maximal Ballot path $\Omega^{\rm B}$
of length $N$ and the factorised operator $H_\alpha$ is defined in (\ref{H_a}), see also Figure \ref{fig:qKZsolB}.
\end{theorem}

The proof of Theorem~\ref{th:facsolB} is given in Section~\ref{sec:ProofFacSolB}.

\subsubsection{Truncation conditions}
For $k=1,\dots ,n=\lfloor \tfrac{N+1}{2}\rfloor$, let $\gamma{(k)}$ denote the path of the
length $N$ with only one minimum, occuring at
the point
$\, 2k-\epsilon_N-1\,$ with ${\gamma(k)}_{2k-\epsilon_N-1}=-1\,$
(recall that $\epsilon_N= N\bmod 2$). Note that $\gamma(k)$ is therefore not a Ballot path. The
associated half-Young diagram has a shape of trapezium.
We introduce the notation $H^{\rm B}_{k}:=H_{\gamma{(k)}}$ for the corresponding
factorised operator. An example of a path $\gamma{(k)}$ with $N=10$ and $k=3$ is
shown in Figure~\ref{gammapath}.

\begin{figure}[h]
\begin{picture}(0,0)
\qbezier[6](-74,-179)(-74,-176)(-74,-173)
\qbezier[6](0,-194)(0,-191)(0,-188)
\qbezier[6](75,-179)(75,-176)(75,-173)
\put(-77,-190){$\scriptstyle\mathit 0$}
\put(2.5,-194){$\scriptstyle\mathit 2k-\mathit 1=\mathit 5$}
\put(56,-190){$\scriptstyle N=\mathit 1\mathit 0$}
\end{picture}
\centerline{\includegraphics[width=150pt]{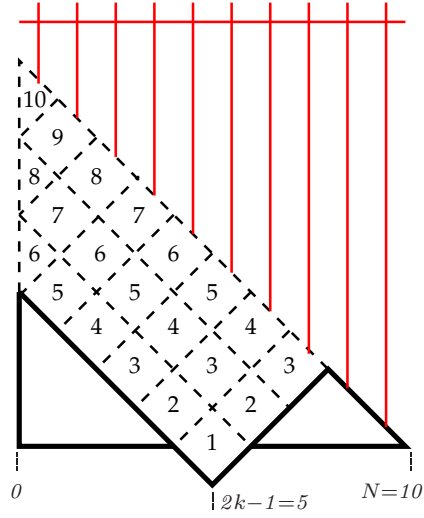}}
\vspace{-2pt}
\caption{The path $\gamma(3)$ of the length $N=10$
is drawn in bold. The corresponding factorised operator $H^{\rm B}_3$
is an ordered product of the Baxterised elements $h_i(u_{ij})$ represented by dashed tiles on the
picture. The additional conditions \eqref{alpha_0-addB} require that this picture vanishes.}
\label{gammapath}
\end{figure}

\begin{prop}
The base coefficient function $\psi_{\Omega}^{\rm B}$
for the solution of the type B qKZ equation satisfies the truncation conditions
\be
\lb{alpha_0-addB}
H^{\rm B}_{k} \,\psi_{\Omega}^{\rm B}\, =\, 0 \quad \forall\, k=1,2,\dots , n.
\ee
Conditions \eqref{psi_alpha_B}, \eqref{psi_a-B} and \eqref{alpha_0-addB}
together are equivalent to the  qKZ equations
\eqref{qKZTL_TypeB1}-\eqref{qKZTL_TypeB3} of type B.
\end{prop}

\begin{proof}
Just as for type A, the relations \eqref{alpha_0-addB} ensure
the absence of $\psi_{\alpha^-}$ in \eqref{TLeiB} if $\alpha^-$ is not a Ballot path.
Explicit examples are considered
in  Appendix \ref{facsolB}.
\end{proof}

For particular values of the boundary parameter
$\lambda$ and the algebra parameter $\omega$
the simplest polynomial solution of the conditions
\eqref{alpha_0-addB} was found in \cite{{ZJ07}}:
\begin{prop}
\lb{prop6}
In case $\lambda = -3/2$ and $\omega =-1/2$ the conditions \eqref{alpha_0-addB}
admit the simple solution $\xi^{\rm B} =1$. In this case the coeficients of the qKZ equation,
when properly normalised,
become polynomials in variables $z_i = q^{x_i}, \;\; i=1,\dots ,N$.
\end{prop}

\subsection{Separation of truncation conditions}
\label{separate}
\subsubsection{Type A}
The equations \eqref{alpha_0-add} impose restrictions on the otherwise arbitrary
symmetric functions $\xi^{\rm A}$ of the ansatz \eqref{psi_alpha_D}.
Based on experience with calculations for small $N$ we observe that
these restrictions can be written in a more explicit way.
Namely, one can separate the functional part (depending on variables $x_i$)
and the permutation part (depending on the permutations $\pi_i$)
in the operators $H^{\rm A}_{k}$ in \eqref{alpha_0-add}.

To formulate this separation in a precise way, let us first define the following set of
Baxterised elements in the group algebra of the symmetric group ${\mathbb C}[S_{N}]\simeq \H^{\rm A}_{N}(1)$:
\[
\pi_i(u) := u^{-1} - \pi_i \quad i=1,\dots ,N-1.
\]
Denote furthermore by $\Pi^{\rm A}_{k}$ the permutation operators obtained from the
operators $H^{\rm A}_{k}$  by substituting
$h_i(u_{i,j})\mapsto \pi_i(u_{i,j})$.

\begin{conj}\lb{prop3}
The left hand side of the truncation condition \eqref{alpha_0-add} for Type A can be written in the
following separated form (remind that $n =\lfloor \tfrac{N+1}{2}\rfloor $)
\be
\lb{alpha_0-add-red}
H^{\rm A}_{k} \psi_{\Omega}^{\rm A} = \frac{\Delta^-_1(x_{k},\dots ,
x_{n+k})}{\Delta^-_0(x_{k},\dots ,
x_{n+k})}\, \Pi^{\rm A}_{k}\,
\frac{\Delta^-_0(x_{k},\dots ,
x_{n})\Delta^-_0(x_{n+1},\dots ,
x_{n +k})}{\Delta^-_1(x_{k},\dots ,x_{n})\Delta^-_1(x_{n+1},\dots ,
x_{n+k})}\,\psi_{\Omega}^{\rm A}.
\ee
\end{conj}

\begin{example}{\rm
In the case $N=5$
there are two truncation conditions on the base function $\psi_{\Omega}^{\rm A}$:
\begin{itemize}
\item{Condition \eqref{alpha_0-add-red} for $k=1$ reads}
\[
h_1(1) h_2(2) h_3(3) \psi_{\Omega}^{\rm A} =
\frac{\Delta^-_1(x_1,\dots ,x_4)}{\Delta^-_0(x_1,\dots ,x_4)}\, \Pi^{\rm A}_{1}
\frac{\Delta^-_0(x_1,x_2,x_3)}{\Delta^-_1(x_1,x_2,x_3)}\psi_{\Omega}^{\rm A} = 0,
\]
or, in terms of $\xi^{\rm A}$
\ba
\nonumber
\Pi_{1}^{\rm A}\Bigl\{\Delta^-_0(x_1,x_2,x_3) \Delta^+_{-1}(x_1,x_2,x_3)
\Delta^-_1(x_4,x_5)\hspace{47mm}
\\[-2mm]
\label{eshe}\hspace{50mm}
\Delta^+_{\lambda +1}(x_4,x_5)
\xi^{\rm A}(x_1,x_2,x_3|x_4+\tfrac{\lambda}{2},x_5+\tfrac{\lambda}{2})\Bigr\}\, =\, 0,
\ea
where $\Pi^{\rm A}_{1}:= (1-\pi_1)(\tfrac{1}{2}-\pi_2)(\tfrac{1}{3}-\pi_3)$. Since the function in braces in
\eqref{eshe} is antisymmetric in the variables $x_1, x_2, x_3,$
the operator $\Pi_1^{\rm A}$
in this formula can be equivalently substituted by
\be
\lb{ex3}
\Pi^{\rm A}_{1}\; \mapsto \; (1-\pi_3+\pi_2\pi_3-\pi_1\pi_2\pi_3).
\ee
\item{Condition \eqref{alpha_0-add-red} for $k=2$ reads}
\[
\hspace{6mm}
h_3(1) h_2(2) h_4(2) h_3(3) \psi_{\Omega}^{\rm A} =
\frac{\Delta^-_1(x_2,\dots ,x_5)}{\Delta^-_0(x_2,\dots ,x_5)}\, \Pi^{\rm A}_{2}
\frac{\Delta^-_0(x_2,x_3)\Delta^-_0(x_4,x_5)}{\Delta^-_1(x_2,x_3)\Delta^-_1 (x_4,x_5)}\psi_{\Omega}^{\rm A} = 0.
\]
or, in terms of $\xi^{\rm A}$
\ba
\nonumber
\Pi_{2}^{\rm A}\Bigl\{\Delta^-_0(x_1+1, x_2,x_3) \Delta^+_{-1}(x_1,x_2,x_3)
\Delta^-_0(x_4,x_5)\hspace{47mm}
\\[-2mm]
\nonumber\hspace{62mm}
\Delta^+_{\lambda +1}(x_4,x_5)\xi^{\rm A}(x_1,x_2,x_3|x_4+\tfrac{\lambda}{2},x_5+\tfrac{\lambda}{2})\Bigr\}
\, =\, 0.
\ea
The operator
$
\Pi^{\rm A}_{2}:= (1-\pi_3)(\tfrac{1}{2}-\pi_2)(\tfrac{1}{2}-\pi_4)(\tfrac{1}{3}-\pi_3)
$
in this formula can be equivalently substituted by
\be
\lb{ex4}
\Pi^{\rm A}_{2}\; \mapsto\;
(1-\pi_3+\pi_2\pi_3+\pi_4\pi_3-\pi_2\pi_4\pi_3+\pi_3\pi_2\pi_4\pi_3).
\ee
\end{itemize}
Formulas \eqref{ex3}, \eqref{ex4} suggest the following proposition
\begin{prop}\lb{prop-another}
In condition \eqref{alpha_0-add-red}
one can substitute the operator $\Pi^{\rm A}_{k}$
by a polynomial in the permutations $(-\pi_i)$, $i=k,\dots ,n+k-1$, with unit coefficients.
The terms of the polynomial are labeled by the sub-diagrams of the
rectangular Young diagram $\{k^{(n-k+1)}\}$ corresponding to the path $\beta(k)$, see subsection \ref{subsec4.1.1}.
Their form is given by formula  \eqref{H_a}, where one has to substitute the factors $h_i(u_{ij})$
by $-\pi_i$.
\end{prop}
}
\end{example}

\subsubsection{Type B}
In this case we found analogues of the expressions \eqref{alpha_0-add-red} for particular
truncation conditions \eqref{alpha_0-addB} only.

Let us supplement the set of Baxterised elements
$\pi_i(u)$ with the boundary Baxterised element
\[
\pi_0(u) := u^{-1} - \pi_0.
\]
The elements $\pi_i(u)$, $i=0,1$ satisfy a reflection equation of the form
\eqref{rea_h}.\smallskip

Denote by $\Pi^{\rm B}_{n}$
the operator obtained from  $H^{\rm B}_{n}$  by the substitutions
\[
{\bar h}_0(u_{i,j})\mapsto \pi_0(1),\quad
h_i(u_{i,j})\mapsto \pi_i(u_{i,j})\;\;\forall\,i\geq 1 .
\]

\begin{conj}\lb{conj2b}
The left hand side of the condition  \eqref{alpha_0-addB} for $k=n$ and arbitrary $N$ can be transformed to
\ba
\nonumber
H^{\rm B}_{n}\, \psi_{\Omega}^{\rm B}\, =\, \Bigl\{\frac{\Theta(x_1,\dots ,x_N)\Delta^-_1(x_1,\dots ,x_N)
}{\Delta^-_0(x_1,\dots ,x_N)}\,
\Pi^{\rm B}_{n}\,
\frac{\Delta^-_0(x_{2},\dots ,x_{N})}{\Theta(x_2,\dots ,x_N)
\Delta^-_1(x_{2},\dots ,x_{N})}\, -\,
\\[2mm]
\frac{[\lfloor N/2\rfloor][\omega+n]}{[\omega+1]}
\Bigr\}\psi_{\Omega}^{\rm B}\, ,\hspace{50mm}
\nonumber
\ea
where $n =\lfloor\tfrac{N+1}{2}\rfloor $ and
$\Theta(x_i,\ldots,x_j) := \prod_{i\leq p\leq j} \frac{k(x_p ,\,\delta)}{[2x_p][\omega+1]}\,$.
\end{conj}

For $N$ odd denote by $\Pi^{\rm B}_{1}$ the operator obtained from
$H^{\rm B}_{1}$  by the substitutions
\[
{\bar h}_0(u_{i,j})\mapsto \pi_0(u_{i,j}),\quad
h_i(u_{i,j})\mapsto \pi_i(u_{i,j})\;\;\forall\,i\geq 1 .
\]

\begin{conj}\lb{conj2}
The left hand side of the condition \eqref{alpha_0-addB} for $k=1$ and $N$ odd can be written
in the following separated form
\be
\lb{add-redB1}
H^{\rm B}_{1} \psi_{\Omega}^{\rm B} = \Theta(x_1,\dots ,x_n)\frac{\Delta^-_1(x_1,\dots ,x_n)
\Delta^+_{-1}(x_1,\dots ,x_n)}{\Delta^-_0(x_1,\dots ,x_n)\Delta^+_0(x_1,\dots ,x_n)}\,
\Pi^{\rm B}_{1}\,
\frac{\Delta^{-}_0 (x_{1},\dots ,x_{n})}{\Delta^-_1(x_{1},\dots ,x_{n})}\,\psi_{\Omega}^{\rm B}.
\ee
\end{conj}

\begin{remark}
Notice that condition \eqref{add-redB1} does not actually depend on the algebra
parameter $\omega$. If we choose $\xi^{\rm B}=1$, then it is satisfied for
$\lambda \in \{-3/2,-2\}$ only.
\end{remark}

\section{Observations and conjectures}
\lb{sec5}

In this section we consider explicit polynomial solutions of the qKZ equation described in Propositions~\ref{prop4} and \ref{prop6}, and we consider the homogeneous limit $x_i\rightarrow 0$. We would like to emphasize the importance of these explicit solutions for experimentation and for the discovery of many interesting new results. In this section we formulate some of these observations. We present new positivity conjectures and relate partial sums over components to single components for larger system sizes. Furthermore, based on our results we have been able to find a compact expression for generalised partial sums in the inhomogeneous case.

\subsection{Type A}
In Appendix~\ref{se:solutions} we have listed the solutions described in Proposition~\ref{prop4} up
to $N=10$ in the limit $x_i\rightarrow 0$, $i=1,\ldots,N$. These solutions were obtained using the factorised forms of the previous section. 

In the following we will write shorthand $\psi_\alpha$ for the limit $x_i\rightarrow 0$ of $\psi_\alpha(x_1,\ldots,x_N)$.  The complete solution is determined up to an overall normalisation. We will choose $\xi^{\rm A}=(-1)^{n(n-1)/2}$ where $n=\lfloor (N+1)/2\rfloor$ for which we have
\[
\psi_{\Omega}^{\rm A} = \tau^{\lfloor {N/2}\rfloor (\lfloor N/2\rfloor -1)/2}.
\]
An immediate observation was already noted in \cite{DF07}:
\begin{obs}
The components $\psi_{\alpha}(x_1,\ldots,x_N)$ of the polynomial solution of the qKZ equation of type A in the limit $x_i\rightarrow 0$, $i=1,\ldots,N$ are, up to an overall factor which is a power of $\tau$, polynomials in $\tau^2$ with positive integer coefficients. Here $\tau=-[2]$.
\end{obs}

We now conjecture a partial combinatorial interpretation, by considering certain natural sums over subsets of Dyck paths. Let us first define the paths $\Omega(N,p) \in \mathcal{D}_N$ whose local minima lie on the height $\tilde{p}$, where
\[
\tilde{p}=\lfloor{(N-1)/2}\rfloor-p.
\]
Figure~\ref{fig:omegan} illustrates the path $\Omega(12,3)$.

\begin{figure}[h]
\centerline{\includegraphics[width=200pt]{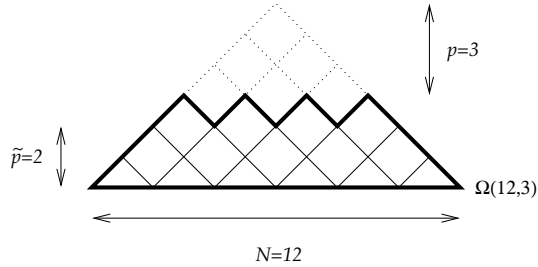}}
\caption{The minimal path $\Omega(12,3)\in \mathcal{D}_{12,3}$ .}
\label{fig:omegan}
\end{figure}

\noindent
For later convenience we also define, in the case of odd $N$, the paths $\widetilde{\Omega}(N,p) \in \mathcal{D}_N$ whose first $p-1$ local minima lie on the height $\tilde{p}$, except for the last minimum which lies at height 0. Figure~\ref{fig:omegat_n} illustrates the path $\widetilde{\Omega}(13,4)$.

\begin{figure}[h]
\centerline{\includegraphics[width=200pt]{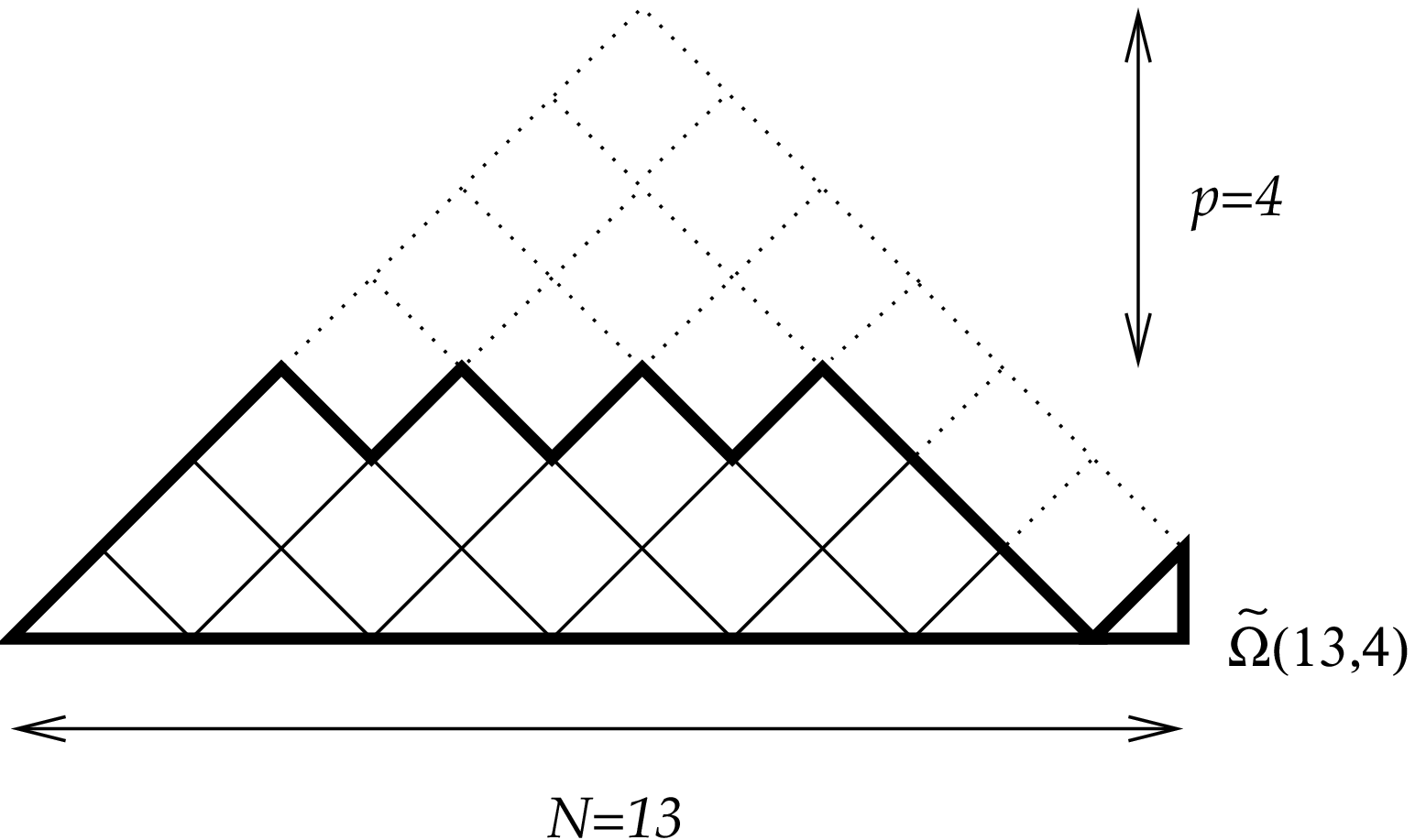}}
\caption{The path $\widetilde{\Omega}(13,4) \in \mathcal{D}_{13,6}$ .}
\label{fig:omegat_n}
\end{figure}

We define the subset $\mathcal{D}_{N,p}$ of Dyck paths of length $N$ which lie above $\Omega(N,p)$, i.e. whose local minima lie on or above height $\tilde{p}$. Formally, this subset is described as
\[
\mathcal{D}_{N,p} = \left\{ \alpha\in\mathcal{D}_N |\ \alpha_i\geq \Omega_i(N,p) = \min(\Omega_i, \tilde{p}) \right\},
\]
where $\Omega_i = \min\{i,N+\epsilon_N-i\}$, $\epsilon_N=N \bmod 2$, are integer heights of
the maximal Dyck path $\Omega^{\rm A} = \Omega(N,0)$. We further define an integer $c_{\alpha,p}$ associated to each Dyck path in the following way (see also Appendix~\ref{se:solutions}). Let $\alpha=(\alpha_0,\alpha_1,\ldots,\alpha_N)\in \mathcal{D}_{N,p}$ be a Dyck path of length $N$ whose minima lie on or above height $\tilde{p}$. Then $c_{\alpha,p}$ is defined as the signed sum of boxes between $\alpha$ and $\Omega(N,p)$, where the boxes at height $\tilde{p}+h$ are assigned $(-1)^{h-1}$ for $h\geq 1$. An example is given in Figure~\ref{fig:cpmdef}, and an explicit expression for $c_{\alpha,p}$ is given by
\[
c_{\alpha,p} = \frac{(-1)^{\tilde{p}}}{2} \left(\sum_{i=1}^{\lfloor N/2\rfloor} (\alpha_{2i}-\Omega_{2i}(N,p))
 -\sum_{i=0}^{\lfloor (N-1)/2\rfloor} (\alpha_{2i+1}-\Omega_{2i+1}(N,p))  \right).
\]

\begin{figure}[h]
\centerline{\includegraphics[width=200pt]{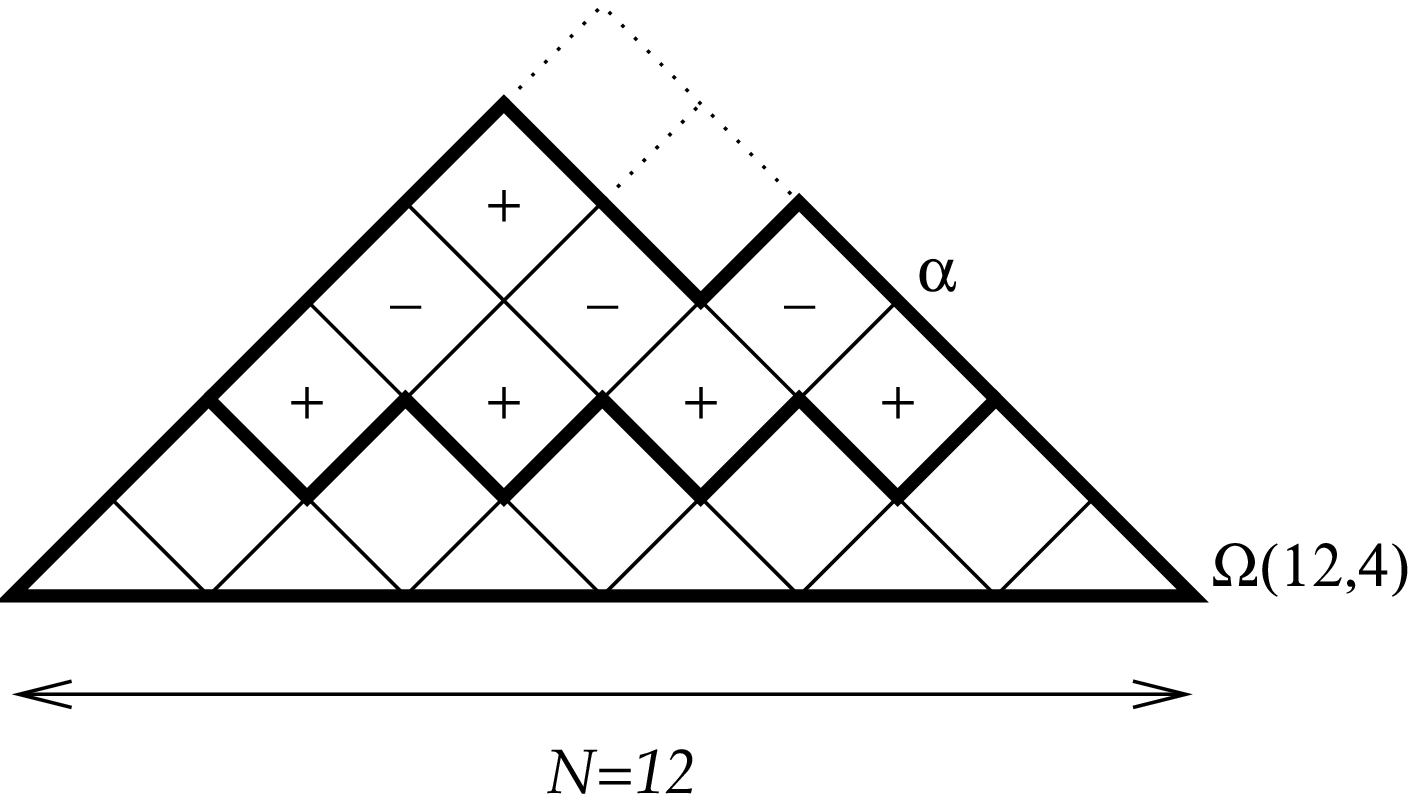}}
\caption{Definition of the number $c_{\alpha,p}$ as the signed sum of boxes between the $\alpha$ and the path $\widetilde{\Omega}(12,4)$. In this figure $N=12$ and $p=4$ and $c_{\alpha,4}=4-3+1=2$.}
\label{fig:cpmdef}
\end{figure}

Consider the partial weighted sums
\be
S_{\pm}(N,p) = \sum_{\alpha\in\mathcal{D}_{N,p}} \tau^{\pm c_{\alpha,p}}\, \psi_\alpha,
\label{eq:Spmdef}.
\ee
It was noted in \cite{Pavel,MNGB} that for $\tau=1$ ($q=\e^{2\pi\i/3}$), these partial sums for system size $N$, correspond to certain individual elements for size $N+1$. Here we observe that this
correspondence holds also for arbitrary $\tau$: the partial sums $S_\pm(N,p)$ are up to an overall normalisation proportional to certain individual components $\psi_{\alpha}$ of the
solution for system size $N+1$:
\[
\renewcommand{\arraystretch}{2}
\begin{array}{|l|l|}
\hline
S_+(4,1) =\psi_{\includegraphics[width=24pt]{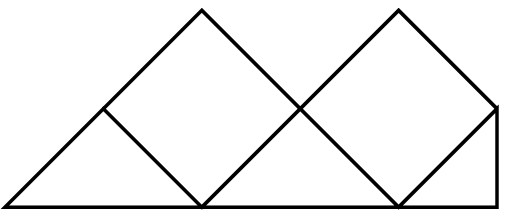}} &
S_-(4,1) =\tau^{-2} \psi_{\includegraphics[width=24pt]{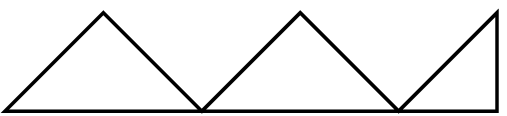}} \\ \hline
& S_-(5,2) =\tau^{-2} \psi_{\,\includegraphics[width=28pt,viewport=0 0 180 40,clip=true]{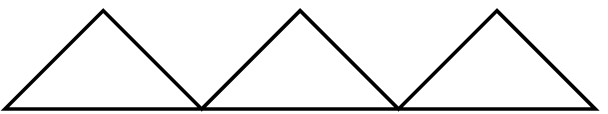}} \\
& S_-(5,1) =\tau^{-2} \psi_{\includegraphics[width=28pt,viewport=0 -10 180 60,clip=true]{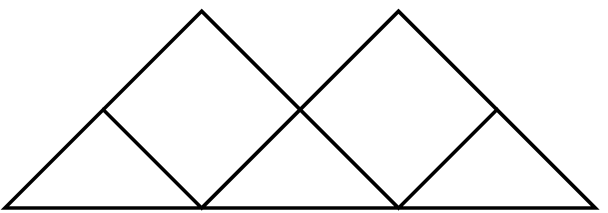}} \\
\hline
S_+(6,2) = \psi_{\includegraphics[width=32pt,viewport=0 -10 210 60,clip=true]{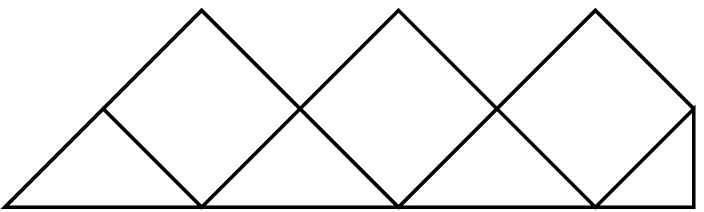}} &
S_-(6,2) =\tau^{-3} \psi_{\includegraphics[width=32pt,viewport=0 -10 210 40,clip=true]{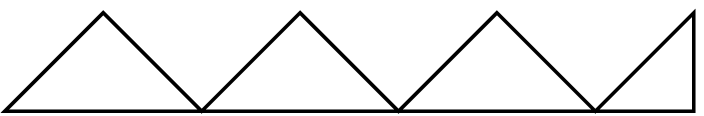}}\\
S_+(6,1) = \psi_{\includegraphics[width=32pt,viewport=0 -10 210 90,clip=true]{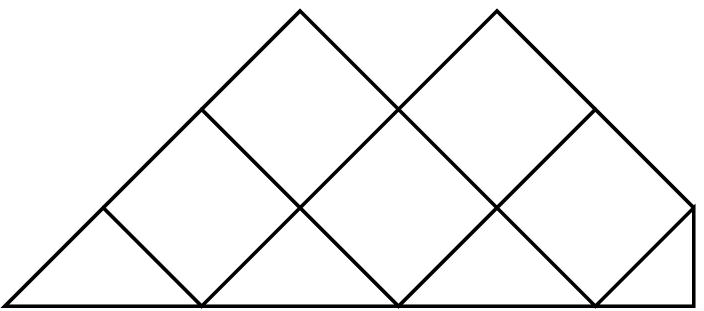}} &
S_-(6,1) =\tau^{-3} \psi_{\includegraphics[width=32pt,viewport=0 -10 210 60,clip=true]{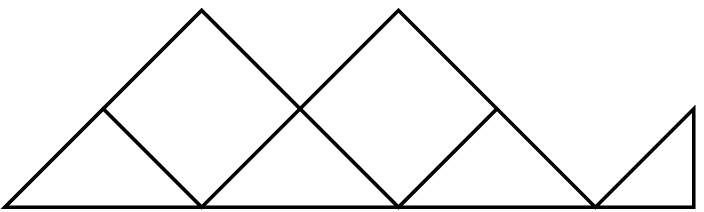}}\\
\hline
& S_-(7,3) =\tau^{-3} \psi_{\includegraphics[width=36pt,viewport=0 -10 240 40,clip=true]{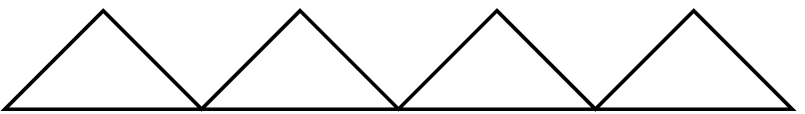}} \\
& S_-(7,2) =\tau^{-3} \psi_{\includegraphics[width=36pt,viewport=0 -10 240 60,clip=true]{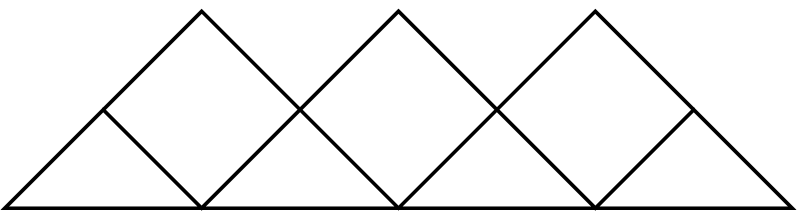}}\\
& S_-(7,1) =\tau^{-3} \psi_{\includegraphics[width=36pt,viewport=0 -10 240 100,clip=true]{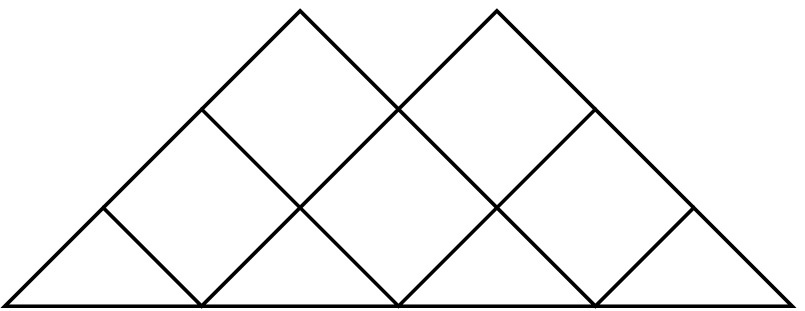}} \\
\hline
S_+(8,3) = \psi_{\includegraphics[width=40pt,viewport=0 -10 270 60,clip=true]{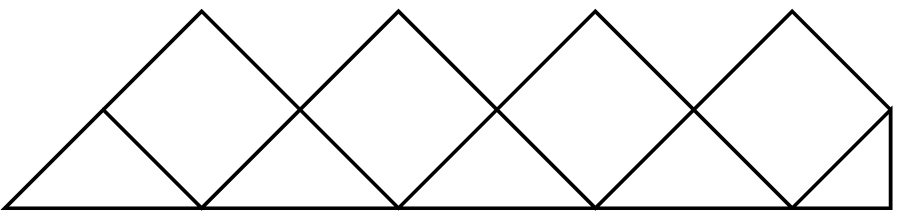}} &
S_-(8,3) = \tau^{-4} \psi_{\includegraphics[width=40pt,viewport=0 -10 270 40,clip=true]{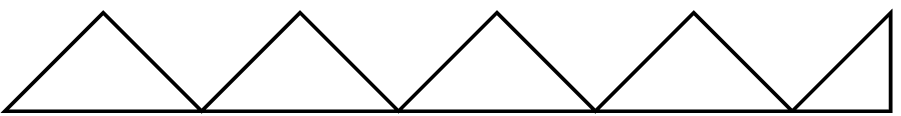}}\\
S_+(8,2) = \psi_{\includegraphics[width=40pt,viewport=0 -10 270 90,clip=true]{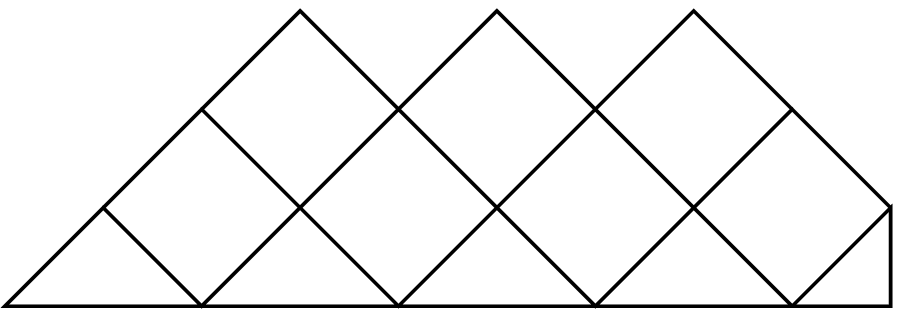}} &
S_-(8,2) = \tau^{-4} \psi_{\includegraphics[width=40pt,viewport=0 -10 270 60,clip=true]{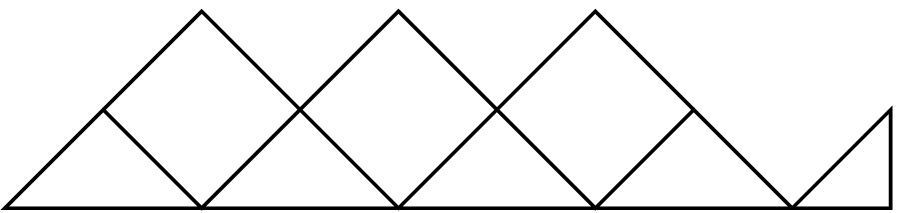}}\\
S_+(8,1) = \psi_{\includegraphics[width=40pt,viewport=0 -10 270 120,clip=true]{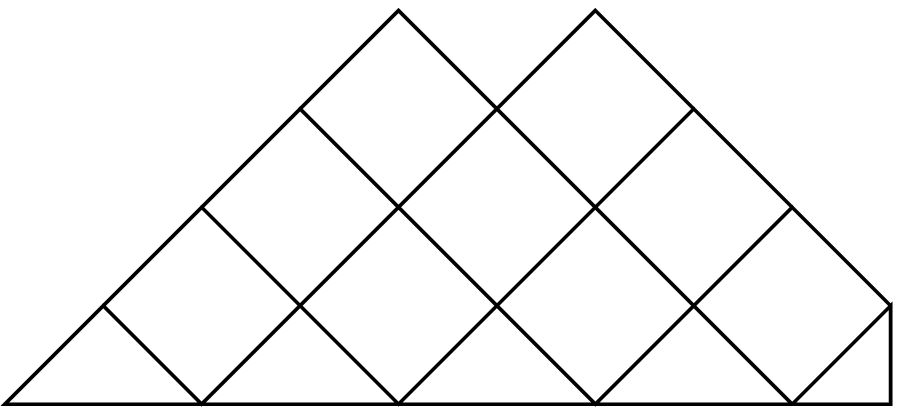}} &
S_-(8,1) = \tau^{-4} \psi_{\includegraphics[width=40pt,viewport=0 -10 270 100,clip=true]{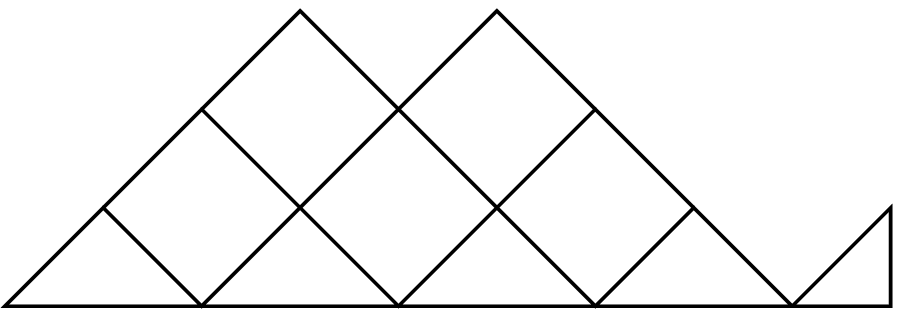}}\\
\hline
\end{array}
\]

\bigskip\noindent
We formalise this observation in the following conjecture:
\begin{obs}
\label{obs:sumsA}
\begin{align*}
S_+(N,p) &= \psi_{\Omega(N+1,p)},\quad N\ {\rm even}\\
S_{-}(N,p) &= \tau^{-N/2}\, \psi_{\widetilde{\Omega}(N+1,p+1)},\quad N\ {\rm even}\\\
S_{-}(N,p) &= \tau^{-(N-1)/2}\, \psi_{\Omega(N+1,p)},\quad N\ {\rm odd}.
\end{align*}
\end{obs}
\bigskip

The weigthed partial sums were defined in \eqref{eq:Spmdef} in an ad-hoc way. This was the way they were discovered when searching for relations as in Observation~\ref{obs:sumsA}. In fact, these partial sums arise in a natural way as we will show now:

\begin{obs}
The partial sums $S_{\pm}(N,p)$ are obtained from factorised expressions. In particular, let
\begin{align*}
P_N^- &= \prod_{j=1}^p \prod_{i=j-1}^{p-1} h_{\tilde{p}+2i+j}(1+j),\\
P_N^+ &= \prod_{j=1}^p \prod_{i=j-1}^{p-1} h_{\tilde{p}+2i+j}(-1),
\end{align*}
where the product is ordered as in Figures~\ref{Spm} and ~\ref{Smp}. Then we have
\be
S_{\pm}(N,p) = \lim_{x_i \rightarrow 0}  P_N^{\pm} \,\psi_{\Omega}^{\rm A}.
\label{eq:Sfact1}
\ee
\end{obs}

\begin{figure}[h]
\centerline{\includegraphics[width=150pt]{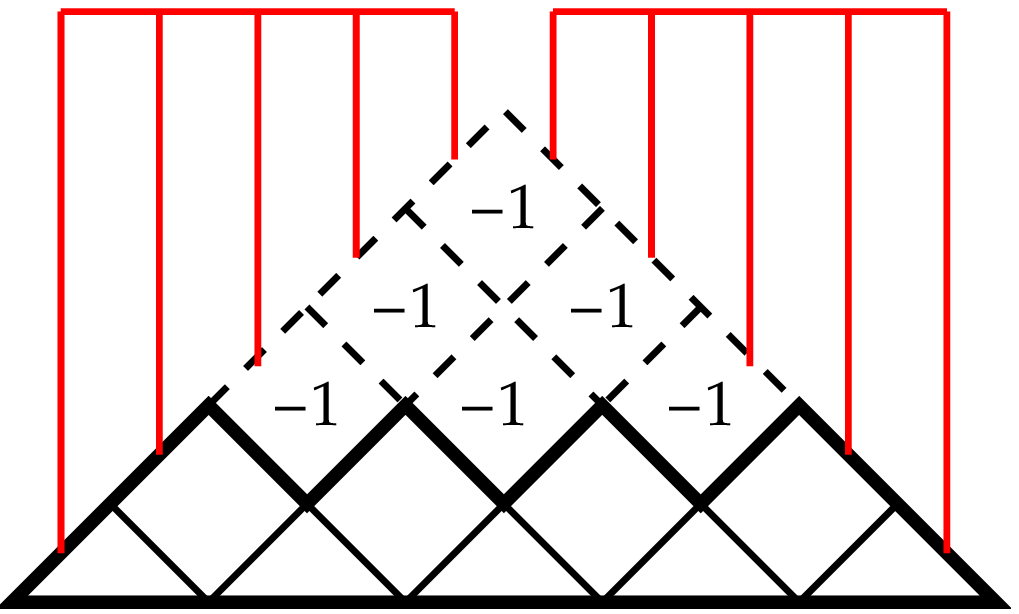}}
\caption{The partial sum $S_+(10,3)$ in factorised form.}
\label{Spm}
\end{figure}

\begin{figure}[h]
\centerline{\includegraphics[width=150pt]{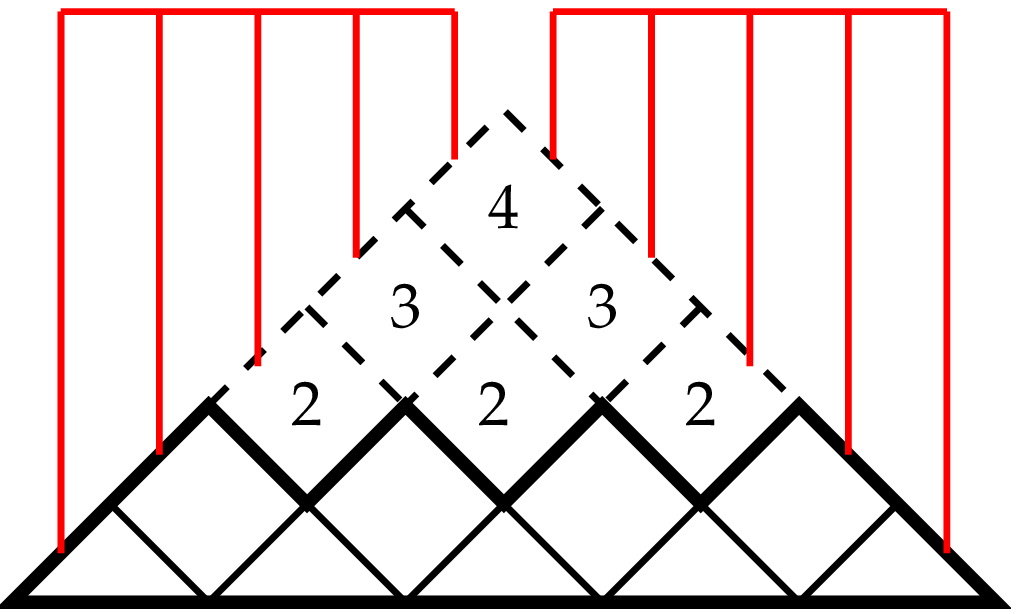}}
\caption{The partial sum $S_-(10,3)$ in factorised form.}
\label{Smp}
\end{figure}

In fact, we conjecture that \eqref{eq:Sfact1} with \eqref{eq:Spmdef} remain valid in the presence of the variables $x_i$:
\begin{obs}
Define
\[
P_N(u) = \prod_{j=1}^p \prod_{i=j-1}^{p-1} h_{\tilde{p}+2i+j}(u+j-1).
\]
The weighted partial sums are special cases of the following identity for polynomials in $x_1,\ldots,x_N$,
\begin{align}
S(N,p,u) := P_N(u)\, \psi_{\Omega}^{\rm A}(x_1,\ldots,x_N)
= \sum_{\alpha\in\mathcal{D}_{N,p}} \left( \frac{[1-u]}{[u]}\right)^{c_{\alpha,p}} \psi_\alpha(x_1,\ldots,x_N).
\label{eq:Su}
\end{align}
\end{obs}

\begin{figure}[h]
\centerline{\includegraphics[width=150pt]{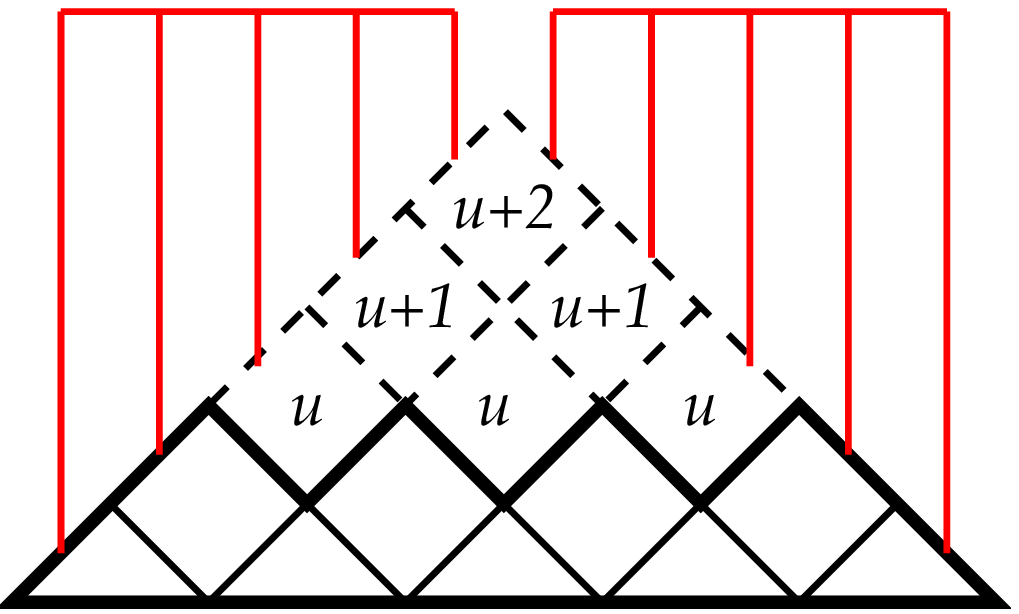}}
\caption{The partial sum $S(10,3,u)$ in factorised form.}
\label{Spu}
\end{figure}
\bigskip

Note that \eqref{eq:Su} has many interesting specialisations, such as $u=0$ and $u=1$ which, when properly normalised, correspond to the single coefficients $\psi_{\Omega(N,\lfloor (N-1)/2\rfloor)}$ and  \;\; $\psi_{\Omega(N,\lfloor (N-3)/2\rfloor)}$ respectively. The standard sum rule where one performs an unweighted sum, corresponds to $u=1/2$. Interestingly, a special case of the generalised sum rule \eqref{eq:Su} is closely related to a result of \cite{DFZJ07}, were a similar generalised sum was considered, based on totally different grounds and in the homogeneous limit $x_i \rightarrow 0$ and for $p=\lfloor (N-1)/2\rfloor$. By computation of a repeated contour integral, it was shown in \cite{DFZJ07} that in this case, $S(N,\lfloor (N-1)/2\rfloor ,u)$ is equal to the generating function of refined $t,\tau$-enumeration of (modified) cyclically symmetric transpose complement plane parititions, where $t=[1-u]/[u]$. Because of the natural way this parameter appears in \eqref{eq:Su}, we hope that this result offers further insights into the precise connection between solutions of the qKZ equation and plane partitions.

\subsection{Type B}
In Appendix~\ref{se:solutions} we have listed solutions of the qKZ equation for type B from Proposition~\ref{prop6} up
to $N=6$ in the limit $x_i\rightarrow 0$, $i=1,\ldots,N$. These solutions were obtained using the factorised forms of the previous section. As in the case of type A, we again find a positivity conjecture, this time in the two variables $\tau'$ and $a$ which are defined by
\[
\tau'^2 =2-\tau=2+[2]=[2]_{q^{1/2}}^2,\qquad a = \frac{[1/2]}{[\frac{2\delta+1}{4}][\frac{2\delta-1}{4}]}.
\]
The complete solution is determined up to an overall normalisation. We will choose $\xi^{\rm B}= a^{\lfloor N/2 \rfloor}(-\tau'^2)^{N(N-1)/2}$, for which we have

\begin{obs}
The solutions $\psi_{\alpha}(x_1,\ldots,x_N)$ of the qKZ equation of type B in the limit $x_i\rightarrow 0$, $i=1,\ldots,N$ are polynomials in $\tau'^2$ and $a$ with positive integer coefficients.
\end{obs}

For $a=1$ this conjecture was already observed in \cite{ZJ07}. As was conjectured in \cite{GR04} for $\tau'=1$, the parameter $a$ corresponds to a refined enumeration of vertically and
horizontally symmetric alternating sign matrices. A sum rule for
this value of $\tau'$ was proved in \cite{ZJ07}. We suspect that
the parameter $\tau'$ is related to a simple statistic on plane
partitions, as it is for type A. We thus have an interesting mix
of statistics, one which is natural for ASMs, and one which is
natural for plane partitions. In a forthcoming paper we hope to
formulate some further results concerning the solutions for type B.

\section{Proofs}
\lb{sec6}

\subsection{Proof of  Theorem~1
}
\label{sec:ProofFacSolA}
We have to show that the vector $\ket \Psi$ whose coefficients are given
by the formulas \eqref{psi_alpha_D}, and \eqref{H_a}, \eqref{psi_a-A}
satisfies the qKZ equations \eqref{qKZTL_TypeB1}-\eqref{qKZTL_TypeB3}
of type A.
Following the preliminary analysis of Section~\ref{subsec3.3.1} we
divide the proof of \eqref{qKZTL_TypeB1} into two parts, depending on whether or
not the word corresponding to the path $\alpha$
begins with $e_i$.
\medskip

\noindent
{{\bf  1.}\;\em $\alpha$ does not have a local maximum at $i$.}
We have to show that, see \eqref{eq:hmonpsi},
\be
-(a_i\psi_{\alpha})(x_1,\ldots,x_N) = (h_i(-1) \psi_{\alpha})(x_1,\ldots,x_N) =0
\label{proj}
\ee
for $\psi_\alpha$ given by \eqref{psi_a-A}.

If $\alpha$ does not have a local maximum at $i$, then either $h_i(-1)$ acts on a local minimum,
or on a slope of $\alpha$.
\begin{itemize}
\item $h_i(-1)$ acts on a local minimum of $ \alpha$.

\noindent
In this case $\psi_\alpha$ is divisible by $h_i(1)$ from the left
and \eqref{proj} follows directly from
\[
\begin{picture}(150,60)(0,20)
\put(-30,40){$h_i(-1) h_i(1)\ =\  -a_i\, s_i\ =\ $}
\put(107,10){\includegraphics[width=36pt]{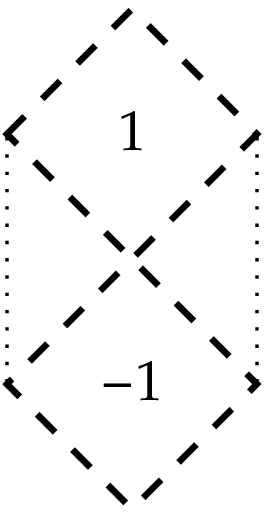}}
\put(157,40){$=\ 0\, .$}
\end{picture}
\]
\item $h_i(-1)$ acts on a slope of $ \alpha$.

\noindent
In this case we use the Yang-Baxter equation \eqref{ybegraphic}  to
push $h_i(-1)$ through the operator $H_\alpha$ \eqref{H_a} in the expression for $\psi_\alpha$
\eqref{psi_a-A}. Then  $h_i(-1)$ vanishes when acting on
$\psi_{\Omega}^{\rm A}$, see \eqref{eq:hmonpsi2}.
This mechanism is illustrated in Figure~\ref{fig:hmonH}.

\begin{figure}[h]
\begin{picture}(340,130)
\put(0,0){\includegraphics[width=135pt]{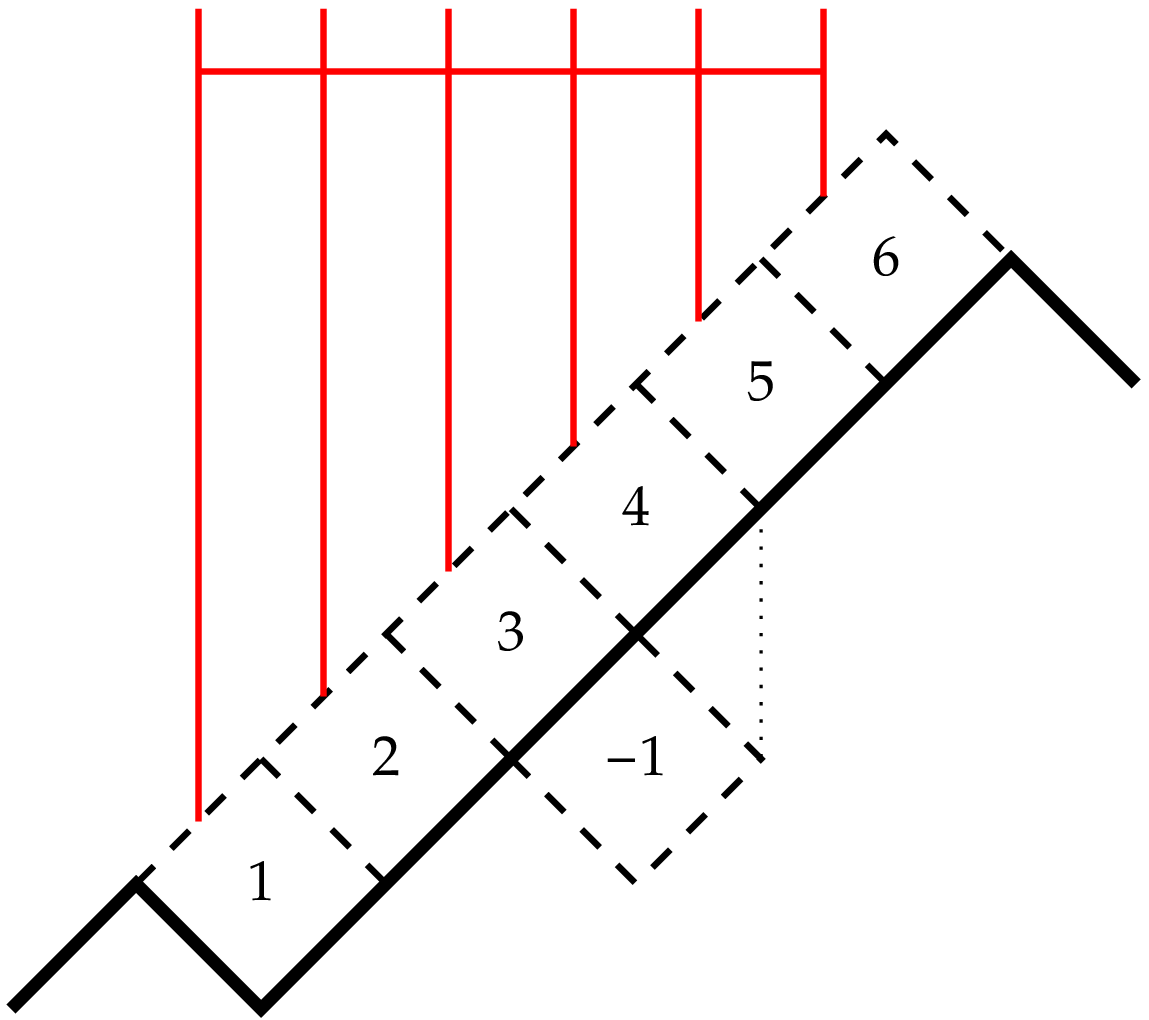}}
\put(150,60){$=$}
\put(170,0){\includegraphics[width=135pt]{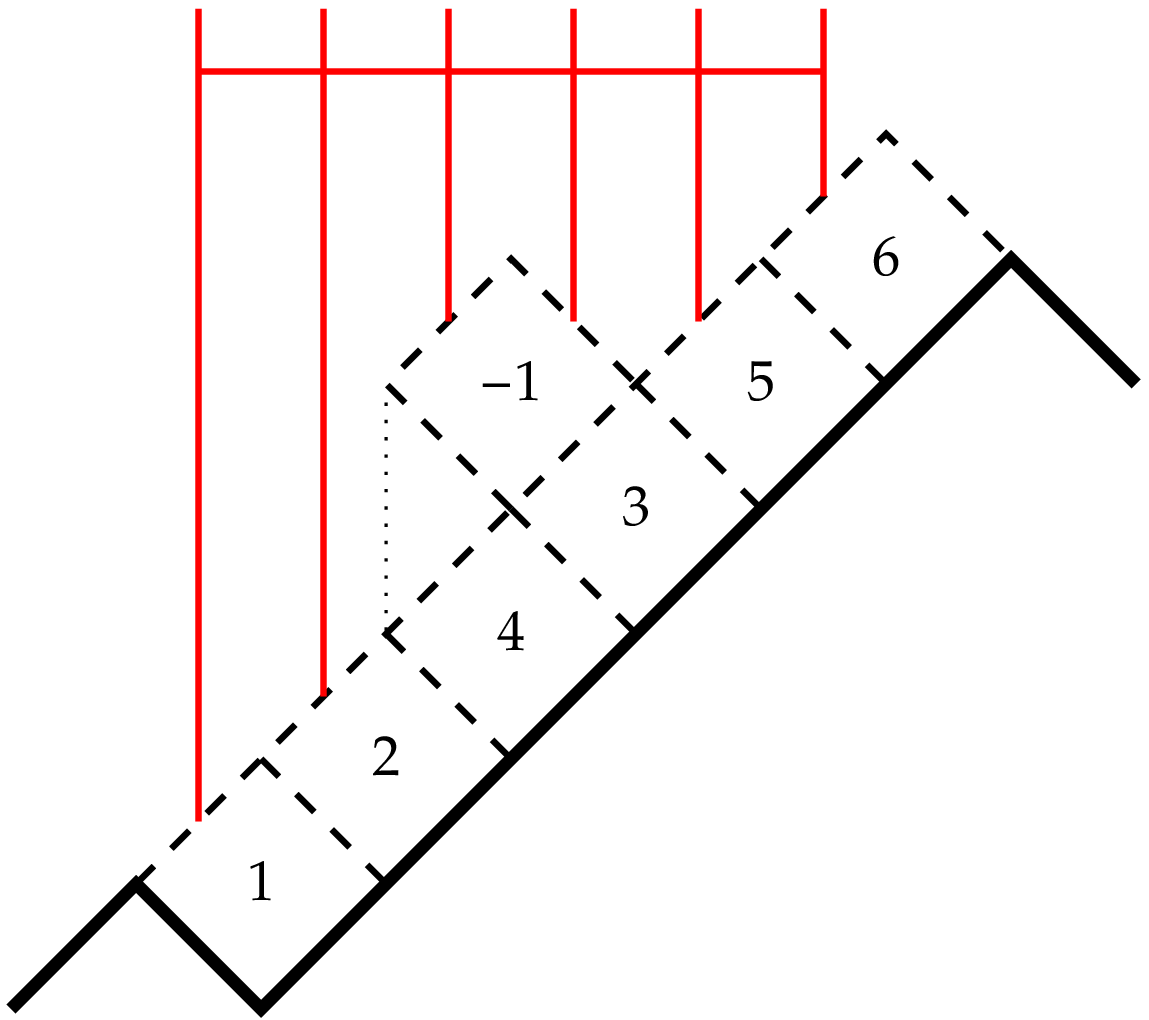}}
\put(320,60){$=0$}
\end{picture}
\caption{Action of $h_i(-1)$ on a slope of a Dyck path $\alpha$. The
first equality follows from the Yang-Baxter equation
\eqref{ybegraphic} and the second
follows from \eqref{eq:hmonpsi2}. For the operators $H_\alpha$ of a more general form
(like, e.g., the one shown in Figure~\ref{fig:qKZsol}) the first part of this transformation
should be repeated until $h_i(-1)$ commutes through all the terms of $H_\alpha$.}

\label{fig:hmonH}
\end{figure}
\end{itemize}

\noindent
{{\bf  2.}\;\em $\alpha$ has a local maximum at $i$.}
The harder part of the proof of Theorem~\ref{th:facsol}, to which we come now, lies in proving
\eqref{TLei} when $h_i(1)$ acts on a component $\psi_\alpha$ where the
path $\alpha$ has local maximum at $i$.
If this maximum at $i$ does
not have a nearest neighbour minimum at $i-1$  or $i+1$ then \eqref{TLei} becomes simply
\[
h_i(1)\psi_{\alpha} = \psi_{\alpha^-}  ,
\]
and the action of $h_i(1)$ is the addition of a tile with content
$1$ at $i$, which is just the prescription of the Theorem \ref{th:facsol},
see Figure~\ref{fig:SimpleMax}.

\begin{figure}[h]
\centerline{\includegraphics[width=0.8\textwidth]{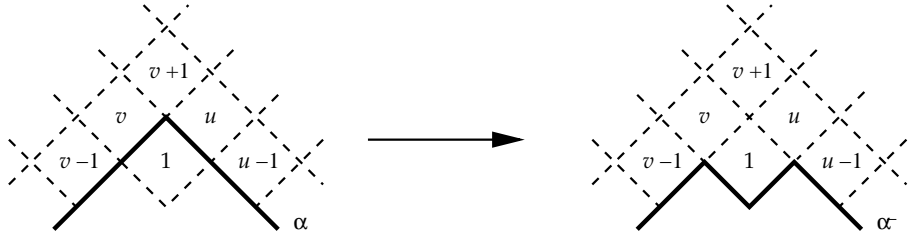}}
\caption{Graphical representation of the equation $h_i(1)
  \psi_{\alpha} = \psi_{\alpha^-}$ corresponding to the addition of a
  tile with content $1$ at a maximum without
  neighbouring minima. In this case $u>1$ and $v>1$.}
\label{fig:SimpleMax}
\end{figure}
\begin{figure}[h]
\includegraphics[width=0.9\textwidth]{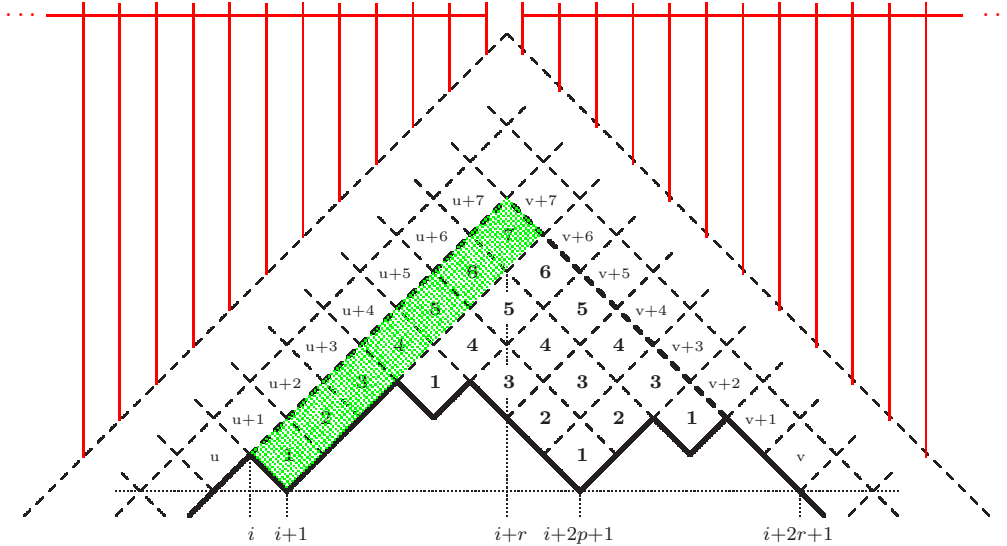}
\caption{
The Dyck path $\alpha$ satisfying conditions a) and b).
Between $i+1$ and $i+2r+1$ this path contains exactly
one local minimum  at $i+2p+1$ which has the same height  as the minimum at $i+1$.
}
\label{NonSimple-H}
\end{figure}

We will now look at the action of $h_i(1)$ on $\psi_\alpha$ \eqref{psi_a-A}, where
the Dyck path $\alpha$ satisfies  conditions (see Figure \ref{NonSimple-H})
\begin{itemize}\label{cond-ab}
\item[{ a)}]
$\alpha$
has a maximum at $i$ with a neighbouring minimum at $i+1$;
\vspace{1mm}
\item[{ b)}]
$\alpha$
crosses the horizontal line at height
$\alpha_{i+1}=\alpha_i -1$, for the first time to the right of $i$, at  $i+2r+1$, $r\geq 1$:
$\alpha_{i+2r}-1=\alpha_{i+2r+1}=\alpha_{i+2r+2}+1=\alpha_{i+1}$.
\end{itemize}
In this case we observe that the factorised expression \eqref{psi_a-A} for $\psi_\alpha$
contains a strip of tiles $H_{i+1,i+r}(1)$,
where
\be
H_{i+1,i+r}(u)\, :=\, h_{i+1}(u)h_{i+2}(u+1)\times\dots \times h_{i+r}(u+r-1).
\label{eq:Hdef}
\ee
This strip is shown shaded on  Figure~\ref{NonSimple-H}.

We are going to rewrite the term $h_i(1)H_{i+1,i+r}(1)$ in the product $h_i(1)\psi_\alpha$
in such a way that we obtain the components $\psi_{\alpha^-}$ and
$\psi_{\alpha^{+k}}$ from the right hand side of \eqref{TLei}, see also Figure~\ref{alphapm}, defined according to the rules
\eqref{H_a}, \eqref{psi_a-A}.
As a first step we prove the following proposition:
\begin{prop}
\label{prop:honpsi}
Let $X$ be an arbitrary element of the algebra $\H^{\rm A}_N(q)$ taken in its faithful representation
\eqref{a-rep}, \eqref{h-rep}.
We denote by $A_{i,j}$, $i\leq j$,  the linear span of terms
$~X a_k = X h_k(-1)~$ $\forall\,k: i\leq k\leq j$. We define additionally
$A_{i,i-1}:=0$,~ $H_{i,i-1}(1):=1$.

The following relation is valid modulo $A_{i,i+r-1}\,$:
\be
\lb{6.3}
H_{i,i+r}(u)
= H_{i,i+r}(u+v)\, +\, \frac{[v]}{[u][u+v]}\, H_{i+1,i+r}(1)\quad \mbox{\rm mod}\; A_{i,i+r-1}\qquad
\forall\,r\geq 0\, .
\ee
\end{prop}
For the proof of Proposition \ref{prop:honpsi} we need the following two simple lemmas.
\begin{lemma}
\label{le:Hshift}
One has
\begin{enumerate}
\item[1.] $h_i(u)H_{i+1,i+r}(u+1) = H_{i,i+r}(u)$;
\vspace{1mm}
\item[2.] $H_{i,i+r}(u)h_{i+r+1}(u+r+1) = H_{i,i+r+1}(u)$;
\item[3.] for generic values of $\; u\, $ and $\, v$
\begin{multline}
\hspace{9mm}H_{i,i+r}(u) = H_{i,i+r}(u+v) \\
 + \sum_{k=0}^r \frac{[v]}{[u+k][u+v+k]}
H_{i,i+k-1}(u+v) H_{i+k+1,i+r}(u+k+1)\, .
\label{s1}
\end{multline}
\end{enumerate}
\end{lemma}

\begin{proof}
The first two parts of the lemma follow immediately from the definition of $H_{i,i+r}(u)$ in
\eqref{eq:Hdef}.
To prove \eqref{s1} we use induction on $r$. For
$r=0$ equation \eqref{s1} reduces to \eqref{huv}. Now we shall make the inductive step by assuming
\eqref{s1} is true for some $r$, and prove it for $r+1$:
\begin{align*}
H_{i,i+r+1}(u)\, &=\, H_{i,i+r}(u)h_{i+r+1}(u+r+1)
\, = H_{i,i+r}(u+v)h_{i+r+1}(u+r+1)\\
&\hspace{-30pt}\hphantom{=} + \sum_{k=0}^r \frac{[v]}{[u+k][u+v+k]}
H_{i,i+k-1}(u+v) H_{i+k+1,i+r}(u+k+1)h_{i+r+1}(u+r+1)
\\[1mm]
&\hspace{-30pt}=\, H_{i,i+r}(u+v)\left(h_{i+r+1}(u+v+r+1) + \frac{[v]}{[u+r+1][u+v+r+1]}\right)\\
&\hspace{-30pt}\hphantom{=} + \sum_{k=0}^r \frac{[v]}{[u+k][u+v+k]}
H_{i,i+k-1}(u+v) H_{i+k+1,i+r+1}(u+k+1)
\\[1mm]
&
\hspace{-30pt}
=\, H_{i,i+r+1}(u+v)
+ \sum_{k=0}^{r+1} \frac{[v]}{[u+k][u+v+k]}
H_{i,i+k-1}(u+v) H_{i+k+1,i+r+1}(u+k+1),
\end{align*}
where we used \eqref{huv} and the induction assumption. This
completes the proof of Lemma~\ref{le:Hshift}.
\end{proof}

\begin{lemma}
\label{lem:HonDelta} One has
\be
\label{6.5}
H_{i,i+r}(u)\,  =\,
\frac{[u+r+1]}{[u]}\quad \mbox{\rm mod}\; A_{i,i+r}\, ,
\ee
In particular,  the coefficient $\psi_\Omega^{\rm A}$ \eqref{psi_alpha_D} of the maximal
Dyck path $\Omega^{\rm A}$ satisfies the relations
\be
\label{6.6}
H_{i,i+r}(u)\, \psi_\Omega^{\rm A}\, =\,
\frac{[u+r+1]}{[u]}\,
\psi_\Omega^{\rm A},
\ee
in case the indices $\,r$ and $i$ are chosen
within the limits $\,0\leq r< n-1\,$, $\,1\leq i < n-r$,
where $n=\lfloor\tfrac{N+1}{2}\rfloor$.
\end{lemma}

\begin{proof} From \eqref{eq:hmonpsi2} and \eqref{huv} we find that
\be
\label{vspomogat}
h_{i}(u) \, = \frac{[u+1]}{[u]}
\, + h_i(-1)\, =\, \frac{[u+1]}{[u]}\quad \mbox{\rm mod}\; A_{i,i}\, ,
\ee
which implies formula \eqref{6.5}. Relation \eqref{6.6} for the coefficient $\psi_\Omega^{\rm A}$
then follows from \eqref{eq:hmonpsi2}.
\end{proof}

\begin{proof}[Proof of Proposition~\ref{prop:honpsi}]
By applying \eqref{s1}
twice, first with the arguments $\{u,v\}$,  and then with the arguments
$\{u,v\}$ substituted by $\{u+1,-u\}$, we find:
\begin{align}
H_{i,i+r}(u) &= H_{i,i+r}(u+v) + \frac{[v]}{[u][u+v]} H_{i+1,i+r}(u+1)\nonumber \\
&\hspace{10mm}
 + \sum_{k=1}^{r} \frac{[v]}{[u+k][u+v+k]} H_{i,i+k-1}(u+v)
H_{i+k+1,i+r}(u+k+1) \nonumber\\[1mm]
&= H_{i,i+r}(u+v) + \frac{[v]}{[u][u+v]} H_{i+1,i+r}(1)\nonumber \\
&\hspace{10mm}
 - \frac{[v]}{[u][u+v]}\sum_{k=0}^{r-1} \frac{[u]}{[u+k+1][k+1]} H_{i+1,i+k}(1)
H_{i+k+2,i+r}(u+k+2)\nonumber\\
&\hspace{10mm}  + \sum_{k=1}^{r}
\frac{[v]}{[u+k][u+v+k]} H_{i,i+k-1}(u+v)
H_{i+k+1,i+r}(u+k+1)\nonumber
\\[1mm]
&= H_{i,i+r}(u+v) + \frac{[v]}{[u][u+v]} H_{i+1,i+r}(1) \nonumber\\
& \hspace{-16mm} + \sum_{k=1}^{r} \frac{[v]}{[u+k]}H_{i+k+1,i+r}(u+k+1)
\left( \frac{1}{[u+v+k]} H_{i,i+k-1}(u+v) -\frac{1}{[u+v][k]}
H_{i+1,i+k-1}(1)\right) .
\label{eq:Hdiff}
\end{align}
It lasts to notice that, by Lemma~\ref{lem:HonDelta}, the operators
between parentheses in the last term of
\eqref{eq:Hdiff} become c-numbers modulo $A_{i,i+r-1}$  and
in fact cancel:
\begin{multline*}
\left( \frac{1}{[u+v+k]} H_{i,i+k-1}(u+v)
-\frac{1}{[u+v][k]} H_{i+1,i+k-1}(1)\right)\\[1mm]
=  \left( \frac{1}{[u+v+k]} \frac{[u+v+k]}{[u+v]}
-\frac{1}{[u+v][k]} \frac{[k]}{[1]}\right)\quad \mbox{\rm mod}\; A_{i,i+r-1}\\
\, =\, 0\quad \mbox{\rm mod}\; A_{i,i+r-1}\, .
\end{multline*}
Hence, \eqref{eq:Hdiff} implies \eqref{6.3}.
\end{proof}

Consider a path $\alpha$ whose all local minima between $i+1$ and $i+2r+1$
lie higher then $\alpha_{i+1}$. For such paths  Proposition \ref{prop:honpsi}
implies the following
\begin{cor}
\label{cor3}
Let $\alpha=(\alpha_0,\ldots,\alpha_N)$ be a Dyck path
satisfying conditions {\rm a)} and {\rm b)} on page \pageref{cond-ab}.
Assume additionally that $\alpha$ has no  local
minima at height $\alpha_{i+1}$ between $i+1$ and $i+2r+1$.
Let further $\alpha^-$  (respectively, $\alpha^{+1}$) denote the path
obtained from $\alpha$ by raising (resp., lowering) the height $\alpha_{i}$
(resp., $\alpha_{i+1}$) by two, see Figure~\ref{alphapm}.
Then for the coefficients $\psi_\alpha$, $\psi_{\alpha^-}$, $\psi_{\alpha^{+1}}$
defined  by \eqref{psi_a-A} we have:
\[
h_i(1) \psi_{\alpha}\,
=\, \psi_{\alpha^-} + \psi_{\alpha^{+1}}\, .
\]
\end{cor}

\begin{proof}
As we noted before, the product $h_i(1)\psi_\alpha$
contains the factor $h_i(1)H_{i+1,i+r}(1)$. Applying \eqref{6.3}
and \eqref{vspomogat} for $u=v=1$,
we can transform this factor in the following way
\begin{align}
h_i(1)H_{i+1,i+r}(1)&=\, h_i(1)H_{i+1,i+r}(2) \, +\,\frac{1}{[2]} h_i(1) H_{i+2,i+r}(1)\quad \mbox{\rm mod}\;
A_{i+1,i+r-1}\nonumber \\[1mm]
&=\, H_{i,i+r}(1)\, +\, H_{i+2,i+r}(1)\quad \mbox{\rm mod}\;
A_{i,i+r-1}\, .
\label{6.10}
\end{align}
Upon substitution of this result back into $h_i(1)\psi_\alpha$
the terms containing the expressions $H_{i,i+r}(1)$ and $H_{i+2,i+r}(1)$
both assume the form of the ansatz \eqref{psi_a-A}. They
correspond, respectively, to the paths
$\alpha^-$ and $\alpha^{^+1}$. The terms from $A_{i,i+r-1}$ vanish
due to the same mechanism as in Figure \ref{fig:hmonH}.
This calculation is graphically
displayed in  Figure~\ref{fig:hponpsi}.
\end{proof}
\begin{figure}[h]
\centerline{
\begin{picture}(350,160)
\put(0,100){\includegraphics[width=80pt]{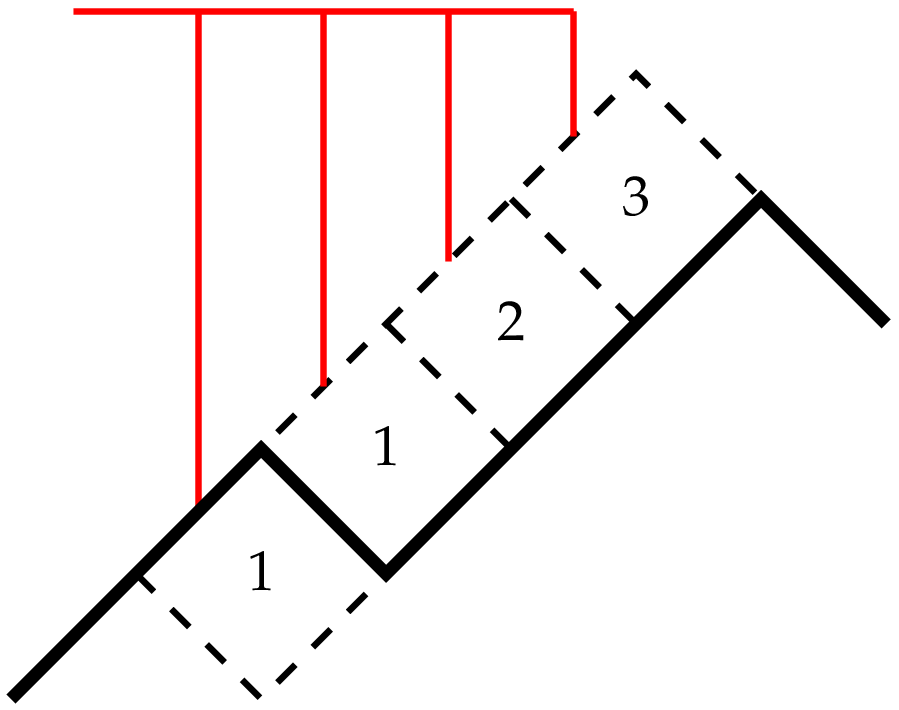}}
\put(80,125){$\scriptstyle\alpha$}
\put(100,130){$=$}
\put(120,100){\includegraphics[width=80pt]{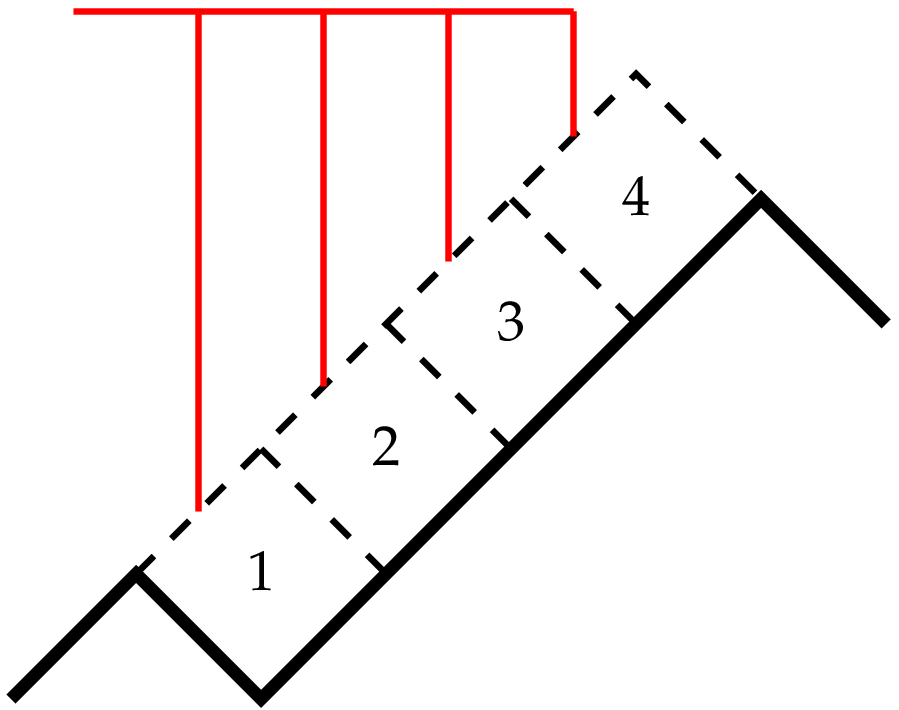}}
\put(210,130){${}+\frac{1}{[2]}$}
\put(240,100){\includegraphics[width=80pt]{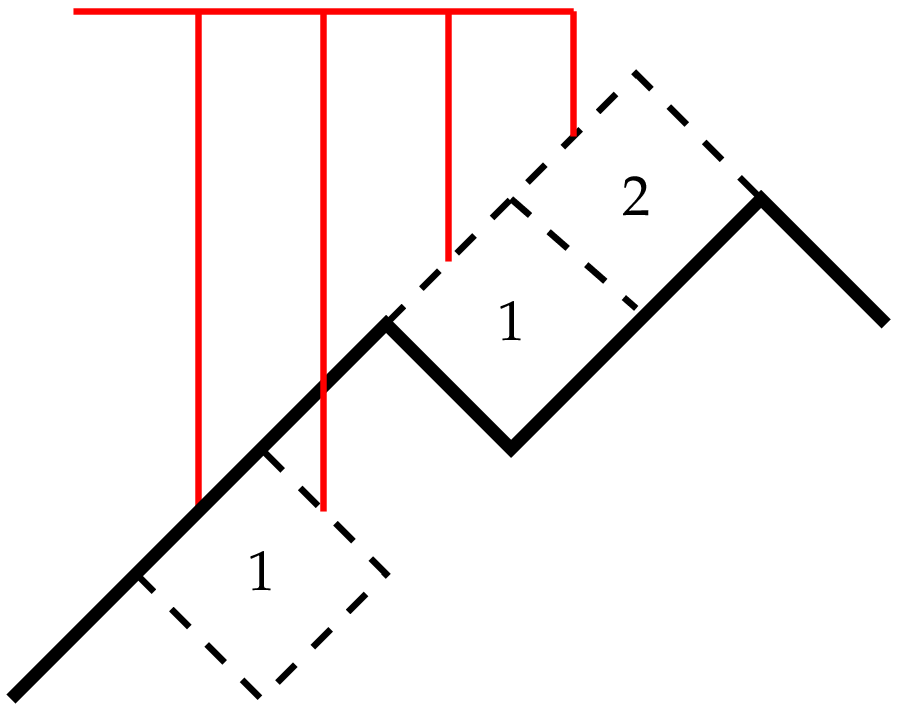}}
\put(100,30){$=$}
\put(120,0){\includegraphics[width=80pt]{hponpsi2.eps}}
\put(200,25){$\scriptstyle\alpha^{-}$}
\put(230,30){$+$}
\put(240,0){\includegraphics[width=80pt]{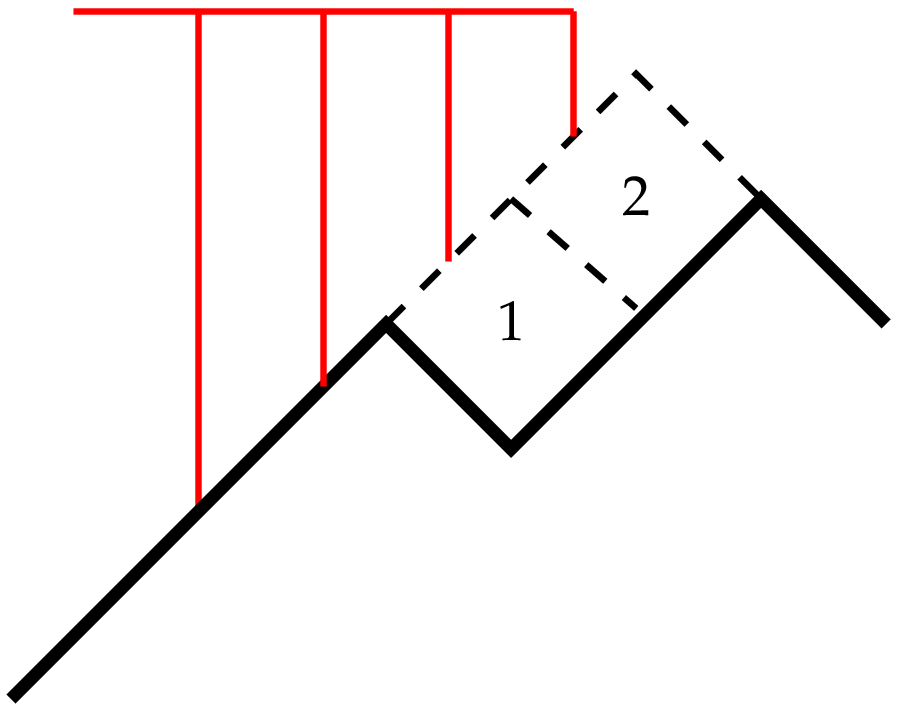}}
\put(320,25){$\scriptstyle\alpha^{+1}$}
\end{picture}
}
\caption{Diagrammatic presentation of
the Corollary \ref{cor3}.
The transformation \eqref{6.10} is displayed in details.}
\label{fig:hponpsi}
\end{figure}

Consider now a path $\alpha$ which has exactly one local minimum
between $i+1$ and $i+2r+1$ of the same height as the minimum at $i+1$, see Figure~\ref{NonSimple-H}.
In this case the following Corollary holds.

\begin{cor}
\label{cor:honpsi2}
Let $\alpha=(\alpha_0,\ldots,\alpha_N)$ be a Dyck path
satisfying conditions {\rm a)} and {\rm b)} on page \pageref{cond-ab}.
Assume additionally that $\alpha$ has
one local minimum placed at the point
$(i+2p+1)$, $0< p< r$, which has exactly the same height
as the minimum at $i+1$: $\alpha_{i+2p+1}=\alpha_{i+1}$, see Figure~\ref{NonSimple-H}.
Let further the paths
$\alpha^-$, $\alpha^{+1}$ be defined as in
Corollary~\ref{cor3}, and denote by $\alpha^{+2}$ the path
obtained from $\alpha$ by raising  the heights
$\alpha_{i+1},\dots ,\alpha_{i+2p+1}$ by two, see
Figure~\ref{alphapm}. Then for the coefficients $\psi_\alpha$,
$\psi_{\alpha^-}$, $\psi_{\alpha^{+1}}$ and $\psi_{\alpha^{+2}}$
defined  by \eqref{psi_a-A} we have:
\be
\label{6.11}
h_i(1) \psi_{\alpha}\, =\,
\psi_{\alpha^-} + \psi_{\alpha^{+1}} + \psi_{\alpha^{+2}}\, .
\ee
\end{cor}

\begin{proof}
We copy the transformation of $h_i(1)\psi_\alpha$ from the proof of Corollary
\ref{cor3} till the end of
\eqref{6.10}. As before, the  term $H_{i,i+r}(1)$ in
the last line of \eqref{6.10} gives rise to the
coefficient $\psi_{\alpha^-}\,$ in \eqref{6.11} whereas the term $A_{i,i+r-1}$
vanishes upon substitution into $h_i(1)\psi_\alpha$. This time however
the term $H_{i+2,i+r}(1)$ when substituted into $h_i(1)\psi_\alpha$ does not give an expression
fitting the ansatz  \eqref{psi_a-A}.
We continue its transformation using again \eqref{6.3} for $u=p$, $v=1$,
and \eqref{6.5} for $u=1$:
\begin{align}
H_{i+2,i+r}(1)\,&=\, H_{i+2,i+p}(1)H_{i+p+1,i+r}(p) \nonumber
\\[1mm]
&\hspace{-9mm}=\, H_{i+2,i+p}(1)\Bigl(H_{i+p+1,i+r}(p+1)\,+\, \frac{1}{[p][p+1]}\,  H_{i+p+2,i+r}(1)\Bigr)
\quad \mbox{\rm mod}\;
A_{i+p+1,i+r-1}\nonumber
\\[1mm]
&\hspace{-9mm}=H_{i+2,i+p}(1)H_{i+p+1,i+r}(p+1)\,+\,\frac{1}{[p+1]}\,H_{i+p+2,i+r}(1)
\quad \mbox{\rm mod}\;
A_{i+2,i+r-1}\, .
\label{6.12}
\end{align}
The first term in \eqref{6.12} upon substitution into $h_i(1)\psi_\alpha$
gives rise to the coefficient $\psi_{\alpha^{+1}}$, whereas the last term vanishes. It remains
to consider the effect of the second term.

Let us introduce a further extension of the notation \eqref{eq:Hdef},
\begin{align}
\label{6.13}
H_{i,i+r}^{\phantom{i,}i-s}(u)&\, :=\,
H_{i,i+r}(u) H_{i-1,i+r-1}(u+1)\times\dots\times
H_{i-s,i+r-s}(u+s)\quad \forall \, r,s\geq 0 ,
\\[1mm]
H_{i,i-1}^{\phantom{i,}i-s}&\, :=\, H_{i,i+r}^{\phantom{i,}i-1}\, :=\, 1\, .
\nonumber
\end{align}
Graphically $H_{i,i+r}^{\phantom{i,}i-s}(u)$ can be displayed as a rectangular
block of tiles of a size $(r+1)\times(s+1)$ with the bottom corner tile corresponding
to $h_i(u)$. We also use the following shorthand symbols for uphill and downhill strips
(the case of block with either $s=0$, or $r=0$):
\[
H_{i,i+r}^{\phantom{i,}i}(u)\, :=\, H_{i,i+r}(u)\, , \qquad
H_{i,i}^{\phantom{i,}i-s}(u)\, :=\, H_{i}^{i-s}(u)\, .
\]

We now notice that in the assumptions of the Corollary
the strip of tiles $H_{i+1,i+r}(1)$ in expression \eqref{psi_a-A} for $\psi_\alpha$
is in fact multiplied from the left by the block of tiles
$\raisebox{-2mm}{\rule{0pt}{6.5mm}}H_{i+2p+1,i+p+r}^{\phantom{i+2p+2,}i+p+2}(1)$.
Therefore, we can continue the transformation of
the second term in \eqref{6.12} by multiplying it from the left  by
the term $H_{i+2p+1,i+p+r}^{\phantom{i+2p+1,}i+p+2}(1)$. The transformation is essentially
a permutation of these two terms which makes use of the Yang-Baxter equation \eqref{ybegraphic}:
\begin{align}
H_{i+2p+1,i+p+r}^{\phantom{i+2p+1,}i+p+2}(1)&\left(
\frac{1}{[p+1]}\,H_{i+p+2,i+r}(1)\right)\nonumber
\\[1mm]
=\,& \frac{1}{[p+1]}\,H_{i+2p+2,i+p+r}(1)\,
H_{i+2p+1,i+p+r-1}^{\phantom{i+2p+1,}i+p+2}(2)\,
H_{i+p+r}^{i+r+1}(1)
\nonumber
\\[1mm]
=\,&
H_{i+2p+2,i+p+r}^{\phantom{i+2p+2,}i+p+2}(1)
\quad \mbox{\rm mod}\;
A_{i+r+1,i+p+r}\, .
\label{6.16}
\end{align}
Here in the last line we evaluate factor $H_{i+p+r}^{i+r+1}(1)$ using \eqref{6.5}.
The transformation \eqref{6.16} is illustrated in Figure~\ref{fig14}.
\begin{figure}[h]
\vspace{-4mm}
\includegraphics[width=0.9\textwidth]{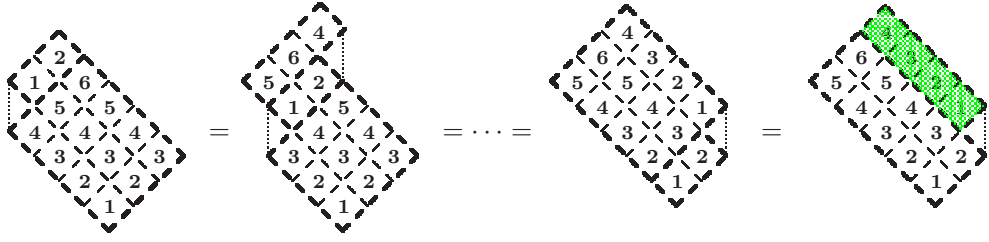}
\caption{
Diagrammatic presentation of the transformation \eqref{6.16}
for the case $p=4$ and $r=7$ (this situation occurs in the third line of the calculation
in Figure~\ref{fig15}).
Using the Yang-Baxter equation one moves the strip $H_{i+p+2,i+r}(1)$ \smallskip
down and right. \smallskip
The contents of the tiles
in the block $H_{i+2p+2,i+p+r}^{\hspace{30pt}i+p+2}(1)$
are shifted cyclically
in the uphill layers
during the permutation. The shaded downhill strip
$H_{i+p+r}^{i+r+1}(1)$ equals $[p+1]$
modulo $A_{i+r+1,i+p+r}$.
}
\label{fig14}
\end{figure}
\begin{figure}[h]
\centerline{\hspace{-15mm}
\includegraphics[width=0.85\textwidth]{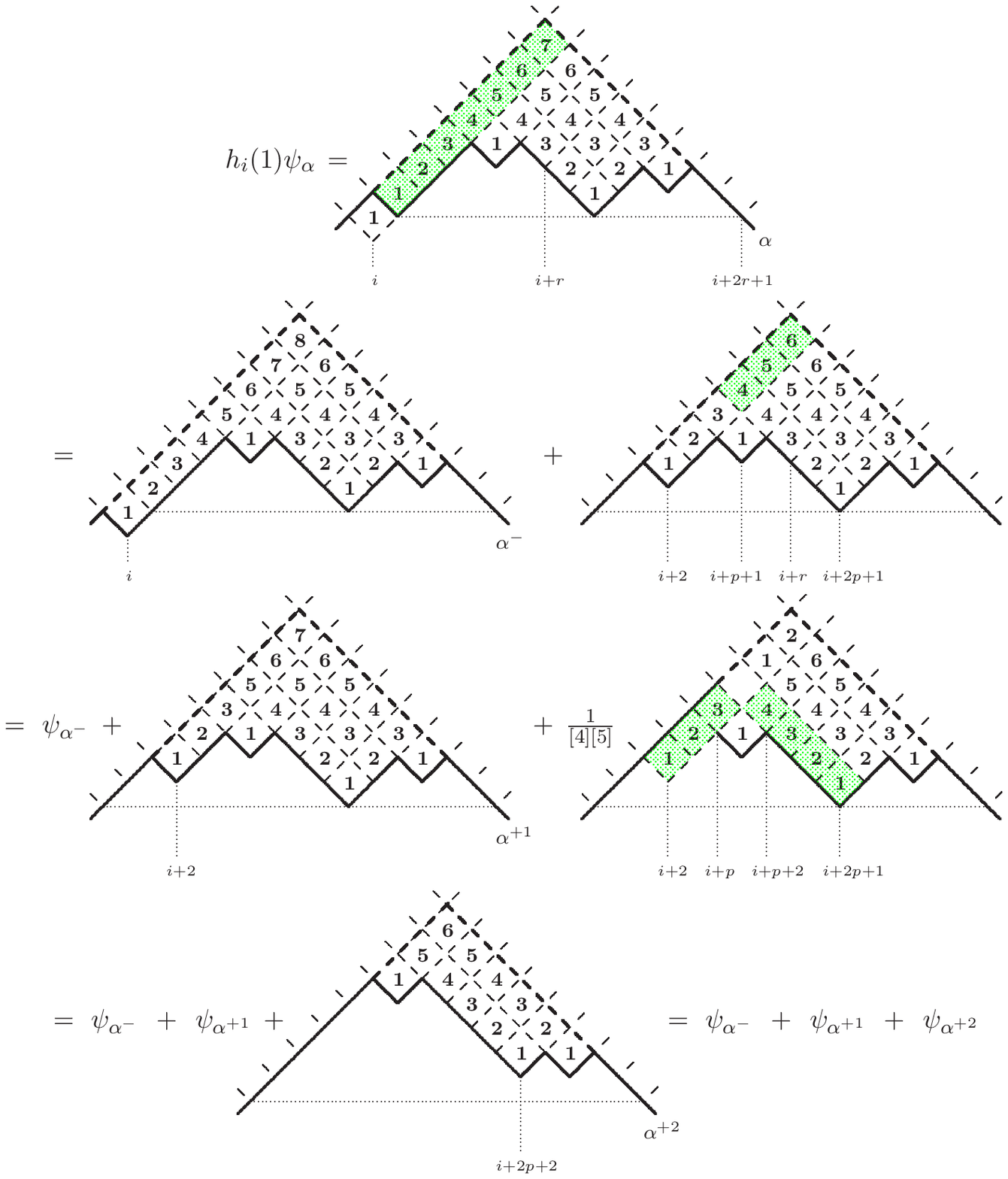}
\vspace{-1mm}
}
\caption{
Diagrammatic illustration for the proof of Corollary \ref{cor:honpsi2}.
The path $\alpha$ is taken from
Figure\,\ref{NonSimple-H}.
Each time the transformation is applied to the strips which are shown  shaded on the pictures.
The second line on the Figure represents transformation \eqref{6.10}.
The third line corresponds to the transformation \eqref{6.12}.
The shaded uphill strip in this line
can be evaluated as $\, [4]$ (cf. with the last equality in \eqref{6.12}). The
shaded downhill strip has to be moved up and right (see Fig.\,\ref{fig14})
and then evaluated as $[5]$.
\vspace{-1mm}}
\label{fig15}
\end{figure}

Substitution of the result of \eqref{6.16} back into $h_i(1)\psi_\alpha$
gives exactly the expression for the coefficient $\psi_{\alpha^{+2}}$.
The whole calculation is graphically
displayed in  Figure~\ref{fig15}.
\end{proof}

The general structure now is clear and we can formulate
\begin{prop}
\label{prop:honpsi3}
Let $\alpha=(\alpha_0,\ldots,\alpha_N)$ be a Dyck path
satisfying conditions {\rm a)} and {\rm b)} on page \pageref{cond-ab}.
Assume additionally that $\alpha$ has
$K\geq 1$ local minima placed at the points
$i+2p_k+1$, $0= p_1<\dots <p_K< r$, which have the same height
$h=\alpha_{i+1}=\alpha_{i+2p_k+1}\; \forall k$, see Figure~\ref{alphapm}.
Then for the coefficients $\psi_\alpha$,
$\psi_{\alpha^-}$, $\psi_{\alpha^{+k}}$, $k=1,\dots ,K,$
defined  by \eqref{psi_a-A} we have:
\begin{align*}
h_i(1) \psi_{\alpha} = \psi_{\alpha^-} + \sum_{k=1}^K
\psi_{\alpha^{+k}},
\end{align*}
where $\alpha^-$ and $\alpha^{+k}$ are the Dyck paths defined
in Figure~\ref{alphapm}.
\end{prop}

For the proof of Proposition~\ref{prop:honpsi3} we need a generalisation of
the formula \eqref{6.3} for the case of blocks (but now for $v=1$ only):
\begin{lemma}
\label{lem4}
For generic values of $\; u\, $  one has
\begin{align}
H_{i,i+r}^{\phantom{i,}i-s}(u) = H_{i,i+r}^{\phantom{i,}i-s}(u+1)\,
+\,\frac{1}{[u][u+1]}\,&
H_{i+1,i+r}(1)\, H_{i-1}^{i-s}(1)\, H_{i,i+r-1}^{\phantom{i,}i-s+1}(u+2)\nonumber
\\
&\hspace{7mm} \mbox{\rm mod}\; A_{i-s,i+r-s-1}\oplus A_{i+r-s+1,i+r}\, .
\label{block}
\end{align}
\end{lemma}

The graphical presentation of relation \eqref{block} given in
Figure \ref{fig16} is probably more clarifying.
\begin{figure}[h]
\vspace{-5mm}
\includegraphics[width=\textwidth]{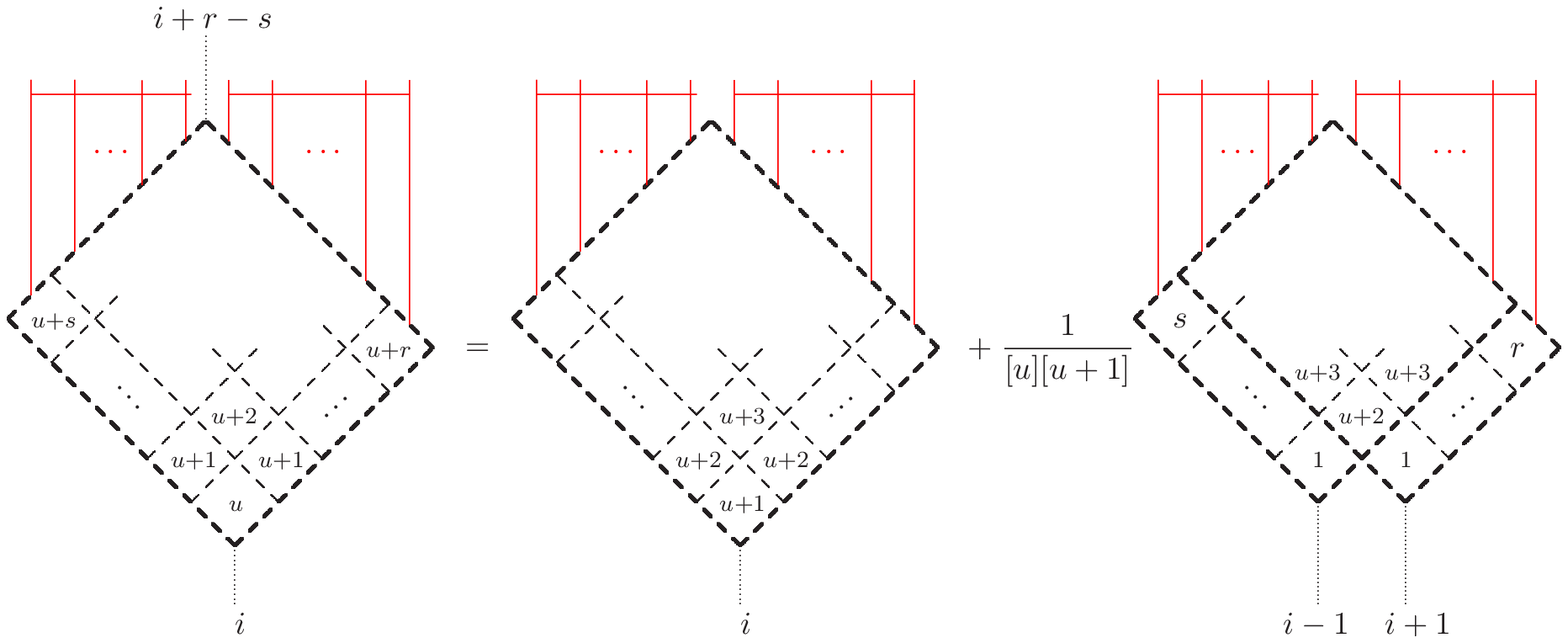}
\caption{
Diagrammatic presentation of  relation \eqref{block}.
The colored vertical lines on top of the tilted rectangles
symbolize annihilator of the terms $A_{i-s,i+r-s-1}$ and $A_{i+r-s+1,i+r}$.
In the expressions for $\psi_\alpha$ \eqref{psi_a-A}
this role is played by the coefficient $\psi^{\rm A}_\Omega$.
}
\label{fig16}
\end{figure}
\begin{proof}
We use induction on $s$. For $s=0$ relation \eqref{block} reduces to
\eqref{s1}.
We now check it for some $s>0$
assuming it is valid for smaller values of $s$:
\begin{align*}
H_{i,i+r}^{\phantom{i,}i-s}(u)\, =&\; H_{i,i+r}(u)H_{i-1,i+r-1}^{\phantom{i-1,}i-s}(u+1)
\\[1mm]
=&\;  \Bigl(H_{i,i+r}(u+1)\, +\, \frac{1}{[u][u+1]}\,
H_{i+1,i+r}(1)\Bigr)H_{i-1,i+r-1}^{\phantom{i-1,}i-s}(u+1)
\quad\mbox{\rm mod}\; A_{i-s,i+r-s-1}
\end{align*}
\begin{align*}
=&\;  H_{i,i+r}(u+1)\,H_{i-1,i+r-1}^{\phantom{i-1,}i-s}(u+2)
\\
& +\, \frac{1}{[u+1][u+2]}\Bigl(H_{i,i+r}(u+1)\,H_{i,i+r-1}(1)\Bigr) H_{i-2}^{i-s}(1)\,
H_{i-1,i+r-2}^{\phantom{i-1,}i-s+1}(u+3)
\\
& +\, \frac{1}{[u][u+1]}\, H_{i+1,i+r}(1)
H_{i-1,i+r-1}^{\phantom{i-1,}i-s}(u+1)\;\;\mbox{\rm mod}\; A_{i-s,i+r-s-1}\oplus A_{i+r-s+1,i+r-1}
\\[1mm]
=&\; H_{i,i+r-1}^{\phantom{i,}i-s+1}(u+1)
\\
&+\,\frac{1}{[u][u+1]}\, H_{i+1,i+r}(1)\Bigl(\frac{[u]}{[u+1]}H_{i-2}^{i-s}(1)
\, +\,H_{i-1}^{i-s}(u+1)\Bigr) H_{i,i+r-1}^{\phantom{i,}i-s+1}(u+2)
\\
&\hspace{76mm}\mbox{\rm mod}\; A_{i-s,i+r-s-1}\oplus A_{i+r-s+1,i+r}
\end{align*}
Here in the second line we used \eqref{6.3} for $v=1$ and take into account
the fact
\[
A_{i,I+r-1}\, H_{i-1,i+r-1}^{\phantom{i-1,}i-s}(u+1)\,\subset\, A_{i-s,i+r-s-1}\, .
\]
When passing to the third line we used the induction assumption and then
permuted two terms in parentheses in the fourth line using the Yang-Baxter equation
\eqref{ybegraphic}. The result of this permutation contains the rightmost factor $h_{i+r}(u+1)$
which can be evaluated as $\, [u+1]/[u]$ modulo $A_{i+r}$.
Finally, we notice that by obvious symmetry arguments
the  mirror images of relations \eqref{6.3} are valid for
the downhill strips $H_i^j(u)$.
Therefore, the term taken in parentheses in the last line equals $H_{i-1}^{i-s}(1)$.
\end{proof}
\begin{proof}[Proof of Proposition~\ref{prop:honpsi3}]
The simpler cases $K=1,2$ were
already considered in Corollaries \ref{cor3} and \ref{cor:honpsi2}.
In general the calculation of $h_i(1)\psi_\alpha$ can be carried out in the following steps:
\medskip

\noindent
{\em Step 1.\,} Using the transformation \eqref{6.10} we extract the term $\psi_{\alpha^-}$
from $h_i(1)\psi_\alpha$. The residual term equals $\psi_{\alpha^{+1}}$ in case
$K=1$, see Corollary~\ref{cor3}.
\medskip

\noindent
{\em Step 2.\,} In case $K>1$ the residue needs further transformation.
Namely, to fit the ansatz \eqref{psi_a-A}
one has to rise by one the arguments in all tiles
of the strip contained between
the uphill lines starting at height $h=\alpha_{i+1}$ at points $i-1$ and $i+2p_1+1\equiv i+1$
and the downhill lines starting at the same height $h$ at points $i+2p_2+1$ and $i+2r+1$
(see the dashed strip in the picture in the second line of Figure~\ref{fig15}). Acting in this
way we extract the coefficient $\psi_{\alpha^{+1}}$ from the first step residue, see \eqref{6.12}, and the rest, in case $K=2$, can be easily transformed to the form of
$\psi_{\alpha^{+2}}$, see Figure~\ref{fig15}.
\medskip

\noindent
{\em Step 3.\,} In case $K>2$ the residual term of the second step does not fit the
the ansatz \eqref{psi_a-A} and has to be further transformed.
This time one has to rise by one the arguments in the block of tiles
contained between the uphill lines crossing the points $i-1$ and $i+2p_2+1$
at the height $h$ and the downhill lines crossing the points $i+2p_3+1$ and $i+2r+1$
at the same height. An example of the second step residue is given in Figure~\ref{fig17}.
\begin{figure}[t]
\includegraphics[width=0.5\textwidth]{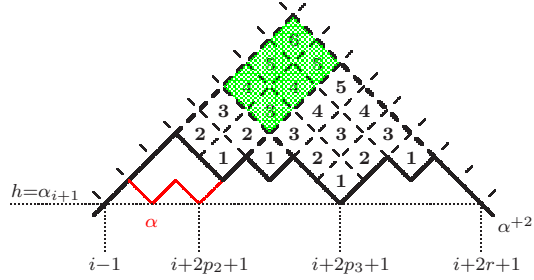}
\caption{The case $K=3$,
a typical diagram of  the second step residue.
The term which needs further transformation is the shaded block of tiles.
}
\label{fig17}
\end{figure}
\begin{figure}[t]
\includegraphics[width=0.9\textwidth]{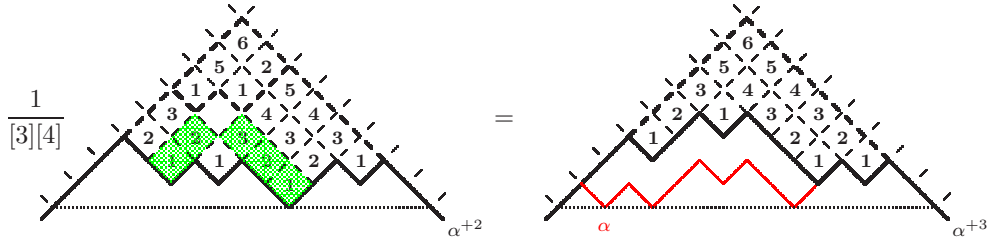}
\caption{
Transformation of the third step residue in case when the second step residue is
given by Figure~\ref{fig17}.
Using the Yang-Baxter equation \eqref{ybegraphic}
shaded uphill/downhill strip can be moved
up and left/right (see  explanation
in Figure~\ref{fig14}) and then evaluated with the use of \eqref{6.5}.
Since $K=3$ in this example, the third step residue  equals
$\psi_{\alpha^{+3}}$.
}
\label{fig18}
\end{figure}
We can raise the arguments in the block using the result of Lemma
\ref{lem4}, see \eqref{block} and Figure~\ref{fig16}. The first term from the right hand side
of \eqref{block} gives rise to the coefficient $\psi_{\alpha^{+2}}$.
The second term is the third step residue  which can be further simplified using
the Yang-Baxter equation \eqref{ybegraphic} and the evaluation relation \eqref{6.5}.
In case $K=3$ the result of the transformation coincides with $\psi_{\alpha^{+3}}$.
For the example of Figure~\ref{fig17} the transformation is illustrated in
Figure~\ref{fig18}.
\medskip

{}From now on the consideration acquires its full generality and we continue the transformation
until it ends up at {\em Step $K$}.
\end{proof}
Up to now we have finished the proof of the bulk qKZ equation \eqref{TLei}
for the coefficients $\psi_\alpha$ whose corresponding Dyck paths $\alpha$
have a local maximum at $i$ and a neighbouring local minimum at $i+1$. Consideration of the cases where
$\alpha$ has a neighbouring local minima at $i-1$, or both at $i-1$ and $i+1$ is a repetition
of the same arguments.
\medskip

It lasts to check the type A boundary qKZ relations
\eqref{qKZTL_TypeB2} and \eqref{qKZTL_TypeB3}. By Corollary \ref{cor:psi0} these
conditions are verified by the coefficient $\psi_\Omega^{\rm A}$ of the maximal Dyck path. We then notice that none of the factors of
the operator $H_\alpha$ \eqref{H_a} affect the coordinates $x_1$ and $x_N$
and hence commute with the boundary reflections $\pi_0$ and $\pi_N$, see \eqref{pi0} and\eqref{piN}.
Therefore, each of the coefficients $\psi_\alpha = H_\alpha \psi_\Omega^{\rm A}$ \eqref{psi_a-A}
also satisfies the conditions \eqref{qKZTL_TypeB2} and \eqref{qKZTL_TypeB3}.
\medskip

This completes the proof of Theorem~\ref{th:facsol}.

\subsection{Proof of Theorem~2
}
\label{sec:ProofFacSolB}

The proof is analogous to that of Theorem~\ref{th:facsol} for type A,
except that now we also have to make use of the reflection equation \eqref{rea-h-graphic} for $\bar h_0$.
Again, following the preliminary analysis of the Section~\ref{subsec3.3.1} we divide the proof of
\eqref{qKZTL_TypeB1} into two parts.
\medskip

\noindent
{\bf  1.}\; {\em In case $\alpha$ does not have a local maximum at $i$}
we have to show that \eqref{proj} is satisfied for $\psi_\alpha$ given by \eqref{psi_a-B}.
The working is identical to that in type A, see Section~\ref{sec:ProofFacSolA},
except for the case when $h_i(-1)$ acts on an uphill slope starting at the left boundary.
In this case we additionally use \eqref{rea-h-graphic}
to reflect $h_1(-1)$ at the boundary, see the illustration on Figure~\ref{fig23}.

\begin{figure}[h]
\includegraphics[width=0.4\textwidth]{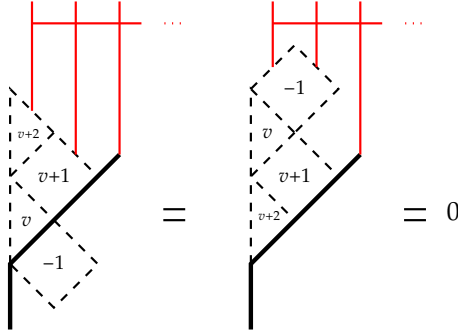}
\caption{
Reflection of $h_1(-1)$ at the origin. The first equality follows from the reflection
equation \eqref{rea-h-graphic} and the second one
is a property of $\psi^{\rm B}_\Omega$ \eqref{psi_alpha_B}:
$h_1(-1)\psi^{\rm B}_\Omega =0$.
}
\label{fig23}
\end{figure}

\noindent
{{\bf  2.}\;\em If $ \alpha$ has a local maximum at $i$}, then
we need to prove that Theorem~\ref{th:facsolB} implies \eqref{TLeiB}.
Here again, the working is identical to that in type A except for the case where
$ \alpha$ has a local minimum at $i-1$ and has no lower local minima
between $0$ and $i-1$. In this case, \eqref{TLeiB} contains
the term $\psi_{\alpha^{+0}}$ which originates from reflections at the left boundary.
The proof still follows basically the same lines as in type A although the calculations
become quite elaborate. Therefore we decided to collect the necessary technical tools
in the Lemma below and then to illustrate the idea of the proof on a few examples in pictures.

Let us introduce the following notation:
\begin{align*}
H_{0,i}(u)\, &:=\,\bar h_0(u)\, H_{1,i}(u+1)\, ,\qquad\quad
H^0_i(u)\, :=\,  H^1_i(u)\, \bar h_0(u+i)\, ,
\\[1mm]
T_{0,i}(u)\, &:=\, H_{0,i}(u)\times H_{0,i-1}(u+2)\times \dots
\times H_{0,1}(u+2i-2)\times \bar h_0(u+2i)\, .
\end{align*}
Pictorially $H_{0,i}(u)$ and  $H^0_i(u)$ can be displayed as
uphill and downhill strips starting with the half-tile at the left boundary,
and $T_{0,i}(u)$ is a right triangle whose hypotenuse lies on the left boundary vertical line.
We also extend the domain of definition for $\bar h_0(u)$ \eqref{half-tile} demanding that
\be
\lb{h_0-reg}
\bar h_0(0) \, :=\, s_0\, -\, \frac{[\omega]}{2[\omega+1]}\, .
\ee
For the newly introduced quantities the following analogues of
equation \eqref{6.10} and Lemma\,\ref{lem4} hold
\begin{figure}[t]
\includegraphics[width=0.7\textwidth]{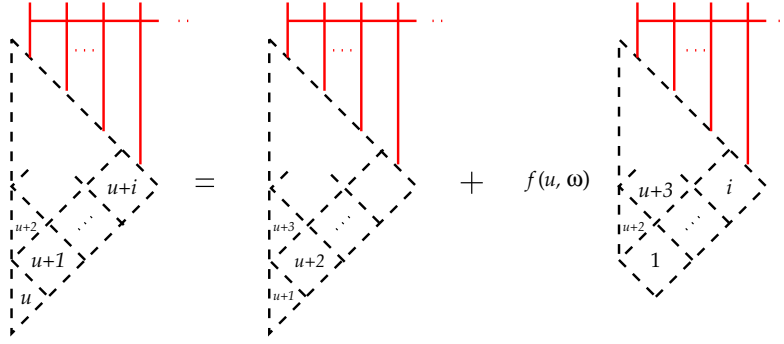}
\caption{
Diagrammatic presentation of \eqref{block}.
The coloured vertical lines on top of the tilted triangles
annihilate the term $A_{1,i}$.
}
\label{bound_rem}
\end{figure}
\begin{lemma}
\label{lem5}
One has\medskip
\begin{enumerate}
\item[]\vspace{-6mm}
\be
\lb{step-1}
\mbox{\rm 1.}\quad h_i(1)\, H_{i-1}^0(1)\, =\, H_{i}^0(1)\, +\, H_{i-2}^0(1)\, +\, c_0(i)\quad
\mbox{\rm mod}\, A_{1,i}\, \quad \forall\, i=1,2,\dots ,
\ee\\[-3mm]
where $c_0(i)=1$ if $i$ is odd, $c_0(i)=-\tfrac{[\omega]}{[\omega+1]}$ if $i$ is even, and we assume
$H_{-1}^0(1)=0$;
\vspace{2mm}
\item[2.] for nonnegative integers $\; i\, $ and $\; u\, $
\be
T_{0,i}(u)\, =\, T_{0,i}(u+1)\, +\, f(u,\omega)\,
H_{1,i}(1)\, T_{0,i-1}(u+2)\, ,
\label{triangle-up}
\ee
where\,\footnote{In principle one can choose a regularisation
for $\bar h_0(0)$ which is different from \eqref{h_0-reg}, but then the prescription for
$f(0,\omega)$ should be changed correspondingly.}\vspace{-2mm}
\be
\lb{f(u,w)}
f(u,\omega)=\left\{
\renewcommand{\arraystretch}{2.2}
\begin{array}{rl}
\displaystyle-\,\frac{[\omega]}{2[\omega+1]}\quad &\mbox{if\; $u=0$,}
\\[2mm]
\displaystyle\frac{[p][p-\omega]}{[2p][2p+1][\omega+1]}\quad &\mbox{if\; $u=2p$\; is even positive,}
\\[2mm]
\displaystyle\frac{[p][p+\omega]}{[2p-1][2p][\omega+1]}\quad &\mbox{if\; $u=2p-1$\; is odd.}
\end{array}
\right.\vspace{1mm}
\ee
and we assume $H_{1,0}(1)=T_{0,-1}(u)=1$.
\end{enumerate}
\end{lemma}
Equation \eqref{triangle-up} is displayed graphically in Figure~\ref{bound_rem}.

\begin{proof} Equations \eqref{step-1} follow easily from
(the mirror images of) \eqref{6.3}, \eqref{6.5}, and from \eqref{half-tile}.
Relations \eqref{triangle-up} can be proved by induction.
The considerations are standard and we just briefly comment on them.

For $i=0$,  \eqref{triangle-up} follows from \eqref{half-tile}.
To check the induction step $i\rightarrow i+1$ one makes a decomposition
\[
T_{0,i+1}(u)\, =\, \bar h_0(u)\, H_{1,i+1}(u+1)\, T_{0,i}(u+2)
\]
and applies formulas \eqref{6.3}, \eqref{half-tile} and the induction assumption to rise
consecutively the contents by one of the factors $H_{1,i+1}(u+1)$, $\bar h_0(u)$ and $T_{0,i}(u+2)$.
Recollecting the (half-)tiles in the resulting expressions with the help of the Yang-Baxter equation
\eqref{ybegraphic} and the reflection equation \eqref{rea-h-graphic},  and using the evaluation formulas
\eqref{6.5} and \eqref{half-tile} together with its consequence
\[
\bar h_0(u)\bar h_0(u+2)\, =\, a\, \bar h_0(u+2)\, +\, b\qquad\mbox{(for some numbers $a$ and $b$),}
\]
one finally reproduces the term $T_{0,i+1}(u+1)$ and a combination of terms $H_{1,i+1}(u+1)T_{0,i}(u+2)$
and $H_{2,i+1}(1)T_{0,i}(u+2)$. The latter can be simplified to $H_{1,i+1}(1) T_{0,i}(u+2)$,
thanks to the relations \eqref{s1}. The unwanted terms $H_{2,i+1}(1)H_{1,i}(u+3)T_{0,i-1}(u+4)$
appearing at the intermediate steps cancel in the final expression.

We mention here that all manipulations with q-numbers, which one
needs during the transformation, can be easily done with the help of
the following identities
\be
\lb{uvk}
[u+k][v\pm k]\, =\, [u] [v] \, +\, [k][v\pm (u+k)]\, .
\ee
\end{proof}\vspace{-5mm}

\noindent
We now continue the proof of \eqref{TLeiB}.

Consider the action of $h_i(1)$ on the Ballot path $\alpha$ which has a local minimum at
$i-1$ and all the local minima in between $0$ and $i-1$ are higher then $\alpha_{i-1}$.
In this case, applying relation \eqref{step-1}, we find
\be
\lb{one-contact}
h_i(1)\psi_\alpha\, =\, \psi_{\alpha^-}\, +\, \psi_{\alpha^{+1}}\, +\, \psi_{\alpha^{+0}}\, ,
\ee
where the paths $\alpha^-$, $\alpha^{+0/1}$ are defined in Figure~\ref{alphapm0} and their
corresponding coefficients are given by \eqref{psi_a-B}. Note that, according to
Lemma~\ref{lem5}.1, in the case
$i=1$ the term $\psi_{\alpha^{+1}}$ should be absent from the right hand side of \eqref{one-contact}.
Altogether these prescriptions are identical to those of \eqref{TLeiB}. The transformation
\eqref{one-contact} is illustrated in Figure \ref{bound_simple}.

\begin{figure}[t]
\includegraphics[width=0.8\textwidth]{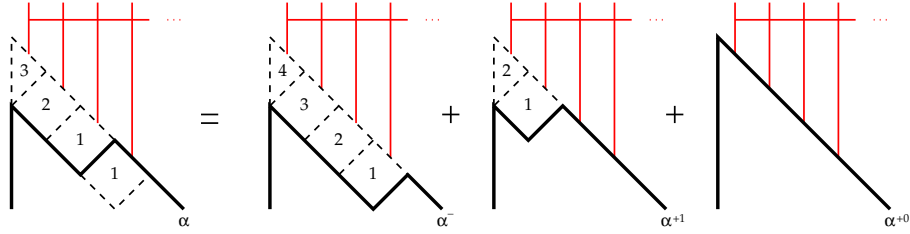}
\caption{
Diagrammatic presentation of  relation \eqref{one-contact}.
In the case shown on Figure $i=3$ and, hence, the coefficient in front of $\psi_{\alpha^{+0}}$
equals $1$.
}
\label{bound_simple}
\end{figure}

We now consider the case where the path $ \alpha$ contains $m\geq 1$ local minima
between $0$ and $i-1$ which are of the same height as the minimum at $i-1$.
The proof can be carried out in $K=m+2$ steps (cf. the proof of Proposition\,\ref{prop:honpsi3}). To explain the first three steps we consider the case $m=1$, i.e., a path
$ \alpha$ with the two local minima placed at $i-1$ and $i-2p-1$, $p>0\,$:
$\alpha_{i-2p-1}=\alpha_{i-1}$. Examples of such paths are given in Figure~\ref{bound_gen}.
\begin{figure}[h]
\begin{eqnarray}
\lefteqn{\hspace{-2mm}\mbox{\small {\bf a)} $i$ odd and $p=\tfrac{i-1}{2}$}
\hspace{3mm}\mbox{\small {\bf b)} $i$ even and $p=\small\tfrac{i}{2} -1$}
\hspace{3mm}\mbox{\small {\bf c)} $i=8$, $p=2$}}
\includegraphics[width=0.75\textwidth]{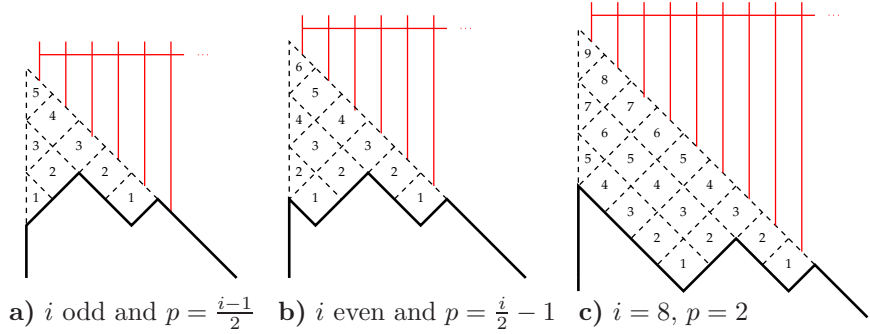}
\nonumber
\end{eqnarray}
\caption{
Ballot paths $\alpha$ with two local minima of the same height, one at $i-1$
and another one at $i-2p-1$.  In these cases one calculates $h_i(1)\psi_\alpha$ in $K=3$ steps.
}
\label{bound_gen}
\end{figure}
The calculation of $h_i(1)\psi_{\alpha}$ for the path shown on Figure~\ref{bound_gen}\,c) is
illustrated in Figures~\ref{bound_hpsi}--\ref{bound_R2}.
\medskip

\noindent
{\em Step 1.\, Extraction of coefficient $\psi_{\alpha^-}$ from $h_i(1)\psi_\alpha$, see Figure~\ref{bound_hpsi}.}

We use \eqref{6.10} to raise by one the contents in the shaded strip $H_{i-1}^1(1)$ on
Figure~\ref{bound_hpsi}. The result is a sum of two terms. The second term is a residue
of the first step which is to be further transformed at the second step. We denote it by $R_1$.
Here we transform the first term
raising by one the content of its top half-tile $\bar h_0(i)$ (shown shaded on Figure~\ref{bound_hpsi}) with the help
of \eqref{triangle-up}.
This results in a sum of $\psi_{\alpha^-}$ and a term proportional to $\psi_{\alpha^{+0}}$, see the second line in Figure~\ref{bound_hpsi}.
The factor $[9]$ in the coefficient  of $\psi_{\alpha^{+0}}$ comes from the evaluation of
the top strip $H_i^1(1)$.
\begin{figure}[h]
\raisebox{42pt}{$h_i(1)\psi_{\alpha}\, =\,$}
\includegraphics[width=0.8\textwidth]{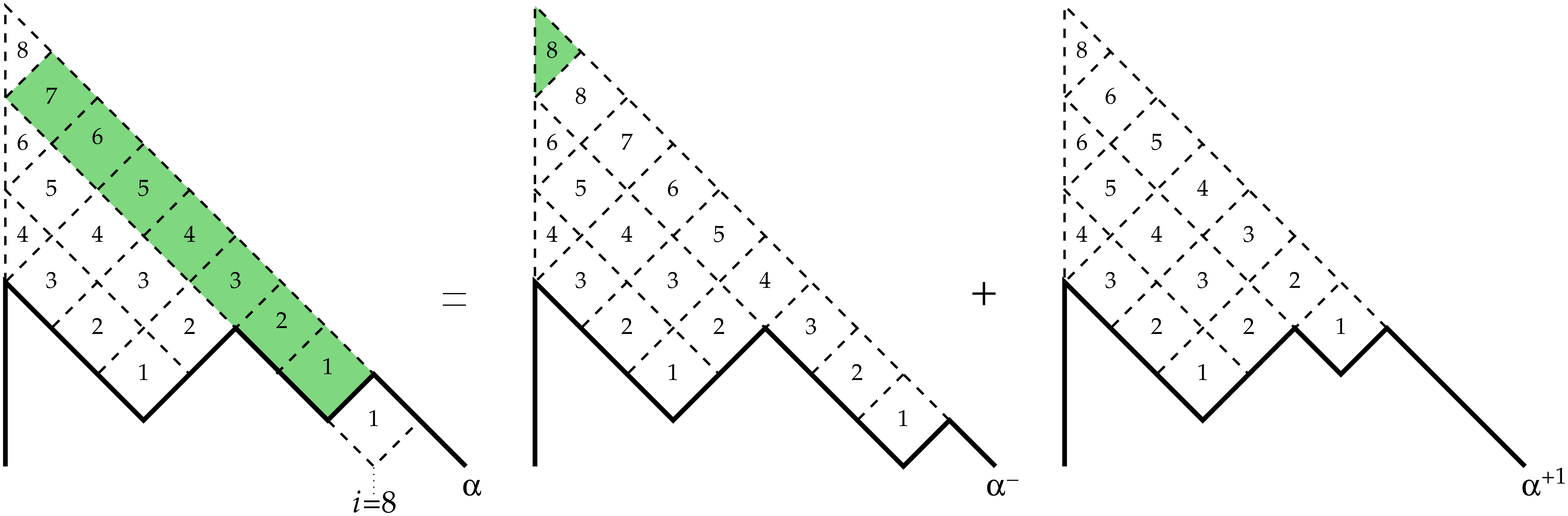}
\vspace{0.5mm}
\[
=\, \psi_{\alpha^-}\, +\, [9]\,\frac{[4][4-\omega]}{[8][9][\omega+1]}\,
\psi_{\alpha^{+0}}\, +\, R_1
\]
\caption{
Extraction of the coefficient $\psi_{\alpha^-}$ from the expression $h_i(1)\psi_{\alpha}$.
}
\label{bound_hpsi}
\end{figure}
\begin{figure}[h]
\begin{eqnarray*}
\lefteqn{\hspace{55mm} \raisebox{46pt}{$\displaystyle \psi_{\alpha^{+1}}\, +\,\frac{1}{[3]}$}}
\raisebox{45pt}{$R_1\; =\;$}
\includegraphics[width=0.7\textwidth]{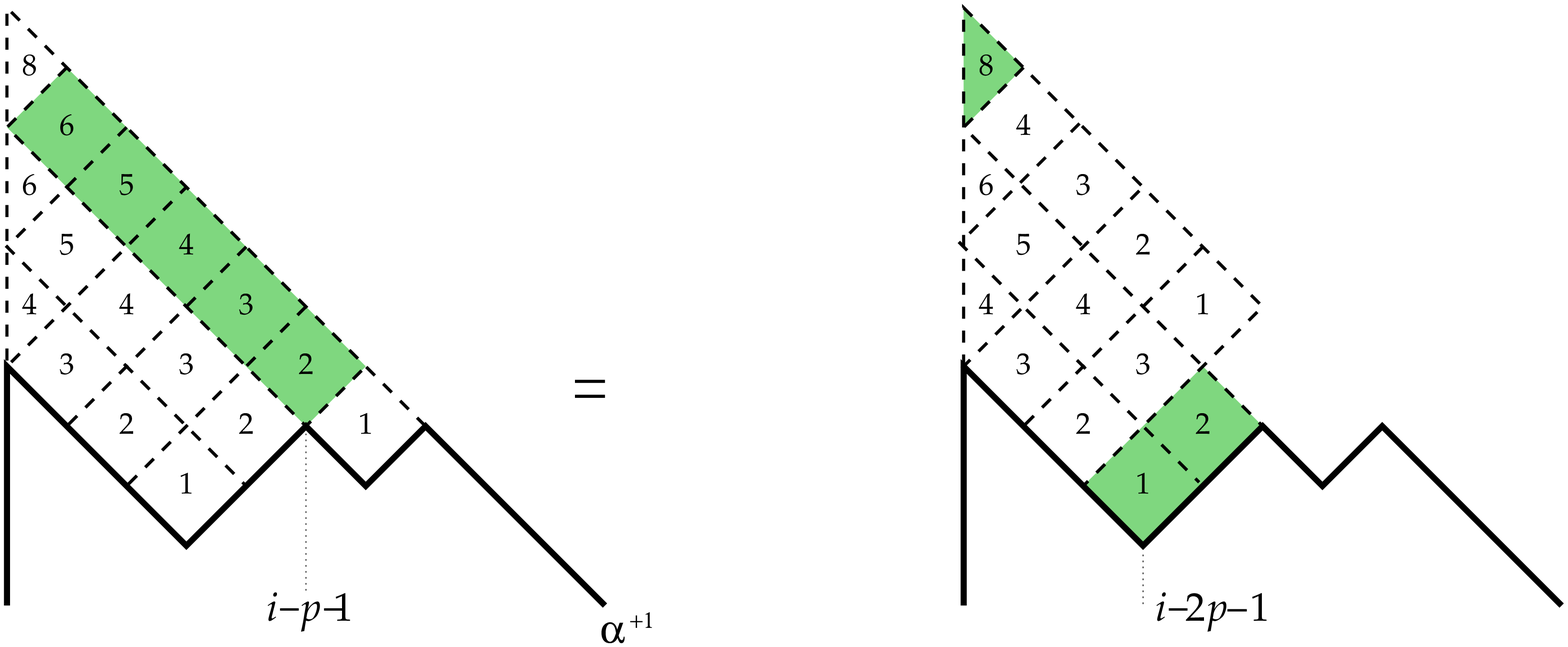}
\nonumber
\end{eqnarray*}
\vspace{-0.5mm}
\[
=\, \psi_{\alpha^{+1}}\, +\, \frac{[5]}{[3]}\,\Bigl(\frac{1}{[2]} -\frac{[4][\omega+4]}{[8][\omega+1]}\Bigr)\,
\psi_{\alpha^{+0}}\, +\, R_2
\]
\caption{
Extraction of coefficient $\psi_{\alpha^{+1}}$ from $R_1$.
}
\label{bound_R1}
\end{figure}

\noindent
{\em Step 2.\, Extraction of coefficient $\psi_{\alpha^{+1}}$ from the first step residue, see Figure~\ref{bound_R1}.}

Again, we use \eqref{6.10} to raise by one the contents in the shaded strip $H_{i-p-1}^1(p)$ on
Figure~\ref{bound_R1}. The result is a sum of $\psi_{\alpha^{+1}}$ and a term which
we continue transforming. Using the definition of $\bar h_0$ \eqref{half-tile}
we lower the content of the top half-tile (shown shaded)
from $i$ to $i-2p-2$ (from $8$ to $2$ in the particular case shown on
Figure~\ref{bound_R1})\,:
\be
\lb{h0-8-2}
\bar h_0(8)\, =\,\left(\tfrac{1}{[2]}-\tfrac{[4][\omega+4]}{[8][\omega+1]}\right)\, +\,\bar h_0(2)\, .
\ee
The constant term resulting from this procedure
gives a contribution proportional to $\psi^{\alpha^{+0}}$, see the second line of Figure~\ref{bound_R1}.
The factor $[5]$ appearing in the coefficient of $\psi_{\alpha^{+0}}$ is due to the evaluation of the
top strip $H_4^1(1)$ (in general one evaluates $H_{i-p-2}^1(1)\rightarrow [i-p-1]$).

Lowering of the content in the top half-tile allows us to reorder the (half-)tiles of the last term
in the first line of Figure~\ref{bound_R1}. Namely,
analogously to the case considered in Figure~\ref{fig14} we can push the uphill strip
$H_{i-2p-1,i-p-2}(1)$ (shown shaded) up and left using the Yang-Baxter equation
until it touches the left boundary. Then we reflect the strip at the boundary as shown
on Figure~\ref{boundary_blockperm}, which is possible because we changed the content of $\bar h_0$ from 8 to 2 in \eqref{h0-8-2}. Finally, after reflection, the strip can be evaluated, cancelling the
numeric factor in front of the picture. We call the result of this transformation
a residue of the second step and denote it by $R_2$.
\begin{figure}[t]
\includegraphics[width=0.5\textwidth]{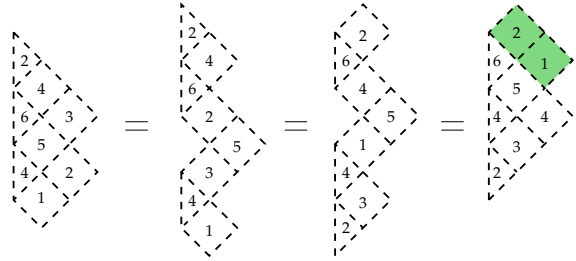}
\caption{Reflection of a strip at the boundary. Here we use the Yang-Baxter equation \eqref{ybegraphic}
in the first and the last equalities and the reflection equation \eqref{rea-h-graphic}
in the second equality. The shaded downhill strip equals $[3]$ modulo $A_{1,2}$.
}
\label{boundary_blockperm}
\end{figure}
\medskip

\noindent
{\em Step 3.\, Extraction of coefficient $\psi_{\alpha^{+2}}$ from the second step residue, see Figure~\ref{bound_R2}.}

We use \eqref{triangle-up}, see also Figure~\ref{bound_rem} to increase the contents of the
triangle $T_{0,p}(i-2p-2)$ (shown shaded on the figure). The result is a sum of $\psi_{\alpha^{+2}}$
and a term which in fact is proportional to $\psi_{\alpha^{+0}}$. To prove this, one has to push
up the downhill strip $H_{i-2p-2}^1(1)$ (shown shaded on the figure) and then evaluate it
in the same way as it was done in the transformation \eqref{6.10}, see Figure~\ref{fig14}.
\begin{figure}[t]
\begin{eqnarray*}
\lefteqn{\hspace{44.5mm} \raisebox{46pt}{$=\,\psi_{\alpha^{+2}} +\frac{[1-\omega ]}{[2][3][\omega +1]}$}}
\raisebox{45pt}{$R_2\; =\;$}
\includegraphics[width=0.75\textwidth]{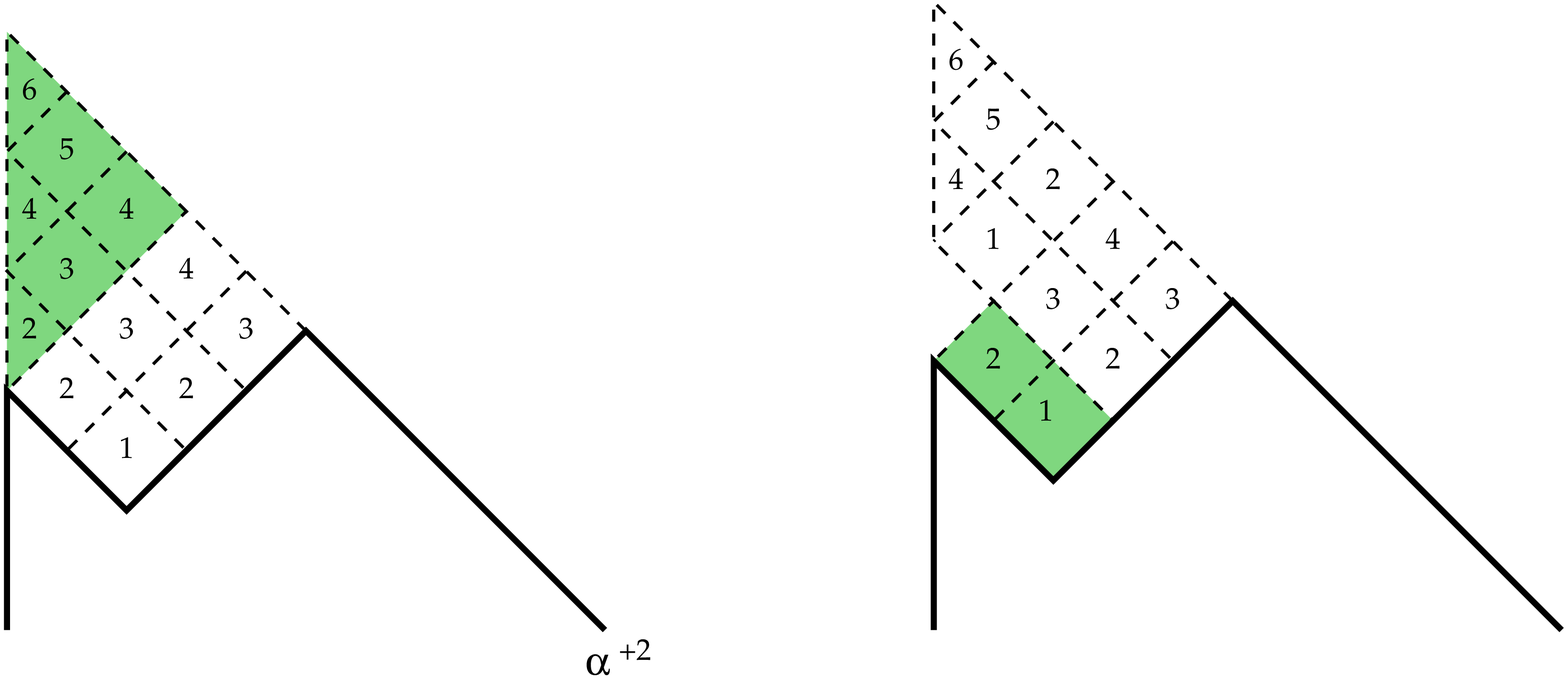}
\nonumber
\end{eqnarray*}
\vspace{-3mm}
\[
=\, \psi_{\alpha^{+2}}\, +\, \frac{[1-\omega]}{[2][\omega+1]}\,
\psi_{\alpha^{+0}}
\]
\caption{
Extraction of coefficient $\psi_{\alpha^{+2}}$ from $R_2$.
}
\label{bound_R2}
\end{figure}

Finally, collecting  the terms $\psi_{\alpha^{+0}}$ from all three steps
we find
\be
\lb{K=1case}
h_i(1)\psi_\alpha\, =\, \psi_{\alpha^-}\, +\, c_0(i)\,\psi_{\alpha^{+0}}\, +\,
\psi_{\alpha^{+1}}\, +\, \psi_{\alpha^{+2}}\, ,
\ee
where $c_0(i)=-\tfrac{[\omega]}{[\omega+1]}$ for the particular case considered in
Figures~\ref{bound_hpsi}--\ref{bound_R2}. This value holds for all cases with $i$ even,
while $c_0(i)=1$ for all cases with $i$ odd. In general, the coefficient $c_0(i)$ can be calculated with the help of \eqref{uvk}.

Equations \eqref{K=1case} coincide with the prescriptions of \eqref{TLeiB}
in case the path $\alpha$ contains $m=1$ local minimum between $0$ and $i-1$
of the same height as the minimum at $i-1$. Before we proceed to cases with $m\geq 2$
let us comment on two particular cases with $m=1$:
these are the cases a) and b) on Figure~\ref{bound_gen}
where the local minimum of the height $\alpha_{i-1}$ appears at $0$, or at $1$. Similar
exceptional cases appear for all values of $m$.

a){\em \, $\,i$ odd and $p=\frac{i-1}{2}$. } In this case the residue $R_2$ vanishes
so that the calculation of $h_i(1)\psi_\alpha$ finishes in two steps.
The term $\psi_{\alpha^{+2}}$ does not appear in \eqref{K=1case} which is in agreement
with  \eqref{TLeiB}. The contributions to $\psi_{\alpha^{+0}}$ from the first two steps sum up
to give the correct value of the coefficient $c_0(i)=1$. The mechanism how the residue $R_2$
vanishes for the path shown on Figure~\ref{bound_gen}\,a) is explained in Figure~\ref{bound_R2anni}.
\begin{figure}[t]
\begin{eqnarray*}
\lefteqn{\hspace{46mm} \raisebox{30pt}{$=\; \frac{[2]}{[3]}$}
\hspace{32mm} \raisebox{30pt}{$=\; 0$}}
\raisebox{30pt}{$R_2\; =\;\frac{1}{[3]}\;$}
\includegraphics[width=0.5\textwidth]{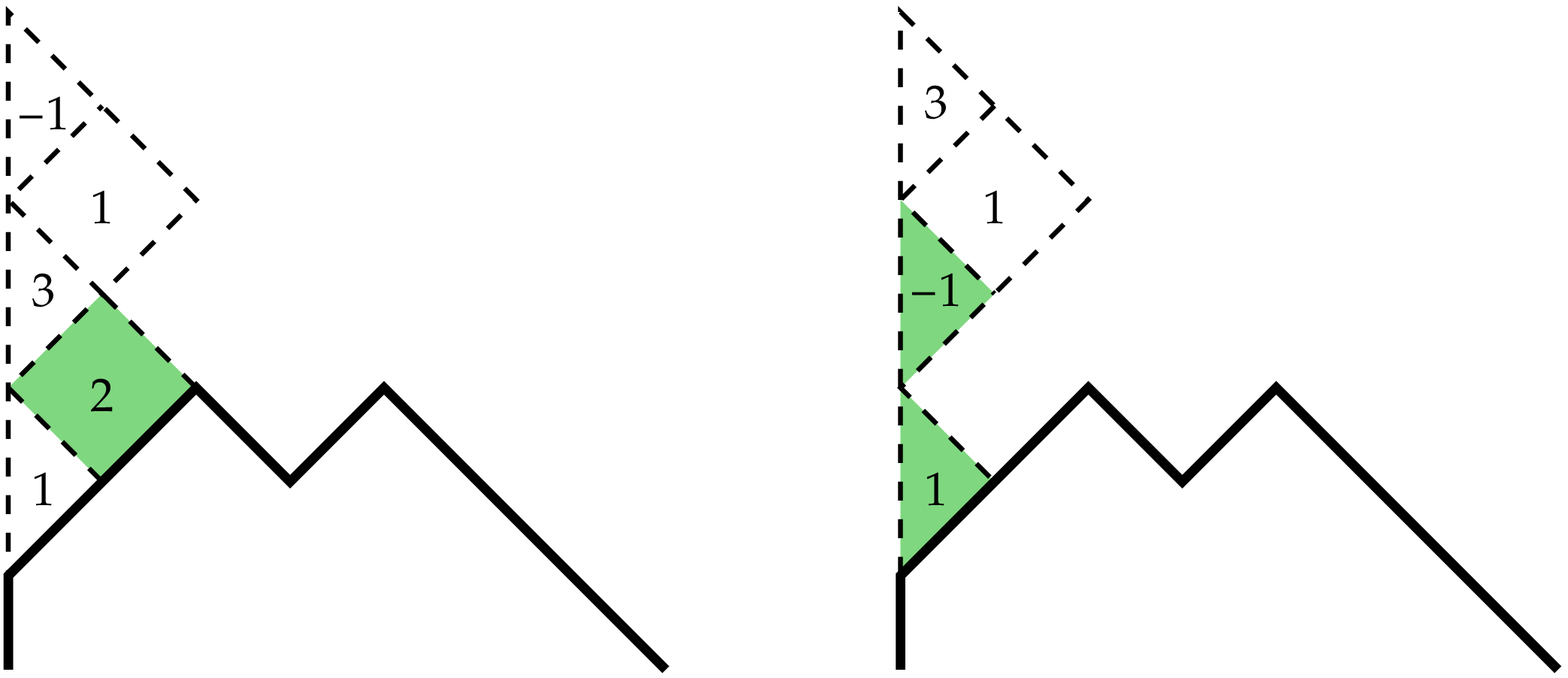}
\nonumber
\end{eqnarray*}
\vspace{-6mm}
\caption{
Vanishing of $R_2$ in $h_i(1)\psi_\alpha$ for the path $\alpha$ given by Figure~\ref{bound_gen}\,a).
In the second step we already changed the content of the top half-tile from $i=5$ to $i-2p-2=-1$.
Now we apply the reflection equation \eqref{rea-h-graphic} to move the shaded tile $h_1(2)$
up and evaluate it
using \eqref{vspomogat}. The two shaded
half-tiles then meet together and annihilate: $\bar h_0(1)\bar h_0(-1)=s_0(- a_0)=0$.
}
\label{bound_R2anni}
\end{figure}

b){\em\, 
$\,i$ even and $p=\frac{i}{2}-1$. } In this case, in the second step,
the content of the top half-tile has to be changed from $i$ to $i-2p-2=0$.
This is why we extended in \eqref{h_0-reg} the domain of definition for $\bar h(u)$
and derived \eqref{triangle-up} and \eqref{f(u,w)} for the case $u=0$.
With these extensions the calculation of $h_i(1)\psi_\alpha$ goes the standard way.
\smallskip
\begin{figure}[t]
\vspace{-3mm}
\begin{eqnarray*}
\raisebox{50pt}{$\displaystyle h_i(1)\psi_\alpha\; =\;\;\;$}
\includegraphics[width=0.3\textwidth]{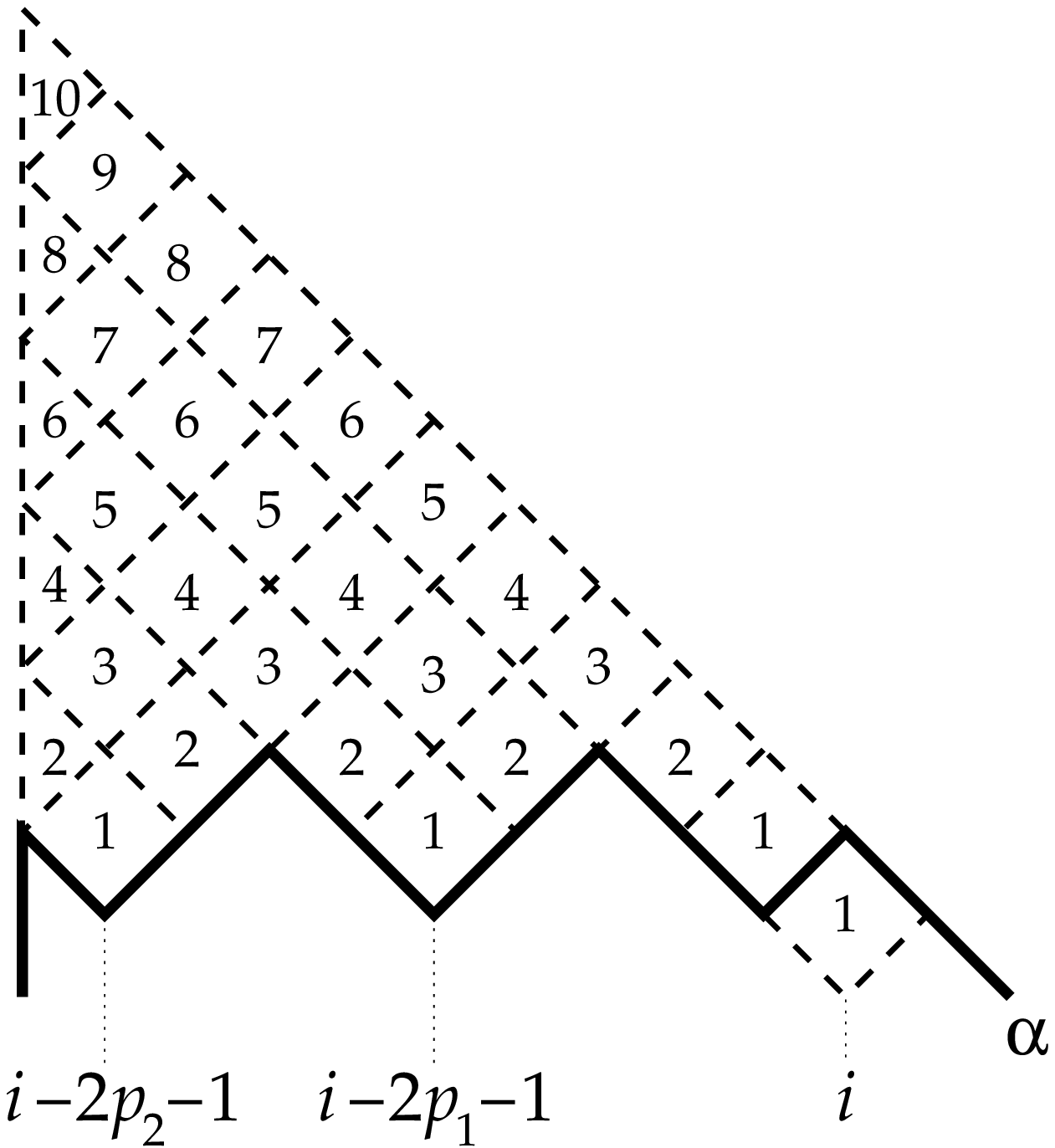}
\hspace{-8mm}
\raisebox{50pt}{$\displaystyle =\;\psi_{\alpha^-}\, - \,\frac{[\omega]}{[\omega +1]}\psi_{\alpha^{+0}}
\, +\,\psi_{\alpha^{+1}}\, +\, R'_2$}
\nonumber
\end{eqnarray*}
\vspace{-6mm}
\caption{
The result of the transformations described in Figures\,\ref{bound_hpsi}--\ref{bound_R2}
in the case $m\geq 2$.
The path $\alpha$  has local minima of the same height at points
$i-2p_k-1$, $k=1,\dots ,m$, and at $i-1$.
In the particular example shown here we have $i=10$, $m=2$, $p_1=2$ and $p_2=4$.
The term $R'_2$
in the right hand side comes in place of
$\psi_{\alpha^{+2}}$ in Figure~\ref{bound_R2}.
}
\label{bound_hpsigen}
\end{figure}

Consider now a path with $m\geq 2$ local minima preceding the minimum at $i-1$ which all
have the same height $\alpha_{i-1}$ (recall that we do not care about higher preceding minima and do not allow lower ones). In this case the transformations of the third step described earlier are not enough to extract the term $\psi_{\alpha^{+2}}$ and so we continue the transformation.
We explain this for the case of the path shown on Figure~\ref{bound_hpsigen}.
\medskip

\begin{figure}[t]
\begin{eqnarray*}
\lefteqn{\hspace{37.5mm}
\raisebox{66pt}{$ =\psi_{\alpha^{+2}}\, + \tfrac{1}{[2][3]}$}
\hspace{33mm}
\raisebox{66pt}{$ +\,\tfrac{[3][3+\omega ]}{[5][6][\omega +1]}$}
}
\raisebox{66pt}{$\displaystyle R'_2\; =\;$}
\includegraphics[width=0.9\textwidth]{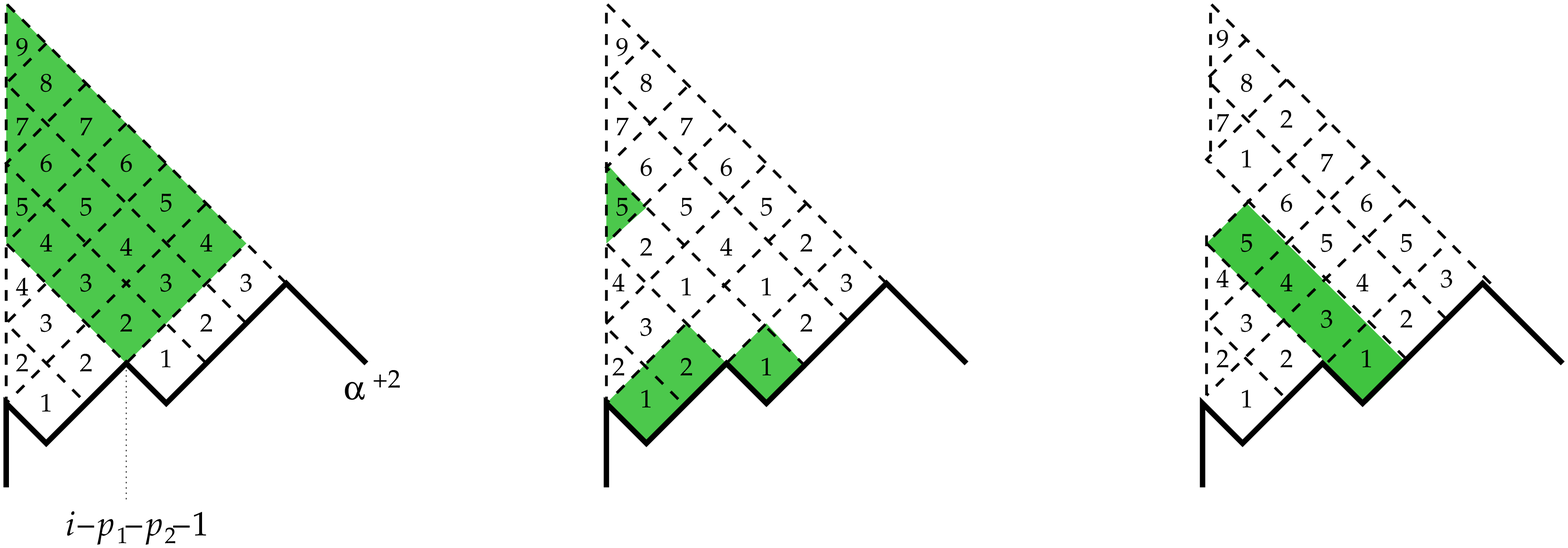}
\nonumber
\end{eqnarray*}
\vspace{-3mm}
\[
=\, \psi_{\alpha^{+2}}\,
+\, \frac{[\omega]}{2[\omega +1]}
\,\chi_{\alpha^{+0}}\, +\,R_3
\]
\caption{
Extraction of $\psi_{\alpha^{+2}}$ from $R'_2$ in the case $m\geq 2$.
}
\label{bound_Rp2}
\end{figure}
\begin{figure}[h]
\begin{eqnarray*}
\nonumber
\lefteqn{\hspace{100mm}
\raisebox{50pt}{$\displaystyle =\,\psi_{\alpha^{+3}}\, - \,\frac{[\omega ]}{2[\omega+1]}\,\chi_{\alpha^{+0}}$}
}
\raisebox{50pt}{$\displaystyle \chi_{\alpha^{+0}}\, =\;$}
\includegraphics[width=0.3\textwidth]{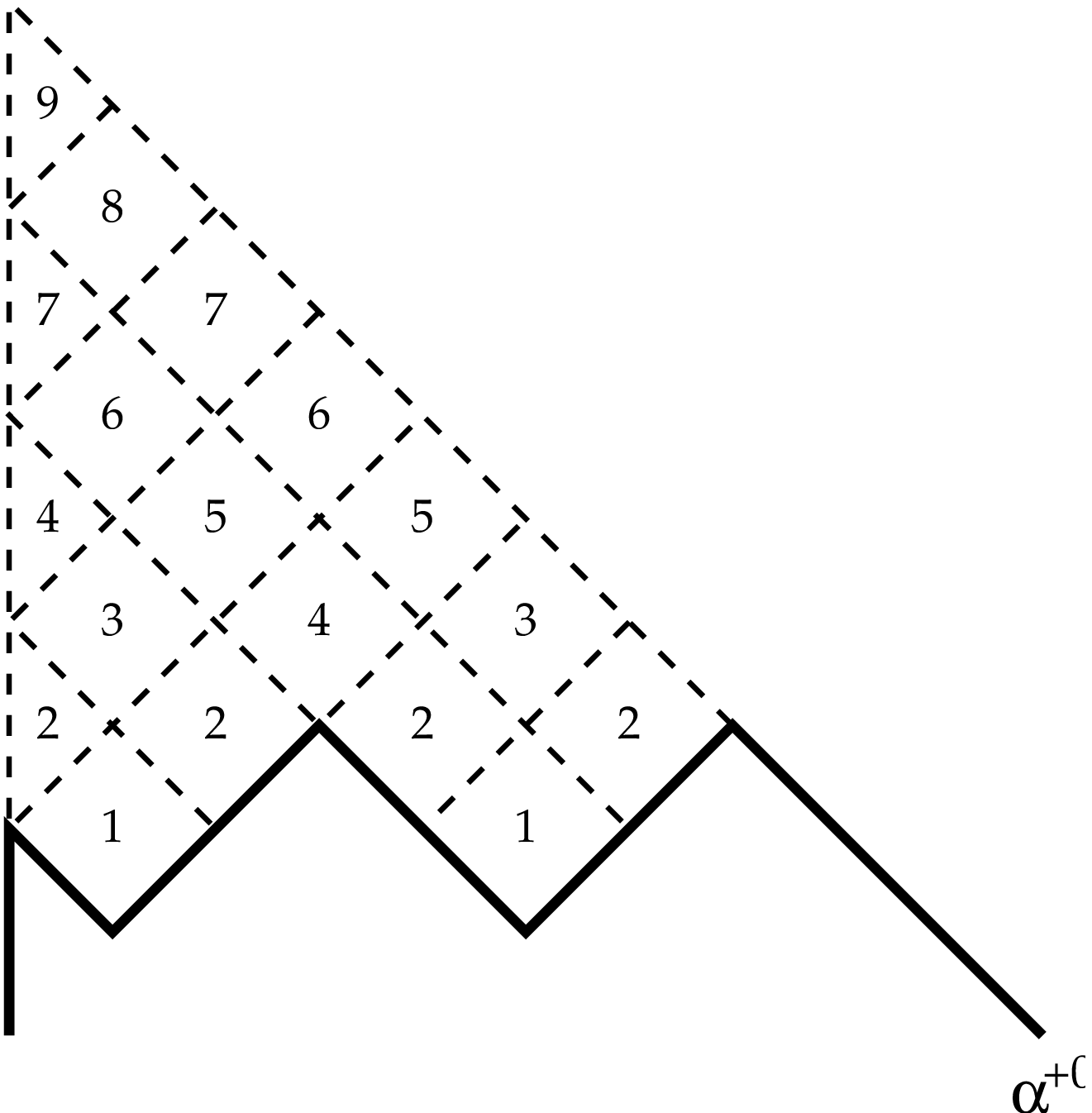}
\raisebox{50pt}{$\displaystyle R_3\, =\;$}
\includegraphics[width=0.3\textwidth]{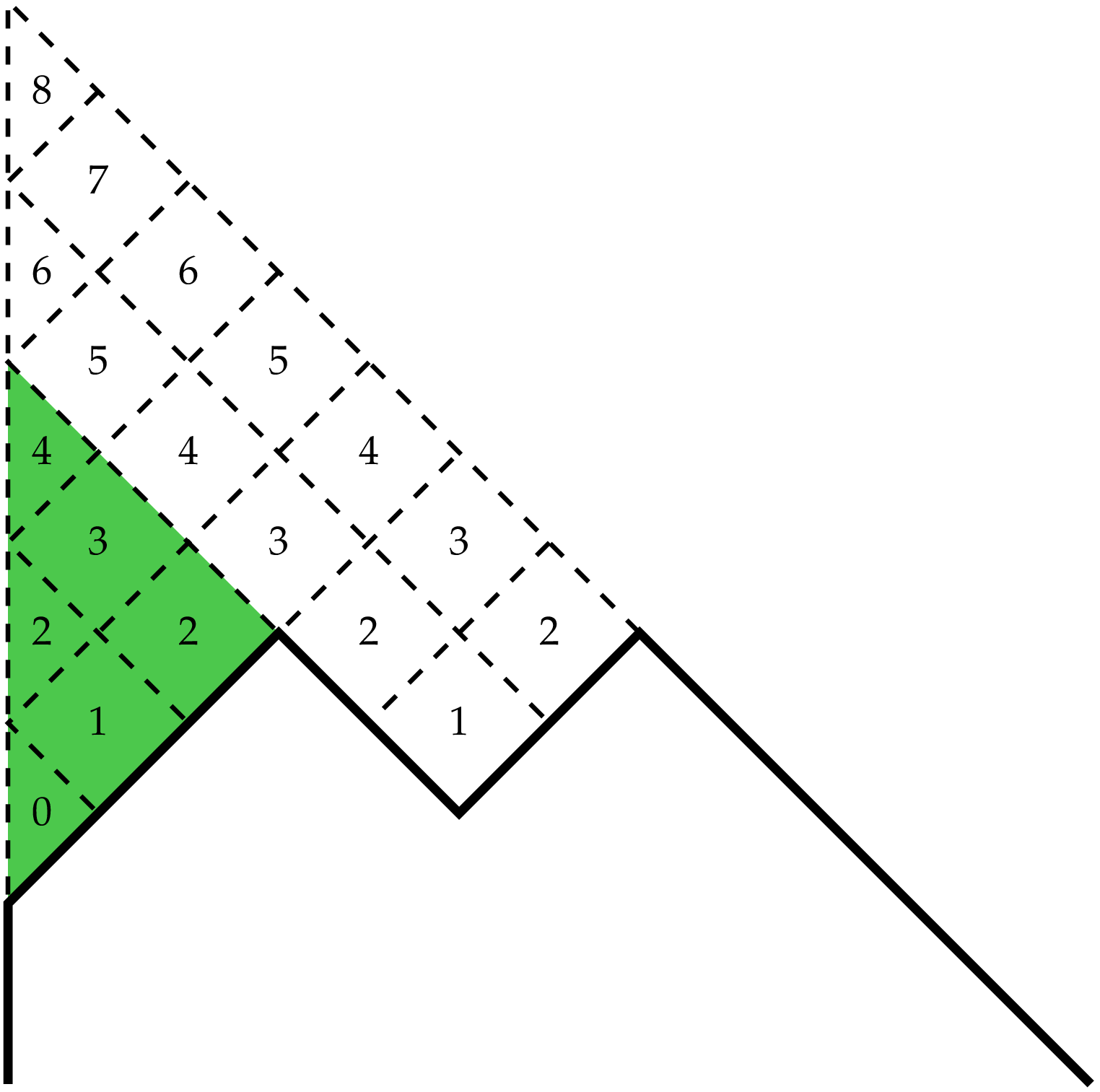}
\hspace{20mm}
\end{eqnarray*}
\caption{
Definition of the term $\chi_{\alpha^{+0}}$ and calculation of the third step residue $R_3$
in the case $m=2$.
}
\label{bound_chi0_R3}
\end{figure}

\noindent
{\em Continuation of the Step 3.\,} As can be seen on Figure~\ref{bound_hpsigen}, the terms
$\psi_{\alpha^-}$, $\psi_{\alpha^{+0}}$ and $\psi_{\alpha^{+1}}$ are already fixed. The term
$R'_2$ displayed in Figure~\ref{bound_Rp2} appears in the place of $\psi_{\alpha^{+2}}$ and we now continue its transformation. To extract the term $\psi_{\alpha^{+2}}$ we increase by one the contents of the
(half-)tiles in the shaded trapezium in the first equality on Figure~\ref{bound_Rp2}.
This trapezium is a composition  of a rectangle and a triangle and we consecutively use \eqref{block} and \eqref{triangle-up} to increase their contents, see also Figures~\ref{fig16} and \ref{bound_rem}). As a result, besides $\psi_{\alpha^{+0}}$ we get two more terms whose pictures are
shown in the second equality on Figure~\ref{bound_Rp2}.
Using the by now standard procedures of lowering the content of the boundary half-tile
(from $i-2p_1-1$ to $i-2p_2-2$ in general,
and from $5$ to $0$ in the specific example on Figure~\ref{bound_Rp2})
and pushing up, reflecting at the boundary and evaluating the strips of tiles,
we extract the third step residue $R_3$ from the middle picture in Figure~\ref{bound_Rp2}.
All the other terms can be reduced to the same form $\chi_{\alpha^{+0}}$.
Both terms $R_3$ and $\chi_{\alpha^{+0}}$ are shown in Figure~\ref{bound_chi0_R3}
(note that $\chi_{\alpha^{+0}}$ is composed of the same factors as $\psi_{\alpha^{+0}}$
but the contents may be different).
\smallskip

\noindent
{\em  Step 4.\,}
Further transformation of $R_3$ is identical to the calculation
of $R_2$, see  Figure~\ref{bound_R2}, and the result, for the case $m=2$, is presented
in Figure~\ref{bound_chi0_R3}. We obtain the term $\psi_{\alpha^{+3}}$
and the term $\chi_{\alpha^{+0}}$ which cancels similar term in the preceding transformation
(cf. the second line in Figure~\ref{bound_Rp2} and the right hand side in Figure~\ref{bound_chi0_R3}).

Collecting the terms in Figures~\ref{bound_hpsigen}--\ref{bound_chi0_R3}
we eventually find that for the case $m=2$
the factorised formulas \eqref{psi_a-B} indeed satisfy the type B qKZ equations
in the bulk \eqref{TLeiB}. Consideration of the cases with $m>3$ goes along the same lines.
\medskip

It lasts to check the type B boundary qKZ relation
\eqref{TLe0} (the boundary relation \eqref{qKZTL_TypeB3}
is valid due to the same arguments used
in the proof of Theorem~\ref{th:facsol}).
Indeed, rewriting \eqref{TLe0} as
\[
 \bar h_0(1)\, \psi_\alpha\, =\,
 \psi_{\alpha^{-0}}\,
\]
one makes the assertion obvious.
\medskip

This completes the proof of Theorem~\ref{th:facsolB}.

\appendix

\section{Factorised solutions for type B}
\label{facsolB}

\subsection{Case $N=2$}

Here we have two paths,
$\Omega^B =\raisebox{0pt}{\includegraphics[height=20pt]{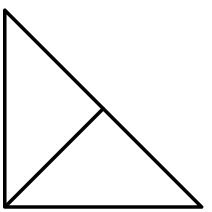}}$ and
$\raisebox{0pt}{\includegraphics[height=10pt]{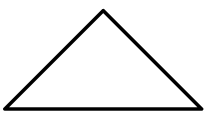}}$. From
the preliminary analysis we know that
$\psi_{\,\includegraphics[width=8pt]{2_210.eps}}$ is given by
\eqref{psi_alpha_B} and satisfies equation
\be
\lb{N=2a}
-a_1 \psi_{\,\includegraphics[width=8pt]{2_210.eps}}\, =\,
h_1(-1) \psi_{\,\includegraphics[width=8pt]{2_210.eps}}\, =\, 0.
\ee
Now we may apply the boundary generator to obtain the component function
$\psi_{\,\includegraphics[width=8pt]{2_010.eps}}$, and from \eqref{TLe0} we find
\be
\lb{N=2b}
s_0 \psi_{\,\includegraphics[width=8pt]{2_210.eps}}\, =\,
\psi_{\,\includegraphics[width=8pt]{2_010.eps}}.
\ee
Notice, that acting by $h_1(1)$  on $\psi_{\,\includegraphics[width=8pt]{2_010.eps}}$ we
can get back to $\psi_{\,\includegraphics[width=8pt]{2_210.eps}}$, see \eqref{TLeiB},
\[
h_1(1)\psi_{\,\includegraphics[width=8pt]{2_010.eps}}\, =\,
\psi_{\,\includegraphics[width=8pt]{2_210.eps}},
\]
which can be equivalently written as
\be
\lb{N=2d}
h_1(1) \Bigl( s_0 -\frac{1}{[2]}\Bigr)\psi_{\,\includegraphics[width=8pt]{2_210.eps}}\, =\, 0,
\ee
where we used \eqref{N=2a}, \eqref{N=2b} and relation \eqref{huv}.
Equation \eqref{N=2d} is the truncation condition
to be satisfied by $\psi_{\,\includegraphics[width=8pt]{2_210.eps}}$.

\subsection{Case $N=3$} In this case there are three paths:
$\Omega^B =\raisebox{0pt}{\includegraphics[height=20pt]{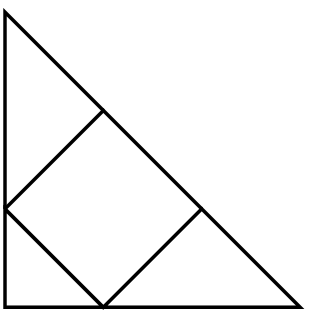}}$,
$\raisebox{0pt}{\includegraphics[height=20pt]{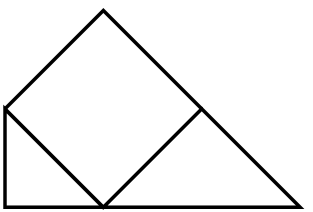}}$ and
$\raisebox{0pt}{\includegraphics[height=20pt]{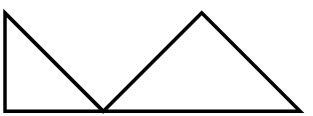}}$.

As before, we start with the component function
$\psi_{\Omega}^{\rm B}=\psi_{\,\includegraphics[width=12pt]{3_3210.eps}}$
given by \eqref{psi_alpha_B} and satisfying relations
\be
\lb{N=3aa}
h_1(-1)\psi_{\,\includegraphics[width=12pt]{3_3210.eps}}\, =\,
h_2(-1)\psi_{\,\includegraphics[width=12pt]{3_3210.eps}}\, =\, 0.
\ee
Then we act with the boundary generator and obtain
\be
\lb{N=3a}
s_0 \psi_{\,\includegraphics[width=12pt]{3_3210.eps}}\, =\,
\psi_{\,\raisebox{-4pt}{\smash{\includegraphics[width=12pt]{3_1210.eps}}}}.
\ee
Next, we apply $h_1(1)$ and find
\be
\lb{N=3b}
h_1(1) \psi_{\,\raisebox{-4pt}{\smash{\includegraphics[width=12pt]{3_1210.eps}}}}\, =\,
\psi_{\,\includegraphics[width=12pt]{3_3210.eps}}\, +\,
\psi_{\,\raisebox{-4pt}{\smash{\includegraphics[width=12pt]{3_1010.eps}}}}.
\ee
This can be rewritten to give the following expression cf. \eqref{N=2d},
\be
\lb{N=3bb}
\psi_{\,\raisebox{-4pt}{\smash{\includegraphics[width=12pt]{3_1010.eps}}}}\, =\,
h_1(1) \Bigl( s_0 -\frac{1}{[2]} \Bigr)\psi_{\,\includegraphics[width=12pt]{3_3210.eps}},
\ee
where we have used \eqref{N=3aa}, \eqref{N=3a} and \eqref{huv}. Finally, we act by operators $s_0$ and $h_2(1)$ to the rightmost term in \eqref{N=3b} and obtain
\begin{align*}
s_0 \psi_{\,\raisebox{-4pt}{\smash{\includegraphics[width=12pt]{3_1010.eps}}}}& =\, 0,
\\
h_2(1) \psi_{\,\raisebox{-4pt}{\smash{\includegraphics[width=12pt]{3_1010.eps}}}}&=\,
\psi_{\,\raisebox{-4pt}{\smash{\includegraphics[width=12pt]{3_1210.eps}}}}\, -\,
\tfrac{[\omega]}{[\omega+1]} \psi_{\,\includegraphics[width=12pt]{3_3210.eps}}.
\end{align*}
The latter relations can be rewritten as the truncation conditions on
$\psi_{\,\includegraphics[width=12pt]{3_3210.eps}}$:
\begin{align}
\lb{N=3e}
s_0 h_1(2) \Bigl(s_0 -\tfrac{[\omega +2]}{[3][\omega +1]} \Bigr)
\psi_{\,\includegraphics[width=12pt]{3_3210.eps}}& =\, 0,
\\
h_2(1) h_1(2)\Bigl(s_0-\tfrac{[\omega +2]}{[3][\omega +1]} \Bigr)
\psi_{\,\includegraphics[width=12pt]{3_3210.eps}}& =\, 0,
\nonumber
\end{align}
where we used \eqref{N=3aa}, \eqref{N=3a}, \eqref{N=3b} and, again, \eqref{huv}
to find factorised expressions.
\medskip

Analysing the factorised expressions for the component functions in the cases $N=2,3$
we see that a proper definition for the dashed boundary half-tile is

\[
\raisebox{-18pt}{\includegraphics[height=36pt]{h0.eps}}\, =\, \bar h_0(k)\, :=\,
h_0(k)|_{\nu=\omega+p_k}\, =\,  s_0\, -\,
\frac{\left[\lfloor{k/2}\rfloor\right]
\left[\omega +\lfloor{(k+1)/2}\rfloor\right]}{[k][\omega+1]},
\]
where $p_k = k \mbox{~mod 2}$.
Note that the Baxterised boundary element ${\bar h}_0(u)$ is defined for integer values of its
spectral parameter $u\in {\mathbb Z}$ as only such values appear in our considerations.

Using this notation, the expressions for the coefficients $\psi_\alpha$ and the truncation conditions
take a simple form. For example, formulas \eqref{N=3a}, \eqref{N=3bb}, \eqref{N=3e}
read
\[
\psi_{\,\raisebox{-4pt}{\smash{\includegraphics[width=12pt]{3_1210.eps}}}} =
\bar h_0(1) \psi_{\,\includegraphics[width=12pt]{3_3210.eps}}, \qquad
\psi_{\,\raisebox{-4pt}{\smash{\includegraphics[width=12pt]{3_1010.eps}}}} =
h_1(1) \bar h_0(2)\psi_{\,\includegraphics[width=12pt]{3_3210.eps}},\qquad
\bar h_0(1) h_1(2) \bar h_0(3)
\psi_{\,\includegraphics[width=12pt]{3_3210.eps}} = 0.
\]
\newpage

\section{Type A solutions}
\label{se:solutions}

Using the factorised expressions of Theorem~\ref{th:facsol}, we
have computed polynomial solutions of the qKZ equation for type A
from Proposition~\ref{prop4} in the limit $x_i\rightarrow 0$ up to $N=10$. These solutions are,
surprisingly, polynomials in $\tau^2$ with positive coefficients,
up to an overall factor which is a power of $\tau$. The
complete solution is determined up to an overall normalisation we have chosen so that
\[
\psi^{\rm A}_{\Omega} = \tau^{\lfloor N/2\rfloor (\lfloor N/2\rfloor-1)/2}.
\]

Let $\alpha=(\alpha_0,\alpha_1,\ldots,\alpha_N)\in \mathcal{D}_{N,p}$ be a Dyck path of length $N$ whose minima lie on or above height $\tilde{p}-1$. Then we define $c_{\alpha,p}$ as the signed sum of boxes between $\alpha$ and $\Omega(N,p)$, where the boxes at height $\tilde{p}+h$ are assigned $(-1)^{h}$. An example is given in the main text in Figure~\ref{fig:cpmdef}. The explicit expression for $c_{\alpha,p}$ is given by
\[
c_{\alpha,p} = \frac{(-1)^{\tilde{p}+1}}{2} \left(\sum_{i=1}^{\lfloor N/2\rfloor} (\alpha_{2i}-\Omega_{2i}(N,p))
 -\sum_{i=0}^{\lfloor (N-1)/2\rfloor} (\alpha_{2i+1}-\Omega_{2i+1}(N,p))  \right).
\]

We furthermore define the subset $\mathcal{D}_{N,p}$ of Dyck paths of length $N$ whose
local minima lie on or above height $\tilde{p}=\lfloor (N-1)/2\rfloor-p$, i.e.
\[
\mathcal{D}_{N,p} = \left\{ \alpha\in\mathcal{D}_N |\ \alpha_i\geq \min(\Omega_i,\tilde{p})\right\}.
\]
These definitions allow us to define the partial sums
\[
S_{\pm}(N,p) = \sum_{\alpha \in \mathcal{D}_{N,p}} \tau^{\pm c_{\alpha,p}} \psi_{\alpha} ,
\]
for which we formulate some conjectures in the main text.

\bigskip

\subsection{$N=4$}

\renewcommand{\arraystretch}{1.4}
\[
\begin{array}{c|lc}
\alpha & \psi_\alpha & \tau^{\pm c_{\alpha,1}}  \\ \hline
\includegraphics[width=24pt]{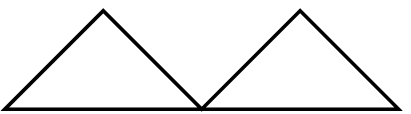} & 1+\tau^2 & 1   \\
\includegraphics[width=24pt]{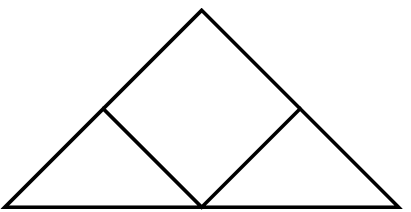} & \tau & \tau^{\pm 1}
\end{array}
\]

\smallskip

\[
\begin{array}{ll}
S_-(4,0) = \tau & S_+(4,0) = \tau\\
S_-(4,1) = 2+\tau^2 & S_+(4,1) = 1+2\tau^2
\end{array}
\]

\subsection{$N=5$}

\[
\begin{array}{c|lcc}
\alpha & \psi_\alpha & \tau^{\pm c_{\alpha,2}} & \tau^{\pm c_{\alpha,1}}\\ \hline
\includegraphics[width=30pt]{5_010101.eps} & \tau^2(2+\tau^2)& 1 & \\
\includegraphics[width=30pt]{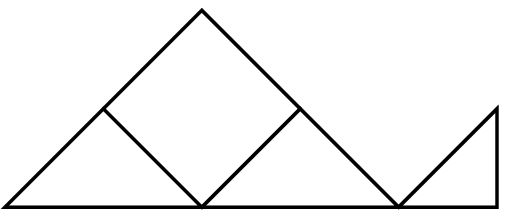} & \tau^3 & \tau^{\pm 1} & \\
\includegraphics[width=30pt]{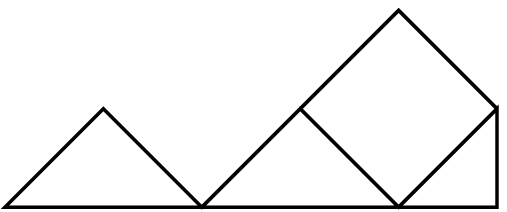} & \tau(2+\tau^2) & \tau^{\pm 1} & \\
\includegraphics[width=30pt]{5_012121.eps} & 1+2\tau^2 & \tau^{\pm 2} & 1 \\
\includegraphics[width=30pt]{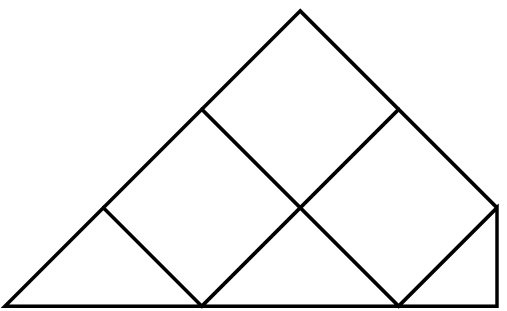} & \tau & \tau^{\pm 1} & \tau^{\pm 1}
\end{array}
\]

\smallskip

\[
\begin{array}{ll}
S_-(5,0) = \tau & S_+(5,0) = \tau\\
S_-(5,1) = 2(1+\tau^2)& S_+(5,1) = 1+3\tau^2 \\
S_-(5,2) = \tau^{-2}(1+5\tau^2 + 4\tau^4 + \tau^6) & S_+(5,2) = \tau^2(6+5\tau^2)
\end{array}
\]

\subsection{$N=6$}

\[
\begin{array}{c|lcc}
\alpha & \psi_\alpha & \tau^{\pm c_{\alpha,2}} & \tau^{\pm c_{\alpha,1}} \\ \hline
\includegraphics[width=36pt]{6_0101010.eps} & 1+5\tau^2+4\tau^4+\tau^6 & 1&  \\
\includegraphics[width=36pt]{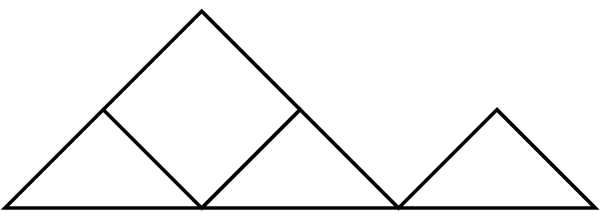} & \tau(2+2\tau^2+\tau^4) & \tau^{\pm 1} & \\
\includegraphics[width=36pt]{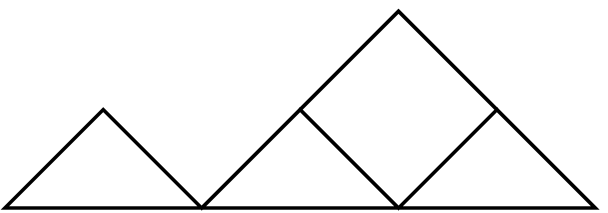} & \tau(1+3\tau^2+\tau^4) & \tau^{\pm 1}  & \\
\includegraphics[width=36pt]{6_0121210.eps} & 2\tau^2(1+\tau^2) & \tau^{\pm 2} & 1 \\
\includegraphics[width=36pt]{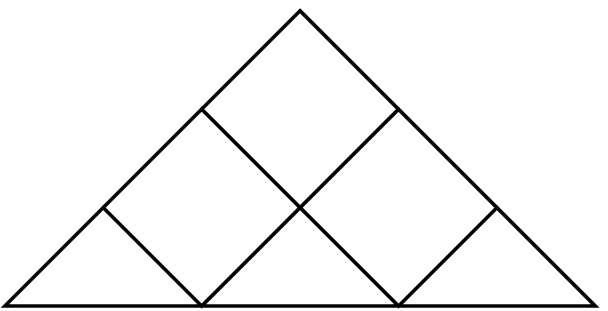} & \tau^3 & \tau^{\pm 1} & \tau^{\pm 1}
\end{array}
\]

\smallskip

\[
\begin{array}{ll}
S_-(6,0) = \tau^3 & S_+(6,0) = \tau^3\\
S_-(6,1) = \tau^2(3+2\tau^2) & S_+(6,1) = \tau^2(2+3\tau^2) \\
S_-(6,2) = 6+13\tau^2+6\tau^4+\tau^6 & S_+(6,2) = 1+8\tau^2+12\tau^4+5\tau^6
\end{array}
\]

\bigskip

\noindent
From now on we abbreviate polynomials of the form $P(\tau) = \tau^p \sum_{k=0}^r a_k \tau^{2k}$ by
\[
P(\tau) = \tau^p(a_0,a_1,\ldots,a_r).
\]
For example,
\[
6\tau^2+13\tau^4+6\tau^6+\tau^8 \equiv \tau^2(6,13,6,1)
\]

\vfill

\subsection{$N=7$}
\[
\begin{array}{c|lccc}
\alpha & \psi_\alpha & \tau^{\pm c_{\alpha,3}} & \tau^{\pm c_{\alpha,2}} & \tau^{\pm c_{\alpha,1}}\\ \hline
\includegraphics[width=42pt]{7_01010101.eps} & \tau^3(6,13,6,1) & 1 \\
\includegraphics[width=42pt]{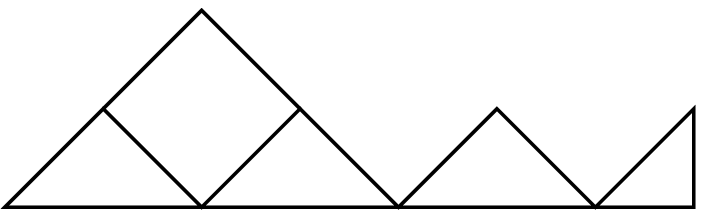} & \tau^4(5,4,1) & \tau^{\pm 1}  \\
\includegraphics[width=42pt]{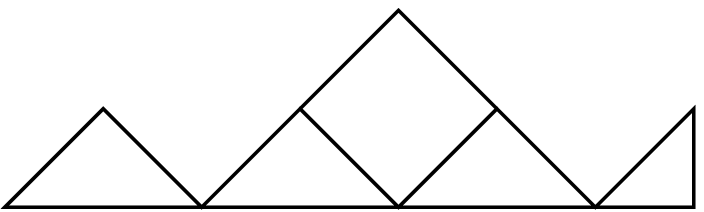} & \tau^4(3,4,1) & \tau^{\pm 1}  \\
\includegraphics[width=42pt]{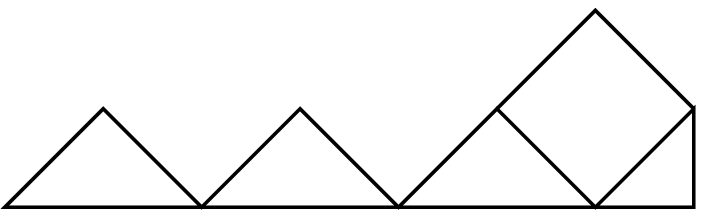} & \tau^2(6,13,6,1) & \tau^{\pm 1} \\
\includegraphics[width=42pt]{7_01212101.eps} & \tau^5(3,2) & \tau^{\pm 2}  \\
\includegraphics[width=42pt]{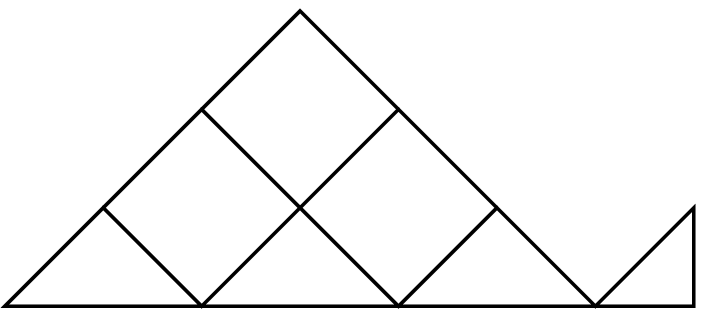} & \tau^6 & \tau^{\pm 1} \\
\includegraphics[width=42pt]{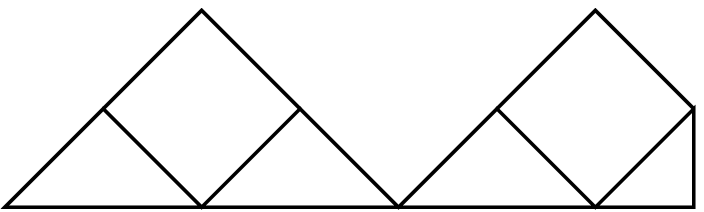} & \tau^3(5,3,1) & \tau^{\pm 2} \\
\includegraphics[width=42pt]{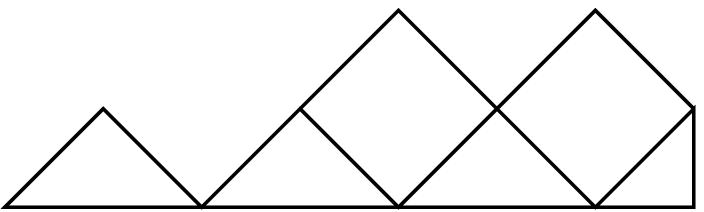} & \tau(3,11,10,2) & \tau^{\pm 2} \\
\includegraphics[width=42pt]{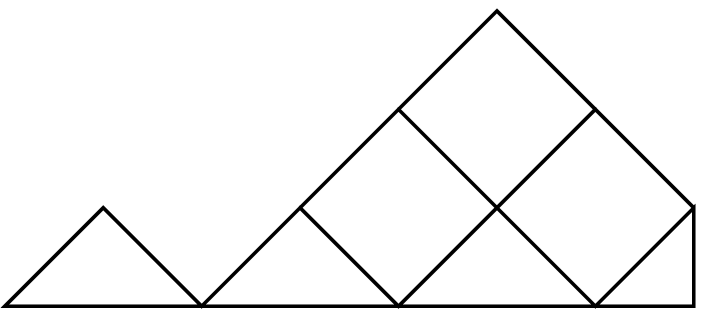} & \tau^2(3,5,1) & \tau^{\pm 1} \\
\includegraphics[width=42pt]{7_01212121.eps} & (1,8,12,5) & \tau^{\pm 3} & 1 \\
\includegraphics[width=42pt]{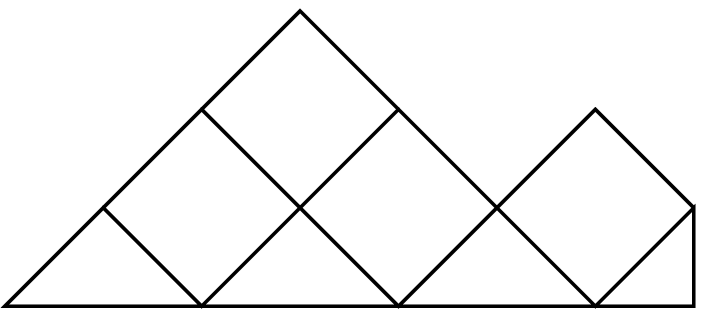} & \tau(2,3,3) & \tau^{\pm 2} & \tau^{\pm 1} \\
\includegraphics[width=42pt]{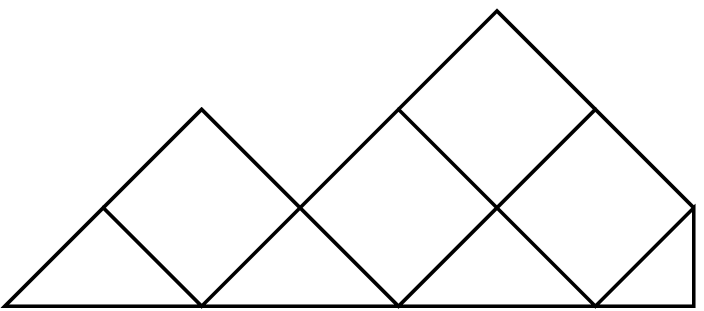} & \tau(1,6,3) & \tau^{\pm 2} & \tau^{\pm 1}\\
\includegraphics[width=42pt]{7_01232321.eps} & \tau^2(2,3) & \tau^{\pm 1} & \tau^{\pm 2} & 1\\
\includegraphics[width=42pt]{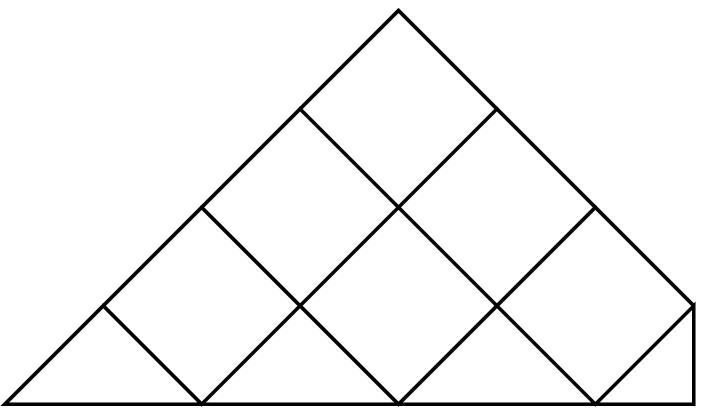} & \tau^3 & \tau^{\pm 2} & \tau^{\pm 1} & \tau^{\pm 1}
\end{array}
\]
%
%
\[
\begin{array}{ll}
S_-(7,0) = \tau^3 & S_+(7,0) = \tau^3\\
S_-(7,1) = \tau^2(3,3) & S_+(7,1) = \tau^2(2,4)\\
S_-(7,2) = (6,21,18,5) & S_+(7,2) = (1,11,24,14)\\
S_-(7,3) = \tau^{-3}(1,14,49,62,34,9,1) & S_+(7,3) = \tau^3(24,76,56,14)
\end{array}
\]

\subsection{$N=8$}
\[
\begin{array}{c|lccc}
\alpha & \psi_\alpha & \tau^{\pm c_{\alpha,3}} & \tau^{\pm c_{\alpha,2}} & \tau^{\pm c_{\alpha,1}}\\ \hline
\includegraphics[width=48pt]{8_010101010.eps} & (1,14,49,62,34,9,1) & 1 \\
\includegraphics[width=48pt]{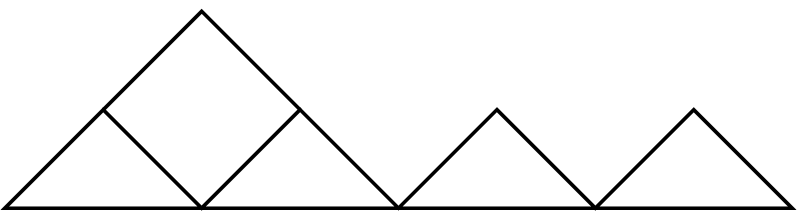} & \tau(3,15,29,20,7,1) & \tau^{\pm 1}  \\
\includegraphics[width=48pt]{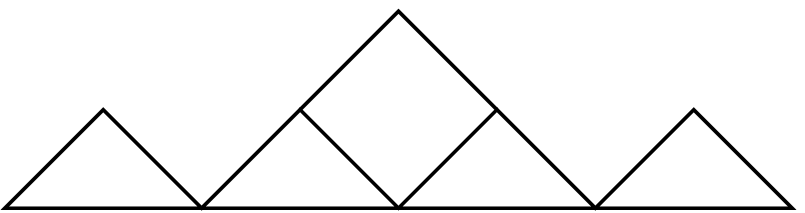} & \tau(2,15,27,19,7,1) & \tau^{\pm 1}  \\
\includegraphics[width=48pt]{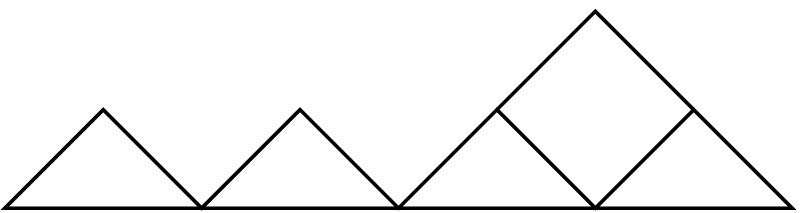} & \tau(1,12,28,25,8,1) & \tau^{\pm 1} \\
\includegraphics[width=48pt]{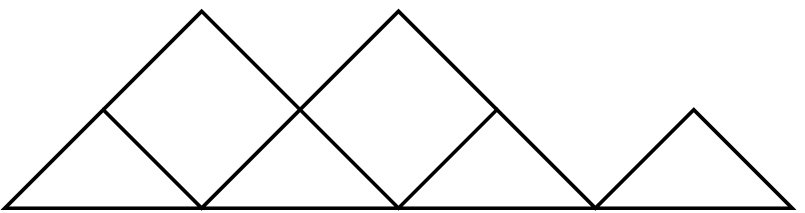} & \tau^2(6,21,18,9,2) & \tau^{\pm 2}  \\
\includegraphics[width=48pt]{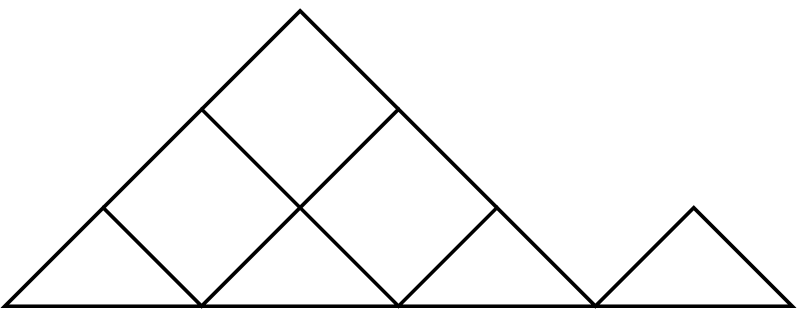} & \tau^3(5,5,3,1) & \tau^{\pm 1} \\
\includegraphics[width=48pt]{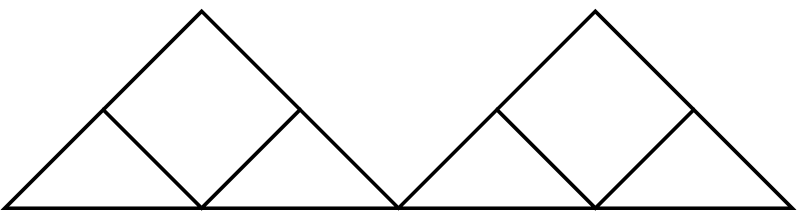} & \tau^2(3,9,12,5,1) & \tau^{\pm 2} \\
\includegraphics[width=48pt]{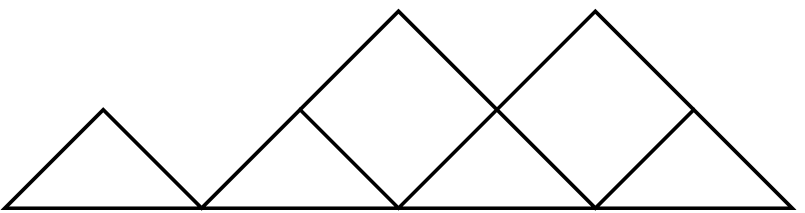} & \tau^2(2,15,24,13,2) & \tau^{\pm 2} \\
\includegraphics[width=48pt]{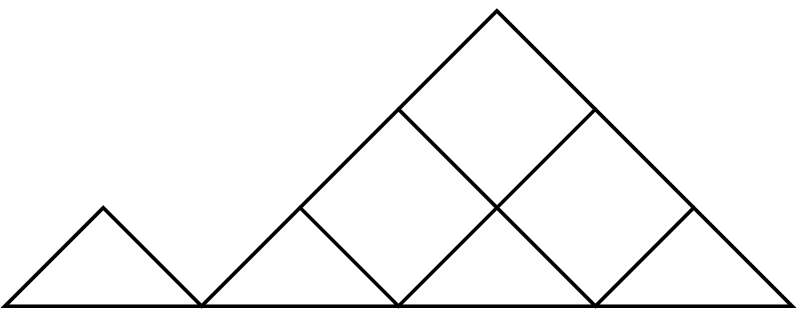} & \tau^3(1,6,6,1) & \tau^{\pm 1} \\
\includegraphics[width=48pt]{8_012121210.eps} & \tau^3(6,21,18,5) & \tau^{\pm 3} & 1 \\
\includegraphics[width=48pt]{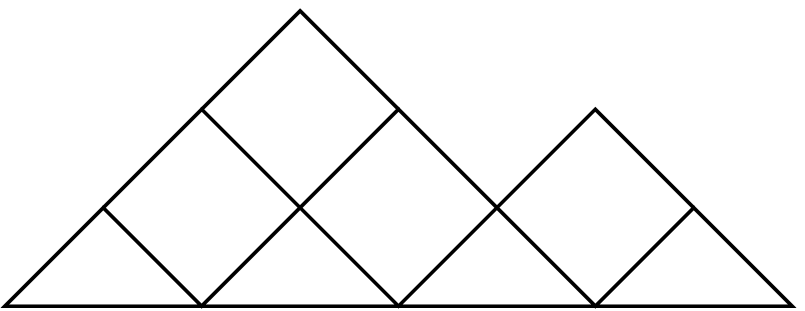} & \tau^4(5,6,3) & \tau^{\pm 2} & \tau^{\pm 1} \\
\includegraphics[width=48pt]{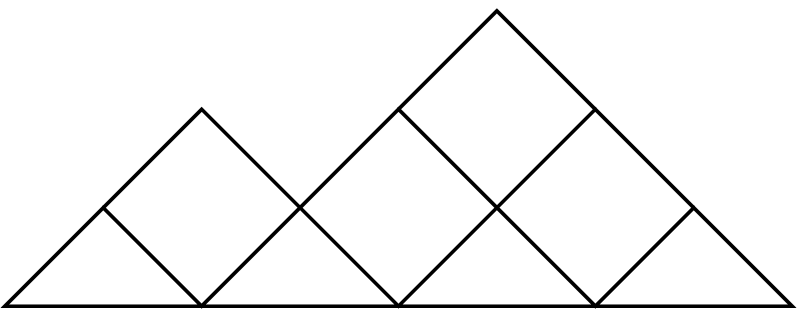} & \tau^4(3,8,3) & \tau^{\pm 2} & \tau^{\pm 1}\\
\includegraphics[width=48pt]{8_012323210.eps} & \tau^5(3,3) & \tau^{\pm 1} & \tau^{\pm 2} & 1\\
\includegraphics[width=48pt]{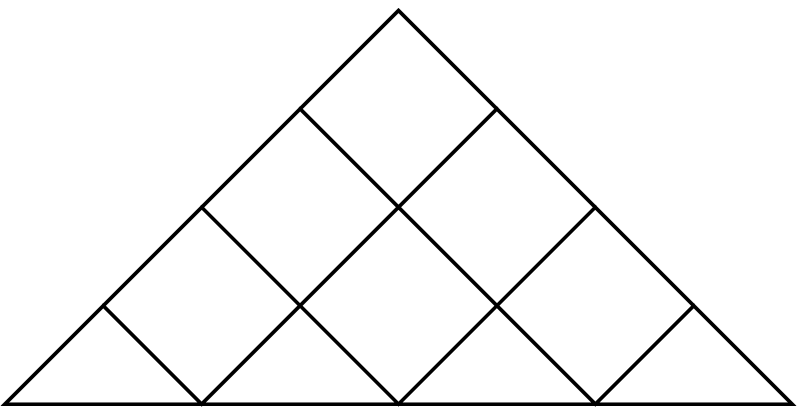} & \tau^6 & \tau^{\pm 2} & \tau^{\pm 1} & \tau^{\pm 1}
\end{array}
\]
%
%
\[
\begin{array}{ll}
S_-(8,0) = \tau^6 & S_+(8,0) = \tau^6\\
S_-(8,1) = \tau^5(4,3) & S_+(8,1) = \tau^5(3,4) \\
S_-(8,2) = \tau^3(17,39,24,5) & S_+(8,2) = \tau^3(6,29,36,14)\\
S_-(8,3) = (24,136,234,176,63,12,1) & S_+(8,3) = (1,20,108,219,200,84,14)
\end{array}
\]

\subsection{$N=9$}
\[
\begin{array}{c|lc}
\alpha & \psi_\alpha & \tau^{\pm c_{\alpha,4}} \\ \hline
\includegraphics[width=54pt]{9_0101010101.eps} & \tau^4(24,136,234,176,63,12,1) & 1 \\
\includegraphics[width=54pt]{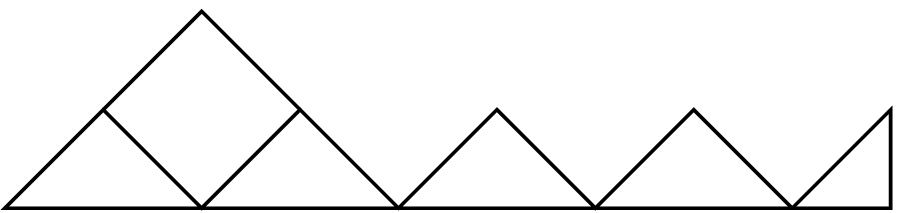} & \tau^5(28,84,94,43,10,1) & \tau^{\pm 1}  \\
\includegraphics[width=54pt]{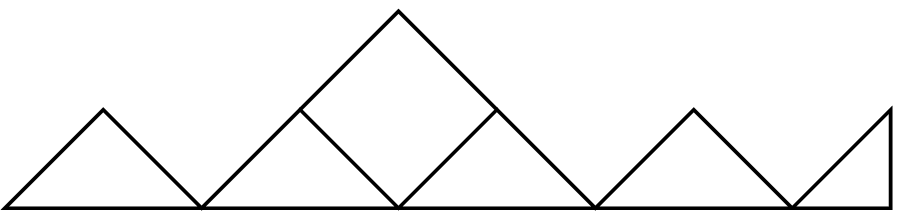} & \tau^5(20,72,84,41,10,1) & \tau^{\pm 1} \\
\includegraphics[width=54pt]{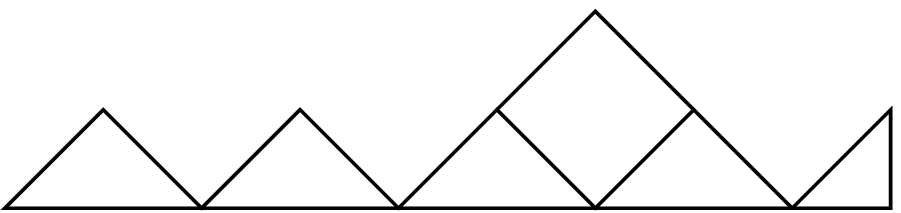} & \tau^5(12,58,74,41,10,1) & \tau^{\pm 1} \\
\includegraphics[width=54pt]{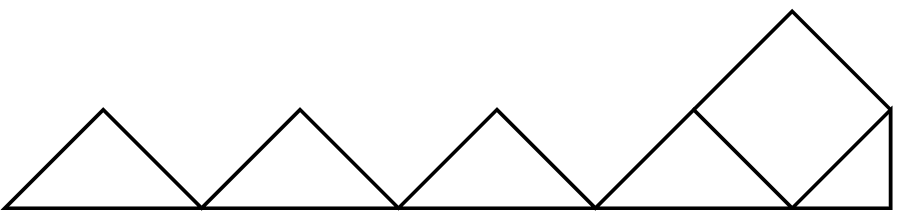} & \tau^3(24,136,234,176,63,12,1) & \tau^{\pm 1} \\
\includegraphics[width=54pt]{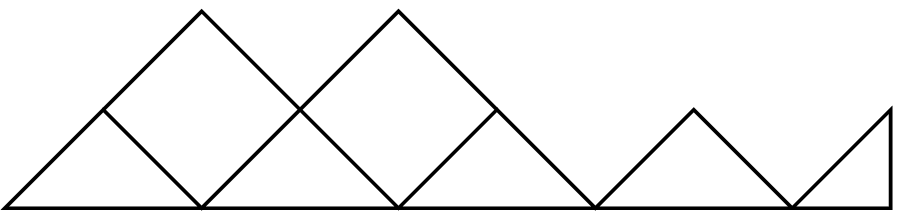} & \tau^6(28,65,45,15,2) & \tau^{\pm 2} \\
\includegraphics[width=54pt]{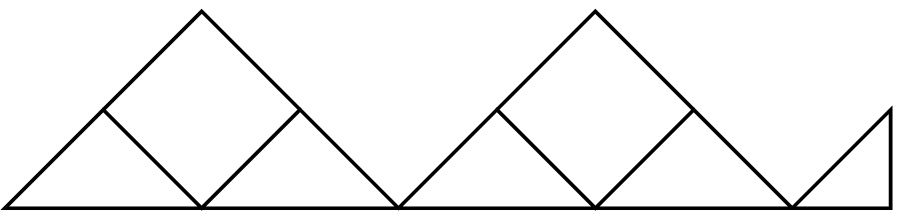} & \tau^6(14,31,23,7,1) & \tau^{\pm 2} \\
\includegraphics[width=54pt]{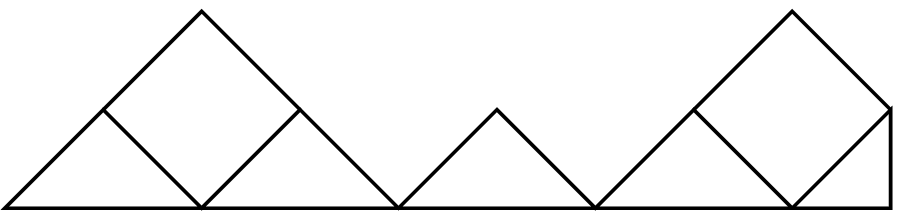} & \tau^4(28,84,90,40,10,1) & \tau^{\pm 2} \\
\includegraphics[width=54pt]{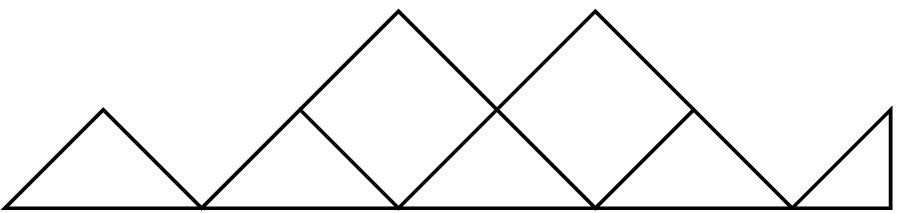} & \tau^6(12,41,41,16,2) & \tau^{\pm 2} \\
\includegraphics[width=54pt]{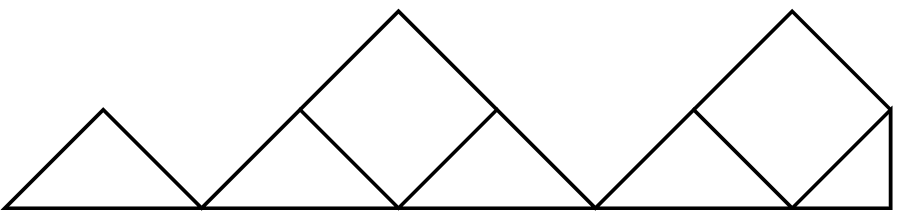} & \tau^4(20,68,74,34,9,1) & \tau^{\pm 2} \\
\includegraphics[width=54pt]{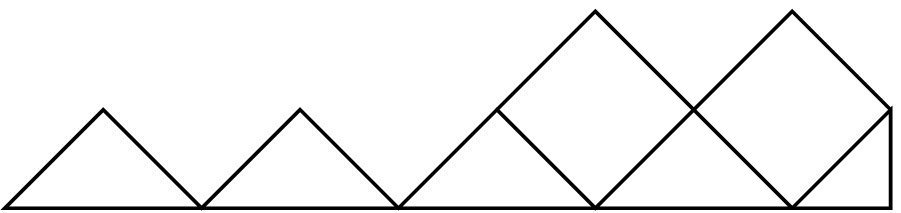} & \tau^2(12,86,208,213,103,22,2) & \tau^{\pm 2} \\
\includegraphics[width=54pt]{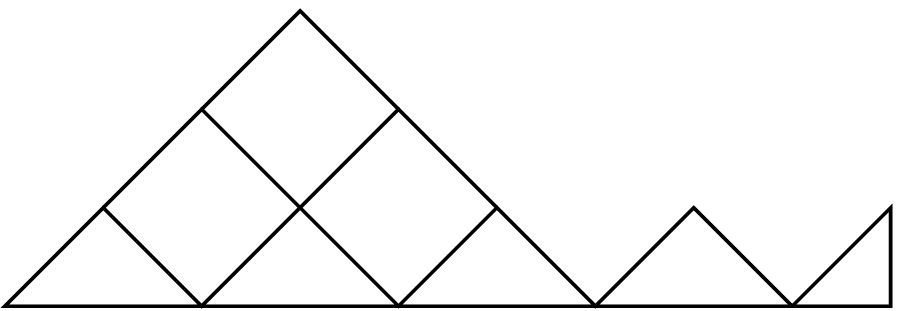} & \tau^7(14,13,6,1) & \tau^{\pm 1} \\
\includegraphics[width=54pt]{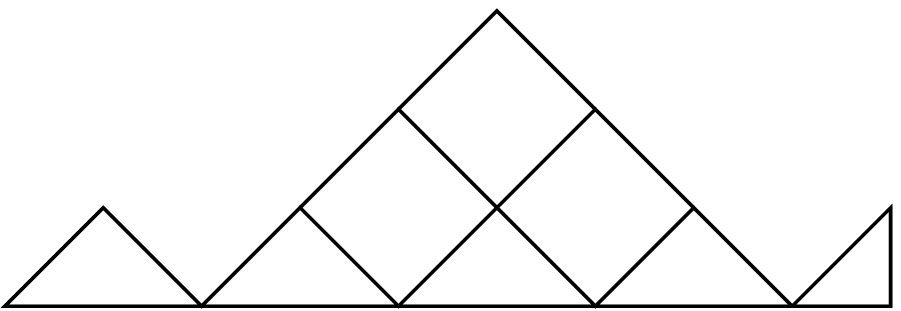} & \tau^7(4,10,7,1) & \tau^{\pm 1}
\\
\includegraphics[width=54pt]{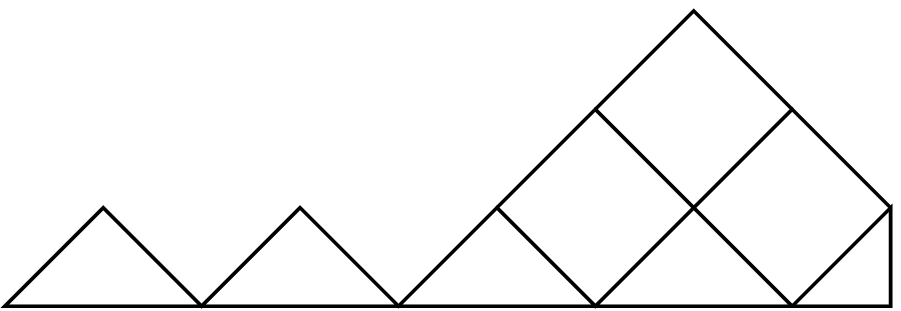} & \tau^3(12,62,88,51,11,1) & \tau^{\pm 1}
\end{array}
\]

\vfill\newpage

\[
\begin{array}{c|lc}
\alpha & \psi_\alpha & \tau^{\pm c_{\alpha,4}} \\ \hline
\includegraphics[width=54pt]{9_0121212101.eps} & \tau^7(17,39,24,5) & \tau^{\pm 3} \\
\includegraphics[width=54pt]{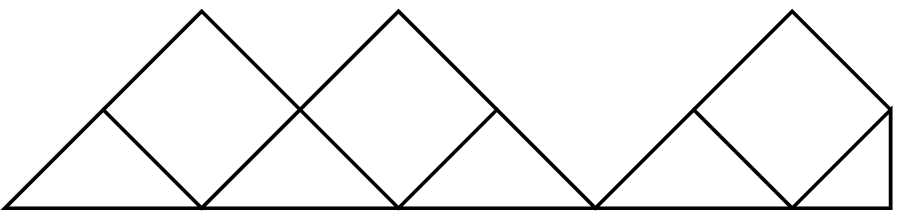} & \tau^5(28,59,33,12,2) & \tau^{\pm 3} \\
\includegraphics[width=54pt]{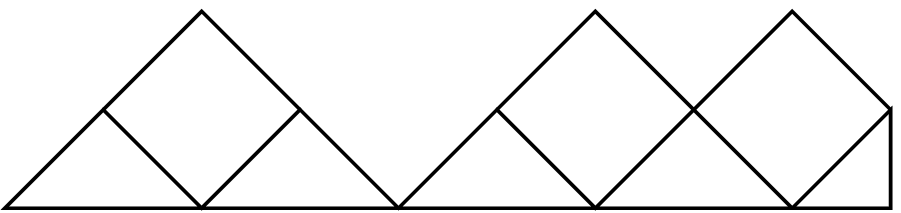} & \tau^3(14,56,84,54,15,2) & \tau^{\pm 3} \\
\includegraphics[width=54pt]{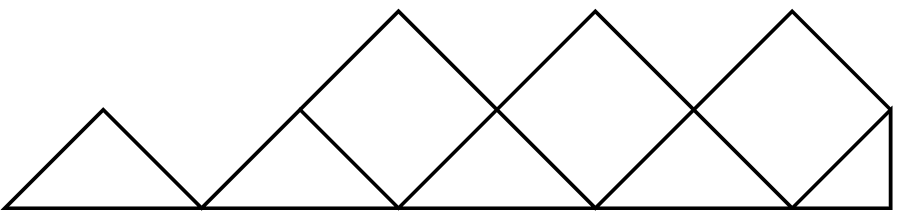} & \tau(4,46,160,230,154,47,5) & \tau^{\pm 3} \\
\includegraphics[width=54pt]{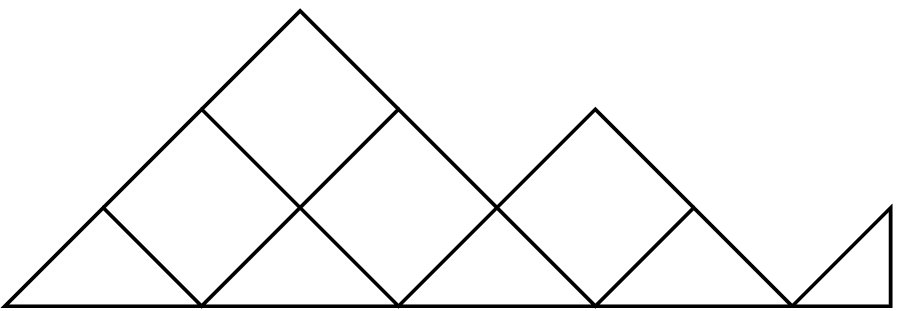} & \tau^8(9,9,3) & \tau^{\pm 2} \\
\includegraphics[width=54pt]{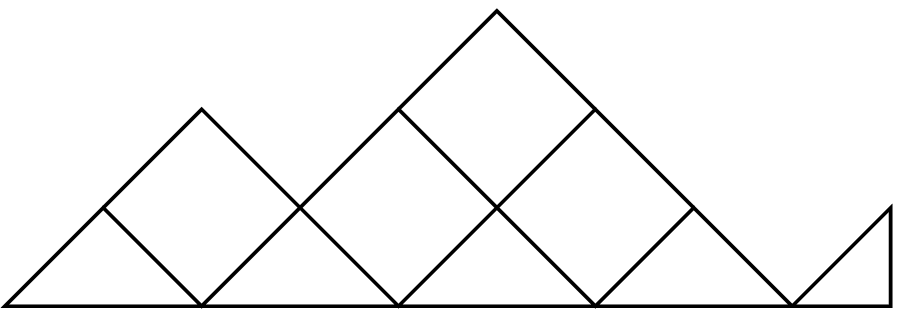} & \tau^8(6,10,3) & \tau^{\pm 2} \\
\includegraphics[width=54pt]{9_0123232101.eps} & \tau^9(4,3) & \tau^{\pm 1} \\
\includegraphics[width=54pt]{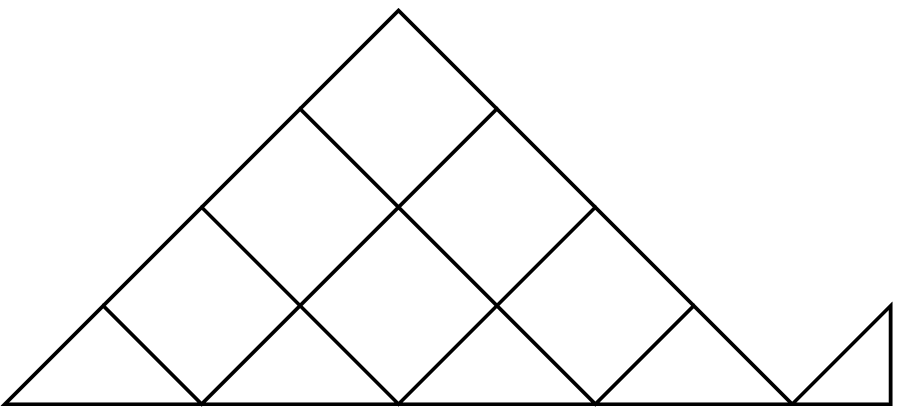} & \tau^{10} & \tau^{\pm 2} \\
\includegraphics[width=54pt]{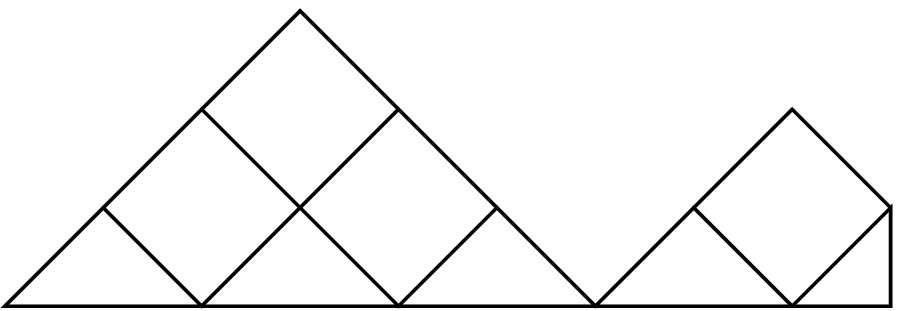} & \tau^6(14,9,4,1) & \tau^{\pm 2} \\
\includegraphics[width=54pt]{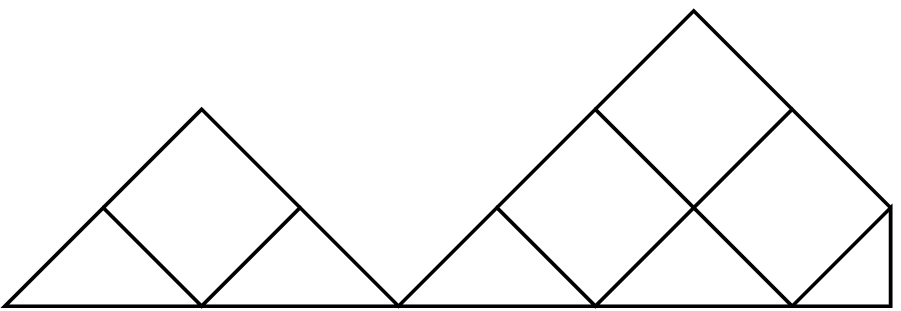} & \tau^4(14,28,25,7,1) & \tau^{\pm 2} \\
\includegraphics[width=54pt]{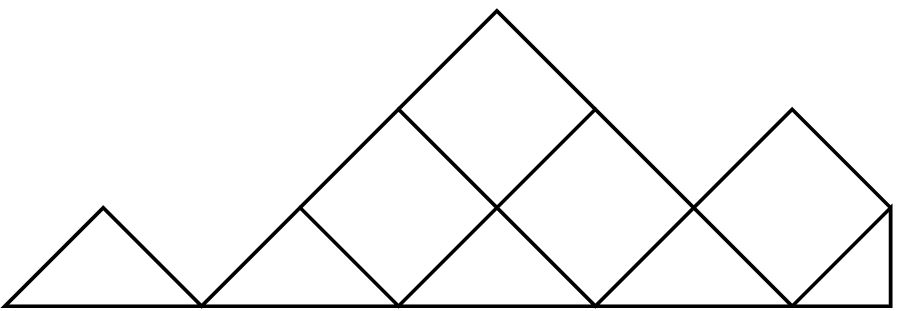} & \tau^2(8,40,68,61,26,3) & \tau^{\pm 2} \\
\includegraphics[width=54pt]{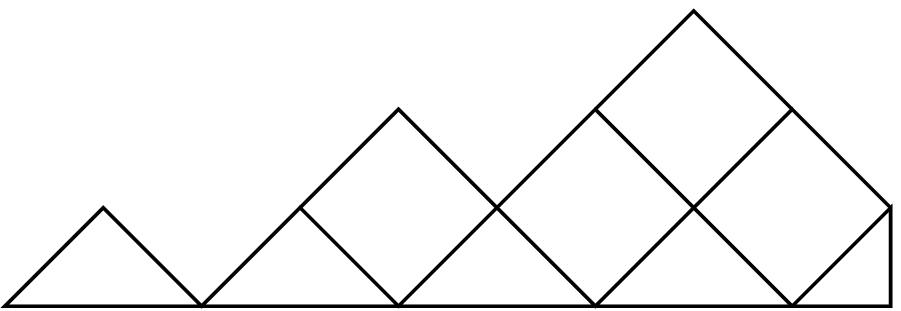} & \tau^2(4,38,96,84,28,3) & \tau^{\pm 2} \\
\includegraphics[width=54pt]{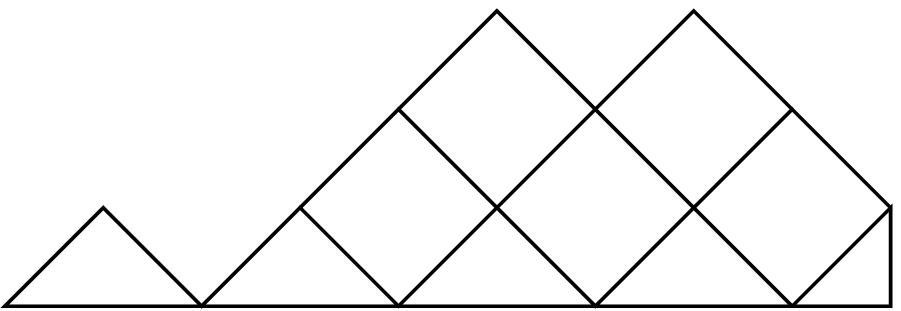} & \tau^3(8,40,56,27,3) & \tau^{\pm 1} \\
\includegraphics[width=54pt]{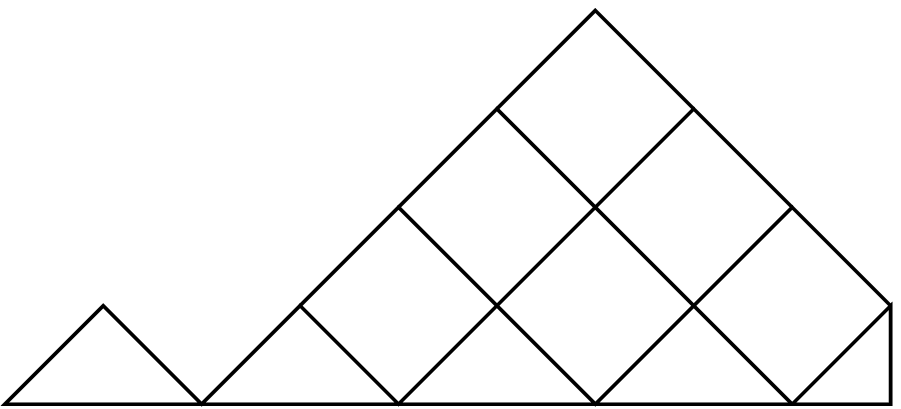} & \tau^4(4,14,9,1) & \tau^{\pm 2}
\end{array}
\]

\vfill\newpage

\[
\begin{array}{c|lcccc}
\alpha & \psi_\alpha & \tau^{\pm c_{\alpha,4}} & \tau^{\pm c_{\alpha,3}} & \tau^{\pm c_{\alpha,2}} & \tau^{\pm c_{\alpha,1}}\\ \hline
\includegraphics[width=54pt]{9_0121212121.eps} & (1,20,108,219,200,84,14) & \tau^{\pm 4} & 1\\
\includegraphics[width=54pt]{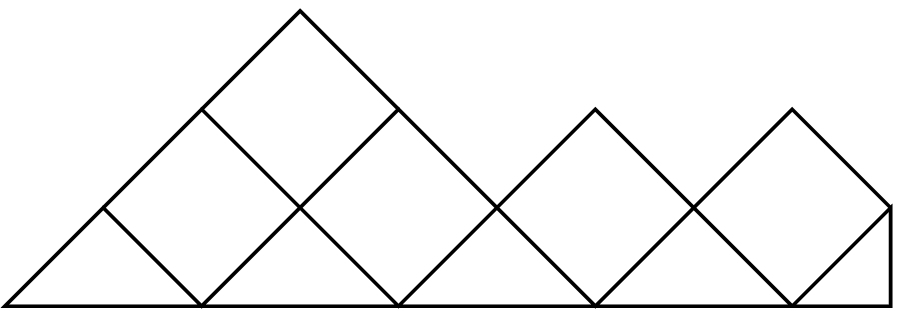} & \tau(3,19,58,69,38,9) & \tau^{\pm 3} & \tau^{\pm 1}\\
\includegraphics[width=54pt]{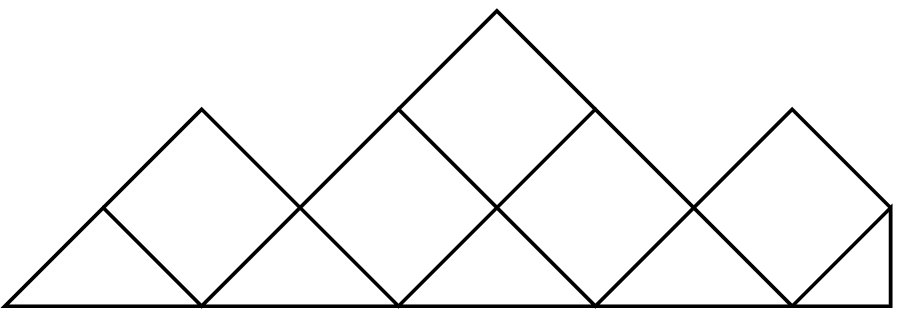} & \tau(2,27,67,75,47,10) & \tau^{\pm 3} & \tau^{\pm 1}\\
\includegraphics[width=54pt]{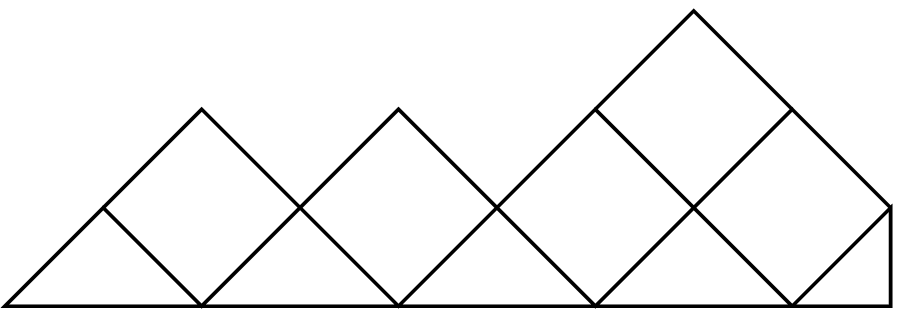} & \tau(1,18,75,106,51,9) & \tau^{\pm 3} & \tau^{\pm 1}\\
\includegraphics[width=54pt]{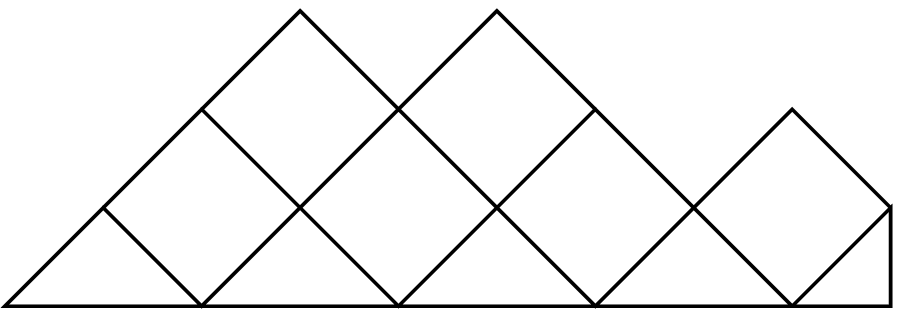} & \tau^2(6,29,36,30,11) & \tau^{\pm 3} & \tau^{\pm 2}\\
\includegraphics[width=54pt]{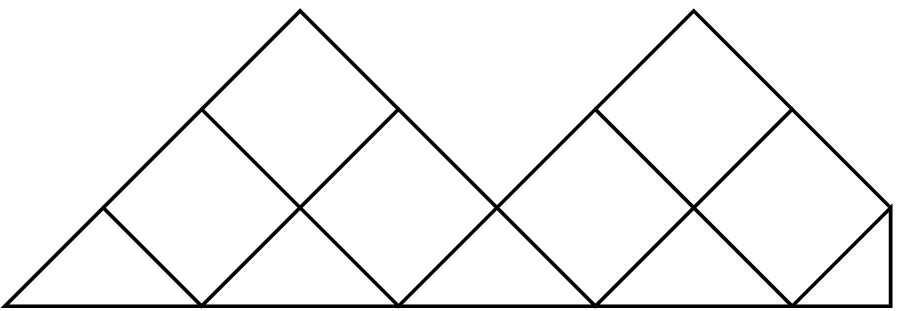} & \tau^2(3,13,33,21,6) & \tau^{\pm 2} & \tau^{\pm 2}\\
\includegraphics[width=54pt]{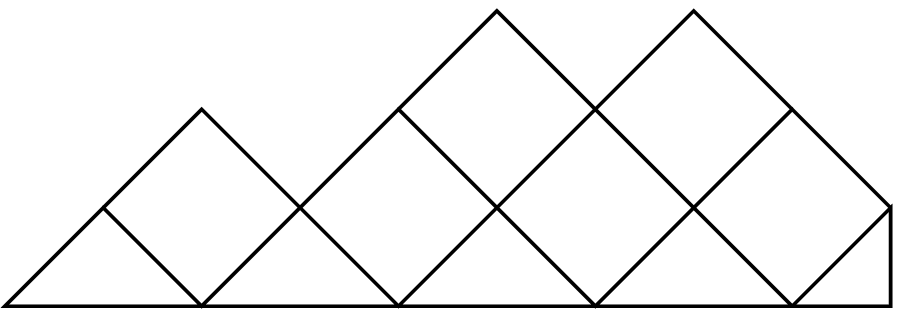} & \tau^2(2,27,64,51,11) & \tau^{\pm 2} & \tau^{\pm 2}\\
\includegraphics[width=54pt]{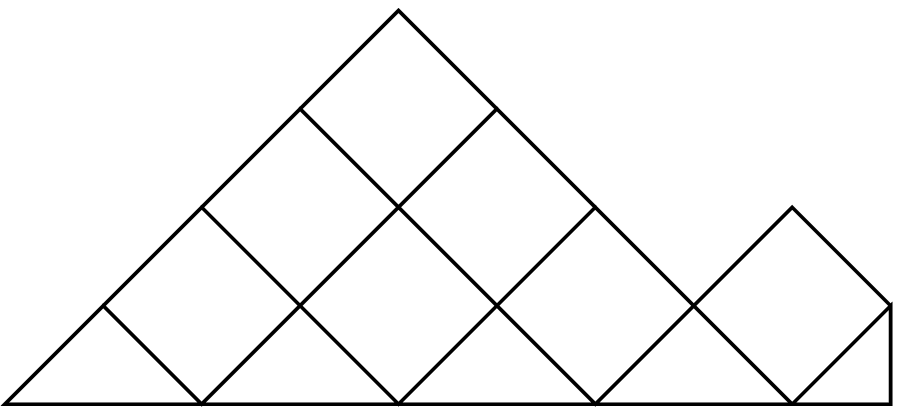} & \tau^3(5,7,6,4) & \tau^{\pm 3} & \tau^{\pm 1} \\
\includegraphics[width=54pt]{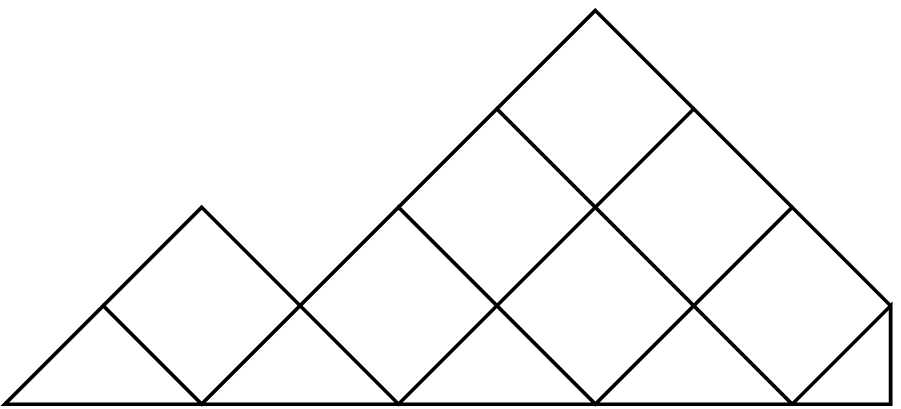} & \tau^3(1,12,17,4) & \tau^{\pm 3} & \tau^{\pm 1} \\
\includegraphics[width=54pt]{9_0123232321.eps} & \tau^3(6,29,36,14) & \tau^{\pm 1} & \tau^{\pm 3} & 1 \\
\includegraphics[width=54pt]{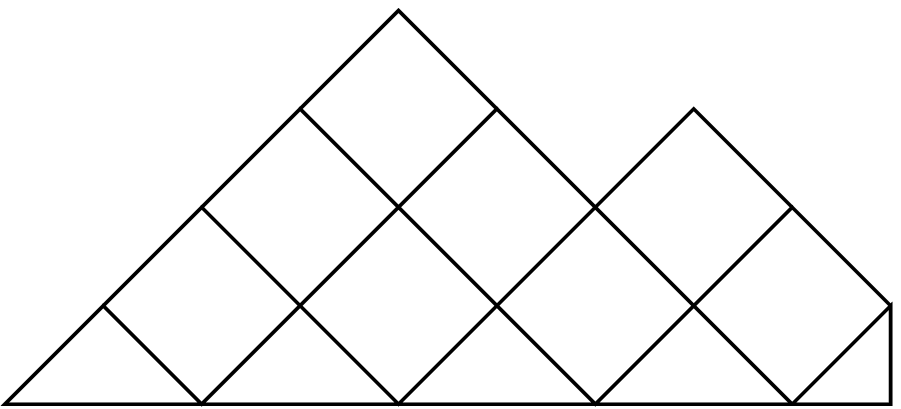} & \tau^4(5,8,6) & \tau^{\pm 2} & \tau^{\pm 2} & \tau^{\pm 1}\\
\includegraphics[width=54pt]{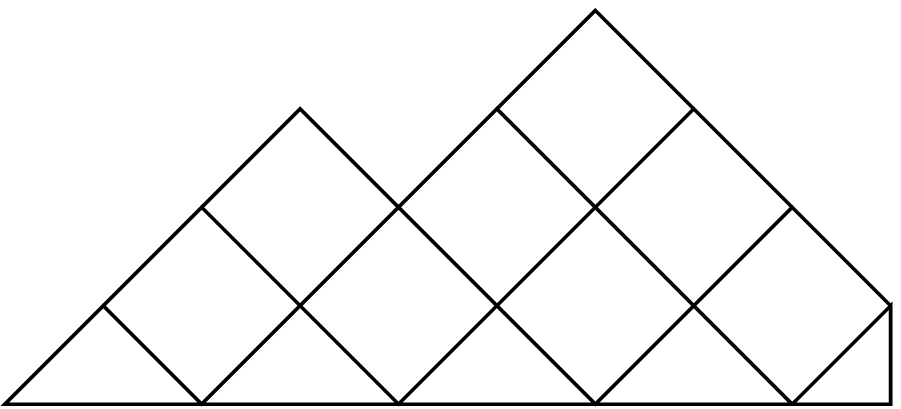} & \tau^4(3,12,6) & \tau^{\pm 2} & \tau^{\pm 2} & \tau^{\pm 1}\\
\includegraphics[width=54pt]{9_0123434321.eps} & \tau^5(3,4) & \tau^{\pm 3} & \tau^{\pm 1} & \tau^{\pm 2} & 1\\
\includegraphics[width=54pt]{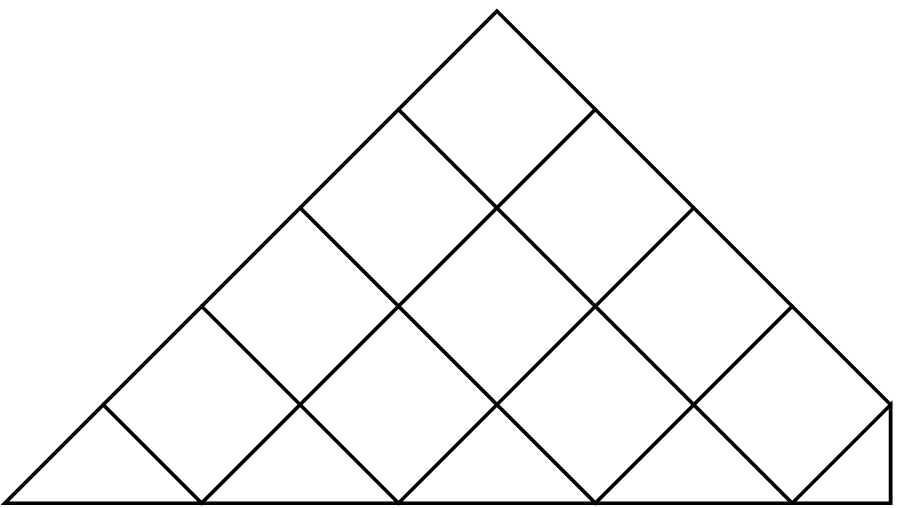} & \tau^6 & \tau^{\pm 2} & \tau^{\pm 2} & \tau^{\pm 1} & \tau^{\pm 1}\\
\end{array}
\]
\bigskip

\begin{align*}
S_-(9,0) &= \tau^6 \nonumber\\
S_-(9,1) &= \tau^5(4,4)\nonumber\\
S_-(9,2) &= \tau^3(17,54,48,14) \\
S_-(9,3) &= (24,196,520,624,372,112,14)\nonumber\\
S_-(9,4) &= \tau^{-4}(1,30,273,1042,2006,2121,1321,501,117,16) \nonumber\\[5mm]
S_+(0,0) &= \tau^6\nonumber\\
S_+(9,1) &= \tau^5(3,5) \nonumber\\
S_+(9,2) &= \tau^3(6,37,60,30)\\
S_+(9,3) &= (1,26,189,524,660,378,84) \nonumber\\
S_+(9,4) &= \tau^4(120,920,2242,2440,1305,360,42)\nonumber
\end{align*}
\vfill\newpage
\subsubsection{$N=10$}

\[
\begin{array}{c|lc}
\alpha & \psi_\alpha & \tau^{\pm c_{\alpha,4}} \\ \hline
\includegraphics[width=54pt]{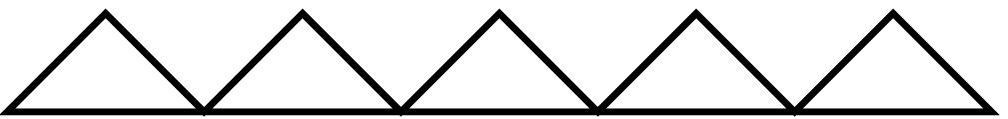} & (1,30,273,1042,2006,2121,1321,501,117,16,1) & 1\\
\includegraphics[width=54pt]{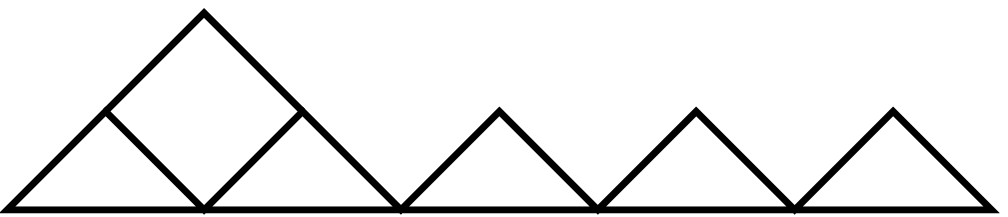} & \tau(4,56,294,738,977,735,327,89,14,1) & \tau^{\pm 1} \\
\includegraphics[width=54pt]{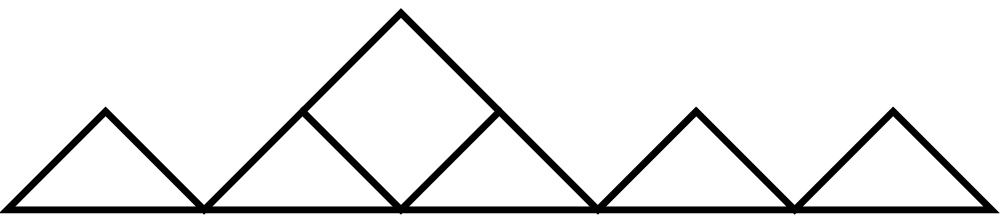} & \tau(3,49,269,683,912,691,312,87,14,1) & \tau^{\pm 1}\\
\includegraphics[width=54pt]{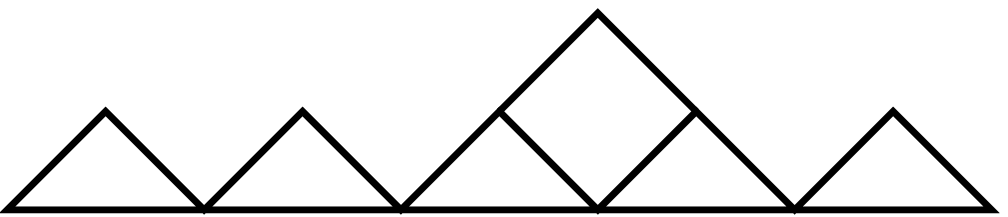} & \tau(2,47,267,686,915,688,313,88,14,1) & \tau^{\pm 1}\\
\includegraphics[width=54pt]{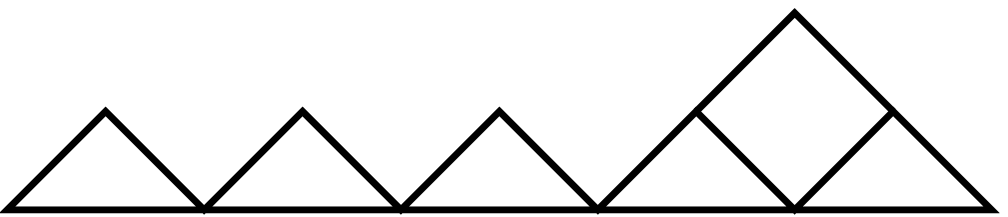} & \tau(1,28,220,669,996,820,384,101,15,1) & \tau^{\pm 1}\\
\includegraphics[width=54pt]{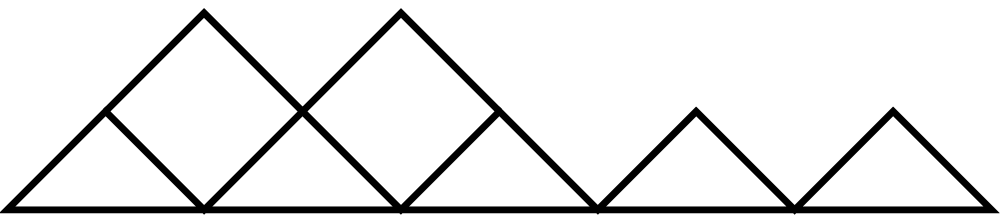} & \tau^2(12,116,396,684,348,117,23,2) & \tau^{\pm 2} \\
\includegraphics[width=54pt]{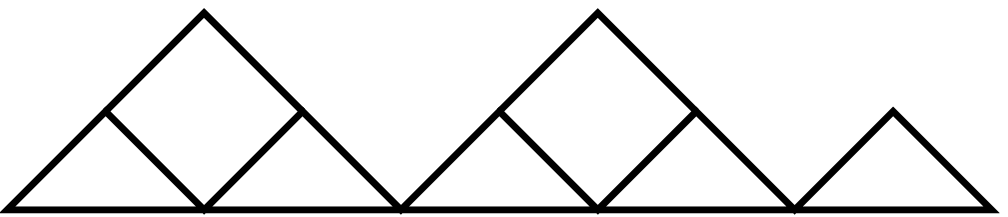} & \tau^2(8,60,206,350,329,176,58,11,1) & \tau^{\pm 2}\\
\includegraphics[width=54pt]{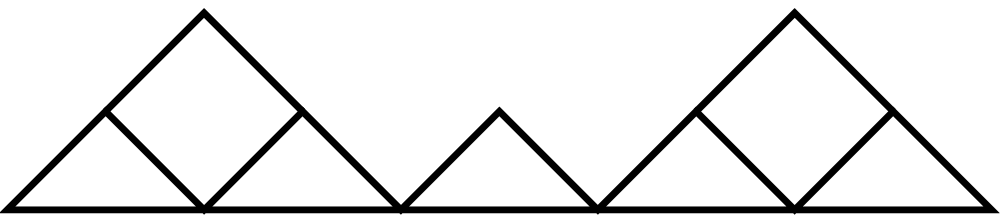} & \tau^2(4,48,210,394,403,230,72,13,1) & \tau^{\pm 2}\\
\includegraphics[width=54pt]{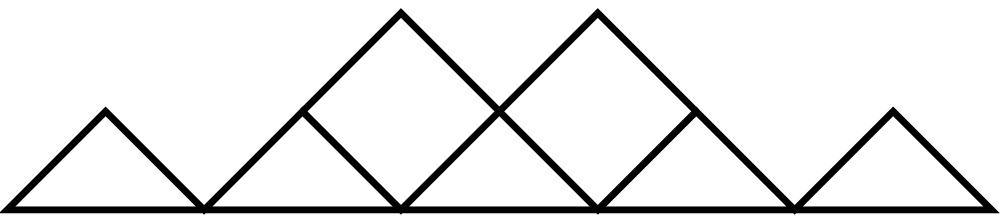} & \tau^2(6,89,368,665,618,342,120,24,2) & \tau^{\pm 2}\\
\includegraphics[width=54pt]{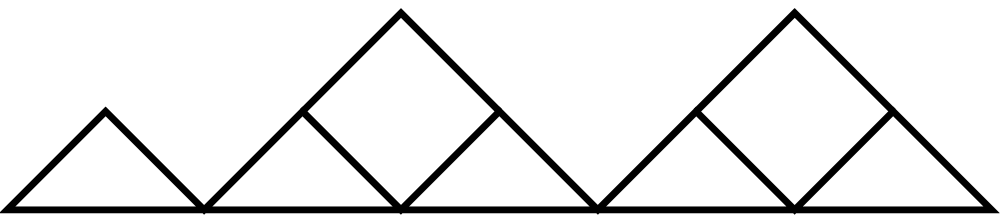} & \tau^2(3,43,184,343,349,200,64,12,1) & \tau^{\pm 2}\\
\includegraphics[width=54pt]{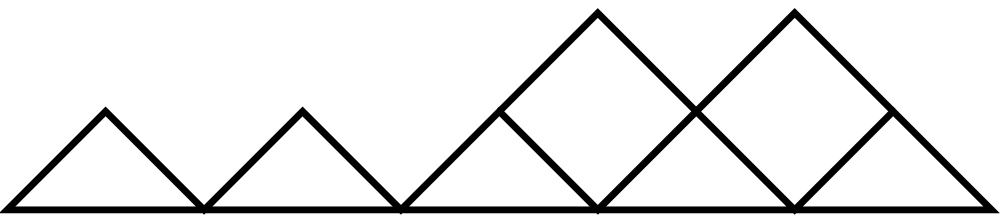} & \tau^2(2,47,264,632,744,469,159,27,2) & \tau^{\pm 2}\\
\includegraphics[width=54pt]{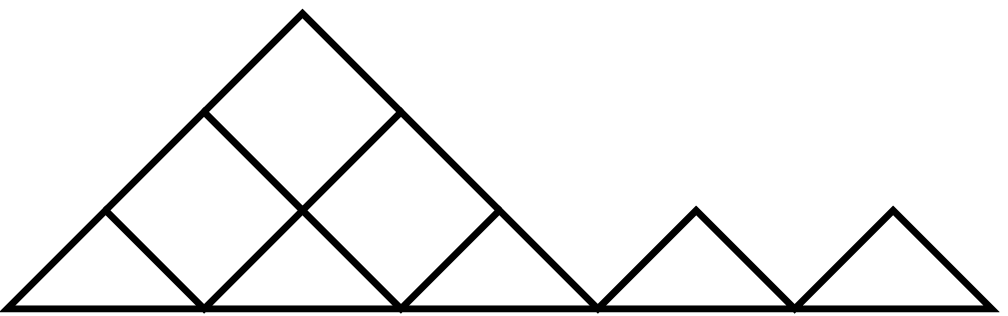} & \tau^3(14,70,151,180,111,43,10,1) & \tau^{\pm 1}\\

\includegraphics[width=54pt]{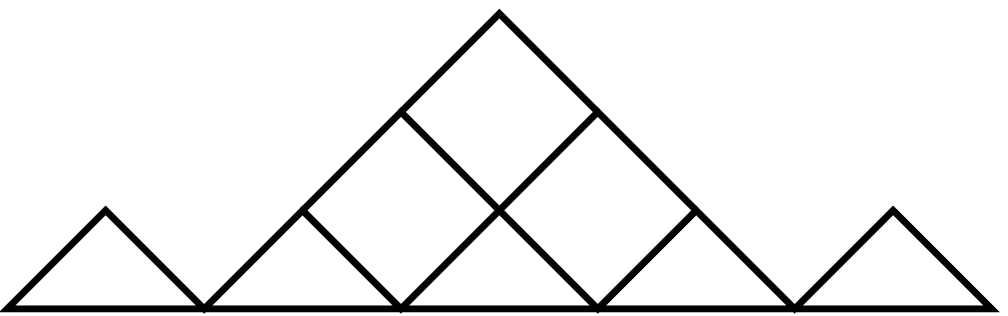} & \tau^3(5,57,156,174,111,46,11,1) & \tau^{\pm 1} \\
\includegraphics[width=54pt]{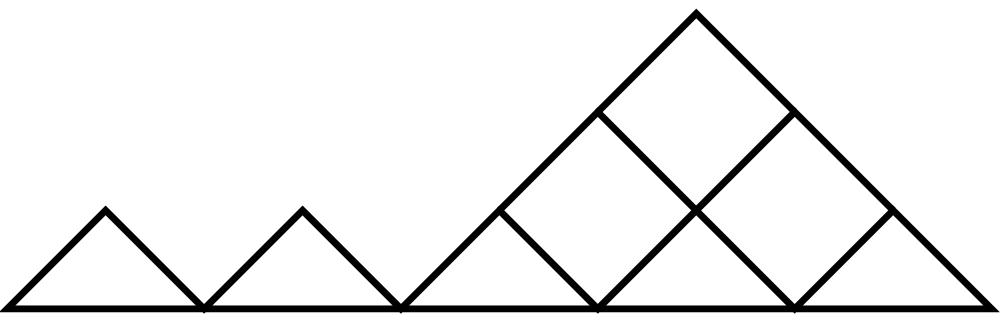} & \tau^3(1,22,102,97,172,72,13,1) & \tau^{\pm 1}
\end{array}
\]

\vfill\newpage

\[
\begin{array}{c|lc}
\alpha & \psi_\alpha & \tau^{\pm c_{\alpha,4}}  \\ \hline
\includegraphics[width=54pt]{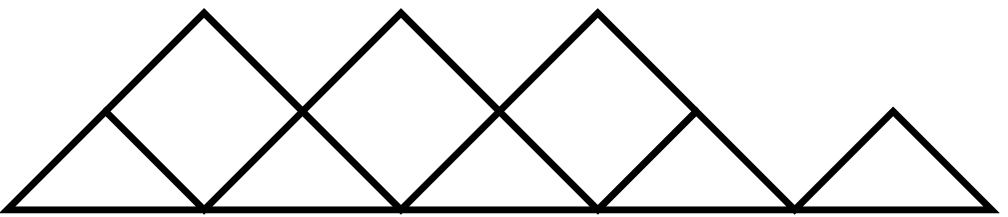} & \tau^3(24,196,520,624,408,174,44,5) & \tau^{\pm 3}\\
\includegraphics[width=54pt]{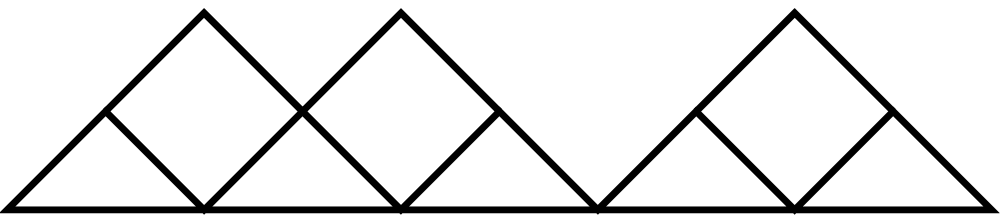} & \tau^3(12,92,216,276,198,76,18,2) & \tau^{\pm 3}\\
\includegraphics[width=54pt]{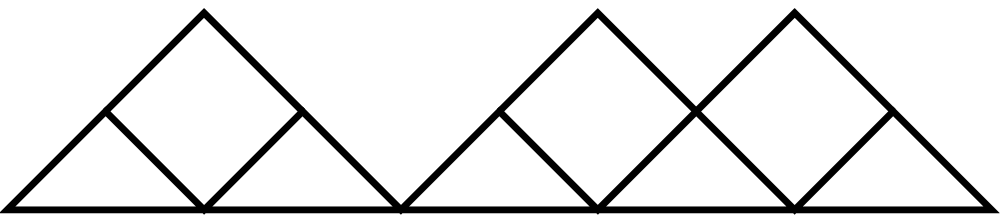} & \tau^3(8,60,194,286,226,94,20,2) & \tau^{\pm 3}\\
\includegraphics[width=54pt]{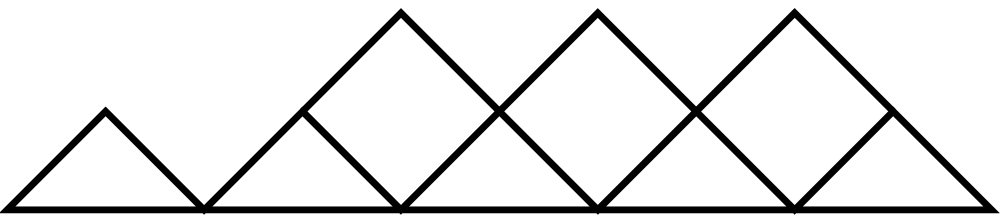} & \tau^3(6,89,368,649,564,256,58,5) & \tau^{\pm 3}\\
\includegraphics[width=54pt]{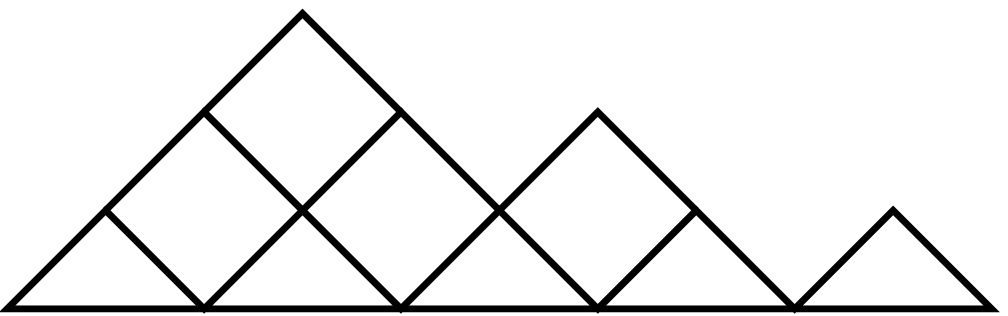} & \tau^4(28,112,187,140,69,21,3) & \tau^{\pm 2}\\
\includegraphics[width=54pt]{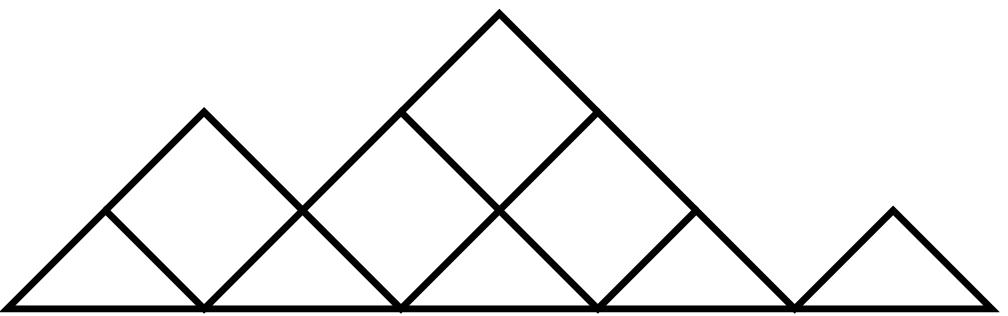} & \tau^4(20,118,189,142,71,22,3) & \tau^{\pm 2}\\
\includegraphics[width=54pt]{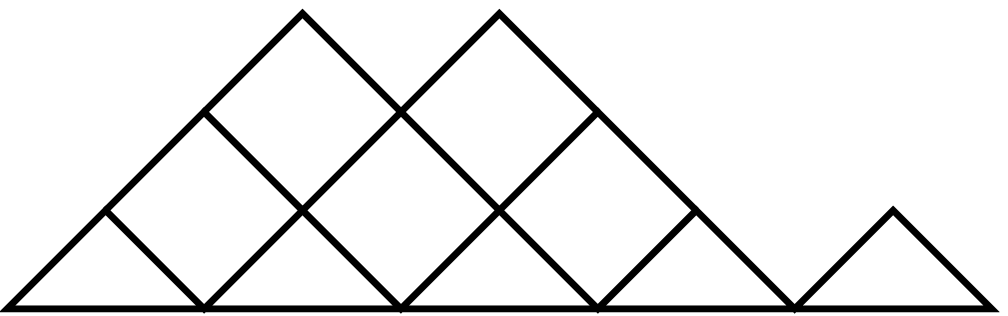} & \tau^5(28,84,73,42,16,3) & \tau^{\pm 1}\\
\includegraphics[width=54pt]{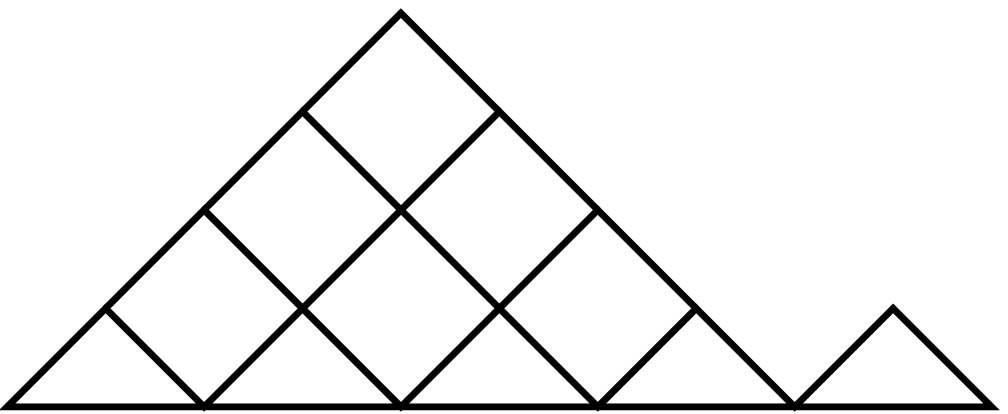} & \tau^6(14,14,9,4,1) & \tau^{\pm 2} \\
\includegraphics[width=54pt]{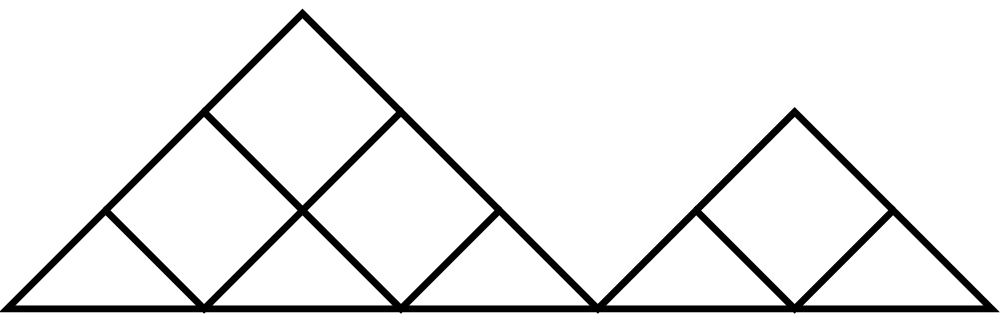} & \tau^4(14,42,64,57,25,7,1) & \tau^{\pm 2}\\
\includegraphics[width=54pt]{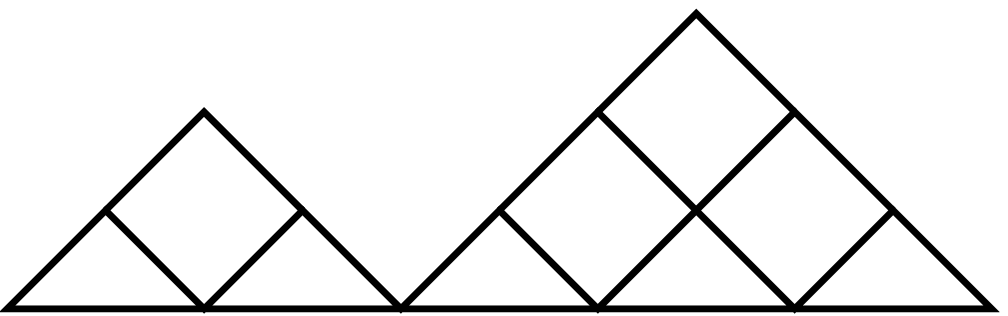} & \tau^4(4,24,63,70,39,9,1) & \tau^{\pm 2}\\
\includegraphics[width=54pt]{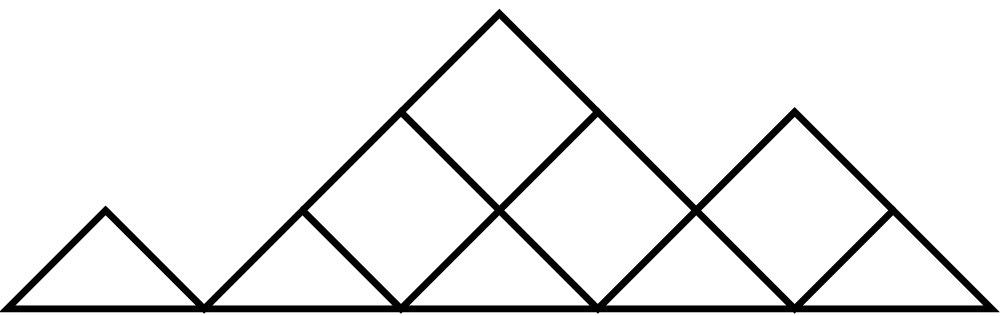} & \tau^4(5,58,166,190,111,32,3) & \tau^{\pm 2}\\
\includegraphics[width=54pt]{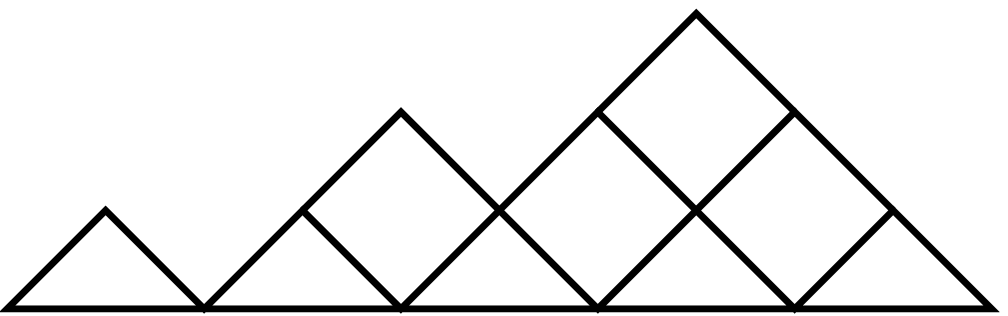} & \tau^4(3,42,147,206,126,33,3) & \tau^{\pm 2}\\
\includegraphics[width=54pt]{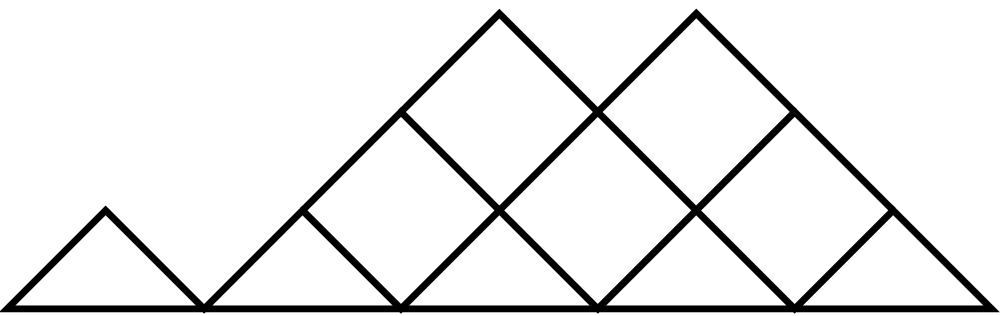} & \tau^5(3,34,90,85,31,3) & \tau^{\pm 1}\\
\includegraphics[width=54pt]{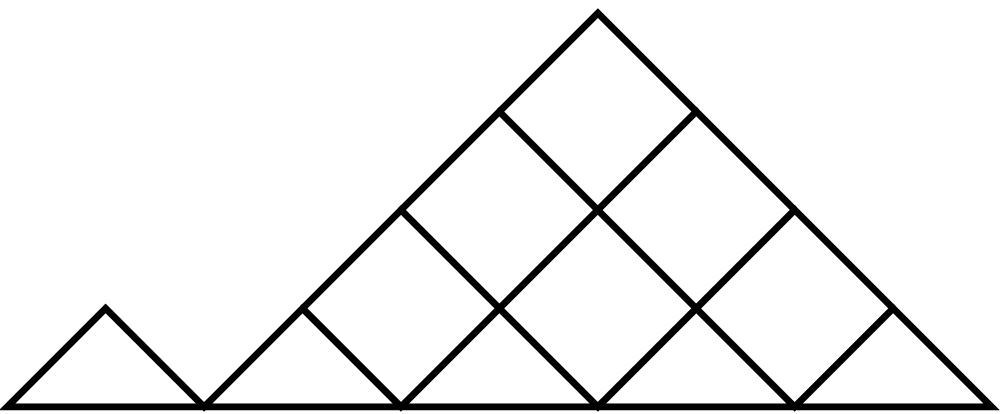} & \tau^6(1,10,20,10,1) & \tau^{\pm 2}
\end{array}
\]

\vfill\newpage

\[
\begin{array}{c|lccccc}
\alpha & \psi_\alpha  & \tau^{\pm c_{\alpha,4}}  & \tau^{\pm c_{\alpha,3}}  & \tau^{\pm c_{\alpha,2}}  & \tau^{\pm c_{\alpha,1}}  \\ \hline
\includegraphics[width=54pt]{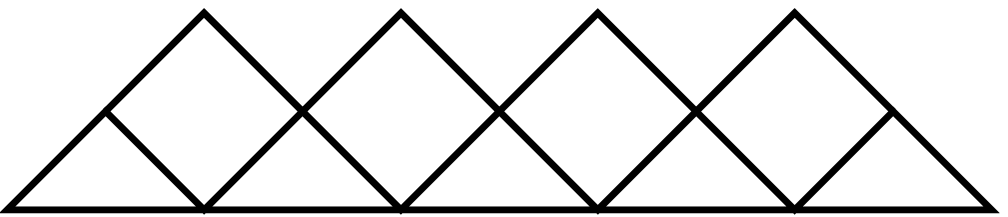} & \tau^4(24,196,520,6224,372,112,14) & \tau^{\pm 4} & 1\\
\includegraphics[width=54pt]{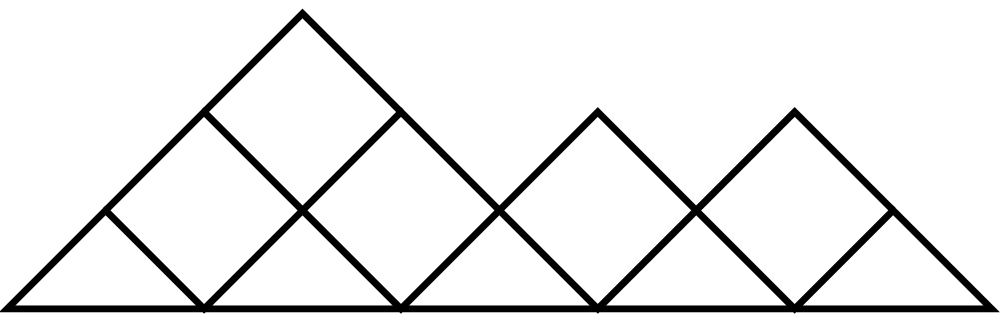} & \tau^5(28,112,191,144,55,9) & \tau^{\pm 3} & \tau^{\pm 1}\\
\includegraphics[width=54pt]{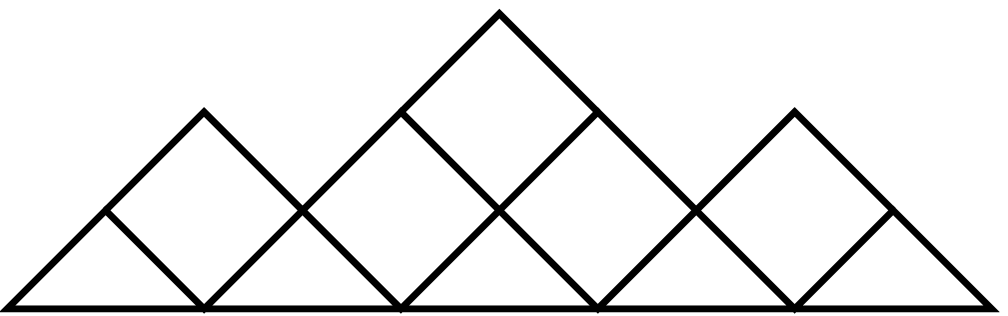} & \tau^5(20,122,209,162,165,10) & \tau^{\pm 3} & \tau^{\pm 1}\\
\includegraphics[width=54pt]{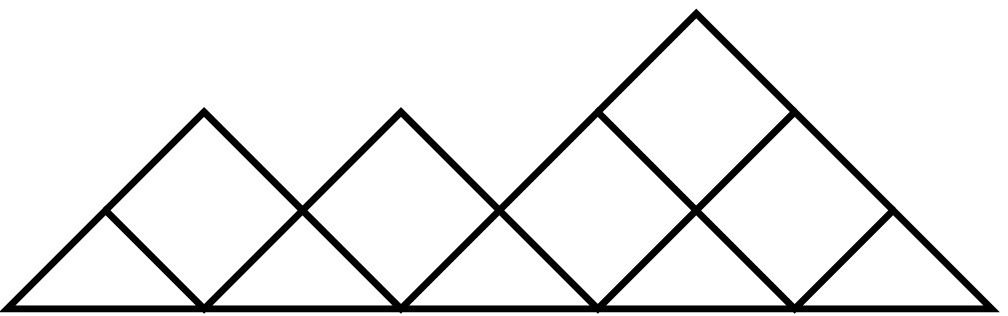} & \tau^5(12,88,192,174,64,9) & \tau^{\pm 3} & \tau^{\pm 1}\\
\includegraphics[width=54pt]{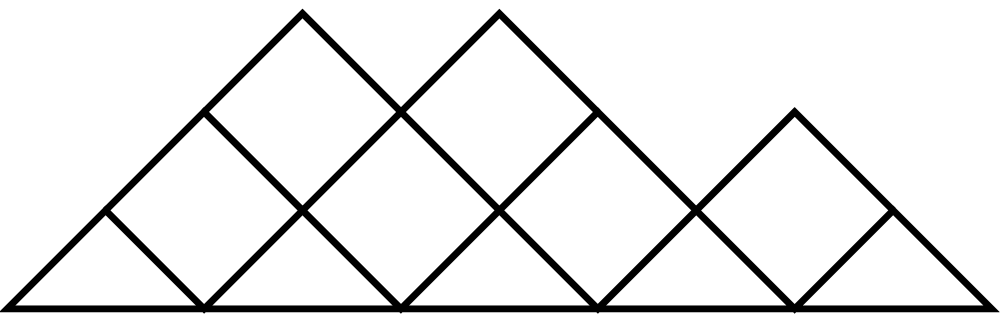} & \tau^6(28,90,80,48,11) & \tau^{\pm 3} & \tau^{\pm 2}\\
\includegraphics[width=54pt]{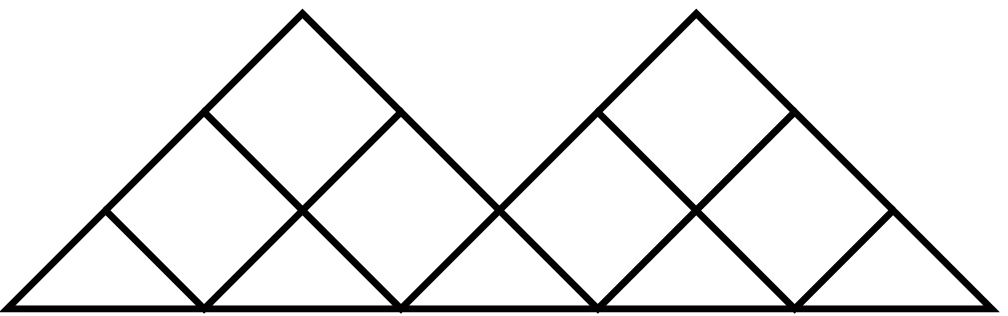} & \tau^6(14,45,60,29,6) & \tau^{\pm 2} & \tau^{\pm 2}\\
\includegraphics[width=54pt]{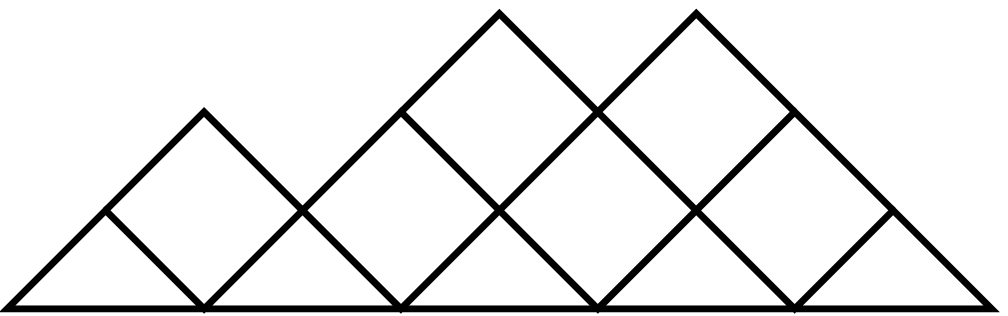} & \tau^6(12,71,110,63,11) & \tau^{\pm 2} & \tau^{\pm 2}\\
\includegraphics[width=54pt]{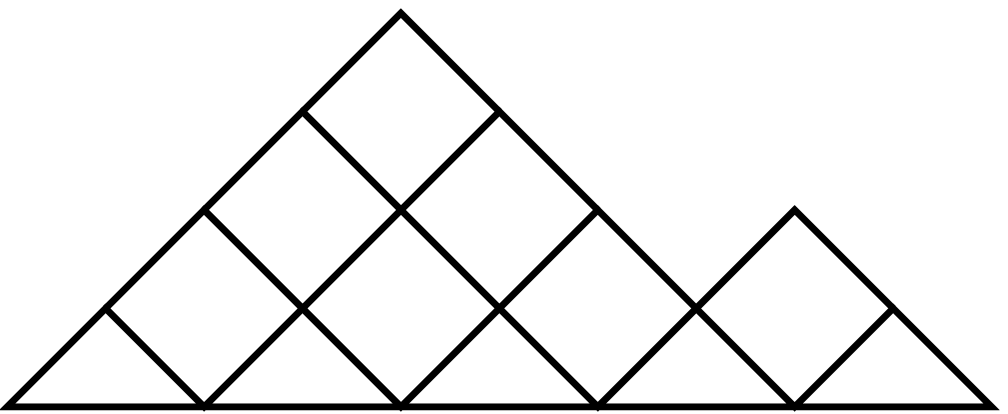} & \tau^7(14,18,12,4) & \tau^{\pm 3} & \tau^{\pm 1}\\
\includegraphics[width=54pt]{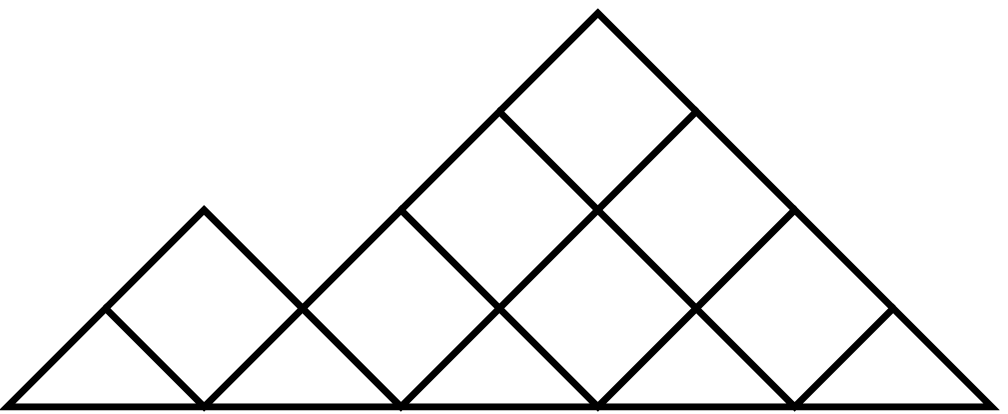} & \tau^7(4,20,20,4) & \tau^{\pm 3} & \tau^{\pm 1}\\
\includegraphics[width=54pt]{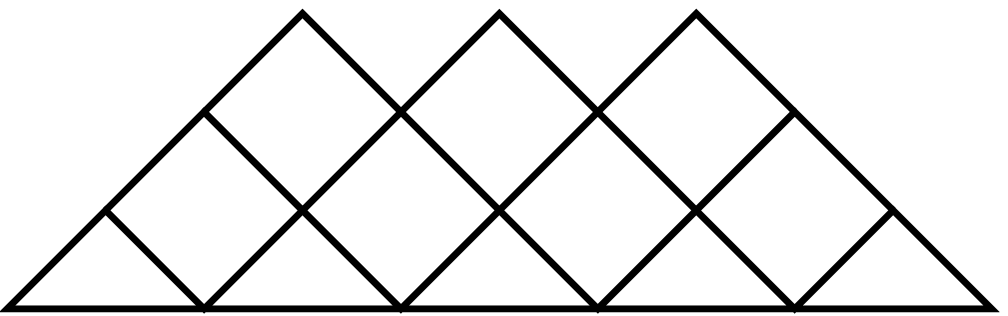} & \tau^7(17,54,48,14) & \tau^{\pm 1} & \tau^{\pm 3} & 1\\
\includegraphics[width=54pt]{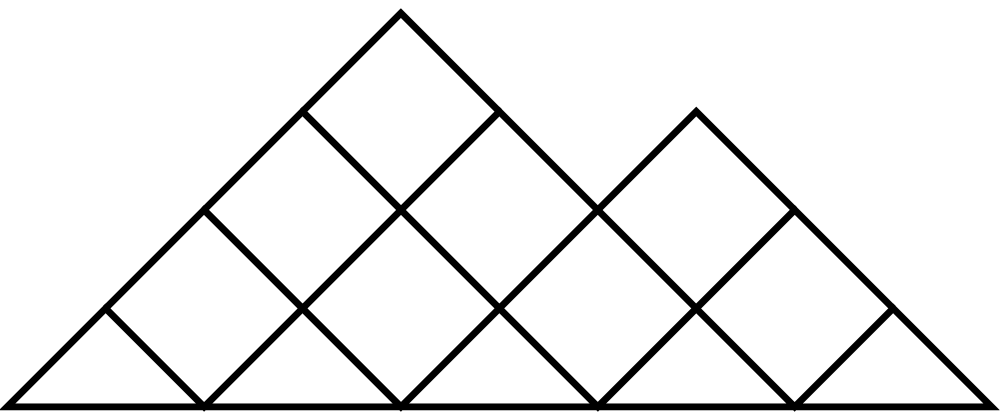} & \tau^8(9,12,6) & \tau^{\pm 2} & \tau^{\pm 2} & \tau^{\pm 1}\\
\includegraphics[width=54pt]{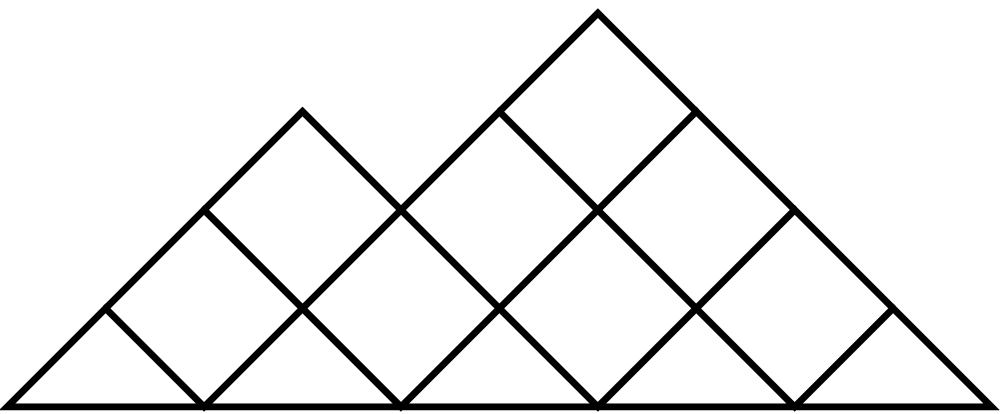} & \tau^8(6,15,6) & \tau^{\pm 2} & \tau^{\pm 2} & \tau^{\pm 1}\\
\includegraphics[width=54pt]{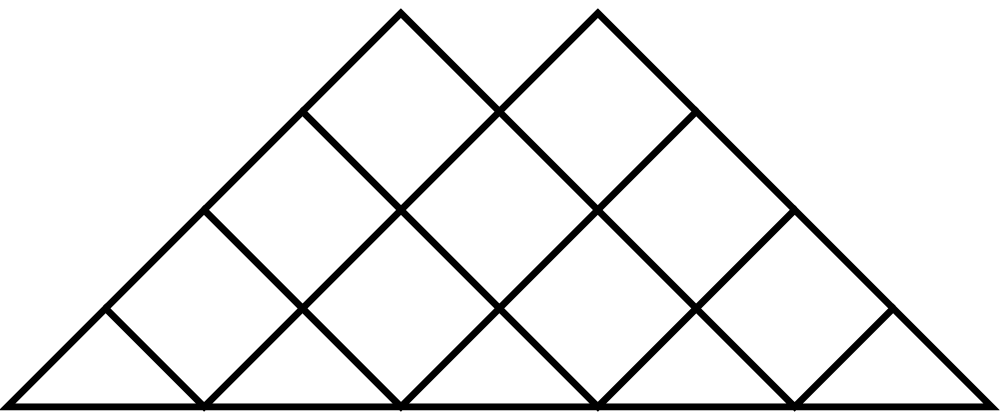} & \tau^9(4,4) & \tau^{\pm 3} & \tau^{\pm 1} & \tau^{\pm 2} & 1\\
\includegraphics[width=54pt]{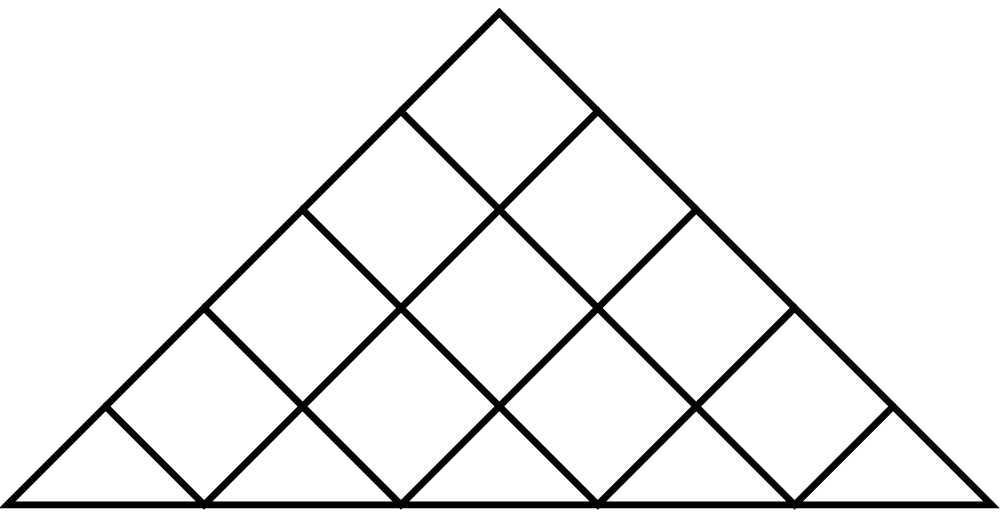} & \tau^{10} & \tau^{\pm 2} & \tau^{\pm 2} & \tau^{\pm 1} & \tau^{\pm 1}\\
\end{array}
\]

\vfill\newpage

\begin{align*}
S_-(10,0) &= \tau^{10} \nonumber\\
S_-(10,1) &= \tau^{9}(5,4)\nonumber\\
S_-(10,2) &= \tau^7(36,86,60,14) \\
S_-(10,3) &= \tau^4(155,811,1490,1306,592,140,14) \nonumber\\
S_-(10,4) &= (120,1400,5754,11584,13071,8900,3805,1044,186,20,1) \nonumber\\[5mm]
S_+(10,0) &= \tau^{10}\nonumber\\
S_+(10,1) &= \tau^9(4,5) \nonumber\\
S_+(10,2) &= \tau^7(17,69,80,30)\\
S_+(10,3) &= \tau^{4}(24,256,914,1496,1230,504,84) \nonumber\\
S_+(10,4) &= (1,40,508,2799,7940,12652,12026,6967,2430,480,42)\nonumber
\end{align*}

\vfill
\newpage
\section{Type B solutions}
\lb{B:solutions}

Using the factorised expressions of Theorem~\ref{th:facsolB}, we have computed polynomial solutions of the qKZ equation for type B from Proposition~\ref{prop6} in the limit $x_i\rightarrow 0$ up to $N=6$. In the following variables, 
\[
\tau'^2 =2-\tau=2+[2]=[2]_{q^{1/2}}^2,\qquad a = -\left.\frac{[\omega+1]}{[\frac{\omega-\delta}{2}][\frac{\omega+\delta}{2}]}\right|_{\omega=-1/2},
\]
see also Remark~\ref{rem:notation}, and up to an overall normalisation, these solutions become polynomials with positive coefficients. We choose the normalisation such that
\[
\psi_{\Omega}^{\rm B} = a^{\lfloor N/2\rfloor}.
\]

\bigskip

\subsection{$N=2$}
\[
\begin{array}{c|l}
\alpha & \psi_\alpha \\\hline
\includegraphics[width=12pt]{2_010.eps} & 1 \\
\includegraphics[width=12pt]{2_210.eps} & a
\end{array}
\]

\subsection{$N=3$}
\[
\begin{array}{c|l}
\alpha & \psi_\alpha \\ \hline
\includegraphics[width=18pt]{3_1010.eps} & 1+\tau'^2+a \\
\includegraphics[width=18pt]{3_1210.eps} &2 \\
\includegraphics[width=18pt]{3_3210.eps} & a
\end{array}
\]

\subsection{$N=4$}
\[
\begin{array}{c|l}
\alpha & \psi_\alpha \\ \hline
\includegraphics[width=24pt]{4_01010.eps} & 5+\tau'^2+a(2+\tau'^2)\\
\includegraphics[width=24pt]{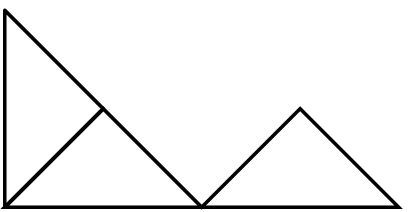} & a(3+2\tau'^2+\tau'^4)+a^2(2+\tau'^2) \\
\includegraphics[width=24pt]{4_01210.eps} & 2+\tau'^2\\
\includegraphics[width=24pt]{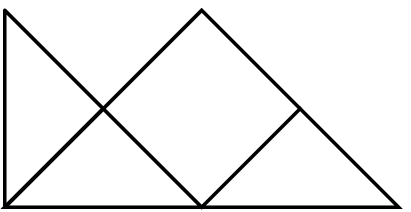} & 2a(2+\tau'^2)+2a^2\\
\includegraphics[width=24pt]{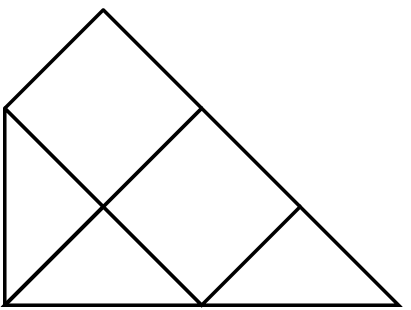} & 3a\\
\includegraphics[width=24pt]{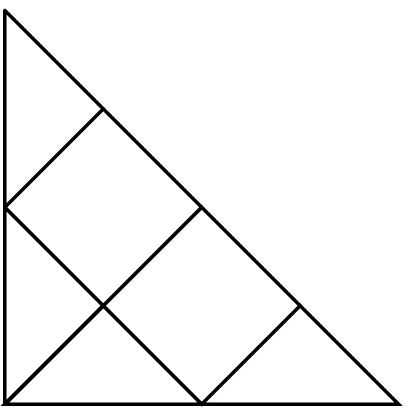} & a^2
\end{array}
\]
\vfill\newpage

\subsection{$N=5$}
\[
\begin{array}{c|l}
\alpha & \psi_\alpha \\ \hline
\includegraphics[width=30pt]{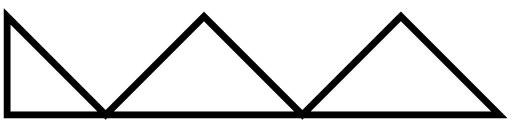} & 9+17\tau'^2+6\tau'^4+\tau'^6 +a(16+17\tau'^2+3\tau'^4)+a^2(7+2\tau'^2)\\
\includegraphics[width=30pt]{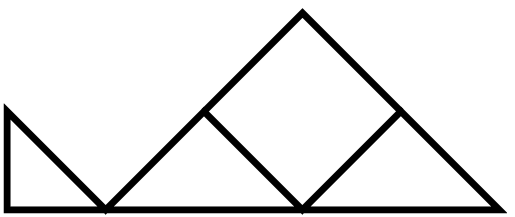} & 2+5\tau'^2+3\tau'^4+\tau'^6 +a(4+6\tau'^2+2\tau'^4) +a^2(2+\tau'^2)\\
\includegraphics[width=30pt]{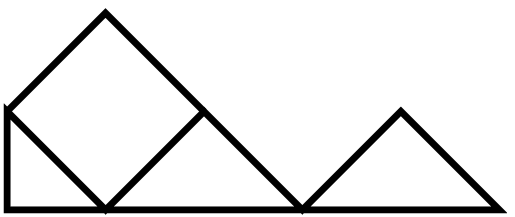} & 12+17\tau'^2+4\tau'^4+a(12+5\tau'^2+\tau'^4)\\
\includegraphics[width=30pt]{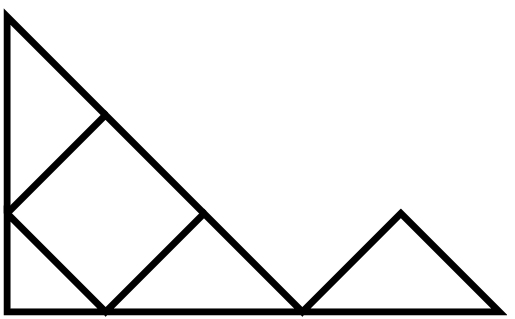} & a(1+\tau'^2)(5+3\tau'^2+\tau'^4)+a^2(5+3\tau'^2+\tau'^4)\\
\includegraphics[width=30pt]{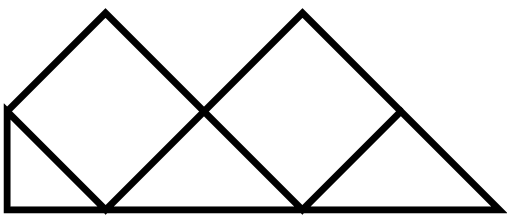} & 24+8\tau'^2+\tau'^4+a(8+3\tau'^2)\\
\includegraphics[width=30pt]{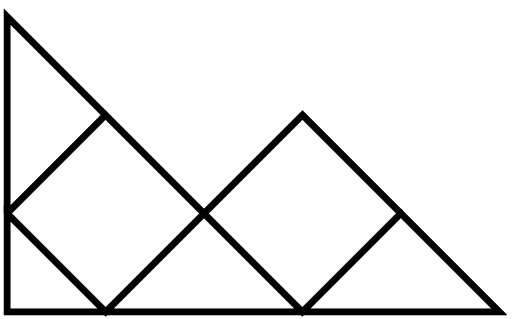} & a(12+7\tau'^2 +2\tau'^4)+2a^2(3+\tau'^2)\\
\includegraphics[width=30pt]{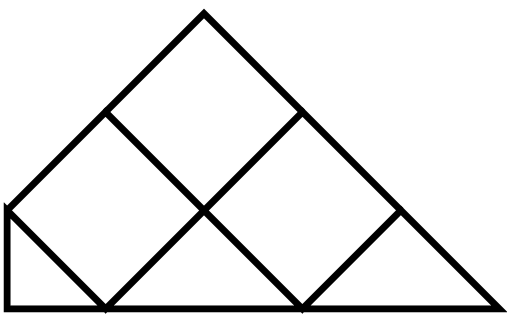} & 8+3\tau'^2\\
\includegraphics[width=30pt]{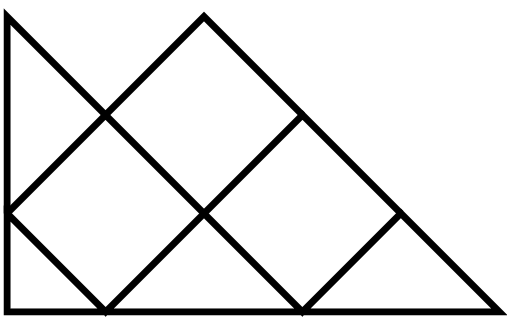} & 3a(3+\tau'^2)+3a^2\\
\includegraphics[width=30pt]{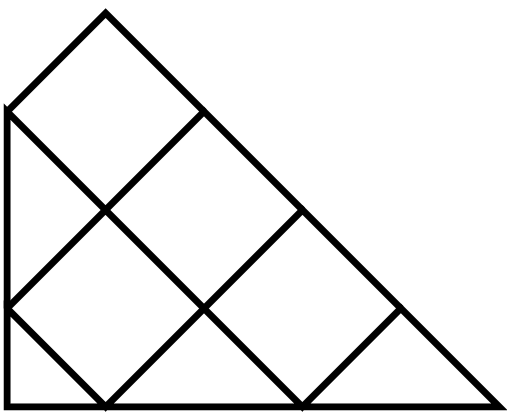} & 4a\\
\includegraphics[width=30pt]{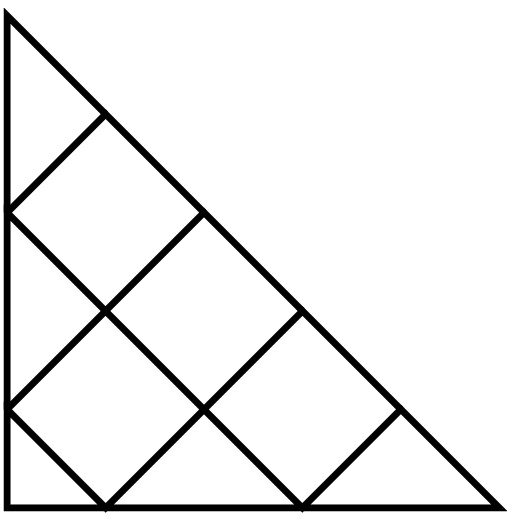} & a^2
\end{array}
\]
\vfill\newpage

\subsection{$N=6$}
\[
\begin{array}{c|l}
\alpha & \psi_\alpha \\ \hline
\includegraphics[width=36pt]{6_0101010.eps} & 149 +107\tau'^2+27\tau'^4+3\tau'^6 + a(126+131\tau'^2+45\tau'^4+9\tau'^6+\tau'^8) \\
& +a^2(32+36\tau'^2+9\tau'^4+\tau'^6) \\
\includegraphics[width=36pt]{6_0101210.eps} & 58+57\tau'^2+14\tau'^4+\tau'^6 +a(32+44\tau'^2+21\tau'^4+6\tau'^6+\tau'^8)\\
& + a^2(1+\tau'^2)(8+4\tau'^2+\tau'^4)\\
\includegraphics[width=36pt]{6_0121010.eps} & 52+50\tau'^2+21\tau'^4+6\tau'^6+\tau'^8 +a(32+36\tau'^2+9\tau'^4+\tau'^6)\\
\includegraphics[width=36pt]{6_0121210.eps} & 2(20+24\tau'^2+7\tau'^4+\tau'^6) +a(1+\tau'^2)(8+4\tau'^2+\tau'^4)\\
\includegraphics[width=36pt]{6_0123210.eps} & (1+\tau'^2)(8+4\tau'^2+\tau'^4)\\
\includegraphics[width=36pt]{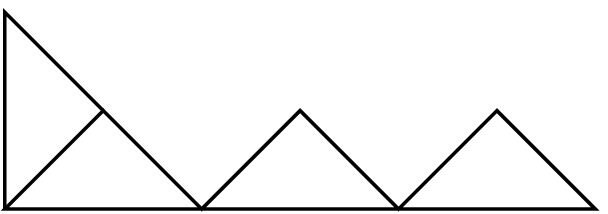} & a(81+101\tau'^2+70\tau'^4+26\tau'^6+7\tau'^8+\tau'^{10})+ a^2(94+127\tau'^2+72\tau'^4+17\tau'^6+2\tau'^8) \\
& +a^3(32+36\tau'^2+9\tau'^4+\tau'^6)\\
\includegraphics[width=36pt]{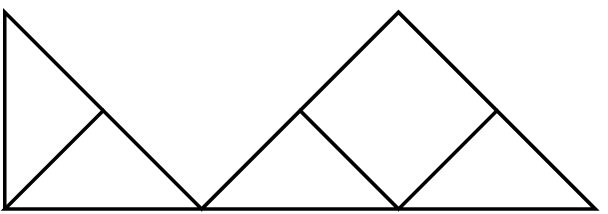} & a(26+37\tau'^2+25\tau'^4+11\tau'^6+4\tau'^8+\tau'^10) +a^2(12+20\tau'^2+14\tau'^4+5\tau'^6+\tau'^8) \\
& +a^3(1+\tau'^2)(8+4\tau'^2+\tau'^4)\\
\includegraphics[width=36pt]{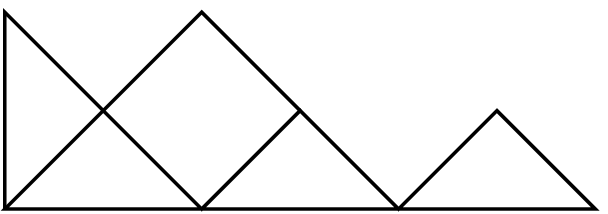} & a(72+104\tau'^2+81\tau'^4+25\tau'^6+4\tau'^8)+ a^2(84+107\tau'^2+55\tau'^4+9\tau'^6) \\
& +a^3(2+\tau'^2)(12+5\tau'^2)\\
\includegraphics[width=36pt]{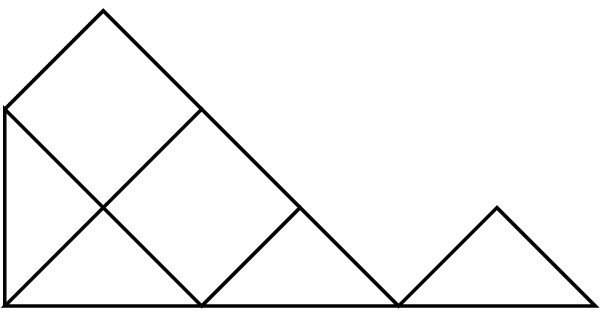} & a(39+58\tau'^2+49\tau'^4+10\tau'^6)+a^2(2+\tau'^2)(17+8\tau'^2+\tau'^4) \\
\includegraphics[width=36pt]{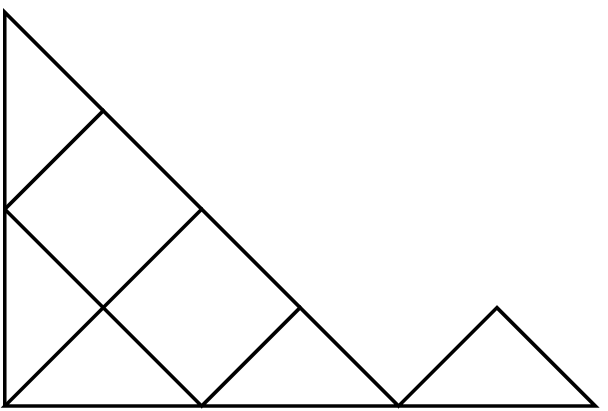} & a^2(11+19\tau'^2+19\tau'^4+7\tau'^6+\tau'^8)+a^3(2+\tau'^2)(5+3\tau'^2+\tau'^4)\\
\includegraphics[width=36pt]{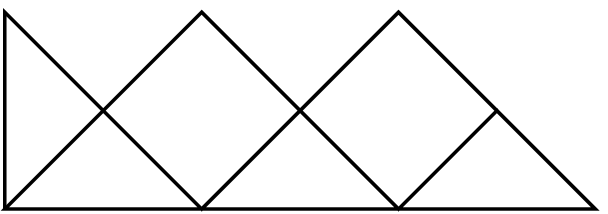} & a(104+128\tau'^2+44\tau'^4+9\tau'^6+\tau'^8) + a^2(104+92\tau'^2+22\tau'^4+2\tau'^6) \\
&+a^3(32+11\tau'^2+\tau'^4)\\
\includegraphics[width=36pt]{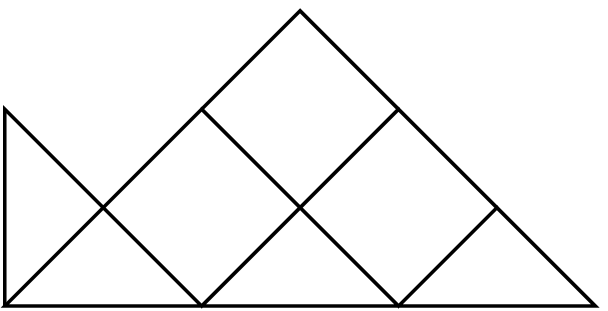} & 3a(1+\tau'^2)(8+4\tau'^2+\tau'^4)+a^2(3+2\tau'^2)(8+3\tau'^2) +a^3(14+3\tau)\\
\end{array}
\]

\vfill\newpage
\[
\begin{array}{c|l}
\alpha & \psi_\alpha \\ \hline
\includegraphics[width=36pt]{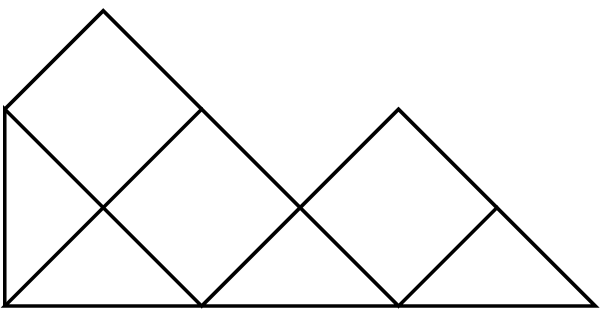} & a(84+78\tau'^2+19\tau'^4+\tau'^6)+a^2(50+20\tau'^2+3\tau'^4)\\
\includegraphics[width=36pt]{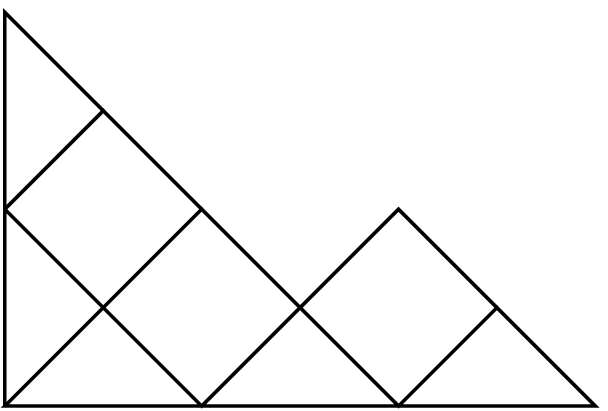} & 2a^2(2+\tau'^2)(7+4\tau'^2+\tau'^4)+2a^3(24+6\tau'^2+2\tau'^4)\\
\includegraphics[width=36pt]{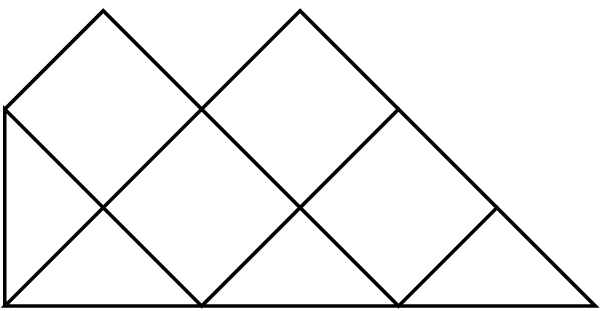} & a(75+26\tau'^2+3\tau'^4)+2a^2(10+3\tau'^2)\\
\includegraphics[width=36pt]{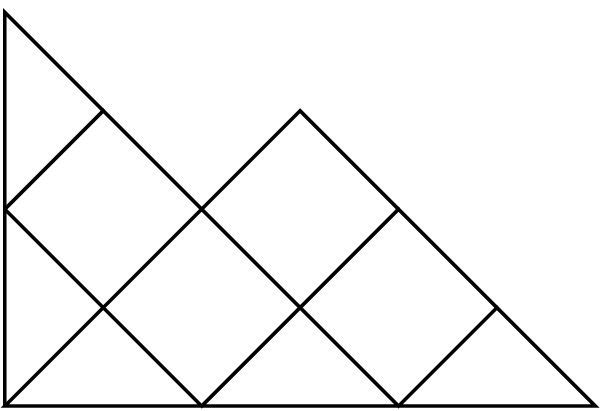} & a^2(31+15\tau'^2 +3\tau'^4)+3a^3(4+\tau'^2)\\
\includegraphics[width=36pt]{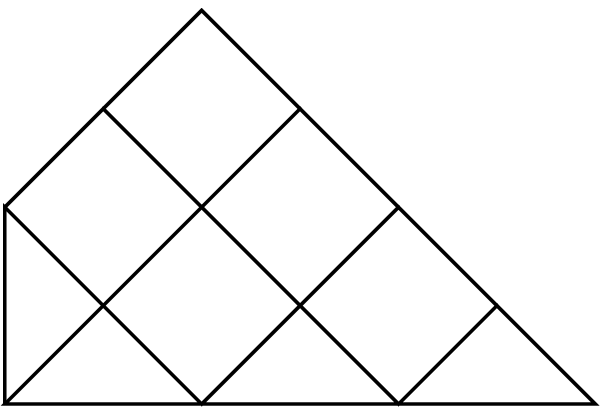} & 2a(10+3\tau'^2)\\
\includegraphics[width=36pt]{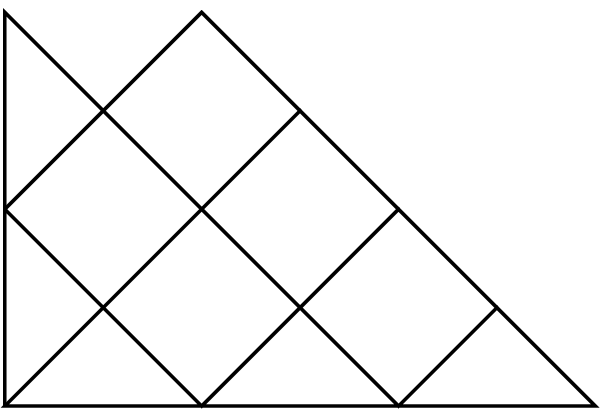} & 4a^2(4+\tau'^2)+4a^3\\
\includegraphics[width=36pt]{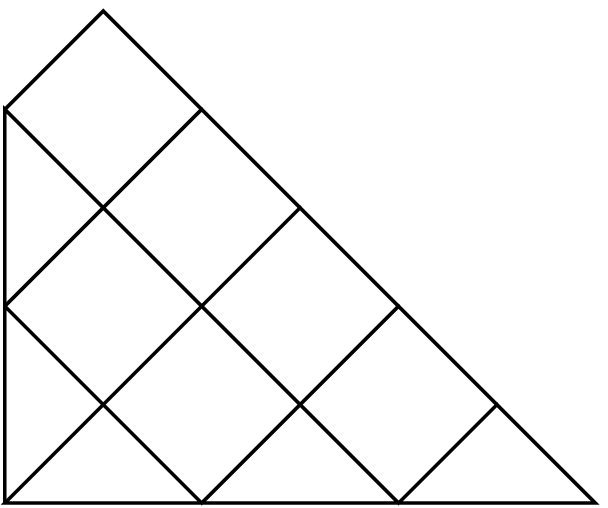} & 5a^2\\
\includegraphics[width=36pt]{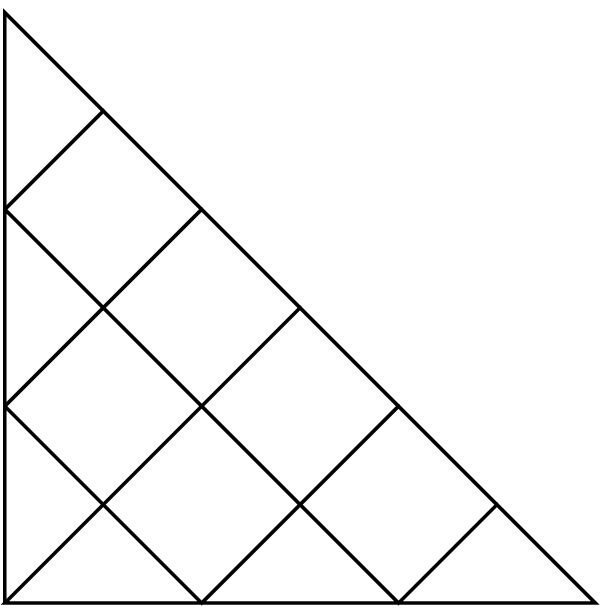} & a^3
\end{array}
\]
%



\bigskip
\begin{thebibliography}{99}

\bibitem{APR} F.C.~Alcaraz, P.~Pyatov, and  V.~Rittenberg,
\textit{Raise and Peel Models of fluctuating interfaces and combinatorics of Pascal's hexagon}, [\href{http://arxiv.org/pdf/0709.4575}{\texttt{arXiv:0709.4575}}].

\bibitem{BatchGN01} M.T. Batchelor, J. de Gier and B. Nienhuis, {\it The
  quantum symmetric XXZ chain at $\Delta=-1/2$, alternating sign
  matrices and plane partitions}, 2001 J. Phys. A {\bf 34} L265--L270,
  \href{http://arxiv.org/pdf/cond-mat/0101385}{\texttt{arXiv:cond-mat/0101385}}.

\bibitem{BA05} M. Beccaria and G. F. De Angelis, {\it Exact Ground State and Finite Size Scaling in a Supersymmetric Lattice Model}, Phys. Rev. Lett. \textbf{94} (2005) 100401, \href{http://arxiv.org/pdf/cond-mat/0407752}{\texttt{arXiv:cond-mat/0407752}};

\bibitem{Bress} D. Bressoud, {\it Proofs and confirmations. The story of the alternating sign matrix conjecture} (Cambridge University Press, 1999).

\bibitem{Ch} I Cherednik, {\it Double affine Hecke algebras}, London Mathematical Society Lecture Notes Series {\bf 319}, (Cambridge University Press, Cambridge, 2005).

\bibitem{Jan} J.~de Gier,
{\em Loops, matchings and alternating-sign matrices},
Discr.~Math. {\bf 298} (2005) 365-388,
\href{http://arxiv.org/pdf/math.CO/0211285}{\texttt{arXiv:math.CO/0211285}}.

\bibitem{Magic} J. de Gier, A. Nichols, P.Pyatov and  V. Rittenberg,
\textit{Magic in the spectra of the XXZ quantum chain with boundaries at
$\Delta=0$ and $Delta=-1/2$}, Nucl.~Phys. \textbf{B729} (2005) 387-418,
\href{http://arxiv.org/pdf/hep-th/0505062}{\texttt{arXiv:hep-th/0505062}}.

\bibitem{GN} J. de Gier and B. Nienhuis, {\it Brauer loops and the commuting variety}, J. Stat. Mech. (2005) P01006, [\href{http://arxiv.org/math/0410392}{\texttt{math-ph/0410392}}]

\bibitem{GNPR} J. de Gier, B. Nienhuis, P.A. Pearce and V. Rittenberg,
  {\it The raise and peel model of a fluctuating interface},
  J.\,Stat.\,Phys. {\bf 114} (2004) 1--35,
\href{http://arxiv.org/pdf/cond-mat/0301430}{\texttt{arXiv:cond-mat/0301430}}.

\bibitem{GNPR03} J. de Gier, B. Nienhuis, P. A. Pearce and  V. Rittenberg, {\it Stochastic processes and conformal invariance}, Phys. Rev. E \textbf{67} (2003) 016101-016104, [\href{http://arxiv.org/cond-mat/0205467}{\texttt{cond-mat/0205467}}].

\bibitem{GR04} J. de Gier and V. Rittenberg, \textit{Refined Razumov-Stroganov conjectures for open boundaries}, J. Stat. Mech. (2004), P09009, [\href{http://arxiv.org/math-ph/0408042}{\texttt{math-ph/0408042}}]

\bibitem{DF05} P. Di Francesco, {\it Boundary qKZ equation and generalized
  Razumov-Stroganov sum rules for open IRF models},  J. Stat. Mech.  (2005) P09004,
  \href{http://arxiv.org/pdf/math-ph/0509011}{\texttt{arXiv:math-ph/0509011}}.

\bibitem{DF06} P. Di Francesco, \textit{Totally Symmetric Self Complementary Plane Partitions and the
quantum Knizhnik-Zamolodchikov equation: a conjecture}, J. Stat. Mech. (2006) P09008, \href{http://arxiv.org/pdf/cond-mat/0607499}{\texttt{arXiv:cond-mat/0607499}}.

\bibitem{DF07} P. Di Francesco, {\it Open boundary Quantum
  Knizhnik-Zamolodchikov equation and the weighted enumeration of
  symmetric plane partitions},  J. Stat. Mech.  (2007) P01024,
  \href{http://arxiv.org/pdf/math-ph/0611012}{\texttt{arXiv:math-ph/0611012}}.

\bibitem{DFZJ04} P. Di Francesco and P. Zinn-Justin, {\it Around the
  Razumov-Stroganov conjecture: proof of a multi-parameter sum rule},
  2005 Electron. J. Combin. {\bf 12}, R6,
\href{http://arxiv.org/pdf/math-ph/0410061}{\texttt{arXiv:math-ph/0410061}}

\bibitem{DFZJ04b} P. Di Francesco, P. Zinn-Justin, {\it Inhomogeneous model of crossing loops and multidegrees of some algebraic varieties}, Commun. Math. Phys. {\bf 262} (2006), 459-487, \href{http://arxiv.org/pdf/math-ph/0412031}{\texttt{arXiv:math-ph/0412031}}.

\bibitem{DFZJ05} P. Di Francesco and P. Zinn-Justin,
  {\it Quantum Knizhnik-Zamolodchikov equation, generalized
  Razumov-Stroganov sum rules and extended Joseph polynomials}, 2005
  J. Phys. A {\bf 38}, L815--L822,
\href{http://arxiv.org/pdf/math-ph/0508059}{\texttt{arXiv:math-ph/0508059}}.

\bibitem{DFZJ07} P. Di Francesco and P. Zinn-Justin, {\it Quantum Knizhnik-Zamolodchikov equation: reflecting boundary conditions and combinatorics},
\href{http://arxiv.org/pdf/math-ph/0709.3410}{\texttt{arXiv:math-ph/0709.3410}}

\bibitem{DFZJZ03} P.~Di~Francesco, P.~Zinn-Justin and J.-B.~Zuber
\textit{A Bijection between classes of Fully Packed Loops and Plane Partitions},
Electron.~J.~Combin. {\bf 11} (2004) 64,
\href{http://arxiv.org/pdf/math/0311220}{\texttt{arXiv:math/0311220}}.

\bibitem{DFZJZ06} P.~Di~Francesco, P.~Zinn-Justin and J.-B.~Zuber
\textit{Sum rules for the ground states of the O(1) loop model on a cylinder and the XXZ spin chain},
Electron.~J.~Combin. {\bf 11} (2004) 64,
\href{http://arxiv.org/pdf/math-ph/0603009}{\texttt{arXiv:math-ph/0603009}}.

\bibitem{DKLLST} G.\,Duchamp, D.\,Krob, A.\,Lascoux, B.\,Leclerc, T.\,Scharf, and J.-Y.\,Thibon,
\textit{Euler-Poincar\'{e} characteristic and polynomial representations of Iwahori-Hecke algebras},
Publ.\,RIMS, {\bf 31} (1995) 179-201.

\bibitem{EFK} P.I. Etingof, I.B. Frenkel and A.A. Kirillov Jr., {\it Lectures on representation theory and Knizhnik-Zamolodchikov equations}, Mathematical Surveys and Monographs {\bf 58} (AMS, Providence, 1998).

\bibitem{FNS} P. Fendley, B. Nienhuis and K. Schoutens, {\it Lattice fermion models with supersymmetry}, J. Phys. A \textbf{36} (2003) 12399-12424, \href{http://arxiv.org/pdf/cond-mat/0307338}{\texttt{arXiv:cond-mat/0307338}}.

\bibitem{FR} I.B.~Frenkel and N.~Reshetikhin, \textit{Quantum affine
  algebras and holonomic difference equations},
  Commun. Math. Phys. \textbf{146} (1992), 1--60.

\bibitem{IO} A.P.~Isaev and O.V.~Ogievetsky,
{\it Baxterized Solutions of Reflection Equation and Integrable Chain Models},
Nucl.~Phys. {\bf B760} (2007) 167-183,
\href{http://arxiv.org/pdf/math-ph/0510078}{\texttt{arXiv:math-ph/0510078}}.

\bibitem{JM} M. Jimbo and T. Miwa, {\it Algebraic analysis of solvable lattice models} (AMS, Providence, 1995).

\bibitem{Jo} V.F.R.~Jones, {\em On a certain value of the Kauffman polynomial},
Comm.\,Math.\,Phys. {\bf 125} (1989) 459--467.

\bibitem{KP} M. Kasatani and V. Pasquier, \textit{On polynomials interpolating between the stationary
state of a O(n) model and a Q.H.E. ground state}, \href{http://arxiv.org/pdf/cond-mat/0608160}{\texttt{arXiv:cond-mat/0608160}}.

\bibitem{KT} M. Kasatani and Y. Takeyama, \textit{The quantum Knizhnik-Zamolodchikov equation and non-symmetric Macdonald polynomials}, \href{http://arxiv.org/pdf/math/0608773}{\texttt{arXiv:math.CO/0608773}}.

\bibitem{KiriL} A. Kirillov Jr. and A. Lascoux,
\textit{Factorization of Kazhdan-Lusztig elements for Grassmanians},
Adv. Stud. \textbf{28} (2000), 143--154;
\href{http://arxiv.org/pdf/math/9902072}{\texttt{arXiv:math.CO/9902072}}.

\bibitem{KZJ} A. Knutson, P. Zinn-Justin, {\it A scheme related to the Brauer loop model}, Adv. in Math. {\bf 214} (2007) 40-77, \href{http://arxiv.org/pdf/math/0503224}{\texttt{arXiv:math.CO/0503224}}.

\bibitem{MS}  P.P.~Martin and H.~Saleur, {\it The blob algebra and the periodic Temperley-Lieb algebra},
Lett.~Math.~Phys. {\bf 30} (1994)  189-206,
\href{http://arxiv.org/pdf/hep-th/9302094}{\texttt{arXiv:hep-th/9302094}}.

\bibitem{MNGB}S.~Mitra, B.~Nienhuis, J.~de Gier and M.T.~Batchelor,
{\it Exact expressions for correlations in the ground state of the dense O(1) loop model},
J.~Stat.~Mech. (2004) P09010
\href{http://arxiv.org/pdf/cond-mat/0401245}{\texttt{arXiv:cond-mat/0401245}}.

\bibitem{P05} V. Pasquier, {\it Quantum incrompressibility and Razumov
  Stroganov type conjectures}, Ann. Henri Poincare' 7, 397-421 (2006)
\href{http://arxiv.org/pdf/cond-mat/0506075}{\texttt{arXiv:cond-mat/0506075}}.

\bibitem{PRGN02} P. A. Pearce, V. Rittenberg, J. de Gier and B. Nienhuis, {\it Temperley-Lieb stochastic processes}, J. Phys. A \textbf{35} (2002) L661-L668, \href{http://arxiv.org/pdf/math-ph/0209017}{\texttt{arXiv:math-ph/0209017}}.

\bibitem{Pavel} P.~Pyatov,
{\em Raise and Peel Models of fluctuating interfaces and combinatorics of Pascal's hexagon},
J.\,Stat.\,Mech. (2004) P09003,
\href{http://arxiv.org/pdf/math-ph/0406025}{\texttt{arXiv:math-ph/0406025}}.

\bibitem{RS06} A.V. Razumov and Yu.G. Stroganov, \textit{Bethe roots and refined enumeration of alternating-sign matrices.}, J. Stat. Mech. (2006) P07004, \href{http://arxiv.org/pdf/math-ph/0605004}{\texttt{arXiv:math-ph/0605004}};

\bibitem{RSPZJ07} A.V. Razumov, Yu.G. Stroganov and P. Zinn-Justin, \textit{Polynomial solutions of qKZ equation and ground state of XXZ spin chain at $\Delta=-1/2$.}, \href{http://arxiv.org/pdf/math-ph/0704.3542}{\texttt{arXiv:math-ph/07043542}};

\bibitem{RazuS00} A.V. Razumov and Yu.G. Stroganov, {\it Spin chains and combinatorics}, 2001 J. Phys. A
  {\bf 34} 3185--3190,
\href{http://arxiv.org/pdf/cond-mat/0012141}{\texttt{arXiv:cond-mat/0012141}}.

\bibitem{RazuS01} A.V. Razumov and Yu.G. Stroganov, {\it Combinatorial
  nature of gound state vector of O(1) loop model}, 2004
  Theor. Math. Phys. {\bf 138} 333--337,
\href{http://arxiv.org/pdf/math/0104216}{\texttt{arXiv:math.CO/0104216}}.

\bibitem{ShigechiU} K. Shigechi and M. Uchiyama, \textit{A$_k$ generalization of the O(1) loop model on a cylinder: Affine HeckeS algebra, $q$-KZ equation and the sum rule}, \href{http://arxiv.org/pdf/math-ph/0612001}{\texttt{arXiv:math-ph/0612001}}

\bibitem{Sk} E. Sklyanin, {\it Boundary conditions for integrable quantum systems}, J. Phys. A {\bf 21} (1988) 2375--2489.

\bibitem{S} F.A. Smirnov, {\it A general formula for soliton form factors in the quantum sine-Gordon model}, J. Phys. A {\bf 19} (1986) L575--L578

\bibitem{S2} F.A. Smirnov, {\it Form factors in completely integrable modles of Quantum Field Theory},
World Scientific, Singapore, 1992.

\bibitem{S00} Yu.G. Stroganov, {\it The importance of being odd},
  2001 J. Phys. A {\bf 34} L179--L185,
\href{http://arxiv.org/pdf/cond-mat/0012035}{\texttt{arXiv:cond-mat/0012035}};

\bibitem{TV} V. Tarasov and A. Varchenko, {\it Jackson integrable representations for solutions of the
quantized Knizhnik-Zamolodchikov equation},
St.Petersburg. Math. J. \textbf{6} (1994) no.2, 275--314.

\bibitem{V} A. Varchenko, {\it Quantized Knizhnik-Zamolodchikov equations, quantum Yang-Baxter equation, and difference equations for $q$-hypergeometric functions}, Commun. Math. Phys. \textbf{162} (1994) 499--528.

\bibitem{VW} G. Veneziano and J. Wosiek, {\it A supersymmetric matrix model: III. Hidden SUSY in statistical systems}, JHEP 0611 (2006) 030, \href{http://arxiv.org/pdf/hep-th/0609210}{\texttt{arXiv:hep-th/0609210}}.

\bibitem{YangF} X. Yang and P. Fendley, \textit{Non-local space-time supersymmetry on the lattice}, J. Phys. A \textbf{37} (2004) 8937, \href{http://arxiv.org/pdf/cond-mat/0404682}{\texttt{arXiv:cond-mat/0404682}}.

\bibitem{ZJ06} P. Zinn-Justin, {\it Combinatorial point for fused loop models},
  Commun. Math. Phys. {\bf 272} (2007), 661-682,
\href{http://arxiv.org/pdf/math-ph/0603018}{\texttt{arXiv:math-ph/0603018}}.

\bibitem{ZJ07} P. Zinn-Justin, {\it Loop model with mixed boundary
  conditions, $q$KZ equation and Alternating Sign Matrices},
  J. Stat. Mech. (2007) P01007,
\href{http://arxiv.org/pdf/math-ph/0610067}{\texttt{arXiv:math-ph/0610067}}.

\end{thebibliography}
\end{document}